\newif\iffullversion
\newif\ifdraft
\newif\ifanonymous
\newif\ifshowproofs
\fullversiontrue
\showproofstrue

\iffullversion
\documentclass{article}
\else
\documentclass[\ifanonymous anonymous\fi]{lipics-hacked-short}
\fi

\usepackage[utf8]{inputenc}
\usepackage[T1]{fontenc}
\usepackage{paralist}
\usepackage{amsmath,amssymb,amsthm,centernot}
\usepackage{stmaryrd}
\usepackage{array}
\usepackage{declmath}
\usepackage{varwidth}
\usepackage{circuits}
\usepackage{mathpartir}
\usepackage{mathtools}
\usepackage{microtype}
\usepackage{xr}
\iffullversion
\usepackage{makeidx}\makeindex
\usepackage[a4paper,margin=1in]{geometry}
\fi
\usepackage{version}
\usepackage[toc]{multitoc}
\iffullversion\usepackage[backend=bibtex,maxbibnames=100,style=alphabetic]{biblatex}\fi
\iffullversion\usepackage[pdfborderstyle={/S/U/W 0.3}]{hyperref}\fi
\usepackage[delaytext]{latexhacks}

\newcommand\singledouble[2]{#1}

\input{macros}

\newcommand\locAxioms{42}
\newcommand\locAxiomsClassical{249}
\newcommand\locAxiomsComplement{24}
\newcommand\locAxiomsComplementQuantum{816}
\newcommand\locAxiomsQuantum{744}

\newcommand\locClassicalExtra{114}
\newcommand\locLaws{625}

\newcommand\locLawsComplement{388}

\newcommand\locMisc{187}
\newcommand\locPureStates{494}
\newcommand\locQHoare{56}
\newcommand\locQuantum{115}
\newcommand\locQuantumExtra{270}
\newcommand\locQuantumExtraTwo{34}
\newcommand\locTeleport{225}
\newcommand\locTensorProductMatrices{200}

\iffullversion
\DeclareFieldFormat{postnote}{#1}
\DeclareFieldFormat{multipostnote}{#1}
\fi

\let\symbolindexmarkhighlight\symbolindexmarkhighlightGRAYBOX

\usetikzlibrary{arrows,matrix}
\tikzset{>=stealth}

\usetikzlibrary{decorations.pathreplacing, decorations.pathmorphing, calc, positioning}

\iffullversion
\newcommand\fullshort[2]{#1}
\includeversion{fullversion}
\excludeversion{shortversion}
\else
\newcommand\fullshort[2]{#2}
\excludeversion{fullversion}
\includeversion{shortversion}
\fi

\newcommand\fullonly[1]{\fullshort{#1}{}}
\newcommand\shortonly[1]{\fullshort{}{#1}}

\shortonly{
  \hideLIPIcs
}

\ifx\definition\undefined\newtheorem{definition}{Definition}\fi
\ifx\lemma\undefined\newtheorem{lemma}{Lemma}\fi
\ifx\conjecture\undefined\newtheorem{conjecture}{Conjecture}\fi
\ifx\theorem\undefined\newtheorem{theorem}{Theorem}\fi

\newcounter{law}
\renewcommand\thelaw{(\roman{law})}
\newcommand\nextlaw{\refstepcounter{law}\textsuperscript\thelaw\,}

\fullonly{\bibliography{references}}

\ifanonymous
  \newcommand\anonymous[2]{#1}
\else
  \newcommand\anonymous[2]{#2}
\fi

\ifshowproofs\else
  \ifdraft\else
    \errmessage{Draft=false, showproofs=false}
  \fi
\fi

\AtBeginDocument{
  \ifx\PreviewMacro\undefined\else
    \let\hyper@@anchor\@gobble
    \gdef\hyper@link#1#2#3{#3}%
    \let\hyper@anchorstart\@gobble
    \let\hyper@anchorend\@empty
    \let\hyper@linkstart\@gobbletwo
    \let\hyper@linkend\@empty
    \def\hyper@linkurl#1#2{#1}%
    \def\hyper@linkfile#1#2#3{#1}%
    \def\hyper@link@[#1]#2#3{}%
    \let\PDF@SetupDoc\@empty
    \let\PDF@FinishDoc\@empty
    \let\@fifthoffive\@secondoftwo
    \let\@secondoffive\@secondoftwo

    \Hy@WarningNoLine{ draft mode on}%
    \PreviewMacro*\section
    \PreviewMacro*\subsection
    \PreviewMacro*\subsubsection
    \PreviewEnvironment*{figure}
    \PreviewEnvironment*{figure*}
    \PreviewEnvironment*{definition}
    \PreviewEnvironment*{lemma}
    \PreviewEnvironment*{theorem}
    \PreviewEnvironment*{corollary}
    \PreviewEnvironment*{algorithm}
    \PreviewMacro*\footnote
    \PreviewMacro*\footnotetext
    \PreviewMacro*\framebox
    \PreviewMacro*\parbox
    \PreviewMacro*\fullshort
    \PreviewMacro*\fullonly
    \PreviewMacro*\shortonly
    \PreviewMacro*\delaytext
    \PreviewMacro*\usedelayedtext
    \PreviewMacro*\index
    \PreviewMacro*\textbf
    \PreviewMacro*\textsf
    \PreviewMacro*\textit
    \PreviewMacro*\textsc
    \PreviewMacro*\pmb
    \PreviewMacro*\symbolindex
  \fi
}

\AtBeginDocument{
  \ifx\PreviewMacro\undefined
    \ifshowproofs\else
      \excludeversion{proof}
    \fi
  \fi
}

\shortonly{
  \externaldocument[FV:]{icalp-submission/references-fv}
}

\begin{document}

\title{Quantum references}
\fullshort{
  \ifanonymous
    \author{}
  \else
    \author{Dominique Unruh\\\small RWTH Aachen, University of Tartu}
    \date{}
  \fi
  \maketitle
}{
  \ccsdesc{Theory of computation~Quantum computation theory}
  \keywords{Quantum programs, references, semantics}
  \anonymous{
    \author{X}{X}{X}{X}{X}
    \authorrunning{Anonymous authors}
  }{
    \author{Dominique Unruh}
  }
  \maketitle
}

\shortonly{
  \renewcommand\paragraph[1]{\medskip\noindent\textbf{#1}}
}

\begin{abstract}
  We present a theory of ``quantum references'', similar to lenses in classical functional programming, that allow to point to a subsystem of a larger quantum system, and to mutate/measure that part.
  Mutable classical variables, quantum registers, and wires in quantum circuits are examples of this, but also references to parts of larger quantum datastructures.
  Quantum references in our setting can also refer to subparts of other references, or combinations of parts from different references, or quantum references seen in a different basis, etc.
  Our modeling is intended to be well suited for formalization in theorem provers and as a foundation for modeling variables in quantum programs.
  We study quantum references in greater detail and cover the infinite-dimensional case as well,
  but also provide a more general treatment not specific to the quantum case.
  We implemented a large part of our results (including a small quantum Hoare logic
  and an analysis of quantum teleportation) in the Isabelle/HOL theorem prover.
\end{abstract}


\makeatletter
\ifdraft
{\catcode`\%=12\immediate\write\@auxout{%
    DRAFT MODE}}
\fi

\ifdraft
\fullonly{
  \begin{center}
    \bfseries\Huge\fboxsep=10pt
    \framebox{%
      \begin{varwidth}\textwidth
        \centering
        THIS IS A DRAFT
        \\
        \Large
        Do not distribute
        \\
        \smallskip
        \footnotesize\mdseries
        (Git revision: {\ttfamily\endlinechar=-1\input{.gitrevision}})
        \\
        \bigskip
        Information about mistakes and typos are appreciated (\href{mailto:unruh@ut.ee}{unruh@ut.ee}).
        \\
        \bigskip
        \textbf{A published version of this paper is \href{https://arxiv.org/abs/2105.10914}{arXiv:2105.10914} [cs.LO].}
      \end{varwidth}%
    }%
  \end{center}
}
\fi

\begin{fullversion}
  \renewcommand*{\multicolumntoc}{2}
  \setlength{\columnseprule}{.4pt}
  \tableofcontents
\end{fullversion}

\section{Introduction}

In this work, we formalize the concept of a quantum reference.
In a nutshell, a quantum reference (in the sense of this paper) is a mathematical object that points to a subsystem of a larger quantum system.
It allows us to access (e.g., modify, measure, initialize) that subsystem.
For example, a variable name in an imperative quantum program can be abstractly seen as a reference:
It tells us which part of the overall memory we are talking about.
Similarly, quantum registers (or wires) in quantum circuits can be seen as references.
More interesting examples (because they are not fixed throughout the execution of the program) would be ``pointers'' into a quantum heap in a functional quantum programming language; those references could be passed around as first-class values.

To understand better what we aim at, consider a related concept from classical functional programming, the \emph{lens}\index{lens} (also known as a \emph{functional reference}\index{functional reference}\index{reference!functional}).
A lens consists of a pair $(g,s)$ of a getter $g:M\to A$ and setter $s:M\times A\to M$.
The lens can be thought to point to some location of type $A$ inside a program state $M$ (e.g., the program's heap as a whole).
Given a concrete value $m$ of the program state, $g(m)$ returns the content of that location, and $s(m,a)$ tells us what the new state of the overall program is when we update the content of that location to the value $a$.
Any read or write access to the memory location can be described in terms of these two functions.

That is, a lens is a mathematical value that points to a subsystem of a larger classical system.
And this value could either be fixed throughout the program (i.e., basically be a semantic description of a statically declared mutable variable),
or it could be created and passed around at runtime (since in a functional language, the functions $(g,s)$ are first-class values).
Of course, the larger system $M$ does not have to be the whole state of a program with a heap.
Lenses can also be used to refer to parts of data structures in purely functional programs.
(E.g., a \emph{name} field in a \emph{person data} data structure can be described by a lens.)
Lenses can even be nested.
For example, if we have a lens describing where the \emph{person data} is located in the overall memory,
and another lens describing where the \emph{name} is in a \emph{person data} data structure,
then we easily can chain them to get a lens describing where the \emph{name} is in the overall memory.

We aim to achieve the same (and a bit more) in the quantum setting.
We want a functional value (i.e., some classical value that can, in principle, be passed around in a program) and that refers to a part of a quantum system, allowing access to that part (similar, but not identical, to the getter and setter for read and write access).

\paragraph{Why quantum references?}
In a nutshell, a quantum reference points to a subsystem (described by a Hilbert space $\calH_A$, the \emph{content Hilbert space}) inside a bigger system (described by a Hilbert space $\calH_M$, then \emph{memory Hilbert space}).
Why are we interested in defining and studying \emph{quantum} references?
There are a number of benefits from a systematic treatment of references:
\begin{itemize}
\item When describing the semantics of quantum programming languages (even simple imperative while languages), we need to refer to individual quantum variables.
  (E.g., we may have statements such as ``apply Hadamard to variable $X$''.)
  It is certainly possible to do this without the abstract concept of quantum references.
  After all, we can simply define the state of the program as the tensor product of the Hilbert spaces of the individual variables, and then each variable refers to one of the factors.
  But doing so forces us to consider specific ``memory layouts'', losing abstraction.
  And it makes the definition of the semantics more complex because all variable accesses need to work with the explicit memory layout.
  (See, e.g., the ``identity-tensor'' approach we describe in a bullet point below.)
  In contrast, using references, we can modularize the definitions better:
  We can declare the program variables to be arbitrary (disjoint) quantum references, and then define the semantics independently of how the these references are arranged into a program state.
  (We have a small case study in \autoref{sec:hoare} for this.)
\item Similar to lenses in the classical case, quantum references can be used in functional programming languages in order to refer to locations in a heap, or locations within other data structures.
  For example, a functional language with classical higher order values, and access to a quantum heap could pass typed references to parts of the heap around, and have one builtin polymorphic function \texttt{alloc} that allocates a new reference with a given content type disjoint from all previously allocated ones.\footnote{At first glance, one might think that a function \texttt{alloc} that returns a qubit pointer would be sufficient.
    But this means that only datatypes whose dimension is a power-of-two could be stored.
    Of course, we can embed, say, a qutrit in two qubits, but then the programmer will have to deal with additional implementation details such as ensuring that the qubits never contain an invalid value (outside the embedded three-dimensional subspace).
    This would contradict the desire of having type-safety in high-level functional languages. }
  We stress that we have not designed such a language in this work, but our quantum references lay the semantic foundations for such applications.
\item Easier formulation of invariants when reasoning about quantum programs:
  For example, in quantum Hoare logics, we often need to express invariants that talk about the state of individual program variables.
  E.g., we might wish to say that $X$ is in state $\ket0$ or similar.
  Quantum references can be used in the descriptions of invariants to refer to individual subsystems and describe the conditions on them.
  (We develop this idea in more detail in \autoref{sec:hoare}.)
\item Formalizing reasoning about quantum registers:
  Often, when a quantum algorithm or quantum circuit or quantum cryptographic system is analyzed, the proofs contain formulas where, e.g., some unitary $U$ is applied to some register $X$.
  This is usually handled in an informal way, e.g., it is said that in the running text that it is understood that $U$ is applied to $X$, and it is left to the reader to understand how $U$ is tensored with the identity to suitably ``address'' $X$.
  While this is arguably satisfactory at the level of formality of pen-and-paper proofs, if we want to formalize it in a theorem prover, or want very rigorous, foundationally clean proofs, then this approach is not suitable and we need some way to refer to individual subsystems and to operators applied to them.
  (There are, of course, other ways to make this formal than using quantum references, but we describe in the next bullet point why these will be extremely cumbersome.)
  Of course, to be able to use this advantage of quantum references, it is necessary to make sure that we can easily model and use them in theorem provers; we have demonstrated this for the Isabelle/HOL theorem prover in \fullshort{\autoref{sec:isabelle}}{\cite{formalization-anonymous}}.
\item The alternative to using references (or some comparable approach) can be very cumbersome:
  Consider the following example:
\fullshort{  \begin{center}
  We have a four-reference system $X,Y,Z,W$ (all qubits), and we apply a CNOT on $ZX$:
    \begin{tikzpicture}
      \initializeCircuit;
      \newWires{X,Y,Z,W};
      \stepForward{2mm};
      \labelWire[\tiny$X$]{X}
      \labelWire[\tiny$Y$]{Y}
      \labelWire[\tiny$Z$]{Z}
      \labelWire[\tiny$W$]{W}
      \stepForward{10mm};
      \node[cnot=X,control=Z] (cnot) {};
      \stepForward{10mm};
      \drawWires{X,Y,Z,W};
    \end{tikzpicture}
  \end{center}
  What is the unitary we apply to the overall system?
  Using our reference approach, we can just say:
  It's $\spair ZX(\CNOT)$.
  (This will be formally defined, but intuitively it means ``the unitary resulting from lifting CNOT through the paired reference $\spair ZX$''.)
}{
  We have a four-reference system $X,Y,Z,W$ (all qubits), and we apply a CNOT on $ZX$ (see the circuit on the right).
  \newline
    \begin{minipage}[b]{.8\linewidth}
  What is the unitary we apply to the overall system?
  Using our reference approach, we can just say:
  It's $\spair ZX(\CNOT)$.
  (This will be formally defined, but intuitively it means ``the unitary resulting from lifting CNOT through the paired reference $\spair ZX$''.)
    \end{minipage}~~~
  \begin{tikzpicture}
      \initializeCircuit;
      \newWires{X,Y,Z,W};
      \stepForward{2mm};
      \labelWire[\tiny$X$]{X}
      \labelWire[\tiny$Y$]{Y}
      \labelWire[\tiny$Z$]{Z}
      \labelWire[\tiny$W$]{W}
      \stepForward{8mm};
      \node[cnot=X,control=Z] (cnot) {};
      \stepForward{8mm};
      \drawWires{X,Y,Z,W};
    \end{tikzpicture}
}

  But without this approach, we would have to write a formula such as $(1_1\otimes \Uswap\otimes 1_1)(\Uswap\otimes 1_2)(\CNOT\otimes 1_2)(\Uswap\otimes 1_2)(1_1\otimes \Uswap\otimes 1_1)$, hardly readable.\footnote{\label{footnote:cnot.example}%
    \newcommand\Ua{U_{\assoc}}%
    And this is if we assume that the tensor product is associative (in other words, we work in a strict category).
    In a more rigid type system (as we would probably have when formalizing these things in a theorem prover) we would have to explicitly use an isomorphism $\Ua$
    between $(a\otimes b)\otimes c$ and $a\otimes(b\otimes c)$.

    In that case, our harmless CNOT becomes $(1_1\otimes\Ua)\penalty0
    \Ua\penalty0
    \pb\paren{(1_1\otimes\Ua)\otimes 1_1}\penalty0
    (\Ua\otimes 1_1)\penalty0
    \pb\paren{(\Uswap\otimes 1_1)\otimes 1_1}\penalty0
    \pb\paren{(\CNOT\otimes 1_1)\otimes 1_1}\penalty0
    \pb\paren{(\Uswap\otimes _1)\otimes 1_1}\penalty0
    (\adj{\Ua}\otimes 1_1)\penalty0
    \pb\paren{(1_1\otimes\adj{\Ua})\otimes 1_1}\penalty0
    \adj\Ua\penalty0
    (1_1\otimes\adj\Ua)$.
  }
  (Here $\Uswap$ is the swapping unitary and $1_1, 1_2$ is the identity on one or two wires, respectively.)

  We call this the \emph{identity-tensor}\index{identity-tensor approach} approach, and we believe that this is often what is implicitly imagined in pen-and-paper reasoning about quantum circuits. But once one tries to be more formal (say in a theorem prover), this approach is extremely unwieldy.
\end{itemize}
We stress that, since quantum references are a foundational tool, and not a specific application, there may be other advantages that we did not think of.
After all, quantum references are an analogue to lenses, and lenses are ubiquitous and highly important in functional programming.
Quantum references are not merely a tool for pointing at a subsystem; we can see them as a possible formalization of what subsystems \emph{are}.

\paragraph{Working with quantum references.}
Having a model of references that allow us to perform operations on their content is nice, but insufficient for many applications.
As we will see later, a quantum reference semantically is a linear function between certain spaces of operators;
defining quantum references by always explicitly constructing these functions would make them hard to use in any practical application.
Instead, we would like to be able to construct references from simpler references and basic building blocks, without needing to explicitly involve the semantic definition.
In particular, to be useful for the application cases we described above, we would like references to support the following natural operations:
\begin{itemize}
\item \textbf{Chaining:}\index{chaining}
  Consider a reference $F$ with memory Hilbert space $\calH_M$ and content Hilbert space $\calH_A$, and a reference $G$ with memory Hilbert space $\calH_A$ and content Hilbert space $\calH_B$.
  That is, $F$ points to a subsystem $A$ of the memory $M$, and $G$ points to a subsystem $B$ of that subsystem $A$.
  (Short: $F:A\to M$, $G:B\to A$.)
  It is natural to expect that then $F$ and $G$ together describe a subsystem $B$ of the memory $M$.
  We call this operation chaining and write the reference describing $B$ in $M$ as $\chain F G$.

  (This is analogous to chaining lenses: We mentioned above that if we have a lens locating \emph{person data} within the memory, and a lens locating \emph{name} within \emph{person data}, we can combine them to a lens locating \emph{name} within the memory.)

  An example of this feature would be: Say $F$ is a two qubit reference (i.e., the content Hilbert space is $\setC^2\otimes\setC^2$).
  It would be natural to consider the first qubit of $F$ as a reference as well.
  And indeed, if $\Fst$ denotes the reference pointing at the first qubit in the space $\setC^2\otimes\setC^2$ (i.e., its memory Hilbert space is  $\setC^2\otimes\setC^2$, and its content Hilbert space is  $\setC^2$), then $\chain F\Fst$ is a reference pointing at the first qubit of $F$.
  This would allow us, for example, to write $\qapply U{\chain Q\Fst}$ instead of  $\qapply{U\otimes 1}Q$ in a quantum program.

  In this simple example, this may not make such a big difference, but by repeated chaining, we can locate arbitrary parts in deeply nested quantum systems.
  With the registers $\Fst$, its analogue $\Snd$, and chaining, we can already build references to factors in arbitrarily nested tensor products.

  It is hard to imagine using lenses without chaining, the same holds for quantum references.
\item \textbf{Pairing:}\index{pairing}
  Let $F$ and $G$ be references with the same memory Hilbert space (i.e., parts of the same memory) but not necessarily the same content Hilbert space.
  E.g., they refer to the first and second qubit of a five qubit system.
  Then it is natural to consider the first two qubits of the five qubit system together as a reference, too (since it is a subsystem).
  That is, we would like to be able to write $\spair FG$ to refer to $F$ and $G$ together.
  (If we think of $F$ and $G$ as regions in the memory, we can think of $\spair FG$ as the concatenation of those regions.)
  Of course, this pair $\spair FG$ should also exist if $F$ and $G$ refer to any non-overlapping parts of the memory, not just if they are consecutive qubits.
  (A classical analogue would be: Say we have lenses \emph{first name} and \emph{last name}, then we should be able to construct a lens \emph{full name} that consists of both, even if \emph{first name} and \emph{last name} are non-consecutive fields in a record.)

  Why would this be needed? Besides being a natural operation, there is an important use case when writing quantum programs:
  In quantum circuits and quantum programs, we often have instructions such as  $\qapply U{\spair FG}$ that applies a unitary to registers $F$ and $G$ jointly.
  While there are different ways to define the semantics of this, it will be easiest to define this if $\spair FG$ can be seen a single reference.
  Then it is sufficient to define the semantics for $\qapply U{F}$ for references $F$, and the semantics for applying a unitary to several variables simultaneously emerges without any extra effort.
  
  Otherwise we will need to give in the semantics an explicit description of how a single unitary is applied to several references, making the formal semantics complex where the intuition is simple.\footnote{%
    Note while the classical program
    $\assign{\pair xy} e$ can always be
    rewritten as $\assign ze$, $\assign x{\mathit{fst}(z)}$,
    $\assign y{\mathit{snd}(z)}$, in the quantum setting such a rewriting is not possible. Applying unitaries to individual variables can never produce entanglement.
    Thus $\qapply U{\spair FG}$ needs to be supported by the language.}

  Another example why pairing is useful would come from quantum information theory:
  There we define the entropy $H(F)$\pagelabel{page:intro.entropy} of a subsystem $F$ (relative to some given density operator).
  However, then we often refer to the entropy $H(FG)$ of several subsystems.
  It is implicitly assumed that we can consider $FG$ together also as a subsystem and then apply the definition of entropy to it, without spelling out the cumbersome details.
  Quantum references with pairing allow to make this formal:
  If we define what $H(F)$ means for a quantum reference $F$ (describing a subsystem), we can immediately give meaning to $H(FG)$ by letting $FG$ stand for the pair $\spair FG$.
  \fullonly{(We illustrate the details of this definition in \autoref{sec:partial.trace}.)}

  Pairing is not meaningful, though, if the registers overlap (e.g., both $F$ and $G$ contain the same qubit).
    To avoid this, pairing must be restricted to references that point to non-overlapping parts of the memory.
  We will call such references \emph{disjoint}\index{disjoint}.
  We will need a formal definition of disjointness, and simple rules to determine when references are disjoint.
  (E.g., if $F,G,H$ are pairwise disjoint, then $\spair FG$ and $H$ should also be disjoint.)

  Once we have such a notion of pairing and disjointness, we can construct quite complex references to non-continuous parts of a quantum memory. For example, if we add chaining and the basic references $\Fst$ and $\Snd$ to the mix, we can write things like $\qapply U{\pairFFF{\chain Q\Snd}{\chain R\Fst}{\chain Q\Fst}}$. This applies $U$ to the second qubit of $Q$, the first of $R$, and the first of $Q$.
  Expressing this without pairing and chaining would be very cumbersome.
\item \textbf{Basis transformation:}\index{basis transformation}\pagelabel{page:mapping.intro.q}
  In quantum programming, quantum systems (such as the content of a reference) are usually expressed with respect to a specific basis (say $\ket0,\ket1$ in case of a qubit).
  But physically, there is usually nothing special about that particular basis.
  A reference might be looked at in a different basis, and there is no a-priori justification why the reference in one basis is a reference, and in another basis, it would not be considered a reference.
  For example, given a qubit reference $F$, we might look at the basis transformed reference $\basistrafo HF$ (where $H$ is the Hadamard transform that transforms $\ket0,\ket1$ into the diagonal basis $\ket+,\ket-$).
  $\basistrafo HF$ should also be considered a reference.
  So in general, given a unitary $U$, we want to have an operation that transforms a reference $F$ into another basis-transformed reference $\basistrafo UF$.
  \showafter{2024}{02}{14}

  A more advanced example (that shows that the concept of references is not limited to programming semantics) is the following:
  We know from physics that position and momentum are two sides of the same coin.
  In fact, we get the momentum by Fourier transforming the position and vice versa.
  So in a system with a single particle, the position of the particle would be described by a real-valued reference (a ``qureal'') $X$.
  And its momentum $P$ could be expressed as the reference $P:=\basistrafo FX$ where $F$ is the Fourier-transform.
  In our approach, $P$ and $X$ could both be first-class references, none of them distinguished as the ``true'' representation: we also have $X=\basistrafo{F^{-1}}P$.\fullonly{\footnote{In a three dimensional system, we can go even further:
    If $X_1,X_2,X_3$ refer to the coordinates of the particle, and $R$ is a rotation in 3D-space, then $\spairFFF{X_1}{X_2}{X_3}$ would be its 3D-position,  $\basistrafo R{\spairFFF{X_1}{X_2}{X_3}}$ in a different frame of reference,  $\chain{\paren{\basistrafo R\spairFFF{X_1}{X_2}{X_3}}}1$ the first coordinate in that frame of reference, and $\basistrafo F{\chain{\paren{\basistrafo R\spairFFF{X_1}{X_2}{X_3}}}1}$ the momentum in that direction.}}
\end{itemize}

\medskip\noindent
We thus ask the question:
\begin{quote}
  \textit{%
    Can we construct a meaningful definition of quantum references that is conceptually simple, amenable to formalization (say in the theorem prover Isabelle/HOL), can be used easily in the definition of programming language semantics and similar, and that supports the operations of chaining, pairing, and basis transformation?
  }
\end{quote}

\paragraph{Our contribution.}
In this work, we answer this question positively.
The intuition is that, if we want to perform some operation $U$ (e.g., a unitary) on the content of some reference $F$, we need to know how this would affect the overall memory, i.e., what operation we effectively perform on the overall memory.
(A bit similar to lenses, where the setter tells us what happens to the overall classical memory when we set the content of the lens.)
The reference provides this information, namely $F(U)$ is that operation on the overall memory.
(This of $F(U)$ as meaning $U\otimes 1$ up to suitable reordering of the subsystems.)
Formally, a quantum reference $F$ with content Hilbert space $\calH_A$ and memory Hilbert space $\calH_M$ is a function that maps an ``update'' on $\calH_A$ (e.g., a unitary operator, or a measurement projector) to an update on $\calH_M$.
Of course, not any function would be a valid quantum reference; we identify certain conditions a function must satisfy to be a valid quantum reference.

Instead of directly developing the theory of \emph{quantum} references, we take a step back and, in \autoref{sec:generic}, define what it means to be a reference in a more general setting (not specifically quantum).
We identify axioms that a category of references (with the objects corresponding to the content/memory types, and the references being the morphisms) should satisfy and derive most laws about references generically (including rules for chaining, pairing, and deriving disjointness).
While quantum references are the main contribution of this work, this two step approach gets us both additional generality and more modular proofs.
And while we only instantiate those axioms for the quantum and for the classical case, having explicit axioms may lay the foundation for future instantiations in other computational settings (e.g., references to hybrid classical/quantum systems, 
or even something more exotic e.g., computations
in non-local models \cite{Popescu1994} or in other physical models,
see, e.g., \cite{Aaronson2005} for some examples).

In \autoref{sec:example}, we illustrate the use of references by generically formulating a small programming language \fullonly{with call-by-reference procedures }and reasoning about it. The generic results are then instantiated in the quantum setting below.

In \autoref{sec:quantum-overall} we instantiate this approach to quantum references (both for finite- and for infinite-dimensional quantum systems), getting our main contribution.

\begin{fullversion}
  And in \autoref{sec:classical} we also show that the same general theory can be instantiated for classical systems, giving us a concept of classical references that actually turns out to be equivalent to lenses.
  However, there is still a novel contribution here:
  To the best of our knowledge, laws for \emph{safe} pairing of lenses have not been studied before;
  our formalism gives them automatically.

  We then illustrate the use of quantum references in
  \autoref{sec:hoare} by introducing a simple quantum Hoare logic and
  verifying the correctness of quantum teleportation.

  In \autoref{sec:complements}, we introduce the additional concept of references with \emph{complements} in which there is a well-defined notion of ``everything outside of reference $F$'' (and we show that the quantum reference category has complements).

  In \autoref{sec:rel.optics}, we compare our concept of references with the general concept of ``optics'' in functional programming (of which lenses are a special case) and discuss how existing general formulations of optics deal with the quantum case.

  In \autoref{sec:lifting}\shortonly{ in the supplement}, we study quantum references in more detail:
  We show that many different quantum mechanical concepts (such as quantum channels, measurements, mixed states, etc.) interoperate well with our notion of references.
  \fullshort{A large fraction of the results presented in this paper have been formally verified in Isabelle/HOL,
    namely the everything from Sections~\ref{sec:generic}, \ref{sec:quantum-overall}--\ref{sec:complements} and parts of \ref{sec:lifting}. }{
    The results presented in this paper have been formally verified in Isabelle/HOL.
    (Note however that the additional quantum results presented in the supplement, \autoref{sec:lifting}, have not all been verified.)
  }
  We present the formalization in \autoref{sec:isabelle}\shortonly{ and attach the theory files as a ZIP to this submission}.
\end{fullversion}

\begin{shortversion}
  We then illustrate the use of quantum references in
  \autoref{sec:hoare} by introducing a simple quantum Hoare logic.

  \paragraph{Additional results and materials.}
  Due to space reasons, we do not include all our results in this extended abstract.
  Attached to the submission is a full version that additionally contains the following results:

  Detailed proofs.
  
  An instantiation of the general theory of references to classical systems, giving us a concept of classical references that actually turns out to be equivalent to lenses. (\autoref{FV:sec:classical})
  However, there is still a novel contribution here:
  To the best of our knowledge, laws for \emph{safe} pairing of lenses have not been studied before;
  our formalism gives them automatically.

  The additional concept of references with \emph{complements} in which there is a well-defined notion of ``everything outside of reference $F$'' (and we show that the quantum reference category has complements). (\autoref{FV:sec:complements})

  Additional detailed discussions, e.g., about the relationship to existing notions of optics (\autoref{FV:sec:rel.optics}).

  An extension of the example from \autoref{sec:example} for a language with procedure calls. (\autoref{FV:sec:example})
  
  We study quantum references in more detail:
  We show that many different quantum mechanical concepts (such as quantum channels, measurements, mixed states, etc.) interoperate well with our notion of references.
  (\autoref{FV:sec:lifting})
  
  A large fraction of the results presented in this paper have been formally verified in Isabelle/HOL, including all sections in this extended abstract except for the example from \autoref{sec:example}, plus more material from the full version. The full version gives an overview. (\autoref{FV:sec:isabelle}) The Isabelle theories are in \cite{formalization-anonymous}.

  An example reasoning in our example Hoare logic of the teleportation protocol. (\autoref{FV:sec:hoare}. Also formalized in Isabelle.)
\end{shortversion}

\paragraph{Related work.}
A number of papers model quantum variables in imperative programs or circuits implicitly or explicitly in different ad-hoc ways, e.g., \cite{Bordg2020, Paykin2017, qrhl, liu19formal, Hietala2021}. In all cases, this is done by considering a memory that has already an explicit tensor product structure in some form (e.g., a sequence of qubits), and variables are modeled by indicating which of the factors we refer to.

Various works use von Neumann algebras to model quantum (sub)systems, somewhat
similar to what we do in the infinite-dimensional setting. E.g., \cite{kornell17quantum,cho16neumann,pechoux20quantum} which bears some superficial resemblance to our modeling in the infinite-dimensional case.
The relationship to our work is subtle and discussed in \fullshort{\autoref{sec:discussion}}{the full version (\autoref{FV:sec:discussion}), }.

Optics (a generalizaton of lenses) have been studied in \cite{riley18optics} for more general categories; the quantum case was not addressed explicitly but might arise as a special case for suitable categories (without support for pairing though).
We discuss the relationship to our work in\fullshort{~\autoref{sec:rel.optics}}{ the full version (\autoref{FV:sec:rel.optics})}.

In follow-up work, our approach is currently being used as the semantic foundation in an overhaul of the qrhl-tool theorem prover \cite{qrhl} (for reasoning about quantum programs) in the development branch.
Our formalization of references is also successfully used as the basis of Isabelle formalization of the one-way-to-hiding theorem \cite{ambainis19semiclassical,katharina-o2h},
as well as the compressed oracle formalism \anonymous{\cite{zhandry18recording,dominique-anon-co-isabelle}}{\cite{zhandry18recording,dominique-co-isabelle}}.

\section{References, generically}
\label{sec:generic}

\paragraph{A reference category.}
(The following text explains the ideas behind the notion of a reference category.
The actual definition is in \autoref{def:reference-category} and \autoref{fig:axioms} only.)

The basic idea behind our formalization of references is that it
transforms updates on the reference's domain into updates on the program
state. For example, if we have a quantum reference~$F$ with content
space $\calH_F$, and some unitary $U$ on $\calH_F$ (which constitutes a specification of
how we want to update the content of $F$), we need to know how~$U$
operates on the whole program state, i.e., we want to transform~$U$
into a unitary $F(U)$ on the program state space $\calH_M$. Or for classical references,
an update on a variable with domain $T$ could be a function
$T\mapsto T$, and thus gets transformed into a function $M\mapsto M$
($M$ is the set of all program states). In the generic setting, we do not
specify explicitly what updates are (e.g., functions, matrices,
\dots). Instead, we only require that the set of updates form a monoid
because we can compose (multiply) updates, and there is a ``do
nothing'' update.  (For any variable type, we have a different monoid
of updates. The overall program state is treated no differently, so we
also have a monoid of updates of program states.)  Said concisely, we
are in a category in which the objects are monoids. We will use bold
letters $\mathbf A,\mathbf B,\dots$ for those objects.

In this setting, a reference is a function that transforms an update of
one type $\mathbf{A}$ into an update of another type $\mathbf{B}$. So the next step in our
formalization is to specify functions between those monoids
(i.e., transformations of updates). Thus in our category, the morphisms are
(not necessarily all) functions between the monoids. For example, in the quantum case, this
could just be linear functions mapping matrices to matrices.
We use the term \emph{pre-reference}\index{pre-reference}
for these morphisms. (Since actual references will need to satisfy some
additional restrictions.)
Pre-references have to
satisfy the usual axioms of categories as well as a few additional
natural ones listed in the first half of \autoref{fig:axioms}.  We
furthermore require that the resulting category has a tensor product, having
some (but not all\footnote{Specifically, one often requires
  that if $F:\mathbf A\to\mathbf A'$, $G:\mathbf B\to\mathbf B'$ are morphisms, then
  there exists a corresponding morphism $F\otimes G:\mathbf A\otimes\mathbf B\to\mathbf A'\otimes\mathbf B'$.
  Or that for any 2-ary multimorphism (e.g., a bilinear map) $F:\mathbf A,\mathbf B\to\mathbf C$,
  there is a morphism $\Hat F:\mathbf A\otimes\mathbf B\to\mathbf C$ with
  $\Hat F(a\otimes b)=F(a,b)$. We do not assume those properties.
  (While in the classical case, and in the finite-dimensional quantum case these two properties
  are easy to achieve, in the infinite-dimensional setting, for many types of morphisms,
  these additional requirements are not satisfied.)}) of the usual properties of a tensor product.
(Notation: $\otimes$ is right-associative. I.e., $\mathbf A\otimes \mathbf B\otimes\mathbf C$ means $\mathbf A\otimes (\mathbf B\otimes\mathbf C)$.)

\begin{figure*}[t]
  \raggedright
  \textbf{General axioms:}
  \begin{compactitem}
  \item \nextlaw\label{ax:monoids}%
    The objects (a.k.a.~\emph{update monoids}\index{update monoid}) $\mathbf A,\mathbf B,\dots$ of $\calR$ are monoids.
    (Which monoids are objects depends on the specific reference category.)
  \item \nextlaw\label{ax:preregs}%
    The pre-references are functions $\mathbf A\to \mathbf B$.
    (Which functions are pre-references depends on the specific reference category.)
    They satisfy the axioms for categories, i.e., they are closed under composition
    (if $F,G$ are pre-references, $F\circ G$ is a pre-reference) and the identity is a pre-reference.
  \item \nextlaw\label{ax:cdot-a}%
    For any $a\in\mathbf A$, the functions $x\mapsto a\cdot x$ and $x\mapsto x\cdot a$ are pre-references $\mathbf A\to\mathbf A$.
  \item $\calR$ has a tensor product \symbolindexmark\tensor{$\tensor$} such that:
    \begin{compactitem}
    \item \nextlaw\label{ax:tensor}%
      For all $\mathbf A,\mathbf B$, $\mathbf A \otimes \mathbf B$ is an object of $\calR$,
      and for $a\in\mathbf A,b\in\mathbf B$, $a\otimes b\in \mathbf{A}\otimes \mathbf{B}$.
    \item \nextlaw\label{ax:tensorext}%
      For pre-references $F,G:\mathbf A\otimes\mathbf B\to\mathbf C$,
      if $\forall a,b: F(a\otimes b)=G(a\otimes b)$, then $F=G$.
    \item \nextlaw\label{ax:tensor.mult}%
      The tensor product is distributive with respect to the monoid multiplication $\cdot$\,, i.e.,
      $(a\otimes b)\cdot(c \otimes d) = (a\cdot c)\otimes (b\cdot d)$.
    \end{compactitem}
  \end{compactitem}
  \textbf{References:}
  \begin{compactitem}
  \item \nextlaw\label{ax:reg.prereg}%
    References are pre-references. (Which pre-references are also references
    depends on the specific reference category.)
  \item \nextlaw\label{ax:reg.morphisms}%
    References satisfy the axioms for morphisms in categories, i.e.,
    they are closed under composition (if $F,G$ are references,
    $F\circ G$ is a reference) and the identity is a reference.
  \item \nextlaw\label{ax:reg.monhom}%
    References are monoid homomorphisms. ($F(1)=1$ and $F(a\cdot b)=F(a)\cdot F(b)$.)
  \item \nextlaw\label{ax:tensor-1}%
    $x\mapsto x\otimes 1$ and $x\mapsto 1\otimes x$ are references.
  \item \nextlaw\label{ax:pairs}%
    If references $F:\mathbf A\to\mathbf C$ and $G:\mathbf B\to\mathbf C$ have commuting ranges (i.e., $F(a),G(b)$ commute for all $a,b$),
    there exists a reference $\spair FG:\mathbf A\otimes\mathbf B\to\mathbf C$ such that $\forall a\, b.\ \spair FG(a\otimes b) = F(a)\cdot G(b)$.
  \end{compactitem}
  \caption{Axioms for a reference category $\calR$}
  \label{fig:axioms}
\end{figure*}

Now we are ready to define references.  A
reference is a function that maps an update to an update, so the pre-references in our category
are potential candidates for references. However, potentially not all
pre-references are valid references. Instead, when instantiating our
theory we specify which pre-references are references.
We require references to satisfy a few simple properties:
\begin{compactitem}
\item References are closed under composition. Intuitively, this means
  that if we can meaningfully transform an update on one type
  $\mathbf A$ into an update on $\mathbf B$, and one on $\mathbf B$
  into one on $\mathbf C$, we can also transform from $\mathbf A$ to
  $\mathbf C$. Besides other things, this is needed to define the
  chaining of references.
\item References $F$ are monoid homomorphisms. This means that doing
  nothing on a variable is translated into doing nothing on the whole
  program state ($F(1)=1$).
  And doing $a$ and then $b$ on the reference has the effect
  of doing $F(a)$ and then $F(b)$ on the whole program state
  ($F(b\cdot a)=F(b)\cdot F(a)$).
\item Mapping an update $x$ to $x\otimes 1$ or $1\otimes x$ is a
  reference. Intuitively, this means: if we can update some system $A$
  with $x$, then we can update a bipartite system $AB$ (or $BA$), too,
  by applying $x\otimes 1$ (or $1\otimes x$).
\item And we require axiom \ref{ax:pairs} to guarantee the existence of ``pairs'',
  see \autoref{def:pair} below.
\end{compactitem}

\noindent To summarize formally:
\begin{definition}\label{def:reference-category}
  A reference category $\calR$ is a category whose objects are sets (possibly with additional structure), satisfying the axioms in \autoref{fig:axioms}.
\end{definition}

That is, if
we want to instantiate the theory of references in a concrete setting
(e.g., classical or quantum references), we have to specify which monoids
constitute the objects (e.g., sets of square matrices), which functions
constitute the pre-references (e.g., linear maps on matrices), and which subset
of the pre-references are the references.\footnote{%
  The reader may wonder why we introduce pre-references in the first place instead of
  directly introducing references and simply requiring that all axioms from
  \autoref{fig:axioms} apply to them.  Unfortunately, this is not possible (at least
  for our choice of axioms): The axiom~\ref{ax:cdot-a} that $x\mapsto a\cdot x$ is a reference is
  incompatible with the axiom that references are monoid homomorphisms.
  On the other hand, removing axiom~\ref{ax:cdot-a} breaks the proofs of laws
  \ref{law:tensor3} and~\ref{law:compat3} in \autoref{fig:laws} below.
  Thus we need the pre-references as intermediate constructs
  to derive laws about references.}

The objects of this category can then intuitively be thought to correspond to types in a programming language.
The precise connection is not specified in the general setting, but can be made precise for specific instantiations of the reference formalism.
For example, in the quantum setting (see \autoref{sec:quantum-overall}), variable types typically correspond to Hilbert spaces (e.g., type the ``qubit'' corresponds to $\setC^2$), and the objects in the corresponding reference category are then the linear operators on these Hilbert spaces.
Hence every variable type with Hilbert space $\calH_A$ corresponds to the object $\mathbf A$ that is the set of bounded operators on $\calH_A$.

References are always functions $\mathbf A\to\mathbf M$.
We will call $\mathbf A$ the \index{content type}\emph{content type}, and $\mathbf M$ the \index{memory type}\emph{memory type} of the reference.
In the simplest setting (an imperative language where references are used to model global program variables), there would be a fixed memory type $\mathbf{M}$ (corresponding to the type of the overall program state), and each variable is represented by a reference $X:\mathbf A\to\mathbf M$ where the content type $\mathbf A$ is the object in the reference category corresponding to the type of the variable.
(And then more complex references could be built up from those ``hardcoded'' variable-references using the operations on references we describe next.)

\begin{fullversion}
  \emph{Remark:} The reader may wonder what common categorical properties a generic reference category satisfies.
  E.g., is it a monoidal category as would be natural in the quantum setting?
  Or should our definition maybe explicitly require additional natural properties such as monoidality?
  While our instantiations (quantum and classical references) are easily seen to be monoidal categories with the tensor product $\rtensor$ defined below (and satisfy many other natural properties since they are subcategories of $\mathbf{Mon}$ and $\mathbf{Hilb}$, respectively), this is not required in the general case by \autoref{def:reference-category}.
  We \emph{intentionally} do not require any such additional properties because we wanted to identify the \emph{minimal} axioms needed to describe references and to derive the laws described in \autoref{fig:laws}.
  This should not stop us, of course, from studying or instantiating, e.g., monoidal reference categories and similar structures.
  Additionally, we also aimed to formulate the axioms in an elementary way to make the definitions accessible to non-category theorists.
\end{fullversion}

\emph{Remark:} The reader may find it confusing what the formal relationship between ``types'' and the objects (update monoids) is. The answer is that in the general definition of a reference category, there are no ``types'', only the update monoids. How exactly they relate to the types of a programming language depends on the specific application. Therefore we refer to the ``types'' in our discussion, but they do not appear in \autoref{def:reference-category}.

\paragraph{Operations on references.} The simplest operation on references
is \emph{chaining}\index{chaining}. For example, say we have a
variable $X$ of type $\mathbf A$ in some memory of type $\mathbf
M$.
(Recall that for simplicity of exposition, we identify types with their corresponding update monoids.)
And assume that $\mathbf A$ is a structured type with fields
$F,G,\dots$. Then $X$ is a reference $\mathbf A\to\mathbf M$.  And $F$
is a reference $\mathbf B\to\mathbf A$ for some $\mathbf B$ because a
record field can be seen as a variable inside a memory of type
$\mathbf A$. Then it makes sense to consider the reference $\chain XF$
that refers to $F$ within $X$. Since $F$ tells us how to transform an
update in $\mathbf B$ into an update in $\mathbf A$, and $X$ tells us
how to transform that into an update in $\mathbf M$, we can simply
define $\chain XF$ as $X\circ F$.
\begin{definition}[Chaining]\label{def:chain}\index{chaining}
  For references $F:\mathbf B\to\mathbf C$, $G:\mathbf A\to\mathbf B$,
  let $\symbolindexmark\chain{\chain FG} := F\circ G$ (composition of functions).
\end{definition}
Chaining is particularly useful combined with the references
$\symbolindexmark\Fst\Fst(a):=a\otimes 1$ and $\symbolindexmark\Snd\Snd(b):=1\otimes b$. These are then
references that refer to the first or second part of a pair. E.g., a
reference
$X:(\mathbf A\otimes \mathbf B)\otimes \mathbf C\mapsto \mathbf M$
would refer to a three-tuple, and $\chain{\chain X\Fst}\Snd$ refers
to the second component of $X$. Thus we can already generically reason about tuples.


Finally, given two references $F:\mathbf A\to\mathbf M$,
$G:\mathbf B\to\mathbf M$, we want to construct the \emph{pair}\index{pair}
$\spair FG:\mathbf A\otimes\mathbf B\to \mathbf M$. What we want is
that operating on the reference $\spair FG$ is the same as operating on
both~$F$ and~$G$.  E.g., in the classical setting, we want that
$\assign{\spair FG}{(x,y)}$ is the same as $\assign Fx;\assign Gy$. And
in the quantum setting $\qapplyOLD{\spair FG}{U\otimes V}$ should be the
same as $\qapplyOLD FU;\qapplyOLD GV$.
Now $\spair FG$ is not, in general, a meaningful construct.
For example, if $F$ is a quantum reference, then it is not
clear what meaning $\spair FF$ could have.
Nor is $\spair FG$ meaningful if $F,G$ share a subsystem.
Thus we need to come up with
a notion of \emph{compatibility} between references so that we can pair
them. This definition will need to be well-behaved. For example, if
$F,G,H$ are pairwise disjoint, we would want that $\spair FG, H$ are
also disjoint, otherwise showing compatibility
of derived references would become too
difficult. It turns out a good (and simple) definition of compatibility is
to require that $F(a),G(b)$ commute for all $a,b$. This intuitively
means that $F$ and $G$ represent different parts of the memory.
\begin{definition}[Compatibility / pairs]\label{def:pair}
  We call references $F,G$ \emph{disjoint}\index{disjoint} if\/
  $\forall a,b.\ F(a)\cdot G(b)=G(b)\cdot F(a)$.  We define the \emph{pair}\index{pair} \symbolindexmark\spair{$\spair FG$} as the unique reference such that $F(a\otimes b)=F(a)\cdot G(b)$.
\end{definition}
(Notation: When writing $\spair FG$ is syntactically ambiguous, we write \symbolindexmark\pair{$\spair FG$} instead.
Pairs are right-associative. I.e., $FGH$ and $\pairFFF FGH$ mean $\pair{F}{\pair{G}{H}}$ and analogously for $F_1\dots F_n$ and $\pairs{F_1}{F_n}$.)
\begin{lemma}\label{lemma:pairs}
  If the references $F,G$ are disjoint, then
  $\spair FG$ exists, is uniquely defined, and is a reference.
\end{lemma}
(The proof is in \autoref{sec:proof:lemma:pairs}.)
Based on pairing, we can also define useful natural references
$\symbolindexmark\swap\swap:\mathbf A\otimes\mathbf B\to\mathbf B\otimes\mathbf A$ and
$\symbolindexmark\assoc\assoc:(\mathbf A\otimes\mathbf B)\otimes\mathbf C \to \mathbf
A\otimes(\mathbf B\otimes\mathbf C)$ and its inverse
$\symbolindexmark\assocp\assocp$.
And for references $F:\mathbf A\to\mathbf{A}',G:\mathbf{B}\to\mathbf{B}'$,
we can define a reference $F\symbolindexmark\rtensor\rtensor G:\mathbf A\otimes\mathbf B\to\mathbf A'\otimes\mathbf B'$,
the tensor product of references.%
\footnote{\label{footnote:def.sigma.alpha.tensor}%
  They are defined as $\swap=\pair\Snd\Fst$,
  $\assoc = \pair{\pair\Fst{\Snd\circ\Fst}}{\Snd\circ\Snd}$,
  $\assocp = \pair{\Fst\circ\Fst}{\pair{\Fst\circ\Snd}{\Snd}}$,
  $F\rtensor G = \pb\pair{\paren{\lambda a. F(a)\otimes 1}}{\paren{\lambda b. 1\otimes G(b)}}$.
  Their characteristic properties which justify these definitions are given in \autoref{fig:laws}.}
(Notation: $\rtensor$ is right-associative. I.e., $F\rtensor G\rtensor H$ means $F\rtensor (G\rtensor H)$.)

Together with the chaining operation, this gives us the
possibility to rewrite expressions involving references. For example, the
quantum program $\qapplyOLD{\spair FG}U;\ \qapplyOLD{\spair GF}V$ would perform
$\spair GF(V)\cdot \spair FG(U)$ on the quantum memory. And say we have
another quantum program $\qapplyOLD{\spair FG}W$, performing $\pair
FG(W)$. And we want to find out whether those programs are
denotationally equivalent.
We can rewrite $\spair GF(V)\cdot \spair FG(U)$ into $\spair FG(\sigma(V)\cdot U)$ (without any recourse to the specific semantics of the quantum case, only using only the general laws
\footnote{$\spair GF(V)\cdot \spair FG(U)
  \quad\txtrel{\ref{law:pair.sigma}}=\quad
  (\chain{(\spair FG)}\sigma)(V)\cdot\spair FG(U)
  \quad\txtrel{Def.~\ref{def:chain}}=\quad
  \spair FG(\sigma(V))\cdot\spair FG(U)
  \ \txtrel{\ref{ax:reg.monhom}}=\ 
  \spair FG(\sigma(V)\cdot U) $.
}). Thus the two programs are denotationally
equal if
$\sigma(V)\cdot U=W$. The important fact here is that we
have reduced the denotational equivalence question to an expression
involving only the matrices $U,V,W$. We do not need to know how $F,G$
are defined, they might be elementary quantum references or complex
reference-expressions (as long as they are disjoint); all the
underlying complexity is handled by the reference formalism.

The laws for reasoning about references are given in \autoref{fig:laws}.

\begin{fullversion}
  We formalized the proofs of all laws in Isabelle/HOL (\autoref{sec:isabelle}), so we do not give them here.
\end{fullversion}

\begin{figure*}
  \raggedright

  \textbf{Tensor product.} ($F,G$ are assumed to be references.)
  \quad \nextlaw\label{law:tensor3}%
    If $F(a\otimes (b\otimes c))=G(a\otimes (b\otimes c))$ for all $a,b,c$,
    then $F=G$.\fullonly\footnotemark
    \quad \nextlaw\label{law:tensor.ab}%
    $F\rtensor G$ is a reference and
    $(F\rtensor G)(a\otimes b) = F(a)\otimes G(b)$

  \textbf{Swaps and associators.}
  \quad \nextlaw\label{law:sigma.alpha.refs} $\sigma,\alpha,\alpha'$ are references.
  \quad \nextlaw\label{law:sigma.ab}%
  $\sigma(a\otimes b)=b\otimes a$
  \quad \nextlaw\label{law:chain.sigma.12} $\chain\sigma\Fst=\Snd$ \quad $\chain\sigma\Snd=\Fst$
  \quad \nextlaw\label{law:alpha.abc}%
  $\alpha\pb\paren{(a\otimes b)\otimes c}=a\otimes(b\otimes c)$
  \quad \nextlaw\label{law:alpha'.abc}%
  $\alpha'\pb\paren{a\otimes(b\otimes c)}=(a\otimes b)\otimes c$

  \textbf{Disjointness.} ($F,G,H,L$ are assumed to be references.)
  \quad \nextlaw\label{law:Fst.Snd.disjoint} $\Fst,\Snd$ are disjoint.
  \quad \nextlaw\label{law:compat3}%
  If $F,G,H$ are pairwise disjoint, $\spair FG$ and $H$ are disjoint.
  \quad \nextlaw\label{law:chain.disjoint.outer} If $F,G$ are disjoint, $\chain FH,G$ are disjoint.
  \quad \nextlaw\label{law:chain.disjoint.inner} If $F,G$ are disjoint, $\chain HF,\chain HG$ are disjoint.
  \quad \nextlaw\label{law:tensor.disjoint} 
  If $F,H$ are disjoint, and $G,L$ are disjoint, then $F\rtensor G, H\rtensor L$ are disjoint.
  
  \textbf{Pairs.}
  ($F,G,H$ are assumed to be pairwise disjoint references, and $C,D$ to be references.)
  \quad \nextlaw\label{law:pair.ref} $\spair FG$ is a reference.
  \quad \nextlaw\label{law:pair.select}
  $\chain{\spair FG}\Fst=F$, $\chain{\spair FG}\Snd=G$, 
  \quad \nextlaw\label{law:fst.snd.id} $\pair\Fst\Snd = \id$
  \quad \nextlaw\label{law:snd.fst.sigma} $\pair\Snd\Fst = \sigma$
  \quad \nextlaw\label{law:pair.sigma}%
  $\chain{\spair FG}\sigma = \spair GF$
  \quad \nextlaw\label{law:pair.alpha}%
  $\chain{\pair F{\spair GH}}\alpha = \pair{\spair FG}H$
  \quad \nextlaw\label{law:pair.alpha'}%
  $\chain{\pair{\spair FG}H}{\alpha'} = {\pair F{\spair GH}}$
  \quad \nextlaw\label{law:pair.chain}%
  $\pair{\chain CF}{\chain CG} = \chain C{\spair FG}$
  \quad \nextlaw\label{law:pair.tensor}%
  $\chain{\spair FG}{(C\rtensor D)}=\pair{\chain FC}{\chain GD}$ 

  \caption{Generic (derived) laws for references. \quad
    All laws have the implicit assumption that the domains/codomains of the references match.
    (E.g., if $\chain FG$ occurs in a law, then the domain of $F$ is assumed to match the codomain of $G$.)
    \quad
    The proofs of these laws (as a consequence of the axioms in \autoref{fig:axioms}) are given in \autoref{sec:proofs:laws}.
  }
  \label{fig:laws}
\end{figure*}
\fullonly{\footnotetext{\label{footnote:separating}More generally, we say $M\subseteq \mathbf A$ is a
      \emph{separating set}\index{separating set} of $\mathbf A$ iff for all $\mathbf C$ and all pre-references $F,G:\mathbf A\to\mathbf C$
      with $\forall x\in \mathbf A.\ F(x)=G(x)$, we have $F=G$.
      The axioms for pre-references then imply that if $M,N$ are separating sets of
      $\mathbf A,\mathbf B$, then $\{a\otimes b:a\in M,b\in N\}$ is a
      separating set of $\mathbf A\otimes \mathbf B$. From this we immediately get
      \ref{law:tensor3} and analogous statements for more than three
      tensor factors.

      In many logics (e.g., simple type theory as in Isabelle/HOL), the
      above definition of a separating set cannot be formulated (because ``for all
      $\mathbf C$'' quantifies over all objects in a category). In that case, we can instead fix $\mathbf C$.
      I.e., we call $M$ a \emph{separating$_{\mathbf C}$ set} of $\mathbf A$
      iff for all pre-references $F,G:\mathbf A\to\mathbf C$
      with $\forall x\in \mathbf A.\ F(x)=G(x)$ we have $F=G$.
      If $M,N$ are separating$_{\mathbf C}$ sets of
      $\mathbf A,\mathbf B$, then $\{a\otimes b:a\in M,b\in N\}$ is a
      separating$_{\mathbf C}$ set of $\mathbf A\otimes \mathbf B$.}}

\paragraph{Iso-references and equivalence.} Some simple categorical notion will come in handy later:
We call a reference $F$ an \emph{iso-reference}\index{iso-reference} iff there exists a reference $G$ such that $F\circ G=\id$ and $G\circ F=\id$.
(I.e., an iso-reference is an isomorphism in the category of references.)
Iso-references have an important intuitive meaning: While references typically give access only to part of the program memory, an iso-reference gives access to the whole memory.\footnote{%
  A notion with a similar intuition could have been obtained if we merely require that $F$ is surjective.
  Then also any operation on the whole memory would correspond to some operation on $F$.
  But the resulting notion seems less nice to work with.}
We call two references $F:\mathbf A\to\mathbf C$, $G:\mathbf B\to \mathbf C$ \emph{equivalent}\index{equivalent references} iff there exists an iso-reference $I$ such that $F\circ I=G$.
(In category-theoretic terms: they are isomorphic as objects of the slice category of the reference category.)
Intuitively, this means that one can perform the same updates on the program memory given $F$ and given $G$.
In particular, if we care only about the effect a reference can have on the program memory, all we need to know is the reference up to equivalence.\fullonly{\footnote{%
  A notion with a similar intuition could have been obtained if we merely require that $F,G$ have the same range.
  But the resulting notion seems less nice to work with.
  Note that specifically in the quantum setting, these two notions coincide, though. \fullshort{(\autoref{lemma:same.range.equiv}.)}{(Shown in the supplement, \autoref{lemma:same.range.equiv}.)}}}

\section{Example}
\label{sec:example}

\fullshort{
  To clarify how references can be used, we will develop a small imperative programming language that uses references to access variables. We support write access to variables (and parts and tuples thereof), as well as procedures with call-by-reference semantics (where we can even pass a reference to part of a variable).
This illustrates that such language features can be obtains ``for free'' using our framework.
}{
  To clarify how references can be used, we will develop a tiny imperative programming language that uses references to access variables.
  In the full version (attached) we extend this example also include procedures with call-by-reference semantics (where we can even pass a reference to part of a variable) to illustrate that such language features can be obtains ``for free'' using our framework.
}

We first design as much of our language as we can using only the generic notion of references (as in \autoref{sec:generic}) in order to demonstrate how far one can get without even having to consider concrete computational models (such as deterministic, probabilistic, or quantum).
We then formalize the special case of a quantum language in \fullshort{\autoref{sec:ex.quantum}, and the classical case in \autoref{sec:ex.classical}.}{\autoref{sec:quantum-overall}, and the classical case in the full version~\autoref{FV:sec:ex.classical}.}
(A reader not interested in this level of abstraction may also first read section \autoref{sec:quantum-finite}
and then read this example section with $\calR:=\Lquantumfin$ in mind.)

We will just assume the following ingredients:
\begin{itemize}
\item A type system for values. We require that there is a type $T\times T'$ for any types $T,T'$.
  Besides that, we put no requirements (but it would be natural to have types for bools, ints, as well as more complex algebraic datatypes).
\item A reference category $\calR$.

  We require that for every program type $T$, there is an associated update monoid $\mathbf T$ in $\calR$.
  Intuitively, $\mathbf T$ are all operations that can be done on a variable of type $T$.

  We assume that the update monoid associated with $T\times U$ is~$\mathbf T\otimes\mathbf U$.
\end{itemize}

We will assume that the update monoids of $\calR$ are also the suitable mathematical structure we use for describing the program semantics. E.g., if the update monoids contain unitary operators, then we will also express the semantics of a program using unitary operators (this would, e.g., disallow having measurements of other non-reversible operations). This is purely for simplicity of exposition, we explain how to avoid this restriction in \fullshort{\autoref{sec:example.more.general}}{the full version (\autoref{FV:sec:example.more.general})}.

\fullonly{
  \subsection{The nano language.}\index{nano language}\index{language!nano}
  \label{sec:nano}
}

\paragraph{Syntax.}
The first (and most minimalistic) version of the language has the following syntax:
\[
  P ::= \skipprog \mid P;P \mid \qapply UX
\]
Here $\skipprog$ does nothing, and $P;P'$ is the sequential composition of programs $P$ and $P'$.
And $\qapply UX$ applies the operation $U$ to $X$.
Here $X$ is a reference, and $U$ is some operation that can be applied to the content of $X$.
E.g., in the classical case, a function $U:T\to T$ that is applied to the content of $X$.
(So a valid program might be $\qapply{\paren{\lambda x.\ x^2}}X$ for an integer variable $X$.)
And in the quantum case a unitary operation $U$.
(So a valid program might be $\qapply HX$, applying a Hadamard gate to the qubit $X$.)
Since $X$ can be an arbitrary reference, we can also use tuples for $X$, without having to introduce them explicitly in language definition. For example:
\[
  P_\mathit{ex} := \qapply{\Uswap}{\spair XY};\ \qapply\Uswap{\spair YX}
\]
There $\Uswap$ would be the swap of two values, the formal definition of such an operation depends on the context.
(Classical: $\Uswap := \lambda(x,y).\ (y.x)$. Quantum: $\Uswap\ket{x,y} := \ket{y,x}$.)
And $X$, $Y$ are disjoint references (i.e., global variables).

So this program would swap $X,Y$ and then swap $Y,X$.

\paragraph{Typing.}
We require programs to be well-typed. For this, fix some type $T_M$ of the global memory.
Then the only requirement is that in an $\qapply UX$ statements, if $U\in\mathbf T_U$, we have that $X$ is a reference $\mathbf T_U\to\mathbf T_M$.
(Note that this also implies that programs like $\qapply U{\spair XX}$ are forbidden because $\spair XX$ is not a reference by \autoref{def:pair}.)

\paragraph{Semantics.}
We can now, completely abstractly, define the denotational semantics $\denot P \in \mathbf T_M$ of a program $P$:%
\footnote{Recall that \emph{for simplicity}, we use the elements of the update monoids also to describe the semantics of programs.}
\[
  \denot\skipprog := 1
  \qquad
  \denot{P;Q} := \denot Q\cdot \denot P
  \qquad
  \denot{\qapply UX} := X(U)
\]
The first two are just the natural way of defining the empty program and sequential composition, given that $\Sem_{T_M}$ is a monoid.

The interesting part is $\denot{\qapply UX}$: Since $X:\mathbf T_U\to\mathbf T_M$ is a reference, we can apply it the operation $U$ to tell us what $U$ does to the whole memory. So $X(U)$ is an update of the whole memory, a suitable description of the semantics of $\denot{\qapply UX}$. This illustrates how the whole complexity of addressing $X$ (even if $X$ contains tuples and projections) is automatically handled by the reference framework.

\paragraph{Example computation.}
We close the discussion of the nano language with a computation of the semantics of $P_\mathit{ex}$ above.
We assume two properties of $\Uswap$: $\Uswap^2=1$ (swapping twice does nothing).
And $\sigma(\Uswap)=\Uswap$.
(Recall from \autoref{sec:generic} that $\sigma$ is the reference that exchanges two parts of the memory.
So this condition means that swapping $X$ and $Y$ is the same as swapping $Y$ and $X$.)
Then we have:
\begin{align*}
  \denot {P_\mathit{ex}}
  &\txtrel{def}= {\spair YX(\Uswap)} \cdot {\spair XY(\Uswap)}
    \starrel= {{\spair XY}(\sigma(\Uswap)) \cdot \spair XY(\Uswap)}
  \\&
  \starstarrel= {{\spair XY}(\Uswap) \cdot \spair XY(\Uswap)}
  \tristarrel= {{\spair XY}(\Uswap\cdot\Uswap)} 
  \starstarrel= {{\spair XY}(1)}
  \tristarrel = 1.
\end{align*}
Here $(*)$ uses law~\ref{law:pair.sigma}.
$(**)$ uses the properties of $\Uswap$.
And $(*{*}*)$ uses that references are monoid homomorphisms (axiom~\ref{ax:reg.monhom}).

The result $\denot{P_\mathit{ex}}$ is what we expect: swapping two variables twice does nothing. The derivation dealt with the various issues of addressing variables and composition of programs etc.

\delaytextsv{micro language}{
\subsection{The micro language.}\index{micro language}\index{language!micro}
\label{sec:micro}

We now extend the nano language to have procedures with call by reference.

\paragraph{Global and local references.}
Since we need to store the procedure parameters somewhere, we first introduce the concept of global and local variables.
The state of a program will be assumed to be bipartite, of type $T_M\times T_L$.
$T_M$ is the type of the global memory (containing global variables).
And $T_L$ is the type of the local memory (containing the procedure parameters).
This allows us to define global and local references:

A reference $X$ is \emph{global}\index{global!reference} if it equals $X=\chain\Fst{\Hat X}$ for some reference $\Hat X$.
This means, the reference points to some part of the global memory.
(This does not mean that we need to write every occurrence of a global reference in this form.
E.g., if $X,Y$ are both global, then $\spair XY$ is global because $\spair XY=\chain\Fst{\pair{\Hat X}{\Hat Y}}$.)

A reference $X$ is \emph{local}\index{local!reference} if it equals $X=\chain{\Snd}{\Hat X}$ for some $\Hat X$.
I.e., it points into the local memory.
(References can also be neither local nor global, e.g., $\spair XY$ for global $X$ and local $Y$.)

\paragraph{Syntax.}
The syntax of the language is now extended to be:
\[
  P ::= \skipprog \mid P;P \mid \qapply UX \mid \proccall pX
\]
Here $p$ is the name of a procedure, and $X$ is a \emph{local} reference.

Why do we restrict the arguments passed to procedures to local references?
If we allow to pass global references $X$ to a procedure, then the procedure can access $X$ both directly (as a global variable) and through its parameter (as a local variable).
We wish to avoid this form of aliasing; the simplest way to do so is the above restriction.
(Of course, it is still possible to pass a global reference to a procedure by swapping its content with a local reference, then passing the local reference, and afterwards swapping back.)

Procedure definitions have the following syntax:
\[
  \procdef pTP
\]
Here $p$ is the procedure name, $T$ the type of the parameters, and $P$ the procedure body (a program).
Note that there is only a single parameter: this is no restriction because $T$ can always be a product type.
Then several parameters can be passed as a tuple (nested pairs of references), and the individual references can be recovered in the procedure by suitable projections (chaining with $\Fst,\Snd$).
Of course, in a full-fledged language, we might want some nicer syntactic sugar for accessing several parameters.
Here we choose to stick to the bare-bones syntax to show what goes on under the hood.
We also note that the parameter is always called $\prm$; again, this is what happens under the hood, syntactic sugar may allow to write different names here.
Inside procedure bodies, $\prm$ will simply be a synonym for $\Snd$.
(Because $\Snd$ accesses the local memory, and the local memory contains only the procedure parameter in our case.)

An example of a program together with a procedure declaration might be:
\begin{align*}
  & \procdef p{\mathtt{bool}\times\mathtt{bool}}{\qapply U{\chain\prm\Fst}} \\
  & P_\mathit{ex2} := \proccall p{\pair{\chain X\Snd}{\chain X\Fst}};\ \qapply V{\chain X\Fst}
\end{align*}
Here $X$ is some local reference (of type $\mathtt{bool}\times\mathtt{bool}$), and $U,V\in\mathbf T_\mathtt{bool}$.

That is, the procedure $p$ takes two arguments (encoded into a single one), and then applies some operation $U$ to the first argument. ($\chain\prm\Fst = \chain\Snd\Fst$ is a local reference pointing to the first part of the local memory, because $\prm=\Snd$ points to the local memory.) With more syntactic sugar, it could be written $\mathbf{proc}\ p(Y,Z)\ \braces{\qapply UY}$. And $P_\mathit{ex2}$ invokes $p$ on $X$, but passing the second part before the first, and afterwards applies some operation $V$ to the first half of $X$.
Intuitively (and also formally as we will see), this means that $U$ will be applied to the second, and $V$ to the first half of the reference $X$.
I.e., $V\otimes U$ to the whole of $X$.

\paragraph{Typing.}
Programs and procedures need to be well-typed, of course.
Since we now have procedure definitions and local memories that have different types in different contexts, we type programs relative to an environment $E=(T_M,T_L,\frakP)$ where $\frakP$ is a set of procedure declarations (with distinct names), and $T_M,T_L$ are the type of the global/local memory. We write $T_M,T_L,\frakP\vdash S$ to mean that $P$ is well-typed in this environment. In essence, the typing rules simply guarantee that all types fit together, and that procedures are called only with local references as parameters, and that we do not have recursion. In detail:
\begin{mathpar}
  \inferrule{~}{T_M,T_L,\frakP \vdash \skipprog}
  \and
  \inferrule{T_M,T_L,\frakP \vdash P,\ Q
  }
  {T_M,T_L,\frakP \vdash P;Q}
  \and
  \inferrule{X:\mathbf T_U\to\mathbf T_M\otimes\mathbf T_L\text{ is a reference}
    \\
    U\in\mathbf T_U
  }
  {T_M,T_L,\frakP \vdash \qapply UX}
  \and
  \inferrule{
    T_M,T_p,\frakP\vdash P
    \\
    X:\mathbf T_p\to\mathbf T_M\otimes\mathbf T_L\text{ is a reference}
    \\
    X\text{ is local}
  }
  {T_M,\ T_L,\ \frakP\cupdot\pb\braces{\procdef p{T_p}P} \ \vdash\  \proccall pX}
\end{mathpar}
Notice how in the rule for procedure calls, the reference $X$ that is passed to $p$ is evaluated in the memory of type $T_M\times T_L$ of the calling program (because the memory space of $X$ is $\mathbf T_M\otimes\mathbf T_L$). But the body of the procedure is then evaluated in the memory of type $T_M\times T_p$ where $T_p$ is the type of $X$ (its content space is $\mathbf T_p$) because the local memory of $p$ contains only the parameter $X$.
(When typing the procedure body $P$, we may encounter the symbol $\prm$.
But recall that this was simply a synonym for $\Snd$, so the typing rules do not need a special provision for that.)

\paragraph{Semantics.}
The semantics of $\skipprog$, $P;Q$, and $\qapply UX$ were already defined for the nano language.
Those definitions do not change.
Defining the semantics of procedure calls in quantum languages can be quite tedious and error-prone:
One needs to worry about where to store the local variables of the caller, allocate variables of the callee, transfer the data to the callee, etc.
In this light, it comes as a pleasant surprise that our definition is very simple:
\[
  \denot{\proccall pX}
  :=
  {\pair\Fst X}\pb\paren{\denot P}
  \qquad
  (\text{given declaration}\ \ 
  \procdef p{T_p}P).
\]
Note that the typing rules imply that $X$ is local and can be written as $\Snd.\Hat X$, hence $\Fst$ and $X$ are disjoint, hence the pairing $\pair\Fst X$ is well-defined here.

How shall we understand this definition? $\denot P$ is the semantics of the program body.
It tells us how the program $P$ operates on a memory of type $T_M\times T_p$.
(That is, $\denot P\in\mathbf T_M\otimes \mathbf T_p$. For example, it is a unitary operation on such memories.)
Now, when invoking ${\proccall pX}$, we want $P$ to operate on the global memory (type $M$), and the content of $X$. Since the global memory is the first part of the memory $T_M\times T_L$, the reference referring to the totality of the global memory is $\Fst$. So $\pair{\Fst}X$ refers to the part of the global and local memory that we want $P$ to operate on.
Therefore we lift $\denot P$ to something operating on the whole memory by applying $\pair{\Fst}X$ to it. 

\paragraph{Example computation.}
We are now ready to check whether $P_\mathit{ex2}$ does what we intuitively assume:
\begin{align*}
  \denot{P_\mathit{ex2}}
  &\txtrel{def}= \chain X\Fst(V) \cdot
    \pair{\Fst}{\pair{\chain X\Snd}{\chain X\Fst}}\pb\paren{\denot{\qapply U{\chain\prm\Fst}}}
  \\
  &\txtrel{def}= \chain X\Fst(V) \cdot
    \pair{\Fst}{\pair{\chain X\Snd}{\chain X\Fst}}{(\chain\Snd\Fst(U))}
  \\
  &=
    {\chain X\Fst(V) \cdot
    \pair{\Fst}{\pair{\chain X\Snd}{\chain X\Fst}}.\Snd.\Fst(U)}
  \starrel=
    {\chain X\Fst(V) \cdot
    {\pair{\chain X\Snd}{\chain X\Fst}}.\Fst(U)}
  \\
  &\starrel=
    {\chain X\Fst(V) \cdot
    {\chain X\Snd}(U)}
  \starstarrel=
    {X\pb\paren{\Fst(V) \cdot
    \Snd(U)}}
  \tristarrel=
    {X\pb\paren{\pair\Fst\Snd(V\otimes U)}}
  \\
  &\fourstarrel=
{X(V\otimes U)}.
\end{align*}
And $(*)$ uses law \ref{law:pair.select}.
And $(**)$ uses that references are monoid homomorphisms (axiom~\ref{ax:reg.monhom}).
And $(*{*}*)$ is by definition of pairs (\autoref{def:pair}).
And $(*{*}{*}*)$ uses law \ref{law:fst.snd.id}.

As we see, $P_{ex2}$ simply applies $V\otimes U$ to $X$.
This was what we had expected from the intuitive semantics.
}

\delaytextsv{example-more-general}{
\subsection{More general semantics}
\label{sec:example.more.general}

(This subsection can be skipped at first reading; it is only required for understanding the examples involving quantum channels in \autoref{sec:ex.quantum} and probabilistic programs in \autoref{sec:ex.classical}.)

In the examples earlier in \autoref{sec:example}, we assumed that $\mathbf T_M$, the update monoid corresponding to the type of the whole memory, is also a suitable mathematical structure for describing the semantics of a program as a whole. This is not always the case. For example, in the quantum setting, we might want to use quantum subchannels (completely positive trace-reducing maps) to describe the semantics of a program (so as to allow, e.g., measurements in our programs).
But $\mathbf T_M$ only contains linear operators such as unitaries.
So we cannot use $\mathbf T_M$ as the mathematical structure for the denotational semantics.
The ideas of the above example still work in that case, if we generalize things a little:

For every program type $T$, let $\Sem_T$ denote the mathematical structure for describing the denotational semantics of imperative programs with a memory of type $T$. (In the examples above, we simply used $\Sem_T:=\mathbf T_M$.
Another example could be, e.g., $\Sem_T$ being the probabilistic functions $T\to T$ in the classical probabilistic case. Or the quantum subchannels.)
In addition to the assumptions made at the beginning of \autoref{sec:example}, we require:
\begin{itemize}
\item $\Sem_T$ is a monoid for every type $T$.
  (Multiplication is interpreted as running two programs in sequence, the unit is a program that does not do anything.)
\item For every type $T$, there is a monoid homomorphism $\Sem:\mathbf T\to\Sem_T$, that tells us how an update (in the sense of the reference category $\calR$) is interpreted as a program semantics.
  (E.g. this could map a bounded operator $U\in\mathbf T$ to a completely positive map $\rho\mapsto U\rho \adj U$.
  In the examples above, we simply had $\Sem:=\id$.)
\item 
  Finally, we require some way to apply references to $\Sem_T$ (not just to $\mathbf T$).
  That is, given a reference~$X$ (which is a function $\mathbf T\to\mathbf M$), we assume that there is a definition of $X^\bullet(s)\in\Sem_M$ for any $s\in\Sem_T$.
  The intuitive interpretation is that, if $s$ is the semantics of a program on states of type $s$, then $X^\bullet(s)$ is the semantics of that program working on the content of the subsystem $X$ of the larger system $M$.
  We require: $\forall a.\ X^\bullet(\Sem(a))=\Sem(X(a))$.
  (This can be less trivial to construct; we make this concrete in Sections~\ref{sec:ex.quantum} and~\ref{sec:ex.classical} in the quantum and classical cases.
  In the examples above, we simply had $X^\bullet:=X$.)
\end{itemize}

\noindent
With these added notions, we can generalize the examples in Sections~\ref{sec:nano} and~\ref{sec:micro}.
Syntax and typing of these languages does not change, but we now describe the semantics of a program $P$ as $\denot P\in\Sem_{T_M}$  or $\denot P\in\Sem_{T_M\times T_L}$, respectively.
The definition of this semantics is in the nano language:
\[
  \denot\skipprog := 1
  \qquad
  \denot{P;Q} := \denot Q\cdot \denot P
  \qquad
  \denot{\qapply UX} := \Sem(X(U))
\]
Note that the only difference is that instead of $X(U)$, we write $\Sem(X(U))$ when applying $U$ to $X$.
This is because $X(U)$ already describes what $U$ does to the program memory, we only need represent this as an element of $\Sem_{T_M}$ instead of $\mathbf T_M$.
$\Sem$ performs this conversion.

And for the micro language, we additionally define:
\[
  \denot{\proccall pX}
  :=
  {\pair\Fst X}^\bullet\pb\paren{\denot P}
  \qquad
  (\text{given declaration}\ \ 
  \procdef p{T_p}P).
\]
The intuition behind this stays the same.
Only now we cannot apply $\pair \Fst X$ to $\denot P$, because $\denot P\in\Sem_{T_M\times T_p}$, not in $\mathbf T_M\otimes \mathbf T_p$.
By transforming  $\pair \Fst X$ into a function $\pair \Fst X^\bullet$ on $\Sem_{T_M\times T_p}$, we solve this.

With these new definitions, the calculations given in the above examples still work, additionally using the properties of $\Sem$ and $\bullet$:
\[
  \denot{P_\mathit{ex}} = \Sem(1) = 1
  \qquad
  \text{and}
  \qquad
  \denot{P_\mathit{ex2}} = \Sem\pb\paren{X(V\otimes U)}.
\]

See Sections~\ref{sec:ex.quantum} and~\ref{sec:ex.classical} below to make this somewhat abstract generalization clearer.
}

\section{Quantum references}
\label{sec:quantum-finite}
\label{sec:quantum-overall}

Instantiating the general theory in the quantum setting turns out to
be very easy (at least when considering finite-dimensional spaces
only). The objects $\mathbf A,\mathbf B,\dots$ have to be monoids that
represent updates of quantum variables. Updates of quantum variables
are unitary operations (and/or projectors etc.), i.e., linear
operators on a finite-dimensional (but not zero-dimensional) complex Hilbert space (a.k.a.~matrices).\footnote{%
  One could argue that not all linear maps should be allowed, only, e.g., unitary ones.
  Or only contracting linear maps. However, it is more convenient to use a set that is linearly closed.
  And we do not loose anything because the references as we define them below preserve
  unitaries, projectors, contracting maps, etc.}
Thus: the objects of the category \symbolindexmark\Lquantumfin{$\Lquantumfin$} are the spaces
of linear operators on~$\calH_A$ for finite dimensional Hilbert spaces $\calH_A$.
(Those are naturally monoids with the identity matrix $1$.)
Throughout this section, we will use the notational convention that $\mathbf A$ is the space of linear operators on~$\calH_A$.
Similarly for $\mathbf B,\mathbf C,\dots$

Pre-references $\mathbf A\to\mathbf B$ are defined to be the linear
functions. It is then easy to check that all axioms from \autoref{fig:axioms} concerning
pre-references are satisfied. The tensor product is the usual tensor product of
matrices from linear algebra.

Finally, we define the references in $\Lquantumfin$ as those linear maps
$F:\mathbf A\to\mathbf B$ such that $F(1)=1$, $F(ab)=F(a)F(b)$, and
$F(\adj a)=\adj{F(a)}$. (The first two properties are required for
references by definition, and the third one makes references more
well-behaved, e.g., $F(U)$ is unitary for a unitary $U$.)
In other words, references are unital $*$-homomorphisms.

For this definition of $\Lquantumfin$, the axioms from
\autoref{fig:axioms} can be shown to hold, and the laws from
\autoref{fig:laws} follow.
\fullonly{We do not give the proofs in the finite-dimensional case here explicitly here.
They can be found in our Isabelle/HOL formalization (\autoref{sec:isabelle}) and also arise as special cases of the infinite-dimensional proofs (see\shortonly{ the supplement,} \autoref{sec:infinite}).}

In addition, we can easily define elementary references. E.g., if our
program state is of the form
$\calH_M:=\calH_A\otimes\calH_B\otimes\calH_C\otimes\dots$, then
$B(b):=1\otimes b\otimes1\otimes\dots$ is a reference
$\mathbf B\to\mathbf M$ that refers to the second subsystem.
However, in addition to those elementary
references, we can construct tuples and more using the constructions from
\autoref{sec:generic}. It is also possible to keep things axiomatic
and to perform an analysis simply stating that $A,B,C,\dots$ are disjoint quantum references without specifying concretely into which program state they embed and how.
This means that the result of the analysis holds for any choice of memory.
\fullonly{(We do this in the teleportation example in \autoref{sec:teleport}.)}

Beside the laws inherited from \autoref{sec:generic}, we get some
useful laws specific to the quantum setting.
For example, $F(a)$ is a projector/unitary if $a$ is a projector/unitary.
If we define $\symbolindexmark\sandwich{\sandwich U}(a):=Ua\adj U$ (short for sandwich), then
$\sandwich U$ is a reference for any unitary $U$. This makes it possible to
``basis transform'' quantum references (see \autopageref{page:mapping.intro.q}); what
was written $\basistrafo UX$ there would simply be $\symbolindexmark\basistrafo{\basistrafo UX}:=\chain X{\sandwich U}$.
We have the law
$F\pb\paren{\sandwich a(b)}=\sandwich{F(a)}\pb\paren{F(b)}$.

Quantum references as defined above allow us to lift unitaries or
projectors from the content space $\calH_A$ of a reference to its memory space
$\calH_M$. However, when reasoning about
quantum programs, we may also want to transport a subspace of
$\calH_A$ to a subspace of $\calH_M$. For example, predicates in Hoare
logic might be represented as subspaces,
 and we
want to express predicates on the whole program state in terms of
predicates about individual variables. This is easily achieved because
a subspace can be represented by the unique projector onto that
subspace, and those projectors can be lifted by~$F$.
(For more properties of this definition see \fullshort{\autoref{sec:lift.sub}}{the full version, \autoref{FV:sec:lift.sub}}.)

The proofs of the axioms from \autoref{fig:axioms} are quite simple, and they are anyway special cases of the infinite-dimensional ones given later, so we do not give them here.



\fullshort{
  \subsection{Example, specialized}
  \label{sec:ex.quantum}
}{
  \paragraph{Example, specialized.}
}%
In \autoref{sec:example}, we described how the semantics of a small programming language can be based on references.
That example was abstract, assuming a generic reference category.
\fullshort{

  In the quantum setting, there are several possibilities to make it concrete:

  \paragraph{Pure state semantics.}
  The simplest possibility is to describe a language where the semantics of programs are linear operators (i.e., the semantics map a pure state to a pure state).
  
}{%
  We now make this concrete:
  The simplest possibility is to describe a language where the semantics of programs are linear operators (i.e., the semantics map a pure state to a pure state).
  In the full version, we also describe this for program semantics based on mixed states, supporting, e.g., measurements etc.
}%
In this case, the types $T$ in the language could correspond to finite sets (e.g., $\mathtt{bool} \mathrel{\hat=} \bit$) with $\times$ being the Cartesian product, the corresponding Hilbert spaces are $\calH_T := \setC^T$ (e.g., $\calH_\mathtt{bool}$ is spanned by $\ket 0$, $\ket 1$), the update monoids $\mathbf T$ are the linear operators on $\calH_T$ (i.e., the matrices $\setC^{T\times T}$). The reference category is $\Lquantumfin$. 
We also restrict the $U$'s in the program syntax to unitaries.%
\footnote{More generally, operations $U$ with $\norm U\leq1$ would also work.
  $\norm U>1$ would be problematic; it leads to programs with termination probability $>1$.}
\fullonly\par
This gives us a language with commands  $\qapply U X$ for applying unitaries $U$ to references $X$.
The semantics of a program on memory of type $T$ is described by a linear operator on $\calH_T$.
\fullonly\par
The various programs and example calculations and results then carry over immediately to this concrete setting. (With $\Uswap\ket{x,y} := \ket{y,x}$.)

\delaytextsv{Mixed state semantics}{
\paragraph{Mixed state semantics.} 
If we want our language to support, e.g., measurements or random choices, pure state semantics are not sufficient.
In this case, we want the semantics to map mixed states (density operators) to mixed states.
That is, the semantics of a program are described by completely positive trace-reducing maps, a.k.a.~quantum subchannels.

Since the update monoids in the quantum reference category $\Lquantumfin$ consist of linear operators, not subchannels, we cannot use the same mathematical structure for the program semantics and the update monoids.
Instead we need to make use of the additional flexibility described in \autoref{sec:example.more.general}:

As before, the types $T$ in the language could correspond to finite sets, $\calH_T := \setC^T$, and the update monoids $\mathbf T$ are the linear operators on $\calH_T$. The reference category is $\Lquantumfin$.
But $\Sem_T$ is now defined as the space of completely positive maps (CPMs) $L(\calH_T)\to L(\calH_T)$ where $L(\calH_T)$ is the set of linear operators on $\calH_T$.
(This contains in particular the quantum subchannels.)

In particular, the semantics of a program with memory type $M$ will map a density operator over $\calH_M$ to a (sub)density operator over $\calH_M$.

To make things work, we additionally need to specify: a monoid homomorphism $\Sem:\mathbf T\to\Sem_T$ and a way to apply references to $\Sem_T$.
The monoid homomorphism is very natural: Any $U\in\mathbf T$ is mapped to the CPM $\rho\mapsto U\rho \adj U$. It is easy to check that this is a monoid homomorphism.

Applying a reference $X$ to $\Sem_T$ is less obvious.
We need to define $X^\bullet(\calE)$ for a CPM $\calE$ over $L(\calH_T)$ and a reference $X:\mathbf T\to\mathbf U$.
$X^\bullet(\calE)$ should be CPM on $L(\calH_U)$, and it should describe what happens on the overall quantum system if we apply $\calE$ to the subsystem $\calH_T$.
(E.g., if $X=\Fst$, then $X^\bullet(\calE)$ should be $\calE\otimes\id$.)
Such an $X^\bullet(\calE)$ is defined in \autoref{sec:quantum.channels}.
(Denoted simply $X(\calE)$ there.)
The required property $\forall a.\ X^\bullet(\Sem(a))=\Sem(X(a))$ is shown in \lemmaref{lemma:lift.channel:sem} there.

The various programs and example calculations and results from \autoref{sec:example} then carry over immediately to the concrete setting here.

Since we have mixed states semantics, we can now easily extend the language, e.g., with a command that measures a qubit: $\qif X1PQ$ measures whether the qubit in $X$ is $1$, and, if so, runs $P$ otherwise $Q$.%
\footnote{Typing: In the nano language, $X$ needs to be a reference $\mathbf T_\mathtt{bool}\to \mathbf T_M$.
  In the micro language, the following typing rule is added:
    \inferrule{T_M,T_L,\frakP \vdash P,\ Q
      \\
      X:\mathbf T_\texttt{bool}\to\mathbf T_M\otimes\mathbf T_L\text{ is a reference}
  }
  {T_M,T_L,\frakP \vdash \qif X1PQ}.} It would have the following semantics:
\[
  \pb\denot{\qif X1PQ}(\rho) := \denot P\pb\paren{X(\selfbutter 1)\rho X(\selfbutter 1)}
  + \denot Q\pb\paren{X(\selfbutter 0)\rho X(\selfbutter 0)}.
\]
}

\subsection{Quantum references -- infinite-dimensional case}
\label{sec:infinite}

Above, we instantiated the general theory of references in the quantum
case by considering finite-dimensional Hilbert spaces only. That is,
this approach does not allow us to consider, e.g., references of type
$\setC^{\setZ}$, i.e., an integer in superposition
(\emph{quint}\index{quint}).  Another example of something that does
not fit into the finite-dimensional framework is the position/momentum
example from the introduction.  We now describe how the instantiation
above needs to be changed to accommodate references whose types
correspond to arbitrary Hilbert spaces. In the finite-dimensional case, the
pairing axiom~\ref{ax:pairs} followed
immediately from the universal property of the tensor
product.\footnote{%
  Formulated in the setting of finite-dimensional Hilbert spaces, it says:
  For any bilinear function
  $F:\mathbf A\times\mathbf B \to\mathbf C$ (where
  $\mathbf A,\mathbf B,\mathbf C$ are the linear operators on the
  finite-dimensional Hilbert spaces $\calH_A,\calH_B,\calH_C$), there
  exists a unique linear function
  $\Hat F:\mathbf A\otimes\mathbf B\to\mathbf C$ such that
  $\Hat F(a\otimes b)=F(a,b)$ for all $a,b$.}  The problem in the
infinite-dimensional case is that this universal property does not, in
general, hold.\footnote{%
  What properties exactly hold depends very much on the spaces of
  operators we consider, as well as the specific variant of the tensor
  product. 
  \fullonly{See \autoref{sec:discussion} and \cite[Chapter 4]{takesaki} for a detailed exposition.}}
Fortunately, for the right
choices of objects and morphisms, we can still prove the pairing axiom without the usual universal property.

\paragraph{The quantum reference category \boldmath$\symbolindexmark\Lquantum\Lquantum$.}
The objects (update monoids) are the sets of bounded operators on Hilbert spaces. That is, for every Hilbert space $\calH$ (the type of a given reference), the space \symbolindexmark\bounded{$\bounded(\calH)$} of all bounded linear operators $\calH\to\calH$ is an object.
Throughout this paper, in the quantum setting, we follow the notational convention that $\mathbf A,\mathbf B,\dots$ always denotes $B(\calH_A),B(\calH_B),\dots$
($\bounded(\calH)$ is a monoid where the multiplication is the composition of operators and the unit is the identity operator $1$.)

The \emph{pre-references}\index{pre-references!quantum, infinite-dimensional}\index{quantum pre-references!infinite-dimensional}
$\mathbf A\to\mathbf B$ are the weak*-continuous bounded linear maps
from $\mathbf A$ to $\mathbf B$.
(The \emph{weak*-topology}\index{weak*-topology} on the space of bounded operators is the coarsest topology where $x\mapsto \abs{\tr ax}$ is continuous for all trace-class operators $a$ (\S20 in~\cite{conway00operator}). In other words, a net $x_i$ weak*-converges to $x$ iff for all trace class $a$, $\tr ax_i\to\tr ax$.
Weak*-continuous bounded linear maps are also known as \emph{normal}\index{normal map} maps \cite[Def.~III.2.15]{takesaki}.)

The tensor product on $\mathbf A,\mathbf B$ is the
\emph{tensor product of von Neumann algebras}%
\index{tensor product of von Neumann algebras}%
\index{Neumann algebra!tensor product of}, that is, $\mathbf A\symbolindexmark\tensor\tensor\mathbf B$ is the set $\bounded(\calH_A\otimes\calH_B)$ of
bounded operators on $\calH_A\otimes\calH_B$.  (In \cite{takesaki},
this tensor product is defined in Definition IV.1.3 as
``$\overline\otimes$''.
The fact that $\mathbf A\otimes\mathbf B=\bounded(\calH_A\otimes\calH_B)$ is given by (10) in the same section.)
For any $a\in\mathbf A$, $b\in\mathbf B$, $a\otimes b$ is defined as the
unique operator such that
$(a\otimes b)(\psi\otimes\phi)=a\psi\otimes b\phi$ for all
$\psi\in\calH_A,\phi\in\calH_B$.  (The tensor product on the rhs of
this equation is the Hilbert space tensor product \cite[Definition IV.1.2]{takesaki}.
Existence of $a\otimes b$ is shown in the discussion after that definition.)  

The \emph{references}\index{references!quantum, infinite-dimensional}\index{quantum references!infinite-dimensional} are
the weak*-continuous unital *-homomorphisms.  (Unital
means $F(1)=1$, $*$-homomorphism means linear, $F(ab)=F(a)F(b)$ and $F(\adj a)=\adj{F(a)}$.)

\medskip

\begin{fullversion}
  We give the proofs that $\Lquantum$ is a reference category in \autoref{sec:proofs.infdim}.

  \medskip

  Note that for finite-dimensional Hilbert spaces, the definitions here coincide with those from \autoref{sec:quantum-finite} since over finite-dimensional Hilbert spaces, all operators are bounded, and all linear maps are bounded and weak*-continuous, and the tensor product of von-Neumann algebras coincides with the algebraic tensor product.
\end{fullversion}

\begin{shortversion}
  We provide additional discussion of design choices, related concepts, and open questions in the full version, \autoref{FV:sec:discussion}
\end{shortversion}

\begin{fullversion}
  \paragraph{Example, specialized}
  The explanations from \autoref{sec:ex.quantum} how to instantiate the example from \autoref{sec:example} in the quantum case also apply in the infinite-dimensional settings, with very minor differences:
  The Hilbert space $\calH_T$ corresponding to type $T$ is the space $\ell_2(T)$ of square-summable sequences $T\to\setC$
  (instead of $\calH_T:=\setC^T$).
  The update monoids are the bounded operators on $\calH_T$
  (instead of the linear operators).
  Program semantics in the pure state semantics are bounded operators on $\calH_T$
  (instead of the linear operators).
  $\Sem_T$ consists of the completely positive \emph{bounded-linear} maps $\tracecl(\calH_T)\to \tracecl(\calH_T)$ where $\tracecl(\calH_T)$ are the trace-class operators over $\calH_T$
  (instead of the linear operators).
  Everything else applies unchanged.

  (Since for finite $T$, $\calH_T=\setC^T$, and the bounded and trace-class operators coincide with the linear operators, and all linear maps are bounded-linear, the instantiation from \autoref{sec:ex.quantum} is a special case of this instantiation.)
\end{fullversion}

\delaytextsv{quantum discussion}{
\fullshort{
  \subsection{Discussion}
}{
  \section{Discussion (quantum case)}
}
\label{sec:discussion}

\paragraph{Using other tensor products (in the infinite-dimensional case).}
In \autoref{sec:infinite}, we chose the tensor product of von Neumann algebras as the tensor product in our category.
As a consequence, the morphisms we considered were \emph{weak*-continuous} bounded linear maps. (This is because the von Neumann tensor product is the weak*-closure of the algebraic tensor product.)
However, this is not the only possible tensor product on spaces of bounded operators. Other tensor products exist that are the closure of the algebraic tensor product with respect to other topologies. For example, the projective C*-tensor product \cite[Definition IV.4.5]{takesaki}
satisfies the property required for Axiom~\ref{ax:pairs} without requiring \emph{weak*-continuous} maps \cite[Proposition IV.4.7]{takesaki}. So it is well conceivable that we can define a category of quantum references also with this tensor product. Or we could forgo any topological considerations and try to use the algebraic tensor product, and allow all linear maps (not just bounded ones). The reason why we chose the tensor product of von Neumann algebras is that it satisfies $\bounded(\calH_A)\otimes\bounded(\calH_B)=\bounded(\calH_A\otimes\calH_B)$.
This makes the resulting theory much more natural to use:
If we have two references $F,G$, with types $\calH_A$ and $\calH_B$, then it would be natural that the pair $\spair FG$ has type $\calH_A\otimes\calH_B$. E.g., a unitary $U:\calH_A\otimes\calH_B\to\calH_A\otimes\calH_B$ is something we expect to be able to apply to the reference $\spair FG$. It seems that the tensor product of von Neumann algebras is the only one that gives us this property.
However, we do not exclude that instantiations of the theory of references with other tensor products might work and give us other advantages that we are not currently aware of.

\paragraph{Mixed classical/quantum references.}
Another interesting question is to extend the theory of quantum references to mixed classical/quantum references.
For example, we could have a reference that contains a list of qubits, where the length of the list is classical.
Our current formalism does not allow that.
(We can only model references that contain lists of qubits where the length of the list is also in superposition.) In particular, with the current theory, classical references (see \autoref{sec:classical}) are not a special case of quantum ones.
At a first glance, it would seem that there is already a ready-made solution available for this problem: Mixed classical/quantum systems are oftentimes represented by von Neumann algebras (i.e., certain subsets of the sets of all bounded operators).
For example, a system consisting of $n$ qubits would be represented by $\bounded(\setC^{2^n})$, same as we do in the present paper.
And a system of classical lists of qubits then would be the direct sum of all possible systems of $n$-qubits for $n\geq 0$. That is, $\bigoplus_{n=0}^\infty \bounded(\setC^{2^n})$.
And a classical bit could be $\bounded(\setC)\oplus\bounded(\setC)$.
In fact, classical systems are exactly the commutative von Neumann algebras!
Unfortunately, this approach seems unsuitable in our setting.

To see this, consider the special case of classical systems where the program state is purely classical.
I.e., it is modeled by some commutative von Neumann algebra~$\mathbf C$.
But that means that for any two references $F:\mathbf A\to\mathbf C$, $G:\mathbf B\to\mathbf C$, the ranges of $F$ and $G$ commute trivially.
Hence any two references in a program with classical state would be disjoint.
This shows us that we are on the wrong track.
The intuitive meaning of disjointness is that $F,G$ corresponds to different parts of the program state, and this should clearly not be the case for any two references. (In particular, $F$ should not be disjoint from itself except in the degenerate case where $F$ has unit type.)

Furthermore, modeling classical references as commutative von Neumann algebras also fails to represent the intuitive meaning of a ``set of updates''. Updates on a classical reference should in some way correspond to classical functions on some set $X$. In particular, different updates should not always commute. But in a commutative von Neumann algebra, everything commutes.

(Note that we have only illustrated why the approach fails for classical data. But since classical data is a necessary special case of mixed quantum/classical data, this also implies that the approach does not work for mixed data.)

What went wrong? And why is the von Neumann algebra formalism (as sketched above) successful in modeling mixed classical/quantum data in other settings? The reason, as we understand it, is that this von Neumann algebra formalism is based on a very different intuition: Elements of the von Neumann algebra represent \emph{observables}\index{observable}. For example, in $\bounded(\calH)$, we find the Hermitian operators $a$ on $\calH$, and each such operator gives rise to a function from quantum states to real numbers, via $\tr a\rho$ or $\adj\psi a\psi$, depending on whether we consider a mixed state $\rho$ or a pure state $\psi$. Physically, this corresponds to the expected value of the real-valued measurement described by $a$. And a classical system with values $X$ is described by the Neumann algebra consisting of all bounded complex sequences $a_x$ ($x\in X$). For any classical ``state'', namely a discrete\footnote{We assume \emph{discrete} measure spaces to keep this example simple.} probability distribution $\mu$ on $X$, any (real-valued) $a_x$ also gives rise to a real number via $\sum_x\mu(x)a_x$, namely the expectation of $a_x$ where $x$ is the state of the system.

So an element of a von Neumann algebra (in this formalism) describes a \emph{measurement} on the system.
I.e., it tells us how to \emph{read} from it. In contrast, in our setting, we also need to know how to \emph{write} to the system (since we describe ``updates'').
This is basically the same difference as between rvalues and lvalues in classical programming languages. Knowing how to read a reference is not, in general, sufficient for knowing how to write it.\footnote{For example, for a classical reference, we can describe how to read it by giving a function $f:M\to X$ from program states $M$ to reference contents $X$. Yet, from this function $f$, it is not possible to reconstruct how updates on the reference should happen.}

So we claim that the fact that our quantum reference category has the
same objects and morphisms as existing work that uses the von Neumann algebra formalism (e.g., \cite{kornell17quantum,cho16neumann,pechoux20quantum} to mention just a few) is coincidental and only happens when we consider purely quantum systems. Namely, in our setting, the objects are the linear span of all unitary transformations (a unitary is a natural representation of an ``update''), while in the other work, the objects are the linear span of all Hermitian operators (a Hermitian is a natural representation of a real-valued measurement). Those two spans happen to coincide, in both cases we get all bounded operators. However, one should be careful to read too much into this coincidence. In particular, if we leave the purely quantum setting, we do not have any such correspondence anymore. (Consider for example the stark difference of our modeling of classical references (partial functions on sets, see \autoref{sec:classical}) and the modeling of classical systems via commutative von Neumann algebras (complex sequences).)\footnote{%
  One further thing that seems to be coincidence is the use of weak*-continuous unital *-isomorphisms. For example, in \cite[Section 5.4]{pechoux20quantum}, a denotational semantics for quantum ``values'' $v$ backed by a memory of qubits is given.
  The underlying intuition seems to be that it describes how an observable on $v$ is transformed into an observable on the program state.
  This denotational semantics $\llbracket v\rrbracket$ of $v$ is given by a weak*-continuous unital *-isomorphism from the von Neumann algebra $\mathbf V$ corresponding to $v$ to the von Neumann algebra corresponding to a collection of qubits (i.e., $\bounded(\setC^{2^n})$), \emph{exactly as in our category}.
  
  Yet we claim that this is another coincidence due to the fact that we consider only pure quantum systems. Namely, in \cite{pechoux20quantum}, $\llbracket v\rrbracket$ is not always an injective function. E.g., $\llbracket\mathbf{left}\dots\rrbracket$ (left injection into a sum type) is noninjective. On the other hand, in our setting, noninjective references would make little sense. This would mean that two updates that have different effects on a specific reference can have the same effect on the overall state. So, in fact, one could argue that references in our setting should actually be described as \textit{injective} weak*-continuous unital *-isomorphisms. It just so happens that in our case (where the only von Neumann algebras we consider are the sets of all bounded operators),
  \textit{injective} weak*-continuous unital *-isomorphisms coincide with
  weak*-continuous unital *-isomorphisms (\autoref{lemma:reg.injective}).
  This is not the case in general!
  So the fact that our set of morphisms coincides with those from \cite{pechoux20quantum} is again just an artifact of the fact that we consider only pure quantum systems.

  (As a final remark note that the modeling from \cite{pechoux20quantum} also does not give rise to a way of pairing two values in the same way as we do with references. In \cite{pechoux20quantum}, when combining two values into a pair, the two values need to be specified on distinct collections of qubits that then get concatenated. See the fifth rule in Figure 7 in \cite{pechoux20quantum}.
  This is more similar to our tensor product of references than to our pairing operation.)
}

So in light of this, to formalize mixed classical/quantum references, we will need to use a different category than that of von Neumann algebras. At this point, it is not clear yet what this category looks like.
}

\delaytextsv{classical refs}{
  \newcommand\INCLUDEONCEfagerighoutsphutvsovtrbgid{}

\section{Classical references}
\label{sec:classical}

We come to the classical case. The most obvious approach would be to say that the updates on a reference taking values in $A$ are the functions $A\to A$. However, it is not possible to satisfy the axioms of a reference category for this choice.\fullonly{\footnote{%
  If updates are the function sets $A\to A$, and the tensor product is defined in the natural way as $\mathbf A\otimes\mathbf B := (A\times B)\to(A\times B)$, $(a\otimes b)(x,y):=(a(x), b(y))$, the axioms are contradictory for any definition of pre-references/references:
  Let $A:=\bit$.
  Let $F_i:\mathbf A\otimes\mathbf A\to
  \mathbf A\otimes\mathbf A$ for $i=0,1$ be defined as $F_i(z):=f\circ z\circ g_i$ where $f(xy):=x0$ and $g_i(00):=0i$ and $g_i=00$ everywhere else.
  By Axioms~\ref{ax:preregs} and~\ref{ax:cdot-a},
  $F_i$ are pre-references. Let $z(xy):=yx$. Then $F_i(z)(00)=i0$, so $F_0\neq F_1$.
  For any $z_1,z_2\in\mathbf A$, we have
  $F_i(z_1\otimes z_2)(xy)=f\circ(z_1\otimes z_2)(00)$ for $xy\neq 00$ and $F_i(z_1\otimes z_2)(00)=z_1(0)0$. Thus $F_0(z_1\otimes z_2)=F_1(z_1\otimes z_2)$ for all $z_1,z_2$.
  By Axiom~\ref{ax:tensorext}, this implies $F_0=F_1$ in contradiction to $F_0\neq F_1$ above.}} Instead, we define the updates on $A$ as the \emph{partial} functions $A\symbolindexmark\partialto\partialto A$.
This choice has the additional advantage that it also directly provides support for partiality in the program semantics; the semantics of a deterministic partial program is often naturally described as a partial function.
\anonymous{}{(An earlier version \cite{arxiv-v1} of this work
  formalized classical updates as relations on $A$ instead. We believe that our new approach is more natural and easier to use. Both approaches are viable, however.)
}In this section, we use the notational
convention that $\mathbf A,\mathbf B,\dots$ denote the sets of partial functions on $A,B,\dots$.

To define references $\mathbf A\to\mathbf B$, let us forget for a moment that a reference should map updates to updates.
Instead, we note that a classical reference with values in $A$ inside a memory of type $B$ naturally allows us to perform two actions:
We can read the value of the reference using a function $g:B\to A$ (the \emph{getter}\index{getter}).
And we can set the value of the reference using a function $s:A\times B\to B$ (the \emph{setter}\index{setter}).
The getter and setter satisfy some natural conditions: We call $(g,s)$ \emph{valid}\index{valid!getter/setter} iff for all $a,a'\in A$ and $b\in B$, we have
$b = s (g(b), b)$ and $g (s(a, b)) = a$ and $s(a, s(a', b)) = s(a, b)$.
This is also known as a lens.
A lens is a natural candidate for formalizing classical references.
We show that this approach indeed fits in our formalism.
A getter/setter pair $(g,s)$ naturally gives rise to a function $F:\mathbf A\to\mathbf B$ via $F(a)(b) := s(a(g(b)),b)$. (I.e., to apply $a$ to the reference, we retrieve its content $g(b)$, apply $a$ to it, and set the result as the new content of the reference.) Note that $F(a)$ can be a partial function if $a$ is, but $F(a)$ is total for total $a$. Thus the \emph{references}\index{reference!classical}\index{classical reference} are the functions $F:\mathbf A\to\mathbf B$ such that there exist valid $(g,s)$ with $\forall a\in\mathbf A, b\in B. \ F(a)(b) = s(a(g(b)),b)$.

In order to define pre-references, we somewhat relax these conditions. We cannot define the pre-references $\mathbf A\to\mathbf B$ as all functions  $\mathbf A\to\mathbf B$ since then Axiom~\ref{ax:tensorext} would not hold. Instead, we define 
\emph{pre-references}\index{pre-reference!classical}\index{classical pre-reference} as the functions $F:\mathbf A\to\mathbf B$ such that there exist a total $g:B\to A$ and a partial $s:A\times B\partialto B$ with $\forall a\in\mathbf A, b\in B. \ F(a)(b) = s(a(g(b)),b)$. (Note: we removed the validity requirement, and we allow $s$ to be partial. With total $s$, Axiom~\ref{ax:cdot-a} would not hold.)

Finally, we need to define the tensor product of updates. $\mathbf A\otimes\mathbf B$ is the set of partial functions on $A\times B$. (This is the only natural choice since we want a pair of references taking values in $A,B$, respectively, to take values in $A\times B$.) Then $(a\otimes b)(x,y) := (a(x), b(y))$, defined iff both $a(x),b(y)$ are defined, makes this into a tensor product satisfying all our axioms. 

Note that the pairs $\spair FG$ in this category are very natural, too: If $F:\mathbf A\to\mathbf C$, $G:\mathbf B\to\mathbf C$ are defined by the getters/setters $(g_F,s_F)$ and $(g_G,s_G)$, respectively, then $\spair FG$ is defined
by the getter/setter $(g,s)$ with $g(c)=(g_F(c),g_G(c))$ and $s((a,b),c)
= s_F(a,s_G(b,c))$.

With these choices, we get a reference category $\symbolindexmark\Lclassical\Lclassical$ satisfying all axioms from \autoref{fig:axioms} and the laws from \autoref{fig:laws} follow.

Furthermore, for singleton sets $A$, have a \emph{unit reference}\index{unit reference!classical} $u:\mathbf A\to\mathbf B$ that is disjoint from all references, and such that setting it has no effect. 

We formalized the proofs for this section in Isabelle/HOL (\autoref{sec:isabelle}), so we do not give them here.

\subsection{Example, specialized}
\label{sec:ex.classical}

In \autoref{sec:example}, we described how the semantics of a small programming language can be based on references.
That example was abstract, assuming a generic reference category.

In the classical setting, there are several possibilities to make it concrete:

\paragraph{Deterministic semantics.}
The simplest possibility is to describe a language where the semantics of programs are deterministic.

In this case, the types $T$ in the language could correspond to finite sets (e.g., $\mathtt{bool} \mathrel{\hat=} \bit$) with~$\times$ being the Cartesian product and the update monoids $\mathbf T$ are the partial functions $T\partialto T$.
 The reference category is $\Lclassical$.

This gives us a language where $\qapply U X$ means $X := U(X)$ for a function $U$.
(Of course, in the classical case a notation such as $X := f(X)$ would be more intuitive.)
The semantics of a program on memory of type $T$ is described by a partial function $T\partialto T$.
(Here it comes in handy that $\Lclassical$ uses \emph{partial} functions in the update monoids since total functions would restrict us to a language where we cannot have non-terminating programs.)

The various programs and example calculations and results then carry over immediately to this concrete setting. (With $\Uswap(x,y) := (y,x)$ being a function.)

\paragraph{Probabilistic semantics.}
If we want our language to support, e.g., random choices, deterministic semantics are not sufficient.
In this case, we want the semantics to map a program state to a distribution over program states.
That is, it is of type $T\to\calD_T$ where $\calD_T$ is the set of functions $f:T\to\setR_{\geq0}$ with $\sum_{x\in T}f(x)\leq 1$.
For simplicity, we consider only programs where all types are finite (i.e., $\abs T<\infty$ for all types $T$).

These ``probabilistic functions'' $f,g:T\to\calD_T$ naturally form a monoid with multiplication $(f\cdot g)(x)(z):=\sum_y g(x)(y)f(y)(z)$ and unit $1\in T\to\calD_T$ with $1(x)(x):=1$, $1(x)(y):=0$ ($x\neq y$).
This multiplication is the monadic bind and $f\cdot g$ intuitively corresponds to the composition of probabilistic functions.

Since the update monoids in the reference category $\Lclassical$ consist of partial functions, not probabilistic functions, we cannot use the same mathematical structure for the program semantics and the update monoids.
Instead, we need to make use of the additional flexibility described in \autoref{sec:example.more.general}:

As before, the types $T$ in the language could correspond to finite sets, and the update monoids $\mathbf T$ are the partial functions $T\partialto T$. The reference category is $\Lclassical$.
But $\Sem_T$ is now defined as the space of probabilistic functions $T\to\calD_T$.

In particular, the semantics of a program with memory type $M$ will be a probabilistic function $M\to\calD_M$, mapping memories to distributions over memories.

As described in \autoref{sec:example.more.general}, we additionally need to specify:
a monoid homomorphism $\Sem:\mathbf T\to\Sem_T$, and a way to apply references to $\Sem_T$.
The monoid homomorphism corresponds to interpreting a partial function $f$ as a probabilistic function $\Sem(f)$:
$\Sem(f)(x)(y) := 1$ if $f(x)=y$ and $:=0$ if $f(x)\neq y$.
It is easy to check that this is a monoid homomorphism.

Applying a reference $X$ to $\Sem_T$ is possible, too:
Let $\Hat\calD_T$ denote all functions $T\to\setR$ (dropping the positivity and the $\leq 1$ requirements).
Then the functions $T\to\Hat\calD_T$ form a real vector space space.
And the probabilistic functions $T\to\calD_T$ are a subset of the $T\to\Hat\calD_T$.
Furthermore, the functions $\Sem(x\mapsto y)$ with $x,y\in T$ form a basis of $T\to\calD_T$.
(Here $x\mapsto y$ is the partial function defined only at $x$, mapping it to $y$.)
We define $X^\bullet$ to be the unique linear function on $T\mapsto \Hat\calD_T$ with $X^\bullet(\Sem(x\mapsto y)):=\Sem(X(x\mapsto y))$.
Since $T\to\calD_T$ is the convex hull of the $\Sem(x\mapsto y)$, and $\Sem(X(x\mapsto y))\in T\to\calD_T$,
$X^\bullet$ maps $T\mapsto \calD_T$ to $T\mapsto \calD_T$.
And we immediately that the required condition $X^\bullet(\Sem(f))=\Sem(X(f))$ holds for $f$ of the form $x\mapsto y$.
We encourage the reader to check that this also holds for all other $f$ (because they are a union of functions of the form $x\mapsto y$).

The various programs and example calculations and results then carry over immediately to this concrete setting.

Since we have probabilistic semantics, we can now easily extend the language with, e.g., random bits: $X\leftarrow\mathbf{coin}$ assigns a random bit to $X$.%
\footnote{{Typing: In the nano language, $X$ needs to be a reference $\mathbf T_\mathtt{bool}\to \mathbf T_M$.
  In the micro language, the following typing rule is added:
    \inferrule{X:\mathbf T_\texttt{bool}\to\mathbf T_M\otimes\mathbf T_L\text{ is a reference}}
    {T_M,T_L,\frakP \vdash X\leftarrow\mathbf{coin}}.}}
It would have the following semantics:
\[
  \pb\denot{X\leftarrow \mathbf{coin}} :=
  X^\bullet(f_\mathit{coin})
  \qquad
  \text{with}
  \qquad
  \forall x,y.\ f_\mathit{coin}(x)(y):=\frac12.
\]

}

\section{Quantum Hoare logic}
\label{sec:hoare}

To illustrate the use and flexibility of our references, we will now
develop a simple quantum Hoare logic.
Since the main focus is to show the use of quantum variables (and constructions based upon these), the language we analyze will not contain any features that are orthogonal to the question of variables, e.g., loops.
The ideas
described can be easily extended to other Hoare logics that use
subspaces as pre-/postconditions  (i.e., von-Neumann-Birkhoff quantum logic \cite{birkhoff36logic};
used, e.g., in \cite{qrhl, zhou19applied, ghosts, Li2020}).  (And Hoare logics
that use operators as pre-/postconditions \cite{dhondt06weakest, feng07proof, ying12floyd, expectation-qrhl, barthe19relationalquantum}
should be even easier to adapt because references
naturally translate operators from variables to overall program
states.)

\paragraph{Language definition.} Our simple language has two
commands: applying a unitary, and branching based on a measurement.
Fix a space $\calH_M$ for the program state.
\emph{Statements}\index{statement} are of the form:
\[
 P,Q ::= \skipprog \ |\ P;Q \ |\ \qapply UF \ |\  \qif GxPQ
\]
We assume that the following type constraints are satisfied: $U$ is a
unitary on some Hilbert space $\calH_A$; $F$ is a reference $\mathbf A\to\mathbf
M$ for the same $\calH_A$; $x\in B$ for some $B$ with $\abs B=2$; and $G$ is a reference $\mathbf B\to \mathbf M$ where $\calH_B:=\setC^B$.
(Here $\mathbf M$ is fixed throughout and corresponds to the program state.)

We use pure state semantics. That is, given a pure state $\psi\in\calH_M$, \symbolindexmark\denot{$\denot{P}$} maps $\psi$ to a list of pure states $\psi_i\in\calH_M$.
The interpretation is that the final state is $\psi_i/\norm{\psi_i}$ with probability $\norm{\psi_i}^2$.
Or equivalently that the final state is described by the density operator $\sum_i\psi_i\adj{\psi_i}$.

$\symbolindexmark\skipprog\skipprog$ denotes the empty program.
Hence $\denot\skipprog(\psi):=(\psi)$ where $(\psi)$ denotes a one-element list.

$P;Q$ the execution of $P$ followed by $Q$.
That is, if $\denot P(\psi)=(\psi_1,\dots,\psi_n)$, then $\denot{P;Q}(\psi):=\denot{Q}(\psi_1)\Vert\dots\Vert\denot{Q}(\psi_n)$.
($\Vert$ is the concatenation of lists.)

The command \symbolindexmark\qapply{$\qapply UF$} applies the unitary $U$ to the variable~$F$.
Recall that, since $F$ is a quantum reference, $F(U)$ describes the operation applied to the overall program state when we apply $U$ to the content of $F$.
Thus the denotational semantics is simply $\denot{\qapplyOLD FU}(\psi):=(F(U)\psi)$.
Note that we do not need to define a statement $\qapply U{F,G}$ that applies a unitary on two registers since we can get the same effect by just applying $U$ to the pair $\spair FG$.
That is, without having any explicit provision for it in the definition of the language, we can write $\qapply U{\spair FG}$ to apply $U$ to more than one variable.

And \symbolindexmark\qif{$\qif GxPQ$} measures $G$ in the computational basis, and if the outcome is $x$, executes $P$, and $Q$ otherwise.

The denotation is $\denot{\qif GxPQ}(\psi) := \denot P(G(\selfbutter x)\psi) \ \big\Vert\ \denot Q(G(\selfbutter{y}))$ where  $y\neq x$.
(Recall that $x\in B$ for some $\abs B=2$, so $y\neq x$ is uniquely defined.)
Note that $G(\selfbutter x)\psi$ is the non-normalized post-measurement state after measuring the content of $G$ with outcome $x$.

\paragraph{Hoare judgments.}\index{Hoare judgment} For subspaces $A,B$ of $\calH_M$ and a
program $P$, we write \symbolindexmark\hoare{$\hoare APB$} to mean: for all $\psi\in A$,
$\denot P(\psi)\subseteq B$.
(Recall that $\denot P(\psi)$ is the list of all possible final states.)

The following rules follow directly from the semantics:%
\index{Seq@\textsc{Seq}}\index{Weaken@\textsc{Weaken}}%
\index{Skip@\textsc{Skip}}\index{Apply@\textsc{Apply}}%
\index{If@\textsc{If}}
\begin{gather*}
  \inferrule[Seq]{\hoare A{P_1}B\\\hoare B{P_2}C}{\hoare A{P_1;P_2}C}
  \qquad
  \inferrule[Weaken]{A \subseteq A' \\ B' \subseteq B \\ \hoare {A'}P{B'}}{\hoare APB}
  \qquad
  \inferrule[Skip]{A \subseteq B}{\hoare A\skipprog B}
  \\[4pt]
  \inferrule[Apply]{F(U)\cdot A \subseteq B}{\hoare A{\qapply UF}B}
  \qquad
  \inferrule[If]{
    \pb\hoare{R(\selfbutter x)\cdot A}P{B} \\
    \pb\hoare{R(\selfbutter y)\cdot A}Q{B} \\
    x\neq y
  }{\hoare A{\qif GxPQ}B}
\end{gather*}
Note how in the last two rules, references are used not only in programs, but
also to describe pre-/postconditions. We do not need to introduce any
additional definitions for that since $F(U)$ and
$G\pb\paren{P_S}$ already have meaning as operators by
definition of quantum references. (And $\cdot$ means the multiplication
of an operator with a subspace.
And $S^\complement$ is the complement of $S$.)

We can also very easily express the predicate ``the content of reference $F$ has state $\psi$'' (short $F\quanteq\psi$).
Having state~$\psi$ is represented by the subspace $\Span\braces\psi$.
Or equivalently, $\im\psi\adj\psi$.
And $F$ transports this to a subspace of $\calH_M$.
Namely: for a reference $F:\mathbf A\to\mathbf M$ and a state
$\psi\in\calH_A$, we define $F\symbolindexmark\quanteq{\!\strut\quanteq} \psi$ as the subspace
$\im F(\psi\adj\psi)$ of $\calH_M$.
(Recall the discussion about subspaces in \autoref{sec:quantum-finite} about predicates and references, and see also the more detailed information in \fullshort{\autoref{sec:lift.sub}}{the full version, \autoref{FV:sec:lift.sub}}.)

Finally, we will often need the intersection of subspaces. See \fullshort{\autoref{sec:lift.sub}}{the full version (\autoref{FV:sec:lift.sub})} for facts for reasoning about this in
terms of references and projectors.

In the next section, we show what all of this looks like in a concrete
situation.

\delaytextsv{teleport}{
  \newcommand\INCLUDEONCEfsiqwowofldsaiotwer{}

\subsection{Teleportation}
\label{sec:teleport}

We illustrate the use of quantum references and of our minimal Hoare logic by analyzing quantum teleportation~\cite{bennett93teleport}.
We chose this example because, while the program is simple, it uses a number of variables that interact nontrivially through unitaries and measurements.
As a quantum circuit, it is illustrated in \autoref{fig:teleport} left.
The overall effect of
this quantum circuit is to ``move'' the qubit from wire $X$ to wire
$\Phi_2$, assuming $\Phi_1,\Phi_2$ are preinitialized with an EPR
pair $\symbolindexmark\bell\bell=\frac1{\sqrt2}\ket{00} + \frac1{\sqrt2}\ket{11}$.  Note
that there are also two wires $A,B$ that represent any additional
state that Alice and Bob may hold.

\begin{figure}
  \centering
  \begin{tabular}{cc}
    \small
    \begin{tikzpicture}[baseline=(current bounding box.center)]
      \initializeCircuit
      \newWires{X,A,Phi1,Phi2,B}
      \draw[thick, decoration={brace,mirror,amplitude=2mm}, decorate]
      ($(\getWireCoord{Phi1}) + (-1mm,1mm)$) -- ($(\getWireCoord{Phi2}) + (-1mm,-1mm)$)
      node [pos=0.5, anchor=east, xshift=-2mm] {\small$\bell$};
      \stepForward{2mm}
      \labelWire[\tiny$X$]{X}
      \labelWire[\tiny$A$]{A}
      \labelWire[\tiny$B$]{B}
      \labelWire[\tiny$\Phi_1$]{Phi1}
      \labelWire[\tiny$\Phi_2$]{Phi2}
      \stepForward{2mm}
      \drawWires{A,B}
      \node at (\getWireCoord{A}) {\tiny/};
      \node at (\getWireCoord{B}) {\tiny/};
      \stepForward{1mm}
      \node[cnot=Phi1, control=X, inner sep=0pt, minimum size=2.5mm] (cnot) {};
      \stepForward{3mm}
      \node[gate=X] (H) {$\hada$};
      \node[gate={Phi1}, inner sep=.5mm] (MPhi1) {$\centernot\frown$};
      \node[above right=0mm and 3mm of MPhi1, inner sep=0pt] (a) {\tiny$a$}; \draw[double,->] (MPhi1) -- (a);
      \stepForward{2mm}
      \node[gate={X}, inner sep=.5mm] (MX) {$\centernot\frown$};
      \node[below right=0mm and 3mm of MX, inner sep=0pt] (b) {\tiny$b$}; \draw[double,->] (MX) -- (b);
      \node[gate=Phi2, label={[inner sep=1pt, scale=.7]below:\tiny if $a=1$}] (X) {$\pauliX$};
      \stepForward{3mm}
      \node[gate=Phi2, label={[inner sep=1pt, scale=.7]below:\tiny if $b=1$}] (Z) {$\pauliZ$};
      \stepForward{2mm}
      \drawWires{X,Phi1,Phi2,A,B}
    \end{tikzpicture}
    \hspace*{.6in}
    &
    \small 
      \begin{tabular}[c]{l}
        $      \teleport := $ \\[5pt]
        \quad$\qapply \CNOT {\pair X{\chain\Phi\Fst}};$ \\
        \quad$  \qapply \hada X;$ \\
        \quad$  \qif {\chain\Phi\Fst} {1} {\braces{\qapply {\pauliX} {\chain\Phi\Snd}}} \skipprog;$ \\
        \quad$  \qif X {1} {\braces{\qapply {\pauliZ} {\chain\Phi\Snd}}} \skipprog$
      \end{tabular}
  \end{tabular}
  \caption{%
    Teleportation circuit and program. $\symbolindexmarkhighlight\hada$ is the Hadamard matrix, $\symbolindexmarkhighlight\pauliX$ is the Pauli-$X$ matrix, and $\symbolindexmarkhighlight\pauliZ$ is the Pauli-$Z$ matrix.}
  \label{fig:teleport}
  \symbolindexmarkonly\hada
  \symbolindexmarkonly\pauliX
  \symbolindexmarkonly\pauliZ
\end{figure}

In our simple language, we model this with the program
\symbolindexmark\teleport{$\teleport$}, see \autoref{fig:teleport}.
Note
that in this modeling, for the fun of it we decided that instead of
having two qubit variables $\Phi_1,\Phi_2$, we have a single two-qubit
variable $\Phi$, and we access their subparts as
$\chain\Phi\Fst,\chain\Phi\Snd$.  We assume that $X,A,B,\Phi$ are pairwise
disjoint. Besides that, we make no assumptions about how they are
embedded into the overall program memory.

The claim is that teleportation ``moves'' a state $\psi$ from $X$ to
$\chain\Phi\Snd$, if $\Phi$ is initialized with an EPR pair. Thus,
with the notation from the previous section, this would be expressed
as
$\hoare{X\quanteq\psi \cap \Phi\quanteq\bell}{\teleport}{\chain\Phi\Snd\quanteq\psi}$. However, this only states that if
$X$ contains some state $\psi$, not entangled with anything else, then
$\chain\Phi\Snd$ will contain it.
It does not cover the case where $X$ is entangled with, say, subsystems $A$ and $B$ in Alice's and Bob's labs.
To express a stronger (and harder to analyze in Hoare logic) statement, we can say: If the joint
state of $X,A,B$ is $\psi$, then the joint state of
$\chain\Phi\Snd,A,B$ is $\psi$ afterwards. In other words, we want to show:
\begin{equation}
   \pb\hoare{\XAB\quanteq\psi \cap \Phi\quanteq\bell}{~\teleport
  ~}{\PhiTwoAB\quanteq\psi}
  \label{eq:hoare.teleport}
\end{equation}

\noindent We now prove \eqref{eq:hoare.teleport}.
First, we have
\begin{equation*}
  {\XAB\quanteq\psi \,\cap\, \Phi\quanteq\bell}
  \quad=\quad
  \underbrace{\Phi\pb\paren{\bell\bell^\dagger}}_{=:\,O_1} \cdot \underbrace{\XAB\quanteq\psi}_{=:\,\pre}
\end{equation*}
This follows from the definition of $\quanteq$ and the fact that
$F(S) \cap \im G(V) = F(P_S)G(V)$ where $P_S$ is the projector onto $S$ (\lemmaref{lemma:lift.sub:inter.disjoint}).
Next we have
\begin{gather*}
  \pb\hoare
  {O_1\cdot\pre}
  {~\qapply\CNOT\XPhiOne~}
  {\singledouble{\underbrace{\XPhiOne(\CNOT)\cdot O_1}_{=:\,O_2}}{O_2}\cdot\,\pre}
  \\
  \singledouble{}{\qquad\qquad\text{with }O_2 := {\XPhiOne(\CNOT)\cdot O_1}\\[5pt]}
  \pb\hoare
  {O_2\cdot\pre}
  {~\qapply\hada X~}
  {\underbrace{X(\hada)\cdot O_2}_{=:\,O_3}\cdot\,\pre}
\end{gather*}
Both judgments follow directly from the Hoare rule \textsc{Apply}.

Furthermore, we have $O_3 = \sum_{a,b} \pb Q_{abab}$ where
$Q_{aba'b'} :=\frac12 \PhiTwo(\adj{(\pauliZ^{b'}\pauliX^{a'})}) \cdot \XPhiTwo(\Uswap)
\cdot \Phi\pb\paren{\ket{a,b}\adj\bell}$ and $\symbolindexmark\Uswap\Uswap$  is the two-qubit swap unitary ($U_\sigma\ket{a,b}=\ket{b,a}$).

This can be shown by rewriting both sides of the equation into the form $\XPhi(\dots)$ and then comparing the arguments of $\XPhi$.
That comparison is possible by explicit computation because the arguments are $8\times8$ matrices.
This is automated in Isabelle/HOL.
From  $O_3 = \sum_{a,b} \pb Q_{abab}$ we can conclude:
\[
  O_3\cdot\pre = \sum_{a,b} \pb\paren{Q_{abab}\cdot\pre}.
\]
Here the $\sum$ is the sum of closed subspaces, i.e., the least upper bound.
(Applying an operator to a subspace is distributive in the left argument.)

We stress that we could not have directly proven the latter equation by direct computation because the definition of $\pre$ involves the infinite-dimensional references $A,B$.  Nor could we have directly computed $O_3$ without rewriting it in the form $\pair X\Phi(\dots)$ because $O_3$ is an infinite-dimensional operator.

Note the interpretation of the postcondition $\sum_{a,b} \pb\paren{Q_{abab}\cdot\pre}$:
It means the state is in a superposition of states satisfying $Q_{abab}\cdot\pre$ for different $a,b$.
And each $Q_{abab}\cdot\pre$ means:
The current state is the result of having $\bell$ in reference $\Phi$ and $\psi$ in $\XAB$ (as in $\pre$),
then changing $\bell$ into $\ket{ab}$ (so that $\Phi$ contains $\ket{ab}$),
Then $X$ and $\PhiTwo$ are swapped, and the inverse of the Pauli matrices $\pauliZ^{b}\pauliX^{a}$ is applied.
So \emph{except for the application of the Paulis}, we already have the state $\psi$ in $\PhiTwoAB$ as desired in the end.
And the content of $\Phi$ tells us which Paulis need to be undone.

We next show
\begin{equation}
  \label{eq:cond.X}
  \pB\hoare{\sum_{a,b} \pb\paren{Q_{abab}\cdot\pre}}
  {\ \qif {\chain\Phi\Fst} {1} {\braces{\qapply {\pauliX} {\chain\Phi\Snd}}} \skipprog\ }
  {\sum_{a,b} \pb\paren{Q_{ab0b}\cdot\pre}}
\end{equation}
That is, we claim that in the state after this program fragment, the superfluous $\pauliX$ is not applied anymore, as desired.
(Note the index $0$ in $Q_{ab0b}$ in the postcondition.)
To show this, by rule \textsf{If}, we need to show the following two facts:
\begin{gather}
  \pB\hoare{\PhiOne(\selfbutter1) \cdot \sum_{a,b} \pb\paren{Q_{abab}\cdot\pre}}
  {\ \qapply {\pauliX} {\PhiTwo}\ }
  {\sum_{a,b} \pb\paren{Q_{ab0b}\cdot\pre}}
  \label{eq:cond.X.1}
  \\
  \pB\hoare{\chain\Phi\Fst(\selfbutter0) \cdot \sum_{a,b} \pb\paren{Q_{abab}\cdot\pre}}
  {\ \skipprog\ }
  {\sum_{a,b} \pb\paren{Q_{ab0b}\cdot\pre}}
  \label{eq:cond.X.0}
\end{gather}
To show these, we first establish the following auxiliary facts:
\begin{align*}
  \PhiOne(\selfbutter{a'}) \cdot Q_{abab} &= \delta_{aa'}\, Q_{abab} && \hskip-2cm \forall aa'b
  \\
  \PhiTwo(\pauliX) \cdot Q_{ab1b} &= Q_{ab0b} &&  \hskip-2cm  \forall ab
\end{align*}
where $\delta_{aa'}=1$ if $a=a'$, and $0$ otherwise.
The first follows by unfolding the definition of $Q_{abab}$, commuting $\PhiOne(\selfbutter{a'})$ next to the $\Phi\pb\paren{\ket{a,b}\adj\bell}$ term, and then merging those two terms into  $\Phi\pb\paren{(\selfbutter{a'}\otimes 1)\ket{a,b}\adj\bell}=\delta_{aa'}\Phi\pb\paren{\ket{a,b}\adj\bell}$.
The second follows by unfolding the definition of $Q_{abab}$ and merging $\PhiTwo(\pauliX^1)$ and $\PhiTwo(\adj{(\pauliZ^{b}\pauliX^{1})})$, which becomes $\PhiTwo(\adj{(\pauliZ^{b}\pauliX^{0})})$.

Now we can prove \eqref{eq:cond.X.1}: We have
\begin{equation*}
  \PhiTwo(\pauliX)\cdot\PhiOne(\selfbutter1) \cdot \sum_{a,b} \pb\paren{Q_{abab}\cdot\pre}
  \starrel=
  \sum_b \pb\paren{Q_{1b0b}\cdot\pre}
  \leq
  \sum_{a,b} \pb\paren{Q_{ab0b}\cdot\pre}.
\end{equation*}
Here $(*)$ follows from the two auxiliary facts.
By rule \textsf{Apply}, this implies \eqref{eq:cond.X.1}.
Similarly, by rule \textsf{Skip}, \eqref{eq:cond.X.0} is implied by
\begin{equation*}
  \PhiOne(\selfbutter0) \cdot \sum_{a,b} \pb\paren{Q_{abab}\cdot\pre}
  \starrel=
  \sum_b \pb\paren{Q_{0b0b}\cdot\pre}
  \leq
  \sum_{a,b} \pb\paren{Q_{ab0b}\cdot\pre}.
\end{equation*}
Here $(*)$ uses the first auxiliary fact.

So \eqref{eq:cond.X.1}, \eqref{eq:cond.X.0} hold, \eqref{eq:cond.X} follows.

Quite similarly, we show
\begin{equation*}
  \pB\hoare{\sum_{a,b} \pb\paren{Q_{ab0b}\cdot\pre}}
  {\ \qif X {1} {\braces{\qapply {\pauliZ} {\chain\Phi\Snd}}} \skipprog\ }
  {\sum_{a,b} \pb\paren{Q_{ab00}\cdot\pre}}.
\end{equation*}

Finally, we have $Q_{ab00}\cdot\pre \ \subseteq\ {  \PhiTwoAB \quanteq \psi  }$ and thus 
$\sum_{a,b}\pb\paren{Q_{ab00}\cdot\pre} \ \subseteq\ {  \PhiTwoAB \quanteq \psi  }$ as the following calculation shows:
\begin{align*}
  Q_{ab00}\cdot\pre\ 
  &\starrel=\ 
  \XPhiTwo(\Uswap) \cdot \Phi\pb\paren{\ket{a,b}\adj\bell}
  \cdot \XAB(\psi\adj\psi) \cdot \calH_M
  \\&
  =\
  \XPhiTwo(\Uswap)
    \cdot \XAB(\psi\adj\psi)
    \cdot \underbrace{ \Phi\pb\paren{\ket{a,b}\adj\bell}
    \cdot \calH_M}_{\subseteq\,\calH_M}
  \\&
  \subseteq\
  \XPhiTwo(\Uswap) 
  \cdot \XAB(\psi\adj\psi) \cdot \calH_M
  \\&
  \starstarrel=\
  \PhiTwoAB(\psi\adj\psi)
  \cdot
  \underbrace{\XPhiTwo\pb\paren{U_\sigma}
  \cdot\calH_M}_{\subseteq\,\calH_M}
  \\&
  \subseteq\
  \PhiTwoAB(\psi\adj\psi) \cdot \calH_M
  =\ {
  \PhiTwoAB \quanteq \psi
  }.
\end{align*}
Here $(*)$ uses the definitions of $Q_{ab00}$, $\pre$, and $\quanteq$.
The factor $\frac12$ vanishes because multiplying a subspace with a nonzero scalar has no effect.
And $(**)$ follows by bringing both sides into the form $\pair{\pair{X}{\chain\Phi\Snd}}\dots$ and expanding $\swap=\sandwich\Uswap$.

Combining all Hoare judgments, equalities, and inequalities, and using
rules \textsc{Seq} and \textsc{Weaken}, we get
\eqref{eq:hoare.teleport} which finishes the analysis of teleportation.


}

\delaytextsv{complements}{
  \newcommand\INCLUDEONCEudfasiuwerfsdrgroiswearldfjsd{}

\section{Complements}
\label{sec:complements}

When we think of a reference $F$ as region in memory, it makes sense to consider the complement~$\compl F$ of~$F$, namely everything in the memory that is not contained in $F$.
A priori, there is no reason to assume that in a general reference category, $\compl F$ is definable in a meaningful way, or that it constitutes a reference.
However, we will see that by requiring some additional axioms in a reference category, we get a well-defined notion of complements.
In particular, quantum references have complements, see \autoref{sec:qregs.complements}.
(We suspect classical references have complements, too, but we have not investigated that case.)

Before we formalize complements, let us explain why complements are useful:
\begin{compactitem}
\item
  When formulating predicates referring to states on a complex system (e.g., the predicates that are used in the pre- and postconditions in the quantum Hoare logic in \autoref{sec:hoare}), we can use predicates to refer to the state of ``everything else''.
  For example, we could state that no part of a quantum memory except for explicitly listed reference $F_1,\dots,F_n$ is modified by a program by stating that if $\compl{\pairs{F_1}{F_n}}$ is in state $\psi$ before the execution, then  $\compl{\pairs{F_1}{F_n}}$ is in state $\psi$ afterwards.
\item
  \fullshort{Complements will be used numerous times in the constructions and proofs in \autoref{sec:lifting} where we describe how various quantum mechanical concepts interact with references.
  For example, when describing the state of a complex system by giving the states of all subsystems (Sections~\ref{sec:mixed} and~\ref{sec:pure}), we have to specify the state of \emph{all} references.
  If an explicit list of all references is not known, we can use complements to refer to ``all other references''.
  And both the construction of the partial trace into a reference (\autoref{sec:partial.trace}) as well as the construction how to lift a quantum channel on a reference to a quantum channel on the whole space (\autoref{sec:quantum.channels}) use complements under the hood.}{
  Complements are used as an important technical tool numerous times when we lift various quantum objects (pure states, quantum channels, etc.) through references. (See the supplement, \autoref{sec:lifting}.)
  }
\item
  Last but not least, besides the practical use cases, it is a foundationally interesting and natural property: In a reference category with complements, there is always a meaningful and well-behaved notion of ``the rest of the system''.
  In particular, this applies to the quantum reference category.
\end{compactitem}

\paragraph{Definitions.} We say references $F,G$ are \emph{complements}\index{complement} iff they are disjoint and $\spair FG$ is an iso-reference.
We say a reference category \emph{has complements} iff
\begin{compactitem}
\item\pagelabel{page:ex.compl}
  For every reference $F:\mathbf A\to\mathbf B$, there exists an object $\mathbf C$ and a reference $G:\mathbf C\to \mathbf B$ such that $F,G$ are complements.\fullonly{\footnote{%
    In our Isabelle/HOL formalization (see \autoref{sec:isabelle}), we formalize a stronger property (different order of quantifiers):
    For every $\mathbf A,\mathbf B$ there exists a $\mathbf C$ such that for every $F:\mathbf A\to\mathbf B$, there exists a reference $G:\mathbf C\to\mathbf B$ such that $F,G$ are complements.
    The reason for this is that objects in the reference category are modeled as types in Isabelle/HOL, and that it is not possibly to express the existence of a type $\mathbf C$ that depends on a value $F$.}}
\item If $F,G$ are complements and $F,H$ are complements, then $G,H$ are equivalent.
\end{compactitem}
The second condition justifies to talk about ``the'' complement $\symbolindexmark\compl{\compl F}$ of $F$.
Strictly speaking, $\compl F$ is just one of the possible complements of $F$, determined only up to equivalence.
(Recall that intuitively, two references are equivalent when they correspond to the same part of the program memory.)

In a reference category with complements, the laws in \autoref{fig:compl.laws} hold. We proved all laws in Isabelle/HOL (see \autoref{sec:isabelle}).

\begin{figure}
  \raggedright

  \textbf{Complements.} ($F,G,H$ are assumed to be references.)
  \quad\nextlaw\label{law:compl.sym} $F,G$ are complements iff $G,F$ are complements.
  \quad\nextlaw $\compl F$ is a reference.
  \quad\nextlaw\label{law:compl.is.compl} $F,\compl F$ are complements.
  \quad\nextlaw\label{law:compl.equiv}%
  Assume $F,G$ are complements. Then $G,H$ are equivalent iff $F,H$ are complements.
  \quad\nextlaw $F,\compl{\paren{\compl F}}$ are equivalent.
  \quad\nextlaw\label{law:complement.pair}
  If $F,G$ are disjoint, $F$, $\pair G{\compl{\spair FG}}$ are complements (and also
  $G$, $\pair F{\compl{\spair FG}}$).
  \quad\nextlaw\label{law:complement.chain}
  $\chain FG$ and $\pb\pair{\compl F}{\chain F{{\compl G}}}$ are complements.
  ($\chain F{{\compl G}}$ is to be read as $\chain F{\paren{\compl G}}$.)
  \quad\nextlaw\label{law:complement.tensor} 
  $F\rtensor G$ and $\compl F\rtensor\compl G$ are complements.

  {
  \textbf{Unit references.} ($F$ is assumed to be a reference, $U,U'$ unit references, $I$ an iso-reference.)
  \quad\nextlaw $U,F$ are disjoint.
  \quad\nextlaw $F,(U;F)$ are equivalent.
  \quad\nextlaw $F\circ U$ is a unit reference.
  \quad\nextlaw $U\circ I$ is a unit reference.
  \quad\nextlaw $U,U'$ are equivalent.
  \quad\nextlaw If $U:\mathbf A\to\mathbf C, U':\mathbf B\to\mathbf D$, then $\mathbf A,\mathbf B$ are isomorphic.
  \quad\nextlaw $\compl\id$ is a unit reference (with same codomain as $\id$).
  \quad\nextlaw $\unitreg{\mathbf A}:\mathbf I\to\mathbf A$ is a unit reference (and $\mathbf I$ does not depend on~$\mathbf A$).
  \quad\nextlaw If $U:\mathbf A\to\mathbf B$, then $\mathbf C$ and $\mathbf C\otimes\mathbf A$ are isomorphic (for every $\mathbf C$).
  \quad\nextlaw\label{law:complement.iso.unit}
  $I$ and $U$ are complements.
  }

  \caption{Additional laws in a reference category with complements.
    \quad
    All laws have the implicit assumption that the types of the references match.
  }
  \label{fig:compl.laws}
\end{figure}

Note that we can easily extend the notion of complements to finite sets of references:
We say $F_1,\dots,F_n$ are a \emph{partition}\index{partition} iff they are pairwise disjoint and $\pairs{F_1}{F_n}$ is an iso-reference.
(Note that this definition does not depend on the order of the references $F_1,\dots,F_n$ by Laws~\ref{law:pair.sigma}, \ref{law:pair.alpha}, \ref{law:pair.alpha'} and the fact that $\sigma,\alpha,\alpha'$ are iso-references.\footnote{We have not formalized this fact in Isabelle/HOL because it cannot be stated in simple type theory.})
The intuition of being a partition is that the references $F_1,\dots,F_n$ cover the whole memory.

{
\paragraph{Unit references.} One interesting consequence of complements (that could, however, also be easily axiomatized on its own), are \emph{unit references}\index{unit reference}.
Intuitively, a unit reference is one that corresponds to an empty part of the program memory.
E.g., a classical variable of type \texttt{unit} would be a unit reference; writing into it has no effect.
Formally, we say that a reference $F:\mathbf A\to\mathbf B$ is a unit reference iff $F$ and $\id$ are complements.
(Since $\id$ is a reference that contains the whole memory, this intuitively means that $F$ refers to an empty part of the memory.)
Unit references are useful in languages with side-effects (even if the side-effect is just possible non-termination).
E.g., if $e$ is an expression with side-effects and unit result type, we can write $\assign{()}e$ to evaluate $e$ without storing the result in memory.
(Here $()$ denotes a unit reference.)
If unit references did not exist, the language would need a distinct concept for evaluating expressions without storing the result.

It turns out that a reference category with complements has a unit reference $\symbolindexmark\unitreg{\unitreg{\mathbf{A}}}:\mathbf I\to\mathbf A$ for every $\mathbf A$ (i.e., $\mathbf A$ represents the type of the memory).
Roughly speaking, $\unitreg{\mathbf{A}}$ is defined as the complement of $\id:\mathbf A\to\mathbf A$.
This is trivially a unit reference according to our definition.
However, this would lead to a different domain $\mathbf I$ for every $\mathbf A$.
With some extra effort, we can show that all those domains $\mathbf I$ are isomorphic and then construct a unit reference $\unitreg{\mathbf{A}}$ whose domain is independent of $\mathbf A$.

However, we also need to show that our definition of unit references ``makes sense''.
For example, we want that a unit reference is disjoint from all other references (not just with $\id$).
The properties we derive are listed in 
\autoref{fig:compl.laws}%
; when proving these properties (see our Isabelle formalization, \autoref{sec:isabelle}) we heavily use the properties of complements.

Note that for concrete reference categories (such as the quantum reference category), we may prefer not to use the unit reference $\unitreg{\mathbf{A}}$ that exists generically but instead define it concretely.
For example, in the quantum setting, we can define a unit reference of type $\setC\mapsto\mathbf{A}$, see \autoref{lemma:quantum.unit} below.
(Here $\setC$ is slight abuse of notation for the space of linear functions $\setC\to\setC$.)
Having $\setC$ instead of some unspecified $\mathbf{I}$ as the domain means that the type of expressions that return unit values is nicer.
However, the laws from \autoref{fig:compl.laws} still apply. (And the unit references according to both definitions are equivalent.)
}

{
\subsection{Quantum references}
\label{sec:qregs.complements}

We show that the quantum reference category $\Lquantum$ has complements. For the finite-dimensional case, we have shown this in Isabelle/HOL.
Here, we show the more general infinite-dimensional case (cf.~\autoref{sec:infinite}).

\paragraph{Existence.} Fix a reference $F:\mathbf A\to\mathbf B$.
By \autoref{lemma:normal.tensor1} (that was already used in the proof of Axiom~\ref{ax:pairs}, \autoref{sec:proofs.infdim}),
there is a Hilbert space $\calH_C$ and a unitary $U:\calH_A\otimes\calH_C\to\calH_B$ such that
$F(a)=U(a\otimes 1_{\mathbf C})\adj U$. Let $G(c) := U(1_{\mathbf{A}}\otimes c)\adj U$.
$G$ is a reference (using Axiom~\ref{ax:tensor-1} and~\autoref{lemma:normal}).
Define $I(x):=Ux\adj U$ which is a reference (using \autoref{lemma:normal}).
We have $F(a)G(c) = U(a\otimes c)\adj U = G(c)F(a)$ since $U$ is unitary.
Thus $F,G$ are disjoint. Thus the pair $\spair FG$ exists, and $\spair FG(a\otimes c)
= I(a\otimes c)$.
By Axiom~\ref{ax:tensorext}, this implies that $\spair FG=I$. Since $I^{-1}(y):=\adj UyU$ is also a reference and the inverse of $I$, $\spair FG=I$ is an iso-reference. Thus $F,G$ are complements by definition.

\paragraph{Uniqueness.} Assume that $F:\mathbf A\to\mathbf D$, $G:\mathbf B\to\mathbf D$ are complements, and $F$, $H:\mathbf C\to\mathbf D$ are complements.
By \autoref{lemma:complement.range} (\autoref{app:proofs:complements}), $G(\mathbf B)$ is the commutant of $F(\mathbf A)$. By the same lemma, $H(\mathbf C)$ is the commutant of $F(\mathbf A)$ as well.
Thus $G(\mathbf B)=H(\mathbf C)$. By \autoref{lemma:same.range.equiv} (\autoref{app:proofs:complements}), references with the same range are equivalent.
Thus $G,H$ are equivalent.

\medskip

By definition of ``has complements'', we now have:

\begin{theorem}\label{theo:compl}
  $\Lquantum$ has complements.
\end{theorem}

Hence the laws from \autoref{fig:compl.laws} hold for $\Lquantum$.

Reference categories with complements always have unit references.
However, in the quantum reference category, we have a particularly natural unit reference (with a simple and explicitly specified domain):
\begin{lemma}\label{lemma:quantum.unit}
  $u:\setC\to\mathbf B$, $u(c):=c\cdot 1_\mathbf{B}$ is a unit reference where we identify $\bounded(\setC)$ with~$\setC$.\footnote{%
    Intuitively, an operator in $\bounded(\setC)$ is a complex $1\times1$-matrix, hence simply a complex number.}
\end{lemma}

\begin{proof}
  We have $u(\setC)=\idmult_\mathbf{B}$.
  By \autoref{lemma:unit.iff} (\autoref{app:proofs:complements}), this means that $u$ is a unit reference.
\end{proof}
}

}

\delaytextsv{rel-to-optics}{
  \newcommand\INCLUDEONCEjashduirisdaskdoeretod{}

\section{Relationship to optics}
\label{sec:rel.optics}

In functional programming, a widely used concept are so-called \emph{optics}\index{optic}.
Optics encompass various different methods for accessing subparts of larger data-structures.
The probably best known example of an optic is the lens that we already encountered in \autoref{sec:classical}.
Other types of optics include \emph{prisms}\index{prism} (giving access to one of the options of a sum-type), \emph{optionals}\index{optional} a.k.a., \emph{affine traversals}\index{affine traversal}\index{traversal!affine} (access to a part of a data-structure that may or may not be there).
Quantum references as described in this paper are also optics in the sense that they are functional objects that provide access to a subpart of the memory.

All this leads to two questions:
\begin{compactitem}
\item Can we construct quantum analogoues to other optics besides lenses?
  (E.g., to optionals, for a part of a quantum system that may or may not be present.)
\item Might references as we have formalized them actually be a special case of some existing general formalism of optics?
\end{compactitem}
To the first question, we do not have an answer but we believe this to be an interesting question for future research.

Concerning the second question: Riley \cite{riley18optics} presents a general treatment of optics.
They show that many existing kinds of optics are just special cases of their category of optics, for a suitable choice of base category.
(E.g., lenses arise when we use the category \textbf{Set}.)
We ask whether we recover our quantum references as a special case if we apply Riley's definition to the category of Hilbert spaces with bounded operators as morphisms.
(Because our references operate on bounded operators.)
An \emph{optic}\index{optic!category theoretical} from $\calH_A$ to $\calH_B$ in Riley's sense is a pair $(L,R)$ of morphisms (i.e., bounded operators) $L:\calH_B\to\calH_A\otimes\calH_C$ and $R:\calH_A\otimes\calH_C\to\calH_B$ (for some Hilbert space $\calH_C$) modulo the equivalence relation generated by $\pb\paren{(1\otimes M)L,R} \cong \pb\paren{L,R(1\otimes M)}$\pagelabel{page:def:opticequiv} for arbitrary bounded operators $M$.\footnote{%
  This is not the actual \emph{definition} of optics in \cite{riley18optics}, but \cite{riley18optics} shows that this is equivalent to the definition.
  We use this formulation because we can present it without introducing a number of category-theoretical concepts.}
The intuition is that, if we have data of type $\calH_B$ (e.g., the memory), and we want to access a subpart of type $\calH_A$ (e.g., the reference content), then we apply $L$ to split the data into a $\calH_A$-part (that we want to access) and the rest (type $\calH_C$). Then we operate on $\calH_A$, and then we apply $R$ to put things together again (we get a state in $\calH_B$ again).
That is, like with our references, given an optic $(L,R)$, we can naturally lift a bounded operator on $\calH_A$ to one on $\calH_B$ by $F_{(L,R)}:a\mapsto R(a\otimes 1_C)L$ (we call $F_{(L,R)}$ the \emph{lifting function}\index{lifting function} of the optic).
Obviously, for any bounded operator $M$, the optics $\pb\paren{(1\otimes M)L,R}$ and $\pb\paren{L,R(1\otimes M)}$ lead to the same lifting function. (The converse is not obvious.)
This justifies the choice of equivalence relation when defining optics.

While the definition of an optics category looks very different from the definition of a reference category, and our definition of quantum references looks very different from this, too, it turns out that the latter are closely connected:
By \autoref{lemma:normal.tensor1}, $F$ is a quantum reference iff it can be written as $F(a)=U(a\otimes 1_C)\adj U$ for unitary $U$ and some Hilbert space $\calH_C$.
That is, we can represent $F$ as the optic $(\adj U,U)$.
But for optics (in Riley's sense) and (quantum) references to actually be the same notion, we need to check three things:
\begin{inparaenum}[(a)]
\item\label{item:optics.reg.1}
  If $(L,R)\cong(L',R')$ (i.e., they are the same optic), do they give rise to the same reference (i.e., have the same lifting function)?
\item\label{item:optics.reg.2}
  If $(L,R)$ and $(L',R')$ are the same reference, do we have $(L,R)\cong(L',R')$?
\item\label{item:optics.reg.3}
  Is every optic a reference? (I.e., is the lifting function of every optic a reference?)
\end{inparaenum}

The answer to \eqref{item:optics.reg.3} is trivially no:
$(0,0)$ leads to the function $F(a)=0$ which is not a reference.
However, the latter is not surprising because the category of optics is not intended to be limited to well-behaved objects but simply establishes the type of things we are talking about (similar to our prereferences).
Indeed, \cite{riley18optics} introduces the notion of lawful optics\index{lawful optics}\index{optics!lawful} to restrict the categories of optics; these are characterized by two certain properties.
We do not give the exact definition of these properties here as they again involve category-theoretic notions that we have not introduced here.
It suffices to say that they imply and maybe are equivalent to $F(ab)=F(a)F(b)$ and $F(1)=1$ for the lifting function $F$.
We recognize our Axiom~\ref{ax:reg.monhom} here.
But unfortunately, $\pb\paren{\tiny\begin{pmatrix}1/2\!\!\!\!\\\!\!&1\end{pmatrix},\begin{pmatrix}2\!\!\\&\!\!1\end{pmatrix}}$ is a lawful optic from $\setC^{2\times2}$ to $\setC^{2\times 2}$, yet the lifting function has $F\pb\paren{\tiny\begin{pmatrix}&\!\!i\\-i\!\!\end{pmatrix}}={\tiny\begin{pmatrix}&\!\!\!\!2i\\-i/2\!\!\end{pmatrix}}$, i.e., it maps a unitary to a non-unitary, and thus is not a reference (nor, it seems, anything sensible from a quantum mechanical point of view).

This issue can be resolved by using the category of Hilbert spaces with unitaries as morphims (instead of bounded operators) as the base category.
The resulting category of optics satisfies conditions (\ref{item:optics.reg.1}--\ref{item:optics.reg.3}),\footnote{%
  \eqref{item:optics.reg.1} is trivial even without assuming lawfulness. 
  \eqref{item:optics.reg.3} follows directly from the fact that lawfulness of $(L,R)$ implies that $RL=1$.
  \eqref{item:optics.reg.2} is shown in \autoref{lemma:optics.uni.equiv}.}
so this leads to a notion of optics that stands in 1-1 correspondence with our references.
But then we live in the category of unitaries, so strictly speaking we cannot apply the lifting functions to other physically meaningful operations (such as projectors or isometries).
To overcome this would have to define the notion of optics with respect to two related categories.
While we do not think that this would lead to problems, to the best of our knowledge this is an extension of optics that has not been studied.

Besides the above challenges, there are two more advantages that references have over optics in Riley's sense:
\begin{itemize}
\item
  One of the distinguishing features of references is that they support pairing, i.e., taking two references and see them as one.
  For optics, it is unclear how such a notion would be defined (especially in a way that lawfulness would be preserved).
\item
  To define optics in Riley's sense, we need category-theoretical constructions that might not be available in some mathematical foundations.
  For example, the relation $\cong$ above would be defined as the closure of the relation $\sim$ defined by ``$(L,R)\sim(L',R')$ iff there exists a Hilbert space $\calH_C$ and an $M$ such that \dots''.
  This seemingly harmless definition quantifies over all Hilbert space $\calH_C$ (with no upper bound on its cardinality) and would not, for example, even be well-typed in foundations like HOL (underlying Isabelle/HOL).
  Stronger foundations (ZFC, calculus of inductive constructions \cite{paulin93inductive} underlying e.g., Coq \cite{coq} and Lean \cite{demoura2015lean}, \dots) may well be able to handle it but we believe that there is an advantage in finding solutions that do not require stronger foundations (both for practical reasons and because the stronger the foundation, the more it is possible that it is unsound).
\end{itemize}
Of course, Riley's formalism has advantages, too, that our references do not have.
For example, we have no notion of partial operations (e.g., accessing the fifth element of a quantum list, which may not exist) while in the classical setting, affine traversals (which are covered by Riley's formalism, too) allow this.

We believe that it would be a valueable avenue for future research to find approaches that combine all the above advantages.

}

\delaytextsv{lifting quantum objects}{
\section{Lifting various quantum objects}
\label{sec:lifting}

By definition, a quantum reference $F:\mathbf A\to\mathbf B$ maps a bounded operator on the Hilbert space $\calH_A$ (the space corresponding to the part of the memory that the reference occupies) to a bounded operator on $\calH_B$ (the space corresponding to the whole memory).
We already mentioned in \autoref{sec:quantum-overall} that $F$ maps unitaries to unitaries.
This allows us to interpret unitary operations on a reference as unitaries on the whole memory, and thus to interpret a quantum circuit containing only unitary operations.
But what about other operations/objects that occur in quantum computation?
Can we interpret a quantum state on a reference as a quantum state on the whole memory?
How about a measurement on a quantum state? A quantum channel?
In this section, we answer those questions.

\subsection{Elementary objects}

Since unitaries and projectors $A$ are bounded operators, a reference $F$ lifts them to the whole memory as $F(A)$.
The following lemma tells us that all relevant properties are preserved:
\begin{lemma}\lemmalabel{lemma:elementary}
  Let $F:\mathbf A\to\mathbf B$ be a reference.
  Fix $A\in\mathbf{A}$.
  \begin{compactenum}[(i)]
  \item\itlabel{lemma:preserve.unitary} If $A$ is a unitary, then $F(A)$ is a unitary.
  \item\itlabel{lemma:preserve.isometry} If $A$ is an isometry, then $F(A)$ is an isometry.
  \item\itlabel{lemma:preserve.projector} If $A$ is a projector, then $F(A)$ is a projector. (By projector, we mean an \emph{orthogonal} projector throughout this paper.)
  \item\itlabel{lemma:preserve.norm} $\pb\norm{F(A)} = \norm A$.
  \item\itlabel{lemma:preserve.positive} If $A$ is positive, then $F(A)$ is positive.
  \end{compactenum}
\end{lemma}

The proof is given in \autoref{app:proof:lemma:elementary}.

\subsection{Subspaces (quantum predicates)}
\label{sec:lift.sub}

Subspaces of a Hilbert space are important in quantum mechanics to describe the properties of quantum states (e.g., ``state $\psi$ is in subspace $S$'').
E.g., in von-Neumann-Birkhoff quantum logic \cite{birkhoff36logic} and in many quantum Hoare logics \cite{qrhl, zhou19applied, ghosts, Li2020}.
Most often, we consider (topologically) \emph{closed subspaces}\index{closed subspace}\index{subspace!closed} only.
(Otherwise, the limit of a sequence of quantum states satisfying a property does not necessarily satisfy that property.
In the finite-dimensional case, every subspace is closed \cite[Proposition 3.3]{conway13functional}, so that requirement is usually not stated explicitly in works considering only finite-dimensional systems.)

Closed subspaces stand in a natural 1-1 correspondence with projectors:
If $S$ is a closed subspace, there exists a unique projector $P_S$ with image $S$ (\autoref{lemma:ex.proj}),
and for any projector $P$, its image is a closed subspace \cite[Proposition II.3.2\,(b)]{conway13functional}.
Since we already know how to lift projectors through references (\lemmaref{lemma:preserve.projector}), it is now easy to lift closed subspaces:

For a reference $F:\mathbf A\to\mathbf B$, and a closed subspace $S\subseteq\calH_A$, we define $\symbolindexmark\regsubspace{\regsubspace FS}$ as the image of $F(P_S)$ where $P_S$ denotes the projector with image $S$.

For non-closed subspaces, we do not know if there is a simple way to lift them through references.%
\footnote{%
  Using the tools from \autoref{sec:pure}, we could define $\regsubspace FS := \{F(\psi)\ptensor\compl F(\phi): \psi\in S,\ \phi\in\calH_{A'}\}$ where $\mathbf A'$ denotes the domain of $\compl F$.
  But we have not explored this approach further since we believe that non-closed subspaces are not very relevant in the context of quantum programs.}

The following lemma shows that the lifting of subspaces performs as
expected:

\begin{lemma}\lemmalabel{lemma:lift.sub}
  Let $F:\mathbf A\to\mathbf B$ be a reference.
  Let $S,T$ be closed subspaces of $\calH_A$.
  \begin{compactenum}[(i)]
  \item\itlabel{lemma:lift.sub:subspace} $\regsubspace FS$ is a closed subspace of $\calH_B$.
  \item\itlabel{lemma:lift.sub:ortho} $F(\ortho S)=\ortho{F(S)}$ where $\symbolindexmark\ortho{\ortho{S}}$ denotes the orthogonal complement of $S$.
  \item\itlabel{lemma:lift.sub:zero} $F\pb\paren{\{0\}}=\{0\}$.
  \item\itlabel{lemma:lift.sub:all} $F(\calH_A)=F(\calH_B)$.
  \item\itlabel{lemma:lift.sub:mono} $F(S)\subseteq F(T)$ iff $S\subseteq T$.
  \item\itlabel{lemma:lift.sub:inter} $F(S\cap T)= F(S)\cap F(T)$.
  \item\itlabel{lemma:lift.sub:plus} $F(S+T)= F(S)+F(T)$ where $S+T$ denotes the closure of $\{\psi+\phi:\psi\in T,\phi\in S\}$, or equivalently the smallest closed subspace containing $S,T$.
  \item \itlabel{lemma:lift.sub:apply-op} For a bounded operator $A$, $F(AS)=F(A)F(S)$.
    ($AS$ means the closed image of $S$ under the operator $A$, and analogously $F(A)F(S)$.)
  \item \itlabel{lemma:lift.sub:inter.disjoint}
    Let $G:\mathbf C\to\mathbf B$ be another reference and $V$ a closed subspace of $\calH_C$. Assume $F,G$ are disjoint.
    Then $F(S)\cap G(V)=\im F(P_S)G(P_V)=F(P_S)G(V)$.
    ($P_S,P_V$ are the projectors onto $S,V$.
    $F(P_S)G(V)$ is the image of $G(V)$ under the operator $F(P_S)$.)
  \end{compactenum}
\end{lemma}

The proof is given in \autoref{app:proof:lemma:lift.sub}.

\subsection{Mixed states (density operators)}
\label{sec:mixed}

A \emph{mixed state}\index{mixed state}\index{state!mixed} (roughly speaking, a probabilistic quantum state) is represented by a density operator.
Formally, a \emph{density operator}\index{density operator} on $\calH_A$ is a positive trace-class operator on $\calH_A$ of trace $1$.
(A trace-class operator is one for which the trace is defined, see \cite[Definitions 18.3, 18.10]{conway00operator}.
In the finite-dimensional case, the requirement ``trace-class'' can be omitted.
We write $\symbolindexmark\tracecl{\tracecl(\calH)}$ for the space of all trace-class operators on $\calH$, and $\symbolindexmark\trnorm{\trnorm\cdot}$ for the trace norm.)
And a \emph{subdensity operator}\index{subdensity operator} is a positive trace-class operator of trace $\leq1$.
Every (sub)density operator $\rho$ can be written as $\rho=\sum_i p_i\psi_i\adj{\psi_i}$ with $p_i\geq0$ and $\sum_ip_i=1$ (or $\leq1$) and $\psi_i\in\calH_A$;
the intuition is that the state is $\psi_i$ with probability $p_i$.
Since $\rho$ is a bounded operator on $\calH_A$, it is tempting to think that we can lift $\rho$ to $\calH_B$ as $F(\rho)$.
Unfortunately, this does not work.
$F(\rho)$ will, in general, not be a density operator.\footnote{%
  For example, if $\calH_A=\setC^2$, $F=\Fst:\mathbf A\to\mathbf A\otimes\mathbf A$,
  $\rho={\tiny\begin{pmatrix}1&0\\0&0\end{pmatrix}}$, then $\rho$ is a density operator, but $F(\rho)=\rho\otimes 1_{\mathbf A}$ has trace $2$ and thus is not.}
If we think about the situation, this is not surprising:
If we only know what the state of the reference $F$ is, we cannot tell what the state of the overall system is since we do not know what the state of the remaining parts of the memory should be.
(Note that if $F$ is an \emph{iso-reference}, then $F(\rho)$ is a density operator if $\rho$ is.)
However, if references $F_1,F_2,\dots,F_n$ ``cover the whole memory'' and $\rho_1,\rho_2,\dots,\rho_n$ are mixed states on those references, then it makes sense to talk about the state of the whole system.
Formally, we do this as follows:
And ``$F_1,F_2,\dots,F_n$ cover the whole memory'' simply means that $F_1,F_2,\dots, F_n$ are a partition (see \autoref{sec:complements}).\footnote{%
  In cases where $F_1,\dots,F_n$ are only pairwise disjoint but not necessarily a partition, we can always add $F_{n+1}:=\compl{\pairs{F_1}{F_n}}$ to turn them into a partition.}
By definition, this implies that $\pairs{F_1}{F_n}$ is an iso-reference.
Thus we can lift the states $\rho_1,\rho_2,\rho_3,\dots$ to become a state $\rho$ on the whole memory defined as $\rho := \pairs{F_1}{F_n}\pb\paren{\rho_1\otimes\dots\otimes\rho_n}$.
We write \symbolindexmark\mtensor{$F_1(\rho_1)\mtensor F_2(\rho_2)\mtensor\dots\mtensor F_n(\rho_n)$} for this state~$\rho$.
(This is a slight abuse of notation, because $\mtensor$ is not actually a binary operation.
The whole term $F_1(\rho_1)\mtensor\dots\mtensor F_n(\rho_n)$ is defined as a whole.)
The following lemma ensures that this construction is well-defined (in particular, it produces a density operator), and that lifting density operators is compatible with the operations considered in the previous section.
\begin{lemma}\lemmalabel{lemma:mixed.states} Let  $F_1:\mathbf A_1\to\mathbf B$, \dots, $F_n:\mathbf A_n\to\mathbf B$  be a partition. Assume that $\rho_i$ ($i=1,\dots,n$) are trace-class operators on $\calH_{A_i}$.
  \begin{compactenum}[(i)]
  \item\itlabel{lemma:mixed.states:density}
    If $\rho_1,\dots,\rho_n$ are (sub)density operators,
    then $F_1(\rho_1)\mtensor\dots\mtensor F_n(\rho_n)$ is a (sub)density operator.
  \item\itlabel{lemma:mixed.states:trace} $\tr \pb\paren{F_1(\rho_1)\mtensor\dots\mtensor F_n(\rho_n)} = \tr\rho_1\cdots\tr\rho_n$.
    (And in particular, ${F_1(\rho_1)\mtensor\dots\mtensor F_n(\rho_n)}$ is trace-class.)
  \item\itlabel{lemma:mixed.states:permute} The lifting of mixed states does not depend on the order of the references and operators.
    That is, for a permutation $\pi$ on $1,\dots,n$, we have
    $F_{\pi(1)}(\rho_{\pi(1)})\mtensor F_{\pi(2)}(\rho_{\pi(2)})\mtensor\dots\mtensor F_{\pi(n)}(\rho_{\pi(n)})
    =
    F_1(\rho_1)\mtensor F_2(\rho_2)\mtensor\dots\mtensor F_n(\rho_n)$.
  \item\itlabel{lemma:mixed.states:apply.UV} Let $U,V$ be bounded operators on $\calH_{A_i}$.
    Then $F_1(U\rho_1 V)\mtensor F_2(\rho_2)\mtensor\dots\mtensor F_n(\rho_n) = F_1(U) \pb\paren{F_1(\rho_1)\mtensor\dots\mtensor F_n(\rho_n)} F_1(V)$.\footnote{\label{footnote:mixed.states:reorder}%
      We state this only for the first reference $F_1$ for simplicity, but by \eqref{lemma:mixed.states:permute}, it applies to all other references, too.}
  \item\itlabel{lemma:mixed.states:rank1} $F_1(\rho_1)\mtensor F_2(\rho_2)\mtensor\dots\mtensor F_n(\rho_n)$ has rank 1 iff all $\rho_i$ have rank 1.
  \item\itlabel{lemma:mixed.states:bounded} $\Phi:\rho_1,\dots,\rho_n \mapsto   F_1(\rho_1)\mtensor\dots\mtensor F_n(\rho_n)$ is bounded linear in each argument (with respect to the trace norm).
  \item\itlabel{lemma:mixed.states:pair} 
    $ \pair{F_1}{F_2}(\rho_1\otimes\rho_2)\mtensor F_3(\rho_3)\mtensor\dots\mtensor F_n(\rho_n)
    =
    F_1(\rho_1)\mtensor F_2(\rho_2)\mtensor\dots\mtensor F_n(\rho_n)$.\footnoterepeat{footnote:mixed.states:reorder}
  \item\itlabel{lemma:mixed.states:separating} Let $S_i$ generate $\tracecl(\calH_{A_i})$ for $i=1,\dots,n$.
    Then $\pb\braces{  F_1(\rho_1)\mtensor F_2(\rho_2)\mtensor\dots\mtensor F_n(\rho_n) :
    \rho_i\in S_i }$ generates $\tracecl(\calH_B)$.
    (``Generate'' refers to the linear span, closed with respect to the trace norm.)
    In particular,
    for quantum (sub)channels\footnote{Readers who wish to recall the precise definition of quantum (sub)channels are directed to \autoref{sec:quantum.channels}.} $\calE,\calF$, if
    $\calE \pb\paren{   F_1(\rho_1)\mtensor F_2(\rho_2)\mtensor\dots\mtensor F_n(\rho_n) }
    = \calF\pb\paren{   F_1(\rho_1)\mtensor F_2(\rho_2)\mtensor\dots\mtensor F_n(\rho_n) }$
    for all $\rho_i\in S_i$, then $\calE=\calF$.
  \item\itlabel{lemma:mixed.states:nested} Let $G_1:\mathbf C_1\to\mathbf A_1$, \dots, $G_m:\mathbf C_m\to\mathbf A_1$ be a partition, and $\sigma_i$ trace-class operators on $\calH_{C_i}$.
    Then $  F_1\pB\paren{G_1(\sigma_1)\mtensor\dots\mtensor G_m(\sigma_m)}\mtensor F_2(\rho_2)\mtensor\dots\mtensor F_n(\rho_n)
    =
    \chain{F_1}{G_1}(\sigma_1)
    \mtensor\dots\mtensor
    \chain{F_1}{G_m}(\sigma_m)
    \mtensor F_2(\rho_2)\mtensor\dots\mtensor F_n(\rho_n)
    $.\footnoterepeat{footnote:mixed.states:reorder}
  \end{compactenum}
\end{lemma}

The proof is given in \autoref{app:proof:lemma:mixed.states}.

\paragraph{Example.}
We illustrate how to use lifted mixed states in a small example, to illustrate how the theory developed in this subsection can be used.
We want to evaluate the following circuit:
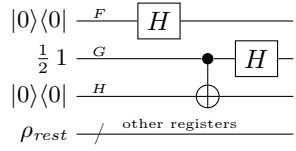
\begin{center}
  \begin{tikzpicture}
    \initializeCircuit
    \newWires{F,G,H,rest}
    \node[wireInput=F] {\small $\selfbutter 0$\,\, };
    \node[wireInput=G] {\small $\frac12\,1$\,\, }; \node[wireInput=H] {\small $\selfbutter 0$\,\,};
    \node[wireInput=rest] {\small $\rho_\textit{rest}$\,\,};
    \stepForward{3mm}
    \labelWire[\tiny$F$]{F}
    \labelWire[\tiny$G$]{G}
    \labelWire[\tiny$H$]{H}
    \drawWire{rest}
    \node at (\getWireCoord{rest}) {\footnotesize/};
    \stepForward{3mm}
    \labelWire[\tiny\rlap{other references}]{rest}
    \stepForward{2mm}
    \node[gate=F] (HF) {$H$};
    \stepForward{2mm}
    \node[cnot=H,control=G] (cnot) {};
    \stepForward{2mm}
    \node[gate=G] (HG) {$H$};
    \stepForward{2mm}
    \drawWires{F,G,H,rest}
  \end{tikzpicture}
\end{center}
The bottom wire represents any other quantum references that are possibly present in the system, i.e., the complement $Z:=\compl{\paren{\spair F{\spair GH}}}$.
References $F$ and $H$ are initially in the classical state $\ket 0$ ($\selfbutter 0$ as a density operator).
And reference $G$ is in the completely mixed state $\frac12\, 1=\frac12\selfbutter0+\frac12\selfbutter1$.
The initial state in this circuit can be expressed as
$F\pb\paren{\selfbutter 0} \mtensor G\pb\paren{\frac12\,1} \mtensor H\pb\paren{\selfbutter0} \mtensor Z(\rho_\mathit{rest})$.
And the circuit consists of the unitaries $F(H)$, $\spair FG(\CNOT)$, $G(H)$.
(Which are to be applied from both sides to the current state, like $\rho\mapsto U\rho\adj U$.)
Thus we can evaluate the circuit as follows:
\begin{align*}
  & F\pb\paren{\selfbutter 0} \mtensor G\pb\paren{\tfrac12\,1} \mtensor H\pb\paren{\selfbutter0} \mtensor Z(\rho_\mathit{rest})
  \\
  \xmapsto{F(H)\textcolor{gray}{(\cdot)}\adj{F(H)}}\ 
  & F\pb\paren{H\selfbutter 0\adj H} \mtensor G\pb\paren{\tfrac12\,1} \mtensor H\pb\paren{\selfbutter0} \mtensor Z(\rho_\mathit{rest})
    \\&\quad
    =
    F\pb\paren{\selfbutter +} \mtensor G\pb\paren{\tfrac12\,1} \mtensor H\pb\paren{\selfbutter0} \mtensor Z(\rho_\mathit{rest})
  \\&
  \quad \starrel=
    \spair GH\pb\paren{\tfrac12\,1 \otimes \selfbutter0}
    \mtensor
    F\pb\paren{\selfbutter +} \mtensor Z(\rho_\mathit{rest})
  \\
  \xmapsto{\spair GH(\CNOT)\textcolor{gray}{(\cdot)}\adj{\spair GH(\CNOT)}}\ 
  &
    \spair GH\pb\paren{\CNOT\paren{\tfrac12\,1 \otimes \selfbutter0}\adj\CNOT}
    \mtensor
    F\pb\paren{\selfbutter +} \mtensor Z(\rho_\mathit{rest})
  \\&
  \quad=
  \spair GH\pb\paren{\tfrac12\selfbutter{00}+\tfrac12\selfbutter{11}}
  \mtensor
  F\pb\paren{\selfbutter +} \mtensor Z(\rho_\mathit{rest})
  \\
  \xmapsto{\spair GH(H\otimes 1)\textcolor{gray}{(\cdot)}\adj{\spair GH(H\otimes 1)}}\
  &
  \spair GH\pb\paren{(H\otimes 1)\pb\paren{\tfrac12\selfbutter{00}+\tfrac12\selfbutter{11}}\adj{(H\otimes 1)}}
  \mtensor
  F\pb\paren{\selfbutter +} \mtensor Z(\rho_\mathit{rest})
  \\&
  \qquad=
  \spair GH\pb\paren{{\tfrac12\selfbutter{+0}+\tfrac12\selfbutter{-1}}}
  \mtensor
  F\pb\paren{\selfbutter +} \mtensor Z(\rho_\mathit{rest}).
\end{align*}
Here $(*)$ is by \lemmaref{lemma:mixed.states:pair},\ \eqref{lemma:mixed.states:permute} which allow us to arbitrarily reorder and pair references in order to put the references involved in the next gate into the first position.
Each of the $\xmapsto{\dots}$ is by \lemmaref{lemma:mixed.states:apply.UV} which applies the current gate to the first tensor factor.
Note that in the third $\xmapsto{\dots}$, we apply $\spair GH(H\otimes 1)$ instead of $G(H)$.
This is because we cannot rewrite the current state into the form $G(\dots)\mtensor\dots$ because the state of $\spair GH$ cannot be written as a tensor product anymore after the CNOT.
However, $\spair GH(H\otimes 1)=G(H)$ by Axiom~\ref{ax:pairs}, so we still apply the same gate, only written differently.

So the final state can be written as $\ket +$ in the $F$ reference and $\ket{+0}$ in the $G,H$ references.
(Due to the CNOT, we cannot separate the states in those references $G,H$ anymore.)

Note that all the reorderings of references necessary for applying the gates are quite straightforward and should be easy to automate in a proof assistant.
(But we have not done so at this point.)

\subsection{Pure states}
\label{sec:pure}

Sometimes, it is simpler to work with pure states than mixed states.
For example, when reasoning about a quantum circuit that contains only unitary operations, keeping track its behavior on pure states may be simpler since the added generality of mixed states might be overkill.
To the best of our knowledge, there are two slightly different definitions of what a pure state on $\calH_A$ is:
\begin{compactitem}
  \item A unit vector $\psi\in\calH_A$.
  \item A unit vector $\psi\in\calH_A$ modulo a global phase factor $c\in\setC$.
    Equivalently: a one-dimensional subspace of $\calH_A$.
    Equivalently: a rank-1 density operator $\psi\adj\psi$.
\end{compactitem}
The second formalization is motivated by the fact that two pure states that differ only by a global phase factor (i.e., $\exists c\in\setC.\ \abs c=1\land\psi_1=c\psi_2$) are physically indistinguishable, i.e., there is no sequence of operations and measurements that can distinguish them.
Since the second formalization is equivalent to restricting ourselves to mixed states of rank 1, we can simply use the mechanisms discussed in the previous section to work with these mixed states.

Thus in this section, we will focus on the first formalization, i.e., a state being a unit vector $\psi\in\calH_A$.
First a word on why this formalization is even important (besides the fact that it is widely used).
After all, if two states differing only in the global phase are physically indistinguishable, why do we even care about the difference between two such states?
The reason (in our opinion) is that one important use case of pure states in quantum computation is to determine the behavior of a unitary by evaluating it on all pure states (or on a basis of pure states).
Namely, if $U_1\psi=U_2\psi$ on all $\psi\in\calH_A$ (or on a basis), then $U_1=U_2$.
We cannot conclude that $U_1=U_2$ if $U_1\psi=U_2\psi$ only holds up to a phase factor.\footnote{%
  And similarly, one may ask why one wants to determine a unitary $U$, and not just a unitary up to a global phase factor.
  One example would be when considering controlled unitaries: The controlled unitary $CI$ corresponding to the identity $I$ is the identity.
  The controlled unitary $C(-I)$ corresponding to the negated identity is the controlled phase gate.
  Those are very different even though $I$ and $-I$ differ only in a phase factor.}

So to cover this use case, we need to see how to lift pure states $\psi\in\calH_A$ through quantum references.
Unfortunately, while the lifting of all other quantum objects (unitaries, projectors, mixed states, superoperators, measurements) turns out to work very naturally, this is not the case with pure states.

\paragraph{Definition.}
Like in the mixed state case (and for the same reasons), we assume a partition $F_1:\mathbf A_1\to\mathbf B$, \dots, $F_n:\mathbf A_n\to\mathbf B$, and pure states $\psi_1\in\calH_{A_1}$, \dots, $\psi_n\in\calH_{A_n}$ on those references.

We define the following auxiliary constructs:
For every $\mathbf A$ (or equivalently, for every Hilbert space~$\calH_A$), we fix an arbitrary unit vector $\symbolindexmark\pureeta{\pureeta{\mathbf A}}\in\calH_A$, with the constraint that $\eta_{\mathbf A\otimes\mathbf B}=\eta_{\mathbf A}\otimes\eta_{\mathbf B}$.
For every $\mathbf{B}$ and every non-zero $b\in\mathbf{B}$, we fix an arbitrary unit vector $\symbolindexmark\purexi{\purexi b}$ in the span of $b$, with the additional constraint that if $\pureeta{\mathbf{B}}$ is in the span of $b$, then $\purexi b:=\pureeta{\mathbf{B}}$. 

Then we can lift the states $\psi_1,\dots,\psi_n$ through $F_1,\dots,F_n$ to $\calH_B$ as $\psi := \pb\paren{F_1(\psi_1\adj{\pureeta{\mathbf A_1}})\mtensor\dots\mtensor F_n(\psi_n\adj{\pureeta{\mathbf A_n}})} \purexi b$ with $b:={\paren{F_1(\pureeta{\mathbf{A}_1}\adj{\pureeta{\mathbf A_1}})\mtensor\dots\mtensor F_n(\pureeta{\mathbf{A}_n}\adj{\pureeta{\mathbf A_n}})}}$.
(Recall the notation $\mtensor$ from the previous section.)
We write \symbolindexmark\ptensor{$F_1(\psi_1)\ptensor\dots\ptensor F_n(\psi_n)$} for this state $\psi\in\calH_B$.
(As for $\mtensor$, this is a slight abuse of notation because $\ptensor$ is not a binary operation.)

We call a reference $F:\mathbf A\to\mathbf B$ \emph{$\eta$-regular}\index{regular!$\eta$-}\index{eta-regular@$\eta$-regular} iff there exists a $c$ such that $(F;\compl F)(\pureeta{\mathbf A}\adj{\pureeta{\mathbf A}}\otimes c) = \pureeta{\mathbf B}$. (Intuitively, $\eta$-regular references are references that do not involve any ``mappings'' (see the introduction), e.g., references constructed from pairs, $\Fst$, $\Snd$.)
This is a technical notion needed for stating \autoref{lemma:lift.pure}\,\eqref{lemma:lift.pure:nested} below.

The following lemma shows that this indeed results in a pure state and shows that the approach is compatible with the lifting of unitaries etc.
\begin{lemma}\lemmalabel{lemma:lift.pure}
  Assume that $F_1:\mathbf A_1\to\mathbf B$, \dots, $F_n:\mathbf A_n\to\mathbf B$ form a partition.
  Let $\psi_1,\dots,\psi_n\in\calH_{A_1},\dots,\calH_{A_n}$.
  \begin{compactenum}[(i)]
  \item\itlabel{lemma:lift.pure:norm} $\pB\norm{F_1(\psi_1)\ptensor\dots\ptensor F_n(\psi_n)} = \norm{\psi_1}\cdots\norm{\psi_n}$.
    In particular, if $\psi_1,\dots,\psi_n$ are pure states, so is $F_1(\psi_1)\ptensor\dots\ptensor F_n(\psi_n)$.
  \item\itlabel{lemma:lift.pure:permute} The construction does not depend on the order of the references and pure states.
    That is, for a permutation $\pi$ on $1,\dots,n$, we have
    $F_{\pi(1)}(\psi_{\pi(1)})\ptensor F_{\pi(2)}(\psi_{\pi(2)})\ptensor\dots\ptensor F_{\pi(n)}(\psi_{\pi(n)})
    =
    F_1(\psi_1)\ptensor F_2(\psi_2)\ptensor\dots\ptensor F_n(\psi_n)$.
  \item\itlabel{lemma:lift.pure:bounded} $\psi_1,\dots,\psi_n \mapsto   F_1(\psi_1)\ptensor F_2(\psi_2)\ptensor\dots\ptensor F_n(\psi_n)$ is bounded linear in each argument.
  \item\itlabel{lemma:lift.pure:pair}
    $ \pair{F_1}{F_2}(\psi_1\otimes\psi_2)\ptensor F_3(\psi_3)\ptensor\dots\ptensor F_n(\psi_n)
    =
    F_1(\psi_1)\ptensor F_2(\psi_2)\ptensor\dots\ptensor F_n(\psi_n)$.\footnote{\label{footnote:lift.pure:reorder}%
      We state this only for the first reference $F_1$ for simplicity, but by \eqref{lemma:lift.pure:permute}, it applies to all other references, too.}
  \item\itlabel{lemma:lift.pure:separating}
    Let $S_i$ generate $\calH_{A_i}$ for $i=1,\dots,n$.
    (``Generate'' refers to the linear span, closed with respect to the Hilbert space norm.)
    Then $\pb\braces{  F_1(\psi_1)\ptensor\dots\ptensor F_n(\psi_n) :
    \psi_i\in S_i }$ generates $\calH_B$.
    In particular, for bounded operators $U,V$, if $U \pb\paren{   F_1(\psi_1)\ptensor \dots\ptensor F_n(\psi_n) }
    = V\pb\paren{   F_1(\psi_1)\ptensor\dots\ptensor F_n(\psi_n) }$ for all $\psi_i\in S_i$, then $U=V$.
  \item\itlabel{lemma:lift.pure:nested} Let $G_1:\mathbf C_1\to\mathbf A_1$, \dots, $G_m:\mathbf C_m\to\mathbf A_1$ be a partition.
    Assume that $G_1,\dots,G_m$ are $\eta$-regular.
    Then $  F_1\pB\paren{G_1(\gamma_1)\ptensor\dots\ptensor G_m(\gamma_m)}\ptensor F_2(\psi_2)\ptensor\dots\ptensor F_n(\psi_n)
    =
    \chain{F_1}{G_1}(\gamma_1)
    \ptensor\dots\ptensor
    \chain{F_1}{G_m}(\gamma_m)
    \ptensor F_2(\psi_2)\ptensor\dots\ptensor F_n(\psi_n)
    $.\footnoterepeat{footnote:lift.pure:reorder}
  \item\itlabel{lemma:lift.pure:regulars}
    $\Fst$, $\Snd$, $\swap$, $\assoc$, $\assoc'$ are $\eta$-regular.
    If $F,G$ are $\eta$-regular and disjoint, then $\spair FG$ is $\eta$-regular.
    If $F,G$ are $\eta$-regular, then $\chain FG$ is $\eta$-regular.
    If $F,G$ are $\eta$-regular, then $F\rtensor G$ is $\eta$-regular.
  \item\itlabel{lemma:lift.pure:U} Let $U$ be a bounded operator on $\calH_{A_1}$.
    Then $F_1(U\psi_1)\ptensor F_2(\psi_i)\ptensor\dots\ptensor F_n(\psi_n) = F_1(U) \pb\paren{F_1(\psi_1)\ptensor\dots\ptensor F_n(\psi_n)}$.\footnoterepeat{footnote:lift.pure:reorder}
  \item\itlabel{lemma:lift.pure:butterfly} Let $\mathord\lozenge(\psi):=\psi\adj\psi$.
    Then $F_1(\mathord\lozenge(\psi_1))\mtensor\dots\mtensor F_n(\mathord\lozenge(\psi_n)) = \mathord\lozenge\pB\paren{F_1(\psi_1)\ptensor\dots\ptensor F_n(\psi_n)}$.
  \item\itlabel{lemma:lift.pure:subspace}
    Let $S$ be a closed subspace of $\calH_{A_1}$.
    Assume $\psi_2,\dots,\phi_n\neq0$.
    Then $F_1(\psi_1)\ptensor\dots\ptensor F_n(\psi_n) \in \regsubspace{F_1}S$ iff $\psi_1\in S$.\footnoterepeat{footnote:lift.pure:reorder}
    (See \autoref{sec:lift.sub} for the definition of $\regsubspace{F_1}S$.)
\end{compactenum}
\end{lemma}

The proof is given in \autoref{app:proof:lemma:lift.pure}.

\paragraph{Limitations of the approach.}
In contrast to the lifting of mixed states, there are a number of limitations to this approach of lifting pure states.

First, the resulting definition does depend on the arbitrary choice of vectors $\pureeta{\mathbf A},\purexi{F}$.
In fact, different choices of $\pureeta{\mathbf A},\purexi{F}$ will map the same pure states on $\psi_1,\dots,\psi_n$ to different states $F_1(\psi_1)\ptensor \dots\ptensor F_n(\psi_n)$.
(Although it is easy to verify that different choices of $\pureeta{\mathbf A},\purexi{F}$ will lead to  $F_1(\psi_1)\ptensor \dots\ptensor F_n(\psi_n)$ that are equal up to a global phase factor.)
In contrast, the lifting of mixed states ($F_1(\rho_1)\mtensor\dots\mtensor F_n(\rho_n)$) does not depend on arbitrary choices.

Second, there are some constraints concerning nested invocations of this construction.
(I.e., when one of the states $\psi_i$ in $F_1(\psi_1)\ptensor \dots\ptensor F_n(\psi_n)$ is again of the form $G_1(\phi_1)\ptensor \dots\ptensor G_n(\phi_n)$.)
In this case, we can flatten the nested application of the construction only when $F_1$ is $\eta$-regular (see \lemmaref{lemma:lift.pure:nested});
$\eta$-regularity guarantees that the arbitrarily chosen $\pureeta{\mathbf A}$ are compatible with each other.
In contrast, when lifting mixed states, we can flatten nested applications with no side-conditions (see \lemmaref{lemma:mixed.states:nested}).

Third, depending on the logical foundations we use, it might not even be able to globally fix arbitrary choices $\pureeta{\mathbf A},\purexi{F}$ throughout the category of Hilbert spaces.
(While satisfying the required constraints between them, e.g., $\pureeta{\mathbf A\otimes\mathbf B}=\pureeta{\mathbf A}\otimes\pureeta{\mathbf B}$.)
We might instead use, e.g., the category of Hilbert spaces $\calH_A$ with a distinguished unit vector $\eta_{\mathbf A}$.
While this seems like a very tiny constraint, it still means that we change our formal foundation just to accommodate the lifting of pure states.\footnote{%
  In our Isabelle/HOL formalization, we solve this issue by restricting all of results about the lifting of pure states to Hilbert spaces $\setC^\alpha$ (spaces with a canonical basis indexed by values of type $\alpha$), where $\alpha$ is of type class \texttt{default}.
  The builtin type class \texttt{default} applies only to types with a specified default element.
  For $\calH_A=\setC^\alpha$, we can then define $\pureeta{\mathbf A} := \ket{\mathtt{default}_{\alpha}}$ where $\mathtt{default}_\alpha$ is the default value of type $\alpha$.
  (And we crucially use the fact that $\mathtt{default}_{\alpha\times\beta}=(\mathtt{default}_\alpha,\mathtt{default}_\beta)$.)
  }

Finding a formalization of quantum references that does not pose the above challenges when lifting pure states (while still preserving the other features described in this paper) is an open problem.

\paragraph{Example.}
We illustrate how to use lifted pure states in a small example:
We will show that three CNOTs are equivalent to a qubit swap.
(This is near trivial in terms of the matrices but the focus of the example is not to show a complex situation but to explain the basic process.
In situations involving many different gates on different references, the present method would shine more as we would avoid complicated chains of swaps and tensor products with the identity.)
We want to show that the following two circuits evaluate the same unitary:
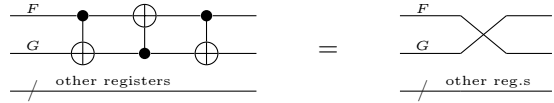
\begin{equation*}
  \begin{tikzpicture}[baseline=(current bounding box.center)]
    \initializeCircuit
    \newWires{F,G,rest}
    \stepForward{3mm}
    \labelWire[\tiny$F$]{F}
    \labelWire[\tiny$G$]{G}
    \drawWire{rest}
    \node at (\getWireCoord{rest}) {\footnotesize/};
    \stepForward{3mm}
    \labelWire[\tiny\rlap{other references}]{rest}
    \stepForward{2mm}
    \node[cnot=G,control=F] (cnot) {};
    \stepForward{5mm}
    \node[cnot=F,control=G] (cnot) {};
    \stepForward{5mm}
    \node[cnot=G,control=F] (cnot) {};
    \stepForward{5mm}
    \drawWires{F,G,rest}
  \end{tikzpicture}
  \qquad=\qquad
  \begin{tikzpicture}[baseline=(current bounding box.center)]
    \initializeCircuit
    \newWires{F,G,rest}
    \stepForward{3mm}
    \labelWire[\tiny$F$]{F}
    \labelWire[\tiny$G$]{G}
    \drawWire{rest}
    \node at (\getWireCoord{rest}) {\footnotesize/};
    \stepForward{3mm}
    \labelWire[\tiny\rlap{other reg.s}]{rest}
    \stepForward{2mm}
    \drawWires{rest,F,G}
    \stepForward{6mm}
    \crossWire{F}{G}
    \crossWire{G}{F}
    \skipWires{F,G}
    \stepForward{6mm}
    \drawWires{F,G,rest}
  \end{tikzpicture}
\end{equation*}
In the language of references, we express this as
\begin{equation}
  \spair FG(\CNOT) \cdot \spair GF(\CNOT) \cdot \spair FG(\CNOT)
  \qquad=\qquad
  \spair FG(\Uswap)
  \label{eq:triple.swap}
\end{equation}
where $\Uswap$ is the qubit swap $\Uswap(\ket x\otimes\ket y)=\ket y\otimes\ket x$.
Note that we did not need to explicitly refer to the ``other references'' because no operation happens on them.
But when we do wish to refer to them, we can refer to them as $\compl{\spair FG}$.

There are many ways to show \eqref{eq:triple.swap}.
For this example, we choose to show it by evaluating the circuit step by step on an initial state from the computational basis.
That is, we show
\begin{multline}
   \pB\paren{\spair FG(\CNOT) \cdot \spair GF(\CNOT) \cdot \spair FG(\CNOT)}
  \pB\paren{F\pb\paren{\ket x}\ptensor G\pb\paren{\ket y}\ptensor \compl{\spair FG}\pb\paren{\ket z}}
  \\
  ={}
    {\pb\spair FG(\Uswap)}\
    \pB\paren{F\pb\paren{\ket x}\ptensor G\pb\paren{\ket y}\ptensor \compl{\spair FG}\pb\paren{\ket z}}.
  \label{eq:triple.swap.basis}
\end{multline}
The left hand side is computed as
\begin{align*}
  &F\pb\paren{\ket x}\ptensor G\pb\paren{\ket y}\ptensor \compl{\spair FG}\pb\paren{\ket z}
    \starrel={} 
  \spair FG\pb\paren{\ket x\otimes\ket y}\ptensor \compl{\spair FG}\pb\paren{\ket z}
  \\ \xmapsto{\spair FG(\CNOT)}{} &
  \spair FG\pb\paren{\ket x\otimes\ket{y\oplus x}}\ptensor \compl{\spair FG}\pb\paren{\ket z}
    \starstarrel={}
  \spair GF\pb\paren{\ket{y\oplus x}\otimes\ket x}\ptensor \compl{\spair FG}\pb\paren{\ket z}
  \\ \xmapsto{\spair GF(\CNOT)}{} &
  \spair GF\pb\paren{\ket{y\oplus x}\otimes\ket y}\ptensor \compl{\spair FG}\pb\paren{\ket z}
    \starstarrel={}
  \spair FG\pb\paren{\ket y\otimes\ket{y\oplus x}}\ptensor \compl{\spair FG}\pb\paren{\ket z}
  \\ \xmapsto{\spair FG(\CNOT)}{} &
  \spair FG\pb\paren{\ket y\otimes\ket x}\ptensor \compl{\spair FG}\pb\paren{\ket z}
    \starrel={}
  F\pb\paren{\ket y}\ptensor G\pb\paren{\ket x}\ptensor \compl{\spair FG}\pb\paren{\ket z}.
\end{align*}
Here $(*)$ follows by \lemmaref{lemma:lift.pure:pair}.
Note that in a more complex situation with more references, we would only have to pull those references into the first tensor factor that are involved in the operation we are applying \emph{in this computation step}.
All other references in the tensor product can be left untouched (like $\compl{\spair FG}\pb\paren{\ket z}$ in our example).
And $(**)$ follows by \lemmaref{lemma:lift.pure:pair},\,\eqref{lemma:lift.pure:permute}.
(By first unwrapping the pair and then rewrapping in a different order.)
Each application of the CNOTs is done using \lemmaref{lemma:lift.pure:U}.
(All these steps are very mechanical and should not be hard to automate.
Our present Isabelle/HOL implementation, for example, proves each of the equalities by a simple tactic invocation.)

And similarly,
\begin{align*}
  &F\pb\paren{\ket x}\ptensor G\pb\paren{\ket y}\ptensor \compl{\spair FG}\pb\paren{\ket z}
  =
  \spair FG\pb\paren{\ket x\otimes\ket y}\ptensor \compl{\spair FG}\pb\paren{\ket z}
  \\ \xmapsto{\spair FG(\Uswap)}{} &
  \spair FG\pb\paren{\ket y\otimes\ket x}\ptensor \compl{\spair FG}\pb\paren{\ket z}
  =
  F\pb\paren{\ket y}\ptensor G\pb\paren{\ket x}\ptensor \compl{\spair FG}\pb\paren{\ket z}.
\end{align*}

Thus the lhs and rhs of \eqref{eq:triple.swap.basis} both evaluate to the same state.
Hence we have shown \eqref{eq:triple.swap.basis}.

From \eqref{eq:triple.swap.basis}, our initial goal \eqref{eq:triple.swap} follows by \lemmaref{lemma:lift.pure:separating}.

\subsection{Quantum channels}
\label{sec:quantum.channels}

Operations on mixed states are commonly described by \emph{quantum channels}\index{quantum channel}\index{channel!quantum} (a.k.a.~quantum operation or CPTPM\index{CPTPM}).
A quantum channel is a completely positive trace-preserving map $\calE:\tracecl(\calH_A)\to\tracecl(\calH_B)$.
(In the finite-dimensional case, the set $\tracecl(\calH)$ of all trace-class operators on $\calH$ is simply the set of all linear operators on $\calH$.)
In particular, $\calE(\rho)$ is a density operator if $\rho$ is.
When we work with subdensity operators, it is more natural to consider \emph{quantum subchannels}\index{quantum subchannel}\index{subchannel!quantum}.
$\calE:\tracecl(\calH_A)\to\tracecl(\calH_B)$ is a quantum subchannel iff it is completely positive and $\tr\calE(\rho)\leq\tr\rho$ for every positive $\rho$.
Quantum (sub)channels have a tensor product compatible with the tensor product of trace-class operators by \autoref{lemma:channel.tensor} (\autoref{app:misc.facts}).

An alternative definition of quantum channels, which we call Kraus-channels to avoid ambiguity, is the following:
$\calE:\tracecl(\calH_A)\to\tracecl(\calH_B)$ is a \emph{Kraus-channel}\index{Kraus-channel}\index{channel!Kraus-} iff there exists a family of bounded operators $M_i:\calH_A\to\calH_B$ such that $\calE(\rho)=\sum_i M_i\rho\adj{M_i}$ and $\sum_i \adj M_i M_i=I$.\footnote{\label{footnote:sum.convergence}%
  Here (and everywhere else in this paper) infinite sums of operators are supposed to converge with respect to the strong operator topology, SOT \cite[Definition IX.1.2]{conway13functional}.
  
  Sums can be over arbitrary, not necessarily countable index sets; see, e.g., \cite[Definition 4.11]{conway13functional} for a precise definition of such sums.
  
  Note that for sums of \emph{positive} operators, convergence of in the SOT is equivalent to convergence in various other topologies by \autoref{lemma:incr.net.lim} (and for sums of positive trace-class operators in even more topologies by \autoref{lemma:incr.net.lim.tr}).
  In particular, the two sums in the definition of Kraus-(sub)channels are with respect to the trace norm and the SOT in \cite{holevo11entropy} and with respect to the trace norm and the WOT \cite[Definition IX.1.2]{conway13functional} in \cite{friedland19choi}; since they are sums of positive operators the definitions from \cite{holevo11entropy,friedland19choi} are nonetheless equivalent to ours.}
A Kraus-channel is always a quantum channel, and the notions coincide if $\calH_A,\calH_B$ are separable (i.e., countably dimensional) Hilbert spaces \cite[Section 3.1]{holevo11entropy}.
A Kraus-subchannel is defined in the same way, except with the condition $\sum_i \adj M_i M_i\leq I$ instead of $\sum_i \adj M_i M_i=I$.
A Kraus-subchannel is always a quantum subchannel, and the notions coincide if $\calH_A,\calH_B$ are separable Hilbert spaces \cite[Theorem 9]{friedland19choi}.

In the following, we will only consider quantum channels where the domain and codomain are the same space. (I.e., quantum channels $\tracecl(\calH_A)\to\tracecl(\calH_A)$).
This is because we will consider quantum channels that operate \emph{on} the content of a reference, not quantum channels between references.
(However, that does not preclude an operation that, e.g., takes data from one reference $F:\mathbf A\to \mathbf C$ and puts the result of the computation into another reference $G:\mathbf B\to\mathbf C$.
This would simply be described as a quantum channel on $\mathbf A\otimes\mathbf B$ applied to $\spair FG$.)

\paragraph{Lifting quantum channels.}
If $F:\mathbf A\to\mathbf B$ is an iso-reference, then there is a very natural way of lifting a quantum channel $\calE:\mathbf A\to\mathbf A$.
By \autoref{lemma:isoreg-decomp}, $F(a)=Ua\adj U$ for some unitary $U:\calH_A\to\calH_B$.
We can thus define $F(\calE):\mathbf B\to\mathbf B$ as $F(\calE)(\rho):=U\calE(\adj U\rho U)\adj U$.
Or equivalently, $F(\calE):=F\circ \calE\circ F^{-1}$.
(The latter gives us a different interpretation:
An iso-reference is a quantum channel (see \autoref{lemma:lift.channel}\,\eqref{lemma:isoreg:channel} below), so $F(\calE)$ is just the composition of three quantum channels.)

For non-iso-references $F$, this definition does not work.
($F$ is not a quantum channel, nor can it be represented as $Ua\adj U$.)
Fortunately, using complements, we can very easily extend the above idea to arbitrary references.
Let $\compl F:\mathbf A'\to\mathbf B$ be the complement of $F$.
We can lift $\calE$ to become a quantum channel on $\mathbf A\otimes\mathbf A'$;
the tensor product $\calE\otimes \id$ of quantum channels achieves this.
And $\pair F{\compl F}$ is an iso-reference $\mathbf A\otimes\mathbf A'\to\mathbf B$ by definition of the complement.
So we can lift  $\calE\otimes \id$ to a quantum channel on $\mathbf B$ through $\pair F{\compl F}$.
Putting this together, we get the following definition:
$\symbolindexmark\regchannel{F(\calE)} := \pair F{\compl F} \circ (\calE\otimes\id) \circ \pair F{\compl F}^{-1}$.
Note that this definition also applies if $\calE$ is a subchannel.

Note that this construction is only applicable when the domain and codomain of $\calE$ are the same.
(I.e., we cannot apply it to $\calE:\mathbf A\to\mathbf B$ with $\mathbf A\neq\mathbf B$.)
This is because $\regchannel F\calE$ intuitively represents ``applying $\calE$ to $F$'', not ``applying $\calE$ to $F$ and storing the result in a different reference''.

The definition uses an arbitrarily chosen complement $\compl F$ of $F$ (recall that the complement is only unique up to equivalence).
The following lemma shows that the definition does not depend on the choice of complement, and that $F(\calE)$ interacts as expected with other constructions.

\begin{lemma}\lemmalabel{lemma:lift.channel}
  Let $\calE$, $\calF$ be quantum subchannels.
  Let $F,G$ be references.
  \begin{compactenum}[(i)]
  \item\itlabel{lemma:lift.channel:channel} If $\calE$ is a quantum (sub)channel, so is $F(\calE)$.
  \item\itlabel{lemma:lift.channel:choice.compl}
    Let $F,G$ be complements.
    Then $\regchannel F\calE = \spair FG \circ (\calE\otimes\id) \circ \spair FG^{-1}$.
    (This means the choice of complement is irrelevant.)
  \item\itlabel{lemma:lift.channel:compose}
    $\regchannel F\calE\circ \regchannel F\calF = \regchannel F{\calE\circ\calF}$.
  \item\itlabel{lemma:lift.channel:pair}
    If $F,G$ are disjoint, then $\spair FG(\calE\otimes\calF) = F(\calE)\circ G(\calF) = G(\calF)\circ F(\calE)$.
  \item \itlabel{lemma:lift.channel:id} $\regchannel F\id=\id$.
  \item \itlabel{lemma:lift.channel:apply}
    For a partition $F_1,\dots,F_n$ and trace-class operators $\rho_1,\dots,\rho_n$, $\pb\regchannel{F_1}\calE\pb\paren{F_1(\rho_1)\mtensor\dots\mtensor F_n(\rho_n)} = F_1(\calE(\rho_1))\mtensor F_2(\rho_2)\mtensor\dots\mtensor F_n(\rho_n)$.%
    \footnote{\label{footnote:lift.channel:reorder}%
      We state this only for the first reference $F_1$ for simplicity, but by \lemmaref{lemma:mixed.states:permute}, it applies to all other references, too.}
  \item\itlabel{lemma:lift.channel:chain} $F\pb\paren{G(\calE)} = (\chain FG)(\calE)$.
  \item\itlabel{lemma:lift.channel:swap} $\sigma(\calE\otimes\calF) = \calF\otimes\calE$.
  \item\itlabel{lemma:lift.channel:assoc} $\assoc\pb\paren{\calE\otimes(\calF\otimes\calG)} = (\calE\otimes\calF)\otimes\calG$ and $\assocp\pb\paren{(\calE\otimes\calF)\otimes\calG} = \calE\otimes(\calF\otimes\calG)$.
  \item \itlabel{lemma:isoreg:channel} If $F:\mathbf A\to\mathbf C$, $G:\mathbf B\to\mathbf D$ are iso-references, then $F|_{\tracecl(\calH_A)}$ ($F$ restricted to trace-class operators) is a bijective quantum channel, and $(F\rtensor G)|_{\tracecl(\calH_A\otimes\calH_B)} = F|_{\tracecl(\calH_A)}\otimes G|_{\tracecl(\calH_B)}$.
  \item \itlabel{lemma:lift.channel:sem}
    For a bounded operator $U$, let $\Sem(U)$ denote the map $\rho\mapsto U\rho \adj U$.
    Then $\Sem(F(U))=F(\Sem(U))$.
  \end{compactenum}
\end{lemma}

The proof is given in \autoref{app:proof:qchannels}.

\paragraph{Lifting Kraus-channels.}
A Kraus-channel $\calE$ can be written as $\calE(\rho)=\sum_i M_i\rho\adj{M_i}$ for bounded operators $M_i$ with $\sum_i \adj M_i M_i=1$.
It is natural to lift a Kraus channel through the reference $F$ using the following definition:
$\symbolindexmark\regchannel{F(\calE)}(\rho) := \sum_i F(M_i)\rho \adj{F(M_i)}$.
But since for given $\calE$, the choice of the Kraus operators $M_i$ is not unique, it could be that this definition is not well-defined.
The following lemma shows that this is not the case, and that this definition is compatible with the one above.
\begin{lemma}\lemmalabel{lemma:lift.kraus}
  Let $F$ be a reference.
  \begin{compactenum}[(i)]
  \item\itlabel{lemma:lift.kraus:channel} If $\calE$ is a Kraus-(sub)channel, then $F(\calE)$ is a Kraus-(sub)channel.
  \item\itlabel{lemma:lift.kraus:indep} The definition of $F(\calE)$ does not depend on the choice of $M_i$ in $\calE(\rho)=\sum_i M_i\rho\adj{M_i}$.
  \item\itlabel{lemma:lift.kraus:samedef} For a Kraus-(sub)channel $\calE$, $F(\calE)$ as defined for Kraus-(sub)channels, and $F(\calE)$ as defined for quantum (sub)channels, coincide.
  \end{compactenum}
\end{lemma}

The proof is given in \autoref{app:proof:qchannels}, \autopageref{page:proof:lemma:lift.kraus}.

Note that by \eqref{lemma:lift.kraus:samedef}, the definition $F(\calE)$ for Kraus-channels inherits the laws from \autoref{lemma:lift.channel} as well.
Furthermore, \eqref{lemma:lift.kraus:samedef} gives us a convenient means of computing $\regchannel F\calE$ for (sub)channels $\calE$ that are explicitly given in Kraus-operator representation.
E.g., if $\calE(\rho)=U\rho\adj U$ for a unitary $U$, then $\regchannel F\calE(\rho) = F(U)\rho \adj{F(U)}$.

\subsection{Partial trace}
\label{sec:partial.trace}

The partial trace is an operation that maps density operators on a system $C$ to density operators on a subsystem $A$.
For example, on a bipartite system $\calH_B:=\calH_A\otimes\calH_C$, we define $\tr_C:\tracecl(\calH_B)\to\tracecl(\calH_A)$ to be the unique operation such that $\tr_C(\rho\tensor\sigma) = \rho\cdot\tr\sigma$.
(We say we ``trace out $C$''.)
Note that the commonly used notation $\tr_C$ indicates which parts of the system $B$ we \emph{remove} (namely $C$), not which part we \emph{keep} (namely $A$).
This is notationally somewhat awkward because we need to be always able to list all other parts of the (potentially quite complex) system besides $A$.
Thus, here we use a slightly different notational convention:
Instead of $\tr_C$ (tracing out $C$), we write $\trin A$ (tracing in $A$).

Since in the context of this work, a reference $F:\mathbf A\to\mathbf B$ describes a subsystem of $\calH_B$, it is natural to want to define $\trin F$, the operation that maps a density operator $\rho\in\tracecl(\calH_B)$ to $\trin F\rho\in\tracecl(\calH_A)$.
Intuitively, $\trin F\rho$ describes the state of the content of the reference $F$ (if we remove the rest of the memory).

To define this, we first define a special case: For Hilbert spaces $\calH_A,\calH_{A'}$, let $\symbolindexmark\trinFst{\trin\Fst}:\tracecl(\calH_A\otimes\calH_{A'})\to\tracecl(\calH_A)$ denote the unique bounded linear map with $\trin\Fst(\rho\tensor\sigma) = \rho\cdot\tr\sigma$.
This is the operation that traces out everything but the content of reference $\Fst$ (the first part of the system), and traces in $\Fst$.
Furthermore, for an iso-reference $I$, we can define $\trin I:=I^{-1}$. (Recall that an iso-reference represents the whole memory, so nothing needs to be traced out; we only need to map the state into the correct representation using $I^{-1}$.)
Note that for any reference $F$, we have that $F=\chain{\pair{F}{\compl F}}\Fst$ (Law~\ref{law:pair.select}).
Thus it makes sense to define $\trin F := \trin\Fst\circ \trin{\pair{F}{\compl F}}$.
Putting all this together, we get the definition:
$\symbolindexmark\trin{\trin F} := {\trinFst} \circ {\pair{F}{\compl F}}^{-1}$.\footnote{Note that this is compatible with the special case definition of $\trinFst$ by \autoref{lemma:partial.tr}\,\eqref{lemma:partial.tr:pair} below.}

\begin{lemma}\lemmalabel{lemma:partial.tr} Let $F:\mathbf A\to\mathbf B$ and $G$ be references.
  \begin{compactenum}[(i)]
  \item\itlabel{lemma:partial.tr:welldef} $\trin F$ is well-defined. (That is, the auxiliary operation $\trinFst$ exists and is unique.)
  \item\itlabel{lemma:partial.tr:channel} $\trin F$ is a quantum channel.
  \item\itlabel{lemma:partial.tr:choice.compl} The definition of $\trin F$ does not depend on the choice of complement. 
    That is, if $F,G$ are complements, $\trin F = {\trinFst} \circ {\pair{F}G}^{-1}$.
  \item\itlabel{lemma:partial.tr:iso} For iso-references $F$, $\trin F = F^{-1}|_{\tracecl(\calH_A)}$.
  \item\itlabel{lemma:partial.tr:pair} If $F,G$ are complements, for all trace-class operators $\rho,\sigma$, $\trin F \pair{F}G(\rho\otimes\sigma) = \rho\cdot\tr\sigma$.
  \item\itlabel{lemma:partial.tr:partition} For a partition $F_1,\dots,F_n$ and trace-class operators $\rho_1,\dots,\rho_n$,
    $\trin {F_i} \pb\paren{F_1(\rho_1)\mtensor\dots\mtensor F_n(\rho_n)} = \rho_i \cdot \prod_{j\neq i}\tr\rho_j$.
  \item\itlabel{lemma:partial.tr:chain} $\trin{\chain FG} = \trin G\circ\trin F$.
  \item\itlabel{lemma:partial.tr:channel.only} For a quantum subchannel $\calE$, ${\trin F} \circ F(\calE) = \calE \circ\trin F$.
  \item\itlabel{lemma:partial.tr:channel.other} For a quantum channel $\calE$, and disjoint $F,G$, ${\trin F}\circ G(\calE) = \trin F$.
  \end{compactenum}
\end{lemma}

The proof is given in \autoref{app:proof:lemma:partial.tr}.

\paragraph{Entropy.} The definition of the partial trace is also useful to define information-theoretic properties that refer to subsystems.
For example, we mentioned in the introduction (\autopageref{page:intro.entropy}) that we can use references to define the Shannon entropy \cite{vonneumann32grundlagen} of the content of a reference $F$.
If $H_\mathit{basic}(\rho):=\rho\log\rho$ (i.e., the entropy of the whole quantum state $\rho$), then we can just define the Shannon entropy $H(F)_\rho$ of the subsystem $F$ (where $\rho$ is the state of the whole system) as follows: $H(F)_\rho := H\pb\paren{\trin F(\rho)}$.
We can also easily define the conditional Shannon entropy \cite[Definition 5.1]{tomamichel16finite} as $H(F|G)_\rho := H\pb\paren{\spair FG}_\rho - H(G)_\rho$ for disjoint $F,G$.

\subsection{Measurements}

There are different formalisms for defining measurements, in order of increasing generality: complete measurements, projective measurements, POVMs, and generalized measurements.
We discuss each of these in turn.
In the following, $F$ will always be a quantum reference $\mathbf A\to\mathbf B$.

\paragraph{Projective measurements.}
A \emph{projective measurement}\index{projective measurement}\index{measurement!projective} on $\calH_A$ is defined by a family of projectors $M=\{P_i\}_{i\in J}$ such that $P_iP_j=0$ for $i\neq j$ and $\sum_i P_i=1$.\footnote{\label{footnote:sum.conv.ref}%
  Convergence of sums is with respect to the SOT, see \autoref{footnote:sum.convergence} for details.}
Here $J$ is the set of outcomes. ($J$ is allowed to be infinite.)
The \emph{probability of outcome $j\in J$ using $M$ given pure state $\psi$ (or mixed state $\rho$)} is $\norm{P_j\psi}^2$ ($\tr P_j\rho\adj{P_j}$).
The \emph{post-measurement state for outcome $j\in J$ using $M$ given pure state $\psi$ (or mixed state $\rho$)} is $P_j\psi/\norm{P_j\psi}$ ($P_j\rho\adj{P_j}/tr P_j\rho\adj{P_j}$).
The \emph{non-normalized post-measurement state for outcome $j\in J$ using $M$ given pure state $\psi$ (or mixed state $\rho$)} is $P_j\psi$ ($P_j\rho\adj{P_j}$).
Note that from the non-normalized post-measurement state, we can compute the post-measurement state and the outcome probability.
Thus our results below will be stated only for the non-normalized post-measurement state for brevity.

We define $\symbolindexmark\regmeas{\regmeas FM}:=\{F(P_i)\}_{i\in J}$.
\begin{lemma}\lemmalabel{lemma:proj.meas}
  Let $M=\braces{P_i}_i$ be a projective measurement.
  Let $F_1,\dots,F_n$ be a partition.
  \begin{compactenum}[(i)]
  \item\itlabel{lemma:proj.meas:is.meas} $F(M)$ is a projective measurement.
  \item\itlabel{lemma:proj.meas:uni}
    For a unitary $U$, let $UM$ denote the projective measurement $\{UP_i\adj U\}$.
    Then $F(U)\,F(M)=F(UM)=(\chain F{\sandwich U})(M)$.
  \item\itlabel{lemma:proj.meas:pure}
    Let $\psi_1,\dots,\psi_n$ be vectors.
    Let $\psi := F_1(\psi_1)\ptensor\dots\ptensor F_n(\psi_n)$.
    Let $\psi_1'$ denote the non-normalized post-measurement state for outcome $j$ using $M$ given $\psi_1$.
    Let $\psi'$ denote the non-normalized post-measurement state for outcome $j$ using $F_1(M)$ given $\psi$.
    Then $\psi' = F(\psi_1')\ptensor F(\psi_2)\ptensor\dots\ptensor F(\psi_n)$.%
    \footnote{\label{footnote:proj.meas:reorder}%
      We state this only for the first reference $F_1$ for simplicity, but by \lemmaref{lemma:lift.pure:permute}/\lemmaref{lemma:mixed.states:permute}, it applies to all other references, too.}
  \item\itlabel{lemma:proj.meas:mixed} Let $\rho_1,\dots,\rho_n$ be trace-class operators. Let $\rho := F_1(\rho_1)\mtensor\dots\mtensor F_n(\rho_n)$.
    Let $\rho_1'$ denote the non-normalized post-measurement state for outcome $j$ using $M$ given $\rho_1$.
    Let $\rho'$ denote the non-normalized post-measurement state for outcome $j$ using $F_1(M)$ given $\rho$.
    Then $\rho' = F(\rho_1')\mtensor F(\rho_2)\mtensor\dots\mtensor F(\rho_n)$.%
    \footnoterepeat{footnote:proj.meas:reorder}
  \end{compactenum}
\end{lemma}
Here \eqref{lemma:proj.meas:uni} describes a common transformation of projective measurements and shows that it is compatible with lifting through references, and with a mapping the reference (replacing $F$ by $\chain F{\sandwich U}$).

The proof is given in \autoref{app:proof:measurements}.

\paragraph{Complete measurements.}
A complete measurement $M$ on $\calH_A$ is described by an orthonormal basis $M=\braces{\psi_j}_{j\in J}$.
The intuitive interpretation is that we measure which of the states $\psi_j$ the reference is in.
A complete measurement is a special case of a projective measurement, namely $M:=\braces{\psi_j\adj{\psi_j}}_{j\in J}$.
Thus a complete measurement can be lifted through a reference simply by interpreting it as a projective measurement and then applying $F(M)$ as defined in the previous paragraph.
Note that this means that a complete measurement on $F$ does not lead to a complete measurement on $\calH_B$.
This is unsurprising: if $\calH_B$ has a higher dimension than $\calH_A$, then there cannot be an orthonormal basis of the same cardinality $\abs J$ on $\calH_B$; no complete measurement on $\calH_B$ with outcomes $J$ exists.

\paragraph{POVMs.}
A \emph{POVM}\index{POVM}\index{measurement!POVM} on $\calH_A$ is defined by a family of positive bounded operators $M=\{M_i\}_{i\in J}$ on $\calH_A$ such that $\sum_i M_i=1$.\footnoterepeat{footnote:sum.conv.ref}
Here $J$ is the set of outcomes.\footnote{There are more general definitions of POVMs where $J$ is a measurable space.
  Hence the name POVM which means ``positive operator-valued \emph{measure}''.
  For example, a measurement of the real-valued position of a particle (with uncountably-many outcomes) cannot be described by a POVM according to the definition we use here.
  POVMs as described here describe always lead to discrete probability distributions of outcomes which makes them easier to handle.}  
($J$ is allowed to be infinite.)
The \emph{probability of outcome $j\in J$ using $M$ given pure state $\psi$ (or mixed state $\rho$)} is $\norm{\sqrt{M_j}\psi}^2$ ($\tr M_j\rho$).
Note that a POVM on its own does not specify the post-measurement-states; the generalized measurements below can be used if post-measurement states are needed.
Note that every projective measurement is also a POVM.

We define $\symbolindexmark\regmeas{\regmeas FM}:=\{F(M_i)\}_{i\in J}$.

\begin{lemma}\lemmalabel{lemma:povm}
  Let $M=\{M_i\}_i$ be a POVM.
  Let $F$ be a reference.
  Let $F_1,\dots,F_n$ be a partition.
  \begin{compactenum}[(i)]
  \item\itlabel{lemma:povm:is.povm}
    $F(M)$ is a POVM.
  \item\itlabel{lemma:povm:uni}
    For a unitary $U$, let $UM$ denote the POVM $\{UM_i\adj U\}_i$.
    Then $F(U)\,{F(M)}=F\pb\paren{UM}=(\chain F{\sandwich U})(M)$.
  \item\itlabel{lemma:povm:pure}
    Let $\psi_k$ be a vector, and let $\psi_1,\dots,\psi_n$ except $\psi_k$ be unit vectors. 
    Let $\psi := F_1(\psi_1)\ptensor\dots\ptensor F_n(\psi_n)$.
    Then the probability of outcome $j$ using $M$ given $\psi_k$ equals the probability of outcome $j$ using $F_k(M)$ given $\psi$.
  \item\itlabel{lemma:povm:mixed}
    Let $\rho_k$ be a trace-class operator, and $\rho_1,\dots,\rho_n$ except $\rho_k$ be density operators. 
    Let $\rho := F_1(\rho_1)\mtensor\dots\mtensor F_n(\rho_n)$.
    Then the probability of outcome $j$ using $M$ given $\rho_k$ equals the probability of outcome $j$ using $F_k(M)$ given $\rho$.
  \item\itlabel{lemma:povm:projmeas}
    Let $\rho$ be a trace-class operator (and $\psi$ be a vector).
    Assume that $M$ is also a projective measurement (i.e., all $M_i$ are projectors).
    Then the probability of outcome $j$ using $F(M)$ given $\rho$ or $\psi$ (where $M$ is interpreted as a POVM)
    equals the  probability of outcome $j$ using $F(M)$ given $\rho$ or $\psi$ (where $M$ is interpreted as a projective measurement)
  \end{compactenum}
\end{lemma}

The proof is given in \autoref{app:proof:measurements}, \autopageref{page:proof:povm}.

\paragraph{General measurements.}
A \emph{general measurement}\index{general measurement}\index{measurement!general} on $\calH_A$ is defined by a family of bounded operators $M=\{M_i\}_{i\in J}$ such that $\sum_i \adj{M_i}M_i=1$.\footnoterepeat{footnote:sum.conv.ref}
Here $J$ is the set of outcomes. ($J$ is allowed to be infinite.)
The \emph{probability of outcome $j\in J$ using $M$ given pure state $\psi$ (or mixed state $\rho$)} is $\norm{M_j\psi}^2$ ($\tr M_j\rho\adj{M_j}$).
The \emph{post-measurement state for outcome $j\in J$ using $M$ given pure state $\psi$ (or mixed state $\rho$)} is $M_j\psi/\norm{M_j\psi}$ ($M_j\rho\adj{M_j}/tr M_j\rho\adj{M_j}$).
The \emph{non-normalized post-measurement state for outcome $j\in J$ using $M$ given pure state $\psi$ (or mixed state $\rho$)} is $M_j\psi$ ($M_j\rho\adj{M_j}$).
Note that from the non-normalized post-measurement state, we can compute the post-measurement state and the outcome probability.
Thus our results below will be stated only for the non-normalized post-measurement state for brevity.

Note that a projective measurement is a special case of a general measurement (with matching definitions of outcome-probabilities and post-measurement-states).
And note that any general measurement~$M$ gives rise to a POVM $M'$ via $M'_i:=\adj{M_i}M_i$ (such that the outcome-probabilities of $M'$ are the same as those of $M$ for the same state).

We define $\symbolindexmark\regchannel{\regchannel FM}:=\{F(M_i)\}_{i\in J}$.

\begin{lemma}\lemmalabel{lemma:general.meas}
  Let $M=\{M_i\}_{i\in J}$ be a general measurement.
  Let $F_1,\dots,F_n$ be a partition.
  \begin{compactenum}[(i)]
  \item\itlabel{lemma:general.meas:is.meas} $F(M)$ is a general measurement.
  \item\itlabel{lemma:general.meas:pure} Let $\psi_1,\dots,\psi_n$ be vectors.
    Let $\psi := F_1(\psi_1)\ptensor\dots\ptensor F_n(\psi_n)$.
    Let $\psi_1'$ denote the non-normalized post-measurement state for outcome $j$ using $M$ given $\psi_1$.
    Let $\psi'$ denote the non-normalized post-measurement state for outcome $j$ using $F(M)$ given $\psi$.
    Then $\psi' = F(\psi_1')\ptensor F(\psi_2)\ptensor\dots\ptensor F(\psi_n)$.%
    \footnote{\label{footnote:general.meas:reorder}%
      We state this only for the first reference $F_1$ for simplicity, but by \lemmaref{lemma:lift.pure:permute}/\lemmaref{lemma:mixed.states:permute}, it applies to all other references, too.}
  \item\itlabel{lemma:general.meas:mixed} Let $\rho_1,\dots,\rho_n$ be trace-class operators. Let $\rho := F_1(\rho_1)\mtensor\dots\mtensor F_n(\rho_n)$.
    Let $\rho_1'$ denote the non-normalized post-measurement state for outcome $j$ using $M$ given $\rho_1$.
    Let $\rho'$ denote the non-normalized post-measurement state for outcome $j$ using $F(M)$ given $\rho$.
    Then $\rho' = F(\rho_1')\ptensor F(\rho_2)\ptensor\dots\ptensor F(\rho_n)$.%
    \footnoterepeat{footnote:general.meas:reorder}
  \item\itlabel{lemma:general.meas:povm}
    Let $M'$ be the POVM corresponding to $M$.
    Then $F(M')$ is the POVM corresponding to $F(M)$.
  \end{compactenum}
\end{lemma}

The proof is given in \autoref{app:proof:measurements}, \autopageref{page:proof:general.meas}.

}

\delaytextsv{isabelle}{
  \newcommand\INCLUDEONCEdfirjfkpsohportwe{}

\section{Isabelle formalization}
\label{sec:isabelle}

We have formalized a large fraction of the results from this paper in Isabelle/HOL
\cite{nipkow2002isabelle, isabelle-webpage},
namely the everything from Sections~\ref{sec:generic}--\ref{sec:complements}.
The source code of our formalization \ifanonymous
  in \cite{formalization-anonymous}.
  \else
  \fullshort{can be found at \cite{References-AFP} (infinite-dimensional update \cite{References-WIP})}{\error}.
\fi
In particular, we formalized the quantum references in the more general (and mathematically more involved) infinite-dimensional case.
\shortonly{It relies on some additional as yet unpublished Isabelle libraries for tensor products which we also attach to this submission. We also provide an overview table with lines-of-code counts, \autoref{fig:isabelle.loc}.}
Note that a large part of the extra formalization effort in the \emph{infinite-dimensional} case lied in the development of the \shortonly{(as yet unpublished) }library for tensor products \showafter{2024}{02}{14}{\fullonly{\cite{tensor-product-library-not-anon}}} by the same authors as this paper.

When formalizing, we encountered the following challenge:
In this paper, we structured our results by first stating generic axioms and deriving
generic laws from them (\autoref{sec:generic}), and then instantiating
these axioms in the quantum and the classical case
(Sections~\ref{sec:quantum-overall} and~\ref{sec:classical}). Unfortunately,
in the logic HOL, this is not possible.
While HOL (as implemented in Isabelle/HOL) allows us to state and prove facts that are type-polymorphic and thus hold for any type, we cannot state a fact that holds for all type constructors.
But the results from
\autoref{sec:generic} are basically of the form: ``for any type constructor \texttt{update} (so that $\alpha$ \texttt{update}, when $\alpha$ ranges over different types, represents all the objects of the category) satisfying axioms X, the laws Y hold''.
And then we
instantiate that type constructor \texttt{update} with the type
constructor of linear maps (quantum case) or of partial functions (classical
case). Therefore, we cannot express this generally in Isabelle/HOL.

We used the following workaround: We have a theory \texttt{Axioms}
that states the general axioms for a type constructor
\texttt{update},\footnote{Restricted to a type class \texttt{domain}.}
and then we derive the laws generically from these axioms in the
theory \texttt{Laws}. Then we instantiate these axioms in the theories
\texttt{Axioms\_Quantum} and~\texttt{Axioms\_Classical}. By
instantiating, we mean that we write an independent theory (that does not import \texttt{Axioms}) that has a proven lemma for each axiom, but
with the type constructor \texttt{update} replaced by the appropriate type constructor for the quantum/classical setting.
(That is, \texttt{cblinfun}, the type of complex bounded linear functions, and
\texttt{map}, the type of partial functions.) Then we use a Python script to
copy the theory \texttt{Laws} into \texttt{Laws\_Quantum} and
\texttt{Laws\_Classical}, importing \texttt{Axioms\_Quantum} and
\texttt{Axioms\_Classical} instead of \texttt{Axiom}. The script then
substitutes \texttt{update} by the appropriate type constructor (and does potential other replacements of constant or lemma names, as desired) based on the specification given in the theory headers.
All proofs from
\texttt{Laws} then automatically carry over, and improvements in \texttt{Laws} are
automatically adapted to \texttt{Laws\_Quantum} and
\texttt{Laws\_Classical}. This constitutes a poor man's substitute for
higher-order polymorphism.

\delaytext{figure isabelle loc}{
\begin{figure*}
  \newcounter{totalLoc}
  \newcommand\loc[1]{##1\addtocounter{totalLoc}{##1}}
  \centering
\begin{tabular}[t]{|>{\ttfamily}lrp{3.05in}|}  
  \hline
  \normalfont \textbf{Theory} & \textbf{LoC} & \textbf{Purpose}\\
  \hline
  Misc   & \loc\locMisc & Various lemmas that are missing in other libraries \\
  \small Tensor\_Product\_Matrices \hspace*{-5mm} &  \loc\locTensorProductMatrices & Theorems for translating finite-dimensional tensor products to explicit matrix representation (for explicit computations) \\
  \hline
  Axioms & \loc\locAxioms & Generic statement of the axioms \\
  Laws   & \loc\locLaws & Laws derived generically from the axioms \\
  Axioms\_Complement & \loc\locAxiomsComplement & Generic statement of the axioms for complements (existence, uniqueness) \\
  Laws\_Complement & \loc\locLawsComplement & Laws derived generically for complements \\
  \hline
  Quantum & \loc\locQuantum & Some useful facts about quantum mechanics (unrelated to references) \\
  Axioms\_Quantum & \loc\locAxiomsQuantum & Instantiation of \texttt{Axioms} for the quantum case \\
  Laws\_Quantum & --- & Autogenerated: \texttt{Laws} for the quantum case \\
  Quantum\_Extra & \loc\locQuantumExtra & Additional laws about references specific to the quantum case \\
  Quantum\_Extra2 & \loc\locQuantumExtraTwo & Additional laws specific to the quantum case (dependent on complements) \\
  QHoare & \loc\locQHoare & Definition of quantum Hoare logic and derived rules \\
  Teleport & \loc\locTeleport & Example: Analysis of quantum teleportation \\
  Axioms\_Complement\_Quantum & \loc\locAxiomsComplementQuantum & Instantiation of \texttt{Axioms\_Complement} for the quantum case \\
  Laws\_Complement\_Quantum & --- & Autogenerated: \texttt{Laws\_Complement} for the quantum case \\
  Pure\_States & \loc\locPureStates & Lifting of pure states through references \\
  \hline
  Axioms\_Classical & \loc\locAxiomsClassical & Instantiation of \texttt{Axioms} for the classical case \\
  Laws\_Classical & --- & Autogenerated: \texttt{Laws} for the classical case \\
  Classical\_Extra & \loc\locClassicalExtra & Additional laws about references specific to the classical case \\
  \hline
  \normalfont \textbf{Total:} & \arabic{totalLoc} & \\
  \hline
\end{tabular}
\caption{\label{fig:isabelle.loc}Isabelle/HOL theories\anonymous{}{ \cite{References-AFP}}. Lines of code are without blanks or comments.}
\end{figure*}
}

\usedelayedtext{figure isabelle loc}

The results from \autoref{sec:generic} are therefore proven in the
theories \texttt{Axioms} and \texttt{Laws}.

The results from \autoref{sec:quantum-overall} (quantum references) are in the
theories \texttt{Axioms\_Quantum} and~\texttt{Quantum\_Extra}.
\texttt{Axioms\_Quantum} instantiates the axioms from \texttt{Axioms}.
Specifically, we use the complex bounded operator library
\cite{bounded-operators} as the mathematical foundation. It defines
the type $(\alpha,\beta)$ \texttt{cblinfun}\index{cblinfun@\texttt{cblinfun}} (short:
$\alpha \symbolindexmark\CBL\CBL \beta$) as the type of bounded linear\footnote{In the
  finite dimensional case that we consider here, bounded linear is the
  same as linear.} functions between complex vector spaces
$\alpha,\beta$. And the type $\alpha$ \texttt{ell2}\index{ell2@\texttt{ell2}} of square-summable
functions $\alpha \to \setC$.\footnote{For finite types $\alpha$, as
  in our case, this is just $\setC^\alpha$.}  The objects of the
category (the ``updates'') are then
$\alpha\ \mathtt{update} := \paren{\alpha\ \mathtt{ell2}\CBL\alpha\
\mathtt{ell2}}$ (for finite $\alpha$), i.e., the linear functions
$\setC^\alpha\to\setC^\alpha$. Based on this formalization, the tensor
product can be defined very nicely, because the tensor product of
$\alpha\ \mathtt{update}$ and $\beta\ \mathtt{update}$ is simply
$(\alpha\times\beta)\ \mathtt{update}$. Using the existing
infrastructure from the bounded operator library, and our own
formalization of finite dimensional tensor products (theory
\texttt{Finite\_Tensor\_Product}, the bounded operator library does
not have a tensor product), defining the required constants (i.e.,
predicates that specify what the pre-references and references are) and
proving the axioms is then very easy. (Just \locAxiomsQuantum\ lines of code, see
\autoref{fig:teleport}.)

All the laws specific to the quantum setting mentioned in
\autoref{sec:quantum-overall} are proven in theory
\texttt{Quantum\_Extra}. Again, the number of lines is very small
(\locQuantumExtra) showing that our formalism is indeed well-suited for
machine-checked formalization.

Finally, the results from \autoref{sec:classical} are formalized in
the theories \texttt{Axioms\_Classical} and \texttt{Classical\_Extra},
with $\alpha\ \mathtt{update}$ instantiated as $(\alpha,\alpha)\ \mathtt{map} = (\alpha\Rightarrow\alpha\ \mathtt{option})$,
the type of partial functions on $\alpha$.

\paragraph{Examples.}
The Hoare logic from \autoref{sec:hoare} is formalized in the theory
\texttt{QHoare}. Again, using references turns out to be very easy, the
(admittedly simple) Hoare logic only needs \locQHoare\ lines of
formalization, including the derivation rules.  And the analysis of
the teleportation protocol from \autoref{sec:teleport} is presented in
\texttt{Teleport}, in \locTeleport\ lines.  The structure of this file is as
follows: First we define a ``locale'' axiomatizing the references
$X,A,B,\Phi$.\footnote{In Isabelle, a locale \cite{ballarin14module, ballarin10tutorial} is a kind
  of package for definitions and lemmas based on some constants and
  axioms such that the constant and axioms (the references and their
  mutual compatibility in our case) can be instantiated later.} In
this locale, we define teleportation and prove the Hoare judgment
\eqref{eq:hoare.teleport}.

\begin{fullversion}
Additionally, we instantiate this locale with some concrete choices for the references.
Namely $A$ and $B$ are 64 and
1000.000 qubit references, respectively. And instead of a two qubit reference $\Phi$, we have two references $\Phi_1,\Phi_2$, and the
references are explicit parts of a quantum memory that has the tensor structure $\calH_{M,\mathit{specific}} :=
\calH_A\otimes\calH_X\otimes\calH_{\Phi_1}\otimes\calH_B\otimes\calH_{\Phi_2}$.
This demonstrates how easy it is to reuse a theorem stated in terms of
references in a different context that makes specific assumptions about
the memory layout.  Had we proven Hoare judgment for some specific
layout (e.g.,
$\calH_A\otimes\calH_X\otimes\calH_{\Phi}\otimes\calH_B$) and by
explicitly tensoring operators with the identity when applying them to
specific references, we would probably have to redo the proof (although
without any relevant changes) if we wanted to reuse the result in a
different context.  (Note that we even chose not to put $\Phi_1$ and
$\Phi_2$ next to each other in $\calH_{M,\mathit{specific}}$, even
though $\Phi$ was treated as a single reference in the proof of the
Hoare judgment. This did not lead to any difficulties.)
\end{fullversion}
\begin{shortversion}
Additionally, we instantiate this locale with some concrete choices for the references (factors in a tensor product of Hilbert spaces) to demonstrate that such a concrete setting can be handled as easily. 
\end{shortversion}


{
\paragraph{Complements.}
We also develop the theory of reference categories with complements.
Analogous to the theory \texttt{Axioms}, we have a theory \texttt{Axioms\_Complements} that is instantiated in the quantum case by \texttt{Axioms\_Complements\_Quantum}.
The additional laws that hold in a reference theory with complements (\autoref{fig:compl.laws}) are proven generically in \texttt{Laws\_Complements}; our Python script then creates the quantum instantiation automatically (\texttt{Laws\_Complements\_Quantum}).

We note that showing that the quantum reference theory has complements is mathematically considerably more involved than proving the rest of the quantum reference axioms.
While the rest (\texttt{Axioms\_Quantum}; \locAxiomsQuantum\ lines of code) is mostly an application of standard facts about linear algebra, in \texttt{Axioms\_Complements\_Quantum} we need to actually do some real proofs.
Consequently, the latter theory is more complex (\locAxiomsComplementQuantum\ lines of code).

We have additionally implemented one of the sections about lifting quantum objects, namely the lifting of pure states (\autoref{sec:pure}).
We chose this one for two reasons:
First, it is a non-trivial application of the theory of complements.
Second, it is one of the most complicated constructions in \autoref{sec:lifting}.
(Recall the various helper constructs such as $\pureeta{\mathbf A}$ and $\purexi b$.)
We wanted to make sure that that construction is amenable to formalization.
The answer is positive, \texttt{Pure\_States} contains both a formalization of the lemmas in \autoref{sec:pure} as well as a formalization of the example from that section (triple CNOT).

When formalizing the existence of complements, we encountered one interesting challenge:
The axiom of the existence of complements (\autopageref{page:ex.compl}) stipulates that for any reference $F:\mathbf A\to\mathbf B$, there exists a $\mathbf C$ and a complement $G:\mathbf C\to\mathbf B$ of $F$.
However, in our formalization, the objects $\mathbf A,\mathbf B,\mathbf C$ of the category are represented by types $\alpha\ \mathtt{update},\beta\ \mathtt{update},\gamma\ \mathtt{update}$.
So the existence axiom requires that for every~$F$, there is a type $\gamma$ with certain properties.
But this cannot be expressed in HOL!
(It would require dependent types.)
Instead, we use the following workaround for \emph{finite-dimensional} references:
We require that for any $\mathbf A,\mathbf B$, there exists a $\mathbf C$ such that for any $F:\mathbf A\to\mathbf B$, we have a complement $G:\mathbf C\to\mathbf B$.
This can actually be formalized because now the type $\gamma$ depends only on other types $\alpha,\beta$.
Concretely, in Isabelle/HOL, we formalize this by requiring the definition of a type constructor \texttt{($\alpha,\beta$) complement\_domain} that we use for $\gamma$, and then require the axiom
\singledouble{
\[
  \mathtt{reference}\ F
  \implies
  \exists G :: (\alpha,\beta)\ \mathtt{complement\_domain}\ \mathtt{update} \Rightarrow \beta\ \mathtt{update}.\ \
  \texttt{complements}\ F\ G
\]
}{
  \begin{multline*}
  \!\!\!\!\!\!\mathtt{reference}\ F
  \implies
  \exists G :: (\alpha,\beta)\ \mathtt{complement\_domain}\ \mathtt{update} \\
  \Rightarrow \beta\ \mathtt{update}.\ \
  \texttt{complements}\ F\ G
\end{multline*}%
}
for $F :: \alpha\ \mathtt{update} \Rightarrow \beta\ \mathtt{update}$.
This works well with the rest of our formalization.
However, it makes the proof of existence in the quantum reference category harder:
We need to show that $\mathbf C$ only depends on $\mathbf A,\mathbf B$, not on $F$.
This, in turn, boils down to showing that the dimension of $\calH_C$ depends only on the dimensions of $\calH_A,\calH_B$.
(In fact, $\dim\calH_C=\frac{\dim\calH_B}{\dim\calH_A}$.)
This extra nontrivial proof would not be necessary if $\mathbf C$ could depend on $F$.

This also shows a problem that we get when formalizing the infinite-dimensional setting:
In that case $\dim\calH_C=\frac{\dim\calH_B}{\dim\calH_A}$ does not hold, and $\calH_C$ does not depend solely on $\calH_A,\calH_B$.
  
Thus we cannot instantiate our axioms for complements in the case of infinite-dimensional quantum references\footnote{Proof: Let $\calH_A$ be a countably infinite-dimensional space, $\calH_B:=\calH_A$. Let $U$ be a unitary from $\calH_B$ to $\calH_A\otimes\setC^n$ (exists for any $n\in\setN$). $F:=\chain{\sandwich U}\Fst$.
  Then $\calH_C$ needs to be isomorphic to $\setC^n$, but there is no way to recover $n$ just from knowing $\calH_A,\calH_B$.} (or infinite classical references)!

Fortunately, there is an alternative, somewhat more complicated to use, solution that also works in the infinite-dimensional case:
It involves a poor man's substitute for dependent types available in Isabelle/HOL:
The \texttt{Types\_To\_Set} extension \cite{kuncar16types} roughly speaking allows us to locally (within a proof) assume the existence of a type $\gamma$ that is isomorphic to some set $S_\gamma$ that is defined only within this proof context.
(Here isomorphic simply means that there is a bijection between them.)
This can be used for our purposes as follows:
Instead of stating something like ``for all types $\alpha$, $\beta$ and all references $F:\alpha\ \mathtt{update}\to\beta\ \mathtt{update}$, there exists a type $\gamma$ and a complement $G:\gamma\ \mathtt{update}\to\beta\ \mathtt{update}$ of $F$'' (which cannot be formulated in Isabelle/HOL), we can formulate a lemma that says
\begin{equation}\label{eq:with-type-complement}
  \begin{gathered}
    \text{for all types $\alpha$, $\beta$, $\gamma$, and all $F:\alpha\ \mathtt{update}\to\beta\ \mathtt{update}$, if $\gamma$ is isomorphic to the set $\mathtt{cdc}\ F$,} \\
    \text{then there is a complement $G:\gamma\ \mathtt{update}\to\beta\ \mathtt{update}$ of $F$.}
  \end{gathered}
\end{equation}
Here $\mathtt{cdc}\ F$ is a set that we construct so that it has the desired cardinality (which may depend on $F$);
after proving \eqref{eq:with-type-complement} we never need to use the definition of \texttt{cdc} again.

Without the \texttt{Types\_To\_Set} mechanism, such a lemma is relatively useless because we cannot construct a suitable type $\gamma$ to satisfy the premise of the lemma, unless $F$ is globally fixed.
But with \texttt{Types\_To\_Set}, it becomes possible to use it:
If, within a proof, we need a complement of a given register $F$, we locally define the type $\gamma$, and then, within the scope of that definition, have access to a complement of $F$ by \eqref{eq:with-type-complement}.
Once we exit that scope, the complement is not accessible (or even typable) anymore, but any consequences from its existence persist beyond the local scope.

Using \texttt{Types\_To\_Set} is quite complicated (and low-level), but the process is much eased by another as yet unpublished Isabelle-library \texttt{With\_Type}.

\paragraph{Conclusion.} Overall, we believe that this formalization shows that our approach is
well-suited for handling quantum references in a
theorem prover, and that the references are easy to handle and lead to
low overhead or boilerplate proof code. We believe the same would hold for implementations in other theorem provers than Isabelle/HOL.


}

\begin{fullversion}
\section{Conclusion}
\label{sec:conclusions}

We have shown how to formalize both classical and quantum references (as
instantiations of a generic theory of references), and how to use them on
the example of a quantum Hoare logic. We have formalized a large part of our results in
Isabelle/HOL (for the quantum case only in finite dimensions) and shown that our theory is easy to use in a formal setting.
Some open research questions include:
\begin{compactitem}
\item How do we model references that contain mixed quantum/classical data?
  (E.g., classical lists of qubits.)
  See \autoref{sec:discussion}.
\item Are there other interesting theories of references besides
  classical or quantum ones? Or are there variants of the quantum or
  classical ones with different nice properties?
\item
  Completing the formalization:  All the results about lifting quantum objects from \autoref{sec:lifting}, currently we only cover lifting of pure states. 
  Additional automation such as tactics, simplification procedures, etc.~for easier reasoning about references. (E.g., all proof steps in the teleportation example, \autoref{sec:teleport}, should ideally be solved by a single proof methods invocation.)
\end{compactitem}
\end{fullversion}

\anonymous{}{
\paragraph{Acknowledgments.} \notanonymous We thank the anonymous PlanQC reviewers for useful feedback and pointers. This work was funded by the ERC consolidator grant CerQuS (819317), by the PRG team grant ``Secure Quantum Technology'' (PRG946) from the Estonian Research Council, and by the Estonian Centre of Exellence in IT (EXCITE) funded by the ERDF.
}

\shortonly{
  \clearpage
  \bibliographystyle{plain}
  \bibliography{references}
}


\appendix

  
    

    

    


\makeatletter
\patchcmd{\hyper@makecurrent}{%
    \ifx\Hy@param\Hy@chapterstring
        \let\Hy@param\Hy@chapapp
    \fi
}{%
  \ifdefstring{\Hy@param}{section}{\def\Hy@param{appendix}}{}%
  \ifdefstring{\Hy@param}{subsection}{\def\Hy@param{appendix}}{}%
}{}{\errmessage{failed to patch}}

\newcommand\INCLUDEONCEdfasflsdjhodfuvppwdasdfsld{}

\section{Proofs of generic laws about references}

\subsection{Proof of \autoref{lemma:pairs}}\label{sec:proof:lemma:pairs}

Assume $F,G$ to be disjoint references. The existence of a reference $\spair FG$ as in \autoref{def:pair} follows immediately from Axiom~\ref{ax:pairs}.
It remains to show that it is unique, i.e., if $\forall a,b.\ X(a\otimes b)=F(a)\cdot G(b)$ for some reference $X$, then $X=\spair FG$.
Since both $\spair FG$ and $X$ are references, they are also pre-references (Axiom~\ref{ax:reg.prereg}).
Furthermore, $X(a\otimes b)=F(a)\cdot G(b)=\spair FG(a\otimes b)$ for all $a,b$, so $X=\spair FG$ (Axiom~\ref{ax:tensorext}).

\subsection{Proofs for \autoref{fig:laws}}\label{sec:proofs:laws}

We show all the laws in \autoref{fig:laws}.
The proofs are not fully in the the same order as the laws in the figure because some laws are used in the proofs of earlier ones and therefore need to be shown first.

\newcommand\lawproof[1]{\paragraph{Proof of Law~\ref{#1}.}}

\paragraph{Proof of \autoref{footnote:separating}.}
Let $M,N$ be separating sets for $\mathbf A, \mathbf B$. Let $\mathit{MN} := \braces{a\otimes b:a\in M,b\in N}$.
We want to show that $\mathit{MN}$ is a separating set for $\mathbf A\otimes \mathbf B$.
For this, fix some pre-references $F,G:\mathbf A\otimes\mathbf B\to\mathbf C$ with $\forall x\in \mathit{MN}.\ F(x)=G(x)$.

Let $F_a(b) := F(a\otimes b)$.
Then $F_a=F\circ\pb\paren{\lambda x.\ x\cdot(a\otimes1)}\circ(\lambda x.\ 1\tensor x)$ by Axiom~\ref{ax:tensor.mult}.
Thus $F_a$ is a pre-register since it is a composition of pre-registers (Axioms~\ref{ax:tensor-1},\ref{ax:reg.prereg},\ref{ax:cdot-a},\ref{ax:preregs}).

Let $F_b(a) := F(a\otimes b)$, $G_a(b):=G(a\otimes b)$, and $G_b(a):=G(a\otimes b)$.
Then $F_b,G_a,G_b$ are also pre-references.

Since $\forall x\in \mathit{MN}.\ F(x)=G(x)$, we have $F(a\otimes b)=G(a\otimes b)$ for all $a\in N,b\in M$.
Hence $F_a(b)=G_a(b)$ for all $a\in N,b\in M$.
Since $M$ is separating, it follows that $F_a=G_a$ for all $a\in N$.

Thus $F(a\otimes b)=G(a\otimes b)$ for all $a\in N, b\in\mathbf{B}$.
Hence $F_b(a)=G_b(a)$ for all $a\in N,b\in\mathbf{B}$.
Since $N$ is separating, it follows that $F_b=G_b$ for all $b\in\mathbf{B}$.

Thus $F(a\otimes b)=G(a\otimes b)$ for all $a\in \mathbf{A}, b\in\mathbf{B}$.
By Axiom~\ref{ax:tensorext}, it follows that $F=G$.

Thus $\mathit{MN}$ is a separating set.

This shows the fact from \autoref{footnote:separating}.

\lawproof{law:tensor3}
We trivially have that $\mathbf A,\mathbf B,\mathbf C$ are separating sets.
By application of \autoref{footnote:separating} (twice), we have that $X:=\pb\braces{a\otimes(b\otimes c) : a\in\mathbf{A}, b\in\mathbf{B}, c\in\mathbf{C}}$ is separating.
By definition of separating sets, this implies that $F=G$ if $F(a\otimes (b\otimes c))=G(a\otimes (b\otimes c))$ for all $a,b,c$.
This shows \ref{law:tensor3}.

\lawproof{law:tensor.ab}
By definition (\autoref{footnote:def.sigma.alpha.tensor}), $F\rtensor G = \pb\pair{\paren{\lambda a. F(a)\otimes 1}}{\paren{\lambda b. 1\otimes G(b)}}$.
First note that $F'(a):= F(a)\otimes 1$ and $G'(b) := 1\otimes G(b)$ are registers (Axioms~\ref{ax:tensor-1},\ref{ax:reg.morphisms}).
Furthermore, $F'$ and $G'$ are disjoint since $F'(a)G'(b) = F(a)\otimes G(b) = G'(b)F'(a)$ (Axiom~\ref{ax:tensor.mult}).
Thus the pair $F\rtensor G = \spair{F'}{G'}$ is well-defined and a reference (\autoref{lemma:pairs})
and $(F\rtensor G)(a\otimes b) =
F'(a)G'(b) = F(a)\otimes G(b)$.

\lawproof{law:Fst.Snd.disjoint}
$\Fst,\Snd$ are references by Axiom~\ref{ax:tensor-1}.
Since $\Fst(a)\Snd(b)=a\otimes b=\Snd(b)\Fst(a)$ (Axiom~\ref{ax:tensor.mult}),
$\Fst,\Snd$ are disjoint.

\lawproof{law:chain.disjoint.outer}
For all $a,b$, we have:
\[
\chain FH(a)G(b) = F(H(a))G(b) \starrel= G(b)F(H(a)) = G(b)\chain FH(a).
\]
Here $(*)$ follows since $F,G$ are disjoint.
Thus $\chain FH$ and $G$ are disjoint.

\lawproof{law:chain.disjoint.inner}
For all $a,b$, we have:
\begin{multline*}
  \chain HF(a)\chain HG(b)
  = H(F(a))H(G(b))
  \txtrel{\ref{ax:reg.monhom}}= H(F(a)G(b))
  \\
  \starrel=   H(G(b)F(a))
  \txtrel{\ref{ax:reg.monhom}}= H(G(b))H(F(a))
  = \chain HG(b)\chain HF(a).
\end{multline*}
Here $(*)$ follows since $F,G$ are disjoint.
Thus $\chain HF$ and $\chain HG$ are disjoint.

\lawproof{law:compat3}%
Since $F,G$ are disjoint, the pair $\spair FG$ is well-defined and a reference (\autoref{lemma:pairs}).

For all $a,b,c$, we have
\begin{equation}\label{eq:law:compat3}
  \spair FG(a\otimes b) H(c)
  =
  F(a)G(b)H(c)
  \starrel=  F(a)H(c)G(b)
  \starrel=   H(c)F(a)G(b)
  = H(c)\spair FG(a\otimes b).
\end{equation}
Here $(*)$ follows since $F,G,H$ are disjoint.

Let $X_c(x) := \spair FG(x)H(c)$.
Since $\spair FG$ is a reference and thus a pre-reference (Axiom~\ref{ax:reg.prereg}), and since $\lambda x.\ xH(c)$ is a pre-reference (Axiom~\ref{ax:cdot-a}), we have that $X_c$ is a pre-reference, too (Axiom~\ref{ax:preregs}).

Let $Y_c(x) := H(c)\spair FG(x)$. Analogously, $Y_c$ is a pre-reference.

From \eqref{eq:law:compat3}, we have that $X_c(a\otimes b)=Y_c(a\otimes b)$ for all $a,b,c$.
By Axiom~\ref{ax:tensorext}, this implies $X_c=Y_c$ for all~$c$.

Thus
\[
\spair FG(x)H(c) = X_c(x) = Y_c(x) = H(c)\spair FG(x)
\]
for all $x,c$.
Hence $\spair FG$, $H$ are disjoint.

\lawproof{law:sigma.alpha.refs}
By definition (\autoref{footnote:def.sigma.alpha.tensor}),
$\swap=\pair\Snd\Fst$,
$\assoc = \pair{\pair\Fst{\Snd\circ\Fst}}{\Snd\circ\Snd}$, and
$\assocp = \pair{\Fst\circ\Fst}{\pair{\Fst\circ\Snd}{\Snd}}$.

$\Fst,\Snd$ are references by Axiom~\ref{ax:tensor-1}.
They are disjoint by Law~\ref{law:Fst.Snd.disjoint}.
Thus the pair $\swap=\pair\Snd\Fst$ is well-defined and a reference (\autoref{lemma:pairs}).

By Law~\ref{law:chain.disjoint.outer}, $\Fst$ and $\Snd\circ\Fst$ are disjoint.
Thus the pair $\pair\Fst{\Snd\circ\Fst}$ is well-defined and a reference (\autoref{lemma:pairs}).
$\Fst$ and $\Snd\circ\Snd$ are disjoint by Law~\ref{law:chain.disjoint.outer}.
$\Snd\circ\Fst$ and $\Snd\circ\Snd$ are disjoint by Law~\ref{law:chain.disjoint.inner}.
Thus by Law~\ref{law:compat3}, $\pair\Fst{\Snd\circ\Fst}$ and $\Snd\circ\Snd$ are disjoint.
Thus the pair $\assoc = \pair{\pair\Fst{\Snd\circ\Fst}}{\Snd\circ\Snd}$ is well-defined and a reference (\autoref{lemma:pairs}).

Analogously for $\assocp$.

\lawproof{law:sigma.ab}
By definition (\autoref{footnote:def.sigma.alpha.tensor}), $\swap=\pair\Snd\Fst$.
Thus
\[
  \swap(a\otimes b) = \pair\Snd\Fst(a\otimes b) = \Snd(a)\Fst(b) = (1\otimes a)(b\otimes 1)
  \txtrel{\ref{ax:tensor.mult}}= b\otimes a.
\]

\lawproof{law:chain.sigma.12}
We have $\chain\swap\Fst(x)=\swap(x\otimes 1)\txtrel{\ref{law:sigma.ab}}= 1\otimes x = \Snd(x)$,
hence $\chain\swap\Fst = \Snd$.

We have $\chain\swap\Snd(x)=\swap(1\otimes x)\txtrel{\ref{law:sigma.ab}}= x\otimes 1 = \Fst(x)$,
hence $\chain\swap\Snd = \Fst$.

\lawproof{law:alpha.abc}
By definition (\autoref{footnote:def.sigma.alpha.tensor}),
$\assoc = \pair{\pair\Fst{\Snd\circ\Fst}}{\Snd\circ\Snd}$.
Thus
\begin{align*}
  \assoc\pb\paren{(a\otimes b)\otimes c}
  &= \pair{\Fst}{\Snd\circ\Fst}(a\otimes b)\cdot
    (\Snd\circ\Snd)(b)
    = \Fst(a) \cdot \paren{\Snd\circ\Fst}(b)\cdot
    (\Snd\circ\Snd)(c)
  \\
  &= (a\otimes 1)\pb\paren{1\otimes (b\otimes 1)}\pb\paren{1\otimes(1\otimes c)}
    \txtrel{\ref{ax:tensor.mult}}= a\otimes (b\otimes c).
\end{align*}

\lawproof{law:alpha'.abc}
By definition (\autoref{footnote:def.sigma.alpha.tensor}),
$\assocp = \pair{\Fst\circ\Fst}{\pair{\Fst\circ\Snd}{\Snd}}$.
Thus
\begin{align*}
  \assocp\pb\paren{a\otimes(b\otimes c)}
  &= \pair{\Fst\circ\Fst}{\pair{\Fst\circ\Snd}{\Snd}}\pb\paren{a\otimes(b\otimes c)}
  = \paren{\Fst\circ\Fst}(a) \cdot {\pair{\Fst\circ\Snd}{\Snd}}(b\otimes c) \\
  &= \paren{\Fst\circ\Fst}(a) \cdot \paren{\Fst\circ\Snd}(b) \cdot \Snd(c)
    = \pb\paren{(a\otimes 1)\otimes 1} \pb\paren{(1\otimes b)\otimes 1} (1\otimes c)
    \txtrel{\ref{ax:tensor.mult}}= \pb\paren{(a\otimes b)\otimes c}.
\end{align*}

\lawproof{law:tensor.disjoint}
Let $\Psi_{cd}(x) := (F\rtensor G)(x) \cdot (H\rtensor L)(c\otimes d)$.
Then $\Psi_{cd}$ is a pre-reference by Law~\ref{law:tensor.ab} and Axioms~\ref{ax:reg.prereg},\ref{ax:cdot-a}.

Let $\Psi'_{cd}(x) :=  (H\rtensor L)(c\otimes d) \cdot (F\rtensor G)(x)$.
Then $\Psi_{cd}$ is a pre-reference by Law~\ref{law:tensor.ab} and Axioms~\ref{ax:reg.prereg},\ref{ax:cdot-a}.

We have for all $a, b$:
\begin{align*}
  \Psi_{cd}(a \otimes b)
  &=    (F\rtensor G)(a\otimes b) \cdot (H\rtensor L)(c\otimes d)
  \txtrel{\ref{law:tensor.ab}}=   \pb\paren{F(a)\otimes G(b)} \cdot \pb\paren{H(c)\otimes L(d)}
  \\
  &\txtrel{\ref{ax:tensor.mult}}=  F(a)H(c) \otimes G(b)L(d)
    \starrel= H(c)F(a) \otimes L(d)G(b)
    \txtrel{\ref{ax:tensor.mult}}=  \pb\paren{H(c)\otimes L(d)} \cdot \pb\paren{F(a)\otimes G(b)}
    \\
    &\txtrel{\ref{law:tensor.ab}}=  (H\rtensor L)(c\otimes d) \cdot (F\rtensor G)(a\otimes b)
    = \Psi'_{cd}(a\otimes b).
\end{align*}
where $(*)$ is because $H,F$ are disjoint and $L,G$ are disjoint, by assumption.

By Axiom~\ref{ax:tensorext}, this implies $\Psi_{cd}=\Psi'_{cd}$.

Let $\Phi_{x}(y) := (F\rtensor G)(x) \cdot (H\rtensor L)(y)$ and $\Phi'_x(y) := (H\rtensor L)(y) \cdot (F\rtensor G)(x) $.
$\Phi_y,\Phi'_y$ are pre-references.

We have for all $c,d$:
\[
  \Phi_x(c\otimes d) = \Psi_{cd}(x) = \Psi'_{cd}(x) = \Phi'_x(c\otimes d).
\]

Since $\Phi_x,\Phi'_x$ are pre-references, this implies $\Phi_x=\Phi'_x$ (Axiom~\ref{ax:tensorext}).

Thus for all $x,y$:
\[
(F\rtensor G)(x) \cdot (H\rtensor L)(y) = \Phi_x(y) = \Phi'_x(y) = (H\rtensor L)(y) \cdot (F\rtensor G)(x).
\]
Hence $F\rtensor G$ and $H\rtensor L$ are disjoint.

\lawproof{law:pair.ref}
Shown in \autoref{lemma:pairs}.

\lawproof{law:pair.select}
We have $\chain{\spair FG}\Fst(x) = \spair{FG}(x\otimes 1) = F(x)G(1) \txtrel{\ref{ax:reg.monhom}}= F(x)$
and $\chain{\spair FG}\Snd(x) = \spair{FG}(1\otimes x) = F(1)G(x) \txtrel{\ref{ax:reg.monhom}}= G(x)$.

\lawproof{law:fst.snd.id}
Since $\Fst,\Snd$ are disjoint, $\pair\Fst\Snd$ is a reference (\autoref{lemma:pairs}) and thus a pre-reference (Axiom~\ref{ax:reg.prereg}). Also $\id$ is a pre-reference (Axiom~\ref{ax:preregs}).

Furthermore, for all $a,b$, we have $\pair\Fst\Snd(a\otimes b)=\Fst(a)\Snd(b)=(a\otimes1)(1\otimes b)\txtrel{\ref{ax:tensor.mult}}= a\otimes b = \id(a\otimes b)$.
Thus $\pair\Fst\Snd=\id$ (Axiom~\ref{ax:tensorext}).

\lawproof{law:snd.fst.sigma}
By definition (\autoref{footnote:def.sigma.alpha.tensor}),
$\swap=\pair\Snd\Fst$.

\lawproof{law:pair.sigma}
$\chain{\spair FG}\swap$ and $\spair GF$ are pre-references.
Furthermore $\chain{\spair FG}\swap(a\otimes b)\txtrel{\ref{law:sigma.ab}}=\spair FG(b\otimes a)=F(b)G(a)=G(a)F(b)=\spair GF(a\otimes b)$.
Hence by Axiom~\ref{ax:tensorext}, $\chain{\spair FG}\swap=\spair GF$.

\lawproof{law:pair.alpha}
$\chain{\pair F{\spair GH}}\assoc$ and $\pair{\spair FG}H$ are pre-references.
Furthermore
\begin{align*}
  \chain{\pair F{\spair GH}}\assoc\pb\paren{(a\otimes b)\otimes c}
  &\txtrel{\ref{law:alpha.abc}}=
  {\pair F{\spair GH}}\pb\paren{a\otimes(b\otimes c)}
    = F(a) \cdot \spair GH(b\otimes c)
    \\
  &= F(a)G(b)H(c)
    = \spair FG(a\otimes b)\cdot H(c)
    = \pair{\spair FG}H \pb\paren{(a\otimes b)\otimes c}.
\end{align*}
By Law~\ref{law:tensor3}, this implies $\chain{\pair F{\spair GH}}\assoc = \pair{\spair FG}H$.

\lawproof{law:pair.alpha'}
$\chain{\pair{\spair FG}H}\assocp$ and ${\pair F{\spair GH}}$ are pre-references.
Furthermore
\begin{align*}
  \chain{\pair{\spair FG}H}\assocp\pb\paren{a\otimes(b\otimes c)}
  &\txtrel{\ref{law:alpha'.abc}}=
    \pair{\spair FG}H\pb\paren{(a\otimes b)\otimes c}
    = \spair FG(a\otimes b) \cdot H(c)
    \\
  &= F(a)G(b)H(c)
    = F(a) \cdot \spair GH(b\otimes c)
    = \pair F{\spair GH} \pb\paren{a\otimes (b\otimes c)}.
\end{align*}
By Law~\ref{law:tensor3}, this implies $\chain{\pair{\spair FG}H}\assocp = {\pair F{\spair GH}}$.

\lawproof{law:pair.chain}
By \autoref{lemma:pairs}, $\spair FG$ is well-defined and a reference.
By Law~\ref{law:chain.disjoint.inner}, $\chain CF,\chain CG$ are disjoint, hence $\pair{\chain CF}{\chain CG}$ is well-defined and a reference (\autoref{lemma:pairs}).

For all $a,b$, we have
\[
  \pair{\chain CF}{\chain CG}(a\otimes b)
  =
  C(F(a))\cdot C(G(b))
  \txtrel{\ref{ax:reg.monhom}}=
  C(F(a)G(b))
  =
  C(\spair FG(a\otimes b))
  =
  \chain C{\spair FG}(a\otimes b).
\]
Since $\pair{\chain CF}{\chain CG}$ and $\chain C{\spair FG}$ are pre-references (Axiom~\ref{ax:reg.prereg}),
this implies   $\pair{\chain CF}{\chain CG} = \chain C{\spair FG}$ (Axiom~\ref{ax:tensorext}).

\lawproof{law:pair.tensor}
$\spair FG$ and $C\rtensor D$ are references (\autoref{lemma:pairs}, Law~\ref{law:tensor.ab}). $\chain FC,\chain GD$ are disjoint (Law~\ref{law:chain.disjoint.outer}) and thus $\pair{\chain FC}{\chain GD}$ is a reference.

For all $a,b$, we have
\[
  \chain{\spair FG}{(C\rtensor D)}(a\otimes b)
  \txtrel{\ref{law:tensor.ab}}=
  \spair{FG}\pb\paren{C(a)\otimes D(b)}
  =
  F(C(a))\cdot G(D(b))
  =
  \chain FC(a)\cdot\chain GD(b)
  =
  \pair{\chain FC}{\chain GD}(a\otimes b).
\]
Since both $\chain{\spair FG}{(C\rtensor D)}$ and $  \pair{\chain FC}{\chain GD}$ are pre-references (Axiom~\ref{ax:reg.prereg}),
this implies that   $\chain{\spair FG}{(C\rtensor D)}=\pair{\chain FC}{\chain GD}$  (Axiom~\ref{ax:tensorext}).


\fullonly{
  
\section{Miscellaneous facts about infinite-dimensional quantum mechanics}
\label{app:misc.facts}

In this section, we give a number of lemmas related to (infinite-dimensional) quantum references (i.e., to operator theory) that are used throughout the paper.
Many of them are likely known facts but we did not find citable references.
We recall our convention that $\mathbf A,\mathbf B,\dots$ always denote $\bounded(\calH_A),\bounded(\calH_B),\dots$, the spaces of bounded operators.
Whenever we say reference/pre-reference we mean quantum reference/pre-reference.

\begin{lemma}\lemmalabel{lemma:ex.proj}
  For every closed subspace $S$ of $\calH$, there exists exactly one projector $P_S$ with image $S$.
\end{lemma}

\begin{proof}
  For the existence see \cite[comment after Definition II.3.1]{conway13functional}.%
  \footnote{We define a projector as a bounded operator with $P^2=P$ and $P=\adj P$.
    \cite{conway13functional} gives a different definition of what they call a ``projection'' but it is equivalent by \cite[Proposition~II.3.3\,(d)]{conway13functional}.}
  Assume there are two projectors $P,Q$ with image $S$.
  Then $\im Q\subseteq\im P$, so by \cite[Exercise~II.3.6]{conway13functional}, $PQ=Q$.
  And $\im P\subseteq\im Q$, so by \cite[Exercise~II.3.6]{conway13functional}, $PQ=P$.
  Thus $P=Q$.
\end{proof}

\begin{lemma}\lemmalabel{lemma:tensor.subspace.proj}
  The tensor product $S\otimes T$ of closed subspaces $S,T$ of Hilbert spaces is defined as the closed span of $\braces{s\otimes t: s\in S, t\in T}$. For a bounded operator $A$ between Hilbert spaces, let $AS$ denote the closure of the image of $S$ under $A$.

  Then $AS\otimes BT = (A\otimes B)(T\otimes S)$.
\end{lemma}

\begin{proof}
  In this proof, let $\overline X$ denote the closure of $X$.
  For an operator $A$ and a subspace $S$, let $A\cdot S$ denotes the image of $S$ under $A$, and $AS := \overline{A\cdot S}$.
  Let $f(x,y) := x\otimes y$ and $g(x) := (A\otimes B)x$.
  
  \begin{align*}
    AS \otimes BT
    &\txtrel{def}= \overline{\Span f(AS\times BT)}
    \txtrel{def}= \overline{\Span f(\overline{A\cdot S}\times \overline{B\cdot T})}
    = \overline{\Span f(\overline{A\cdot S\times B\cdot T})}
    \\&
    \starrel= \overline{\Span\overline{f\pb\paren{A\cdot S\times B\cdot T}}}
    = \overline{\Span{f\pb\paren{A\cdot S\times B\cdot T}}}
    = \overline{\Span{g\pb\paren{f\paren{S\times T}}}}
    \\&
    \starstarrel= \overline{g\pb\paren{{\Span{f\paren{S\times T}}}}}
    \txtrel{def}= \overline{g\pb\paren{S\otimes T}}
    \txtrel{def}= \overline{(A\otimes B)\cdot\paren{S\otimes T}}
    \txtrel{def}= {(A\otimes B)\paren{S\otimes T}}.
  \end{align*}
  Here $(*)$ uses that $f$ is continuous.
  And $(**)$ uses that $g$ is linear.
\end{proof}

\begin{lemma}\lemmalabel{lemma:traceclass.abs}
  Fix a bounded operator $t$.
  \begin{compactenum}[(i)]
  \item\itlabel{lemma:traceclass.abs:tc} $t$ is trace-class iff $\abs t$ is trace-class.
  \item\itlabel{lemma:traceclass.abs:norm} If $t$ is trace-class, $\trnorm{t}=\tr\abs t$
  \item\itlabel{lemma:traceclass.abs:normpos} If $t$ is positive trace-class, $\trnorm{t}=\tr t$.
  \item \itlabel{lemma:traceclass.abs:between} If $t$ is trace-class and $s\leq t$ is a positive operator, then $s$ is trace-class.
  \item \itlabel{lemma:traceclass.abs:opnorm} If $t$ is trace-class, $\norm t\leq 2\trnorm t$.
  \end{compactenum}
\end{lemma}

\begin{proof}
  We show \eqref{lemma:traceclass.abs:tc}.
  By definition of trace-class \cite[Definition 18.3]{conway00operator}, $t$ is trace-class iff $\sum_{e\in E}\adj e\abs t e<\infty$ for some orthonormal basis $E$.
  Since $\abs t = \pb\abs{\abs{t}}$, this holds iff  $\sum_{e\in E}\adj e\pb\abs{\abs t} e<\infty$ which by definition holds iff $\abs t$ is trace-class.

  We show \eqref{lemma:traceclass.abs:norm}.
  By definition \cite[text after Definition 18.3]{conway00operator}, $\trnorm{t}=\sum_{e\in E}\adj e\abs t e$ for some orthonormal basis $E$, and by definition \cite[Definition 18.10]{conway00operator}, $\tr t=\sum_{e\in E}\adj et e$ for all $t$, and in particular $\tr\abs t=\sum_{e\in E}\adj e\abs t e$. Thus $\trnorm t=\tr\abs t$.

  We show \eqref{lemma:traceclass.abs:normpos}. For positive $t$, we have $\abs t=t$. Thus $\trnorm t \eqrefrel{lemma:traceclass.abs:norm}= \tr\abs t = \tr t$.

  We show \eqref{lemma:traceclass.abs:between}.
  By definition \cite[Definition 18.3]{conway00operator}, $t$ is trace-class means that $\sum_{e\in E}\adj e\abs t e$ converges.
  Since $s$ is positive and $t\geq s$, $t$ is positive. Thus
  \[
    \adj e\abs t e = \adj et e\geq \adj es e = \adj e\abs s e \geq 0.
  \]
  Hence $\sum_{e\in E}\adj e\abs s e$ converges as well.
  By definition, this means that $s$ is trace-class.

  We show \eqref{lemma:traceclass.abs:opnorm}.
  Fix positive a trace-class $t$. Then there is a bounded operator $a$ with $\adj aa=t$.
  Fix an arbitrary vector $\psi$ with $\norm\psi=1$. Fix an orthonormal basis $E$ with $\psi\in E$.
  Then
  \begin{equation}
    \norm{a\psi}^2 = \adj\psi\adj aa\psi = \adj\psi t\psi \starrel\leq \sum_{e\in E}\adj ete
    \starstarrel= \tr t.
    \label{eq:norm.apsi}
  \end{equation}
  Here $(*)$ follows since $t$ is positive and hence $\adj et e\geq0$ for all $e$.
  And $(**)$ is by definition of the trace \cite[Definition 18.10]{conway00operator}.
  
  Since this holds for any unit vector $\psi$, we have (with $\psi$ ranging over unit vectors):
  \begin{equation}
    \norm t = \norm{\adj aa}  \starrel= \norm{a}^2
    =\sup_\psi\norm{a\psi}^2\eqrefrel{eq:norm.apsi}\leq \tr t.
    \label{eq:norm.pos}
  \end{equation}
  Here $(*)$ holds by \cite[Proposition II.2.7]{conway13functional}.
  So \eqref{eq:norm.pos} holds for any positive trace-class $t$.

  Fix a Hermitian trace-class $t$. By \cite[Proposition 3.2]{conway00operator}, $t=t_+-t_-$ where $t_+$ and $t_-$ are the ``positive and negative part'' of $t$, and $t_+,t_-$ are positive.
  Furthermore, $\abs t=t_++t_-$ by  \cite[Exercise 3.4]{conway00operator}.
  Since $\abs t$ is trace-class by \eqref{lemma:traceclass.abs:tc}, and $t_+,t_i \leq\abs t$, we have that $t_+,t_-$ are trace-class by \eqref{lemma:traceclass.abs:between}.
  Then
  \begin{equation}
    \norm{t} = \norm{t_+-t_-}
    \leq \norm{t_+} + \norm{t_-}
    \eqrefrel{eq:norm.pos}\leq \tr t_+ + \tr t_- = \tr \abs t \eqrefrel{lemma:traceclass.abs:norm}= \trnorm t.
    \label{eq:norm.herm}
  \end{equation}
  This holds for any Hermitian trace-class $t$.

  Fix a trace-class $t$. Then $t=t_h + it_a$ where $t_h := (t + \adj t)/2$ and $t_a := (t - \adj t)/2i$.
  Both $t_h,t_a$ are easily verified to be Hermitian.
  Thus
  \begin{align*}
    \norm{t} &= \norm{t_h+it_a} 
    \leq \norm{t_h} + \norm{t_a}
    \eqrefrel{eq:norm.herm}\leq \trnorm{t_h} + \trnorm{t_a}
    = \trnorm{t/2+\adj t/2} + \trnorm{t/2 - \adj t/2i}
    \\&
    \leq \trnorm t/2 + \trnorm{\adj t}/2 + \trnorm t/2 + \trnorm{\adj t}/2
    \starrel= 2\trnorm t.
  \end{align*}
  Here $(*)$ follows since $\trnorm{\adj t}=\trnorm t$ by \cite[Theorem 18.11\,(f)]{conway00operator}.
  Thus holds for any trace-class $t$ and thus shows \eqref{lemma:traceclass.abs:opnorm}.
\end{proof}

\begin{lemma}[Circularity of the trace]\lemmalabel{lemma:circ.trace}
  Let $a=a_1\cdots a_n:\calH_A\to\calH_B$ where $a_i$ are bounded operators (not necessarily square).
  Let $b=b_1\cdots b_m:\calH_B\to\calH_A$ where the $b_i$ are bounded operators (not necessarily square), but one of them is a (square) trace-class operator.%
  \footnote{A \emph{square}\index{square operator} is an operator with equal domain and co-domain.
    E.g., $x:\calH_A\to\calH_A$ is square and $y:\calH_A\to\calH_B$ with $\calH_A\neq\calH_B$ is not.}
  \begin{compactenum}[(i)]
  \item\itlabel{lemma:circ.trace:ubu} If $b$ (as defined above) is square and $u$ is an isometry (not necessarily square), then $b$ and $ub\adj u$ are trace-class and $\tr ub\adj u=\tr b$ and $\trnorm {ub\adj u}=\trnorm b$.
  \item\itlabel{lemma:circ.trace:circ} The square operators $ab$ and $ba$ are trace-class and $\tr ab=\tr ba$.
  \end{compactenum}
\end{lemma}

Note that, e.g., \cite[Theorem 18.11\,(e)]{conway00operator} states the same as \lemmaref{lemma:circ.trace:circ} but it only applies when the involved operators are square.

\begin{proof}
  First, we show that for any isometry $u:\calH_t\to\calH_u$, and any $t\in\bounded(\calH_t)$,
  \begin{equation}\label{eq:utu}
    ut\adj u\text{ is trace-class} \iff t\text{ is trace-class}.
  \end{equation}
  Note that $\abs{ut\adj u}=u\abs t\adj u$.
  Fix an orthonormal basis $T$ of $\calH_t$.
  Then $uT$ is orthonormal and can be extended to an orthonormal basis $V\supseteq uT$.
  Then for $\psi\in V\setminus uT$, $\adj u\psi=0$ and thus $\adj\psi u\abs{t}\adj u\psi=0$.
  
  By definition of trace-class \cite[Definition 18.3]{conway00operator}, $ut\adj u$ is trace-class iff $\sum_{\psi\in V}\adj\psi\abs{ut\adj u}\psi<\infty$ iff $\sum_{\psi\in V}\adj\psi u\abs{t}\adj u\psi<\infty$
  iff $\sum_{\psi\in uT}\adj\psi u\abs{t}\adj u\psi<\infty$ 
  iff $\sum_{\phi\in T}\adj{(u\phi)}u\abs{t}\adj u(u\phi)<\infty$  (because $u$ is a bijection between $T$ and $uT$)
  iff $\sum_{\phi\in T}\adj{\phi}\abs{t}\phi<\infty$  (because $\adj uu=1$)
  iff $t$ is trace-class (by definition of trace-class).

  \medskip

  Next, we show any square operator $e:\calH_e\to\calH_e$ that is the product of bounded operators and at least one trace-class operator, we have that
  \begin{equation}
    \text{$e$ is trace-class.}\label{eq:e.tc}
  \end{equation}
  We can write $e$ as $e=ctd$ for $t\in\tracecl(\calH_t)$ and bounded $c:\calH_t\to\calH_e$ and $d:\calH_e\to\calH_t$ (for some Hilbert spaces $\calH_t$).
  Let $U_t:\calH_t\to\calH_t\oplus\calH_e$ and $U_e:\calH_e\to\calH_t\oplus\calH_e$ be arbitrary isometries (e.g., the canonical embeddings).
  Then $U_ectd\adj{U_e} = U_ec\adj{U_t} \cdot U_t t\adj{U_t}\cdot U_td\adj{U_e}$. The right hand side is a product of square operators, and $U_t t\adj{U_t}$ is trace-class by \eqref{eq:utu}.
  Thus their product $U_ectd\adj{U_e}$ is trace-class \cite[Theorem 18.11\,(a)]{conway00operator}.
  By~\eqref{eq:utu}, this implies that $e=ctd$ is trace-class. This shows \eqref{eq:e.tc}.

  We show \eqref{lemma:circ.trace:ubu}.
  We have that $b$ and $ub\adj u$ are trace-class by \eqref{eq:e.tc} (with $e:=b$ and $e:=ub\adj u$, respectively).
  Since $b$ is square, $b:\calH_A\to\calH_A$ and $u:\calH_A\to\calH_u$ for some $\calH_u$.
  Fix an orthonormal basis $A$ of $\calH_A$.
  Then~$uA$ is orthonormal and can be extended to an orthonormal basis $T\supseteq uA$.
  By definition of the trace \cite[Definition~18.10]{conway00operator}, $\tr b = \sum_{\psi\in A}\adj\psi b\psi$ and $\tr ub\adj u=\sum_{\psi\in T}\adj\psi ub\adj u\psi$.
  Thus
  \begin{align*}
    \tr ub\adj u
    &=
    \sum_{\psi\in uA} \adj\psi ub\adj u\psi +
    \sum_{\hskip-1cm\psi\in T\setminus uA\hskip-1cm} \adj\psi ub\adj u\psi
    \starrel=
    \sum_{\psi\in uA} \adj\psi ub\adj u\psi
    \\&
    \starstarrel=
    \sum_{\phi\in A} \adj{(u\phi)} ub\adj u(u\phi)
    =
    \sum_{\phi\in A} \adj{\phi} \adj u ub\adj u u\phi
    \tristarrel=
    \sum_{\phi\in A} \adj{\phi} b\phi
    =
    \tr b.
  \end{align*}
  Here $(*)$ follows because $\psi\in T\setminus uA$ is orthogonal to the image of $u$ and thus $\adj u\psi=0$.
  And $(**)$ follows because $u$ is injective and thus a bijection between $A$ and $uA$.
  And $(*{*}*)$ follows because $u$ is an isometry and thus $\adj uu=1$.

  Finally, we show that $\trnorm{ub\adj u} = \trnorm b$.
  Note that $\abs{ub\adj u} = u\abs b\adj u$.
  Since $b$ is trace-class, $\abs b$ is trace-class (\lemmaref{lemma:traceclass.abs:tc}).
  Thus we can apply the part of \eqref{lemma:circ.trace:ubu} that we have shown so far, namely that $\tr ub\adj u=\tr b$, to $\abs b$ instead of $b$.
  That is, we have $\tr u\abs b\adj u=\tr\abs b$.
  Thus $\trnorm{ub\adj u} \starrel= \tr\abs{ub\adj u} = \tr u\abs b\adj u = \tr\abs b \starrel= \trnorm b$ where $(*)$ is by \lemmaref{lemma:traceclass.abs:norm}.
  
  This shows \eqref{lemma:circ.trace:ubu}.

  \medskip

  We show \eqref{lemma:circ.trace:circ}.
  The fact that $ab$ and $ba$ are trace-class follows from \eqref{eq:e.tc} because $ab$ and $ba$ are square.
  Let $U_a:\calH_A\to\calH_A\oplus\calH_B$ and $U_b:\calH_B\to\calH_A\oplus\calH_B$ be arbitrary isometries (e.g., the canonical embeddings).
  Then we have
  \begin{equation*}
    \tr ab
    \eqrefrel{lemma:circ.trace:ubu}=
    \tr U_bab\adj{U_b}
    \starrel=
    \tr U_ba\adj{U_a} \cdot U_ab\adj{U_b}
    \starstarrel=
    \tr U_ab\adj{U_b} \cdot U_ba\adj{U_a}
    \starrel=
    \tr U_aba\adj{U_a}
    \eqrefrel{lemma:circ.trace:ubu}=
    \tr ba.
  \end{equation*}
  Here the $\eqrefrel{lemma:circ.trace:ubu}=$ use \eqref{lemma:circ.trace:ubu} with $b:=ab$ and $b:=ba$, respectively.
  And $(*)$ is since $U_a,U_b$ are isometries.
  And $(**)$ follows since $U_ab\adj{U_b}$ and $U_ba\adj{U_a}$ are both bounded square operators, and  $U_ab\adj{U_b}$ is trace-class by \eqref{eq:e.tc}, and thus we can apply the circularity of the trace for square operators \cite[Theorem 18.10 (e)]{conway00operator}.

  This shows \eqref{lemma:circ.trace:circ}.
\end{proof}

\begin{lemma}\label{lemma:normal}
  $x\mapsto \adj x$, $x\mapsto Ux\adj U$,
  $x\mapsto xy$, and $y\mapsto xy$ are all weak*-continuous and bounded.
  (Not linear though in the case of $x\mapsto \adj x$.)
  Here $x,y\in\mathbf A$ are bounded operators and $U:\calH_A\to\calH_B$ is unitary.
\end{lemma}

\begin{proof}
  First, we show the auxiliary fact that $\tr\adj x=\conj{\tr x}$.
  ($\symbolindexmark\conj{\conj c}$ is the complex conjugate of $c$.)
  By definition of the trace \cite[Definition 18.10]{conway00operator}, $\tr\adj x=\sum_e \adj e\adj xe=\sum_e \overline{\adj exe}=\overline{\sum_e \adj ex e}=\conj{\tr x}$ where $e$ ranges over an orthonormal basis. Thus $\tr\adj x=\conj{\tr x}$.
  
  We show that   $x\mapsto \adj x$ is weak*-continuous and bounded.
  It is bounded because $\norm{\adj x}=\norm x$ \cite[Proposition 2.7]{conway13functional}. We show  weak*-continuity: Consider a net $x_i$ that weak*-converges to~$ x$. Then $\tr ax_i\to \tr ax$ for all trace-class $a$.
  Then for all trace-class $b$, $\tr b(\adj{x_i}-\adj x)
  = \conj{\tr (x_i-x)\adj b} \starrel= \conj{\tr \adj b (x_i-x)} \to 0$.
  Here $(*)$ is the circularity of the trace \cite[Theorem 18.11\,(e)]{conway00operator}. And $\to0$ follows since $\tr ax_i\to \tr ax$ holds in particular with $a:=\adj b$.
  Thus $\adj{x_i}\to \adj x$ in the weak*-topology.
  So $x_i\to x$ implies $\adj{x_i}\to\adj x$ for every net $x_i$. Hence   $x\mapsto \adj x$ is weak*-continuous.

  We show that $x\mapsto Ux\adj U$ is weak*-continuous and normal.
  It is a *-isomorphism.
  Thus is it completely positive \cite[Example 34.3\,(a)]{conway00operator}.
  Thus it is bounded \cite[Proposition 33.4]{conway00operator}.
  And by \cite[Proposition 46.6]{conway00operator}, *-isomorphisms between von Neumann algebras are weak*-continuous, hence  $x\mapsto Ux\adj U$ is weak*-continuous.
  
  Let $F(y):=xy$. To show that $F$ is weak*-continuous, we need
  to show that for any net $y_i$, if $y_i$ weak*-converges to $y$,
  then $F(y_i)$ weak*-converges to $F(y)$.  Assume $y_i$
  weak*-converges to $y$.  By definition of the weak*-topology,
  $\tr ay_i\to \tr ay$ for any trace-class operator $a$. Thus
  $\tr bF(y_i)=\tr (bx)y_i\to \tr (bx)y=\tr bF(y)$ for any trace-class
  operator $b$. (Since $bx$ is trace-class if $b$ is trace-class and $x$ is bounded \cite[Theorem 18.11\,(a)]{conway00operator}.)
  Thus $F(y_i)$ weak*-converges to $F(y)$.
  Hence $F$ is weak*-continuous.
  Weak*-continuity of $F'(x):=xy$ is proven analogously, using additionally that $\tr ax_i=\tr x_ia$ by the
  circularity of the trace.
  $F,F'$ are $\norm a$-bounded because $\norm{ax},\norm{xa}\leq \norm a\norm x$ for bounded operators $a,x$.
  Hence $F,F'$ are weak*-continuous and bounded.
\end{proof}

\begin{lemma}\label{lemma:tensor.abs}
  For bounded operators $a$, $b$, we have $\abs{a\otimes b}=\abs a \otimes \abs b$.
  In particular, if $a,b$ are positive, so is $a \otimes b$.
\end{lemma}

\begin{proof}
  We first show the ``in particular'' part. If $a$ and $b$ are positive, by definition, there are operators $\hat a,\hat b$ such that $\adj{\hat a}\hat a=a$, $\adj{\hat b}\hat b = b$.
  Then $\adj{(\hat a\otimes\hat b)}(\hat a\otimes\hat b)
  = \adj{\hat a}\hat a\otimes \adj{\hat b}\hat b = a\otimes b$.
  (That $\adj\cdot$ and multiplication commute with $\otimes$ is shown in \cite[Equation IV.1\,(5), (6)]{takesaki}.)
  Thus $a\otimes b$ is positive.

  We now show $\abs{a\otimes b}=\abs a \otimes \abs b$. By definition, $\abs a$ is the unique positive operator with $\adj{\abs a}\abs a = \adj aa$. Analogously for $\abs b$ and $\abs{a\otimes b}$. Then $\abs a\otimes \abs b$ is positive by the ``in particular'' part. Furthermore $\adj{(\abs a\otimes \abs b)}(\abs a\otimes \abs b)
  = \adj{\abs a}\abs a \otimes \adj{\abs b}\abs b = \adj aa \otimes \adj bb
  = \adj{(a\otimes b)}(a\otimes b)$. Thus
  $\abs a \otimes \abs b$ is $\abs{a\otimes b}$.
\end{proof}

\begin{lemma}\label{lemma:alg.dense}
  The \emph{algebraic tensor product}\index{algebraic tensor product}\index{tensor product!algebraic} $\mathbf A\symbolindexmark\atensor\atensor\mathbf B$ is weak*-dense in $\mathbf A\tensor\mathbf B$.
  ($\mathbf A\atensor\mathbf B$ is defined as the subspace of $\mathbf A\otimes\mathbf B$ spanned by the operators $a\otimes b$ for $a\in\mathbf A$, $b\in\mathbf B$.)
\end{lemma}

\begin{proof}
  The \emph{spatial tensor product}\index{spatial tensor product}\index{tensor product!spatial} $\mathbf A\symbolindexmark\stensor\stensor\mathbf B$ is defined as the norm-closure of $\mathbf A\atensor\mathbf B$.
  Thus $\mathbf A\atensor\mathbf B$ is norm-dense in $\mathbf A\stensor\mathbf B$. Since the weak*-topology is coarser than the norm-topology,
  $\mathbf A\atensor\mathbf B$ is also weak*-dense in $\mathbf A\stensor\mathbf B$.

  The tensor product $\mathbf A\otimes\mathbf B$ of von Neumann algebras is defined as the smallest von Neumann algebra containing   $\mathbf A\atensor\mathbf B$ \cite[Definition~IV.1.3]{takesaki}.
  By definition of von Neumann algebras \cite[Definition~13.1]{conway00operator}), this is the SOT-closure (strong operator topology) of $\mathbf A\atensor\mathbf B$.
  Since the norm-topology is finer than the SOT, the SOT-closure of the norm-closure of a set is the SOT-closure. Thus $\mathbf A\otimes\mathbf B$ is also the SOT-closure of $\mathbf A\stensor\mathbf B$.

  We will now show that $\mathbf A\stensor\mathbf B$ is weak*-dense in $\mathbf A\otimes\mathbf B$:
  Fix $x\in \mathbf A\otimes\mathbf B$. W.l.o.g.~$\norm x\leq 1$.
  Then~$x$ is in the intersection of the unit ball $S$ with SOT-closure of $\mathbf A\stensor\mathbf B$.
  By the Kaplansky Density Theorem \cite[Theorem 44.1\,(a)]{conway00operator}, and since $\mathbf A\stensor\mathbf B$ is a C*-algebra, this is the same as the SOT-closure of $S\cap(\mathbf A\stensor\mathbf B)$.  $S\cap(\mathbf A\stensor\mathbf B)$ is a convex subset of $\bounded(\calH_A\otimes\calH_B)$, so its SOT-closure equals its WOT-closure \cite[Corollary 8.2]{conway00operator}. So $x$ is in the WOT-closure
  of $S\cap(\mathbf A\stensor\mathbf B)$.
  Thus there is a net~$x_i$ with $x_i\in S\cap(\mathbf A\stensor\mathbf B)$ that WOT-converges to $x$. 
  Since $x_i\in S\cap(\mathbf A\stensor\mathbf B)$ is a bounded subset of $\bounded(\calH_A\otimes\calH_B)$,
  the WOT and the weak*-topology coincide on it \cite[Proposition 20.1\,(b)]{conway00operator}).
  So $x_i$ weak*-converges to $x$. So $x$ is in the weak*-closure of $\mathbf A\stensor\mathbf B$. Since this holds for any $x$, $\mathbf A\stensor\mathbf B$ is weak*-dense in $\mathbf A\otimes\mathbf B$.

  So $\mathbf A\atensor\mathbf B$ is weak*-dense in  $\mathbf A\stensor\mathbf B$ which is weak*-dense in $\mathbf A\otimes\mathbf B$.  So $\mathbf A\atensor\mathbf B$ is weak*-dense in $\mathbf A\otimes\mathbf B$.  
\end{proof}

\begin{lemma}\label{lemma:tensor.overlap.vec}
  Let $U_\alpha:(\calH_A\otimes\calH_B)\otimes\calH_C\to \calH_A\otimes(\calH_B\otimes\calH_C)$ be the unitary with $U_\alpha\pb\paren{(\psi\otimes\phi)\otimes\xi}=\psi\otimes(\phi\otimes\xi)$.
  (This is the canonical isomorphism between those spaces.)
  Fix $\psi\in\calH_A$, $\overlpx\in\calH_B\otimes\calH_C$, $\overlpp\in\calH_A\otimes \calH_B$, $\xi\in\calH_C$.
  Assume that $\psi,\xi\neq0$ and $\overlpp\otimes \xi = \adj{U_\alpha} (\psi\otimes \overlpx)$.
  Then there exists a $\phi\in\calH_B$ such that $\overlpp=\psi\otimes\phi$ and $\overlpx=\phi\otimes\xi$.
\end{lemma}

\begin{proof}
  We have
  \begin{align*}
    \overlpp\otimes\xi
    &=
    \adj{U_\alpha}(\psi\otimes\overlpx)
    =
    \adj{U_\alpha}(\psi\adj\psi\otimes 1)(\psi\otimes\overlpx)
    =
    \adj{U_\alpha}(\psi\adj\psi\otimes(1\otimes1))(\psi\otimes\overlpx) \\
    &=
    \adj{U_\alpha}(\psi\adj\psi\otimes(1\otimes1))U_\alpha(\overlpp\otimes\xi) 
    =
    ((\psi\adj\psi\otimes1)\otimes1)(\overlpp\otimes\xi)
    = (\psi\adj\psi\otimes1)\overlpp \otimes \xi.
  \end{align*}
  Since $\xi\neq0$, this implies $\overlpp = (\psi\adj\psi\otimes1)\overlpp$.
  Thus $\overlpp$ is in the image of the projector $\psi\adj\psi\otimes1$.
  Hence $\overlpp$ is of the form $\overlpp=\psi\otimes\phi$ for some $\phi\in\calH_B$ as desired.

  Then
  \begin{align*}
    \psi \otimes \overlpx
    =
    U_\alpha (\overlpp \otimes \xi)
    =
    U_\alpha ((\psi\otimes\phi)\otimes\xi)
    = \psi \otimes (\phi\otimes\xi).
  \end{align*}
  Since $\psi\neq0$, this implies $\overlpx = \phi\otimes\xi$ are desired.
\end{proof}

\begin{lemma}\label{lemma:tensor.overlap}
  Let $U_\alpha:(\calH_A\otimes\calH_B)\otimes\calH_C\to \calH_A\otimes(\calH_B\otimes\calH_C)$ be the unitary with $U_\alpha((\psi\otimes\phi)\otimes\xi)=\psi\otimes(\phi\otimes\xi)$.
  (This is the canonical isomorphism between those spaces.)
  Fix $a\in\mathbf A$, $\overlbc\in\mathbf B\otimes\mathbf C$, $\overlab\in\mathbf A\otimes \mathbf B$, $c\in\mathbf C$.
  Assume that $a,c\neq0$ and $\overlab\otimes c = \adj{U_\alpha} (a\otimes \overlbc) U_\alpha$.
  Then there exists some $b\in\mathbf B$ such that $\overlab=a\otimes b$ and $\overlbc=b\otimes c$.
\end{lemma}

\begin{proof}
  Let $\xi_0\in\calH_C$ be some vector with $c\xi_0\neq0$.
  Let $\psi_0\in\calH_A$ be some vector with $a\psi_0\neq0$.
  
  For any $\psi\in\calH_A,\phi\in\calH_B$, we have
  \[
    \overlab(\psi\otimes\phi)\otimes c\xi_0 =
    (\overlab\otimes c)((\psi\otimes\phi)\otimes \xi_0) =
    \adj{U_\alpha} (a\otimes \overlbc) U_\alpha ((\psi\otimes\phi)\otimes \xi_0)
    = 
    \adj{U_\alpha} (a\psi \otimes \overlbc(\phi\otimes\xi_0)).
  \]
  By \autoref{lemma:tensor.overlap.vec} (with $\overlpp:=\overlab(\psi\otimes\phi)$, $\xi:=c\xi_0$, $\overlpx:= \overlbc(\phi\otimes\xi_0)$, $\psi:=a\psi$), this implies that there exists an $\eta$ with $\overlab(\psi\otimes\phi)=a\psi\otimes\eta$.
  (In the case $a\psi=0$, the lemma does not apply, but in this case $\overlab(\psi\otimes\phi)=0$ and thus $\overlab(\psi\otimes\phi)=a\psi\otimes\eta$ holds for any $\eta$.)
  
  Thus we can fix some $\eta_{\psi\phi}$ such that $\overlab(\psi\otimes\phi)=a\psi\otimes\eta_{\psi\phi}$ for all $\psi,\phi$.
  And analogously some $\theta_{\phi\xi}$ such that $\overlbc(\phi\otimes\xi)=\theta_{\phi\xi}\otimes c\xi$ for all $\phi,\xi$.

  Let $b(\phi) := \eta_{\psi_0\phi}$.
  
  Note that we have $\overlab(\psi_0\otimes\phi)=a\psi_0\otimes b(\phi)$ which determines $ b(\phi)$ uniquely since $a\psi_0\neq0$.
  Furthermore,  $\overlab(\psi_0\otimes(s\phi+\phi'))=a\psi_0\otimes(sb(\phi)+b(\phi'))$ for all $\phi,\phi'$ and $s\in\setC$.
  Thus $b(s\phi+\phi') = sb(\phi)+b(\phi')$.
  Thus $b$ is linear.
  Then for all $\phi$,
  \[
    \norm{b(\phi)}
    = \frac{\norm{a\psi_0\otimes b(\phi)}}{\norm{a\psi_0}}
    = \frac{ \norm{\overlab(\psi_0\otimes\phi)} }{ \norm{a\psi_0} }
    \leq \frac{\norm \overlab \cdot \norm {\psi_0\otimes\phi}}{\norm{a\psi_0}}
    = \frac{\norm\overlab\cdot\norm{\psi_0}}{\norm{a\psi_0}}
    \cdot \norm {\phi}.
  \]
  Thus $b$ is bounded.
  Since $b$ is a bounded operator, we will from now on write $b\phi$ instead of $b(\phi)$.

  For any $\psi,\phi,\xi$, we have
  \begin{multline*}
    (a\psi \otimes \eta_{\psi\phi}) \otimes c\xi 
    = (\overlab (\psi\otimes\phi)) \otimes c\xi
    = (\overlab\otimes c)((\psi\otimes\phi)\otimes\xi)
    \\
    = \adj{U_\alpha} (a\psi \otimes \overlbc(\phi\otimes\xi))
    = \adj{U_\alpha} (a\psi \otimes (\theta_{\phi\xi}\otimes c\xi))
    = (a\psi \otimes \theta_{\xi\phi}) \otimes c\xi.
  \end{multline*}
  Thus whenever $a\phi,c\xi\neq 0$, we have $\eta_{\psi\phi}=\theta_{\phi\xi}$.

  For $\psi,\phi$ with $a\psi\neq0$, we then have
  \[
    \overlab(\psi\otimes\phi)
    =
    a\psi \otimes\eta_{\psi\phi}
    =
    a\psi \otimes\theta_{\phi\xi_0}
    =
    a\psi \otimes\eta_{\psi_0\phi}
    =
    a\psi\otimes b\phi
    =
    (a\otimes b)(\psi\otimes\phi).
  \]
  And for $\psi,\phi$ with $a\psi = 0$, we have
  \[
    \overlab(\psi\otimes\phi) = a\psi\otimes\eta_{\psi\phi} = 0 = a\psi\otimes b\phi
    =
    (a\otimes b)(\psi\otimes\phi).
  \]
  as well. Thus $\overlab=a\otimes b$ on all $\psi\otimes\phi$. Thus $\overlab=a\otimes b$ as desired.

  For $\phi,\xi$ with $c\xi\neq0$, we then have
  \[
    \overlbc(\phi\otimes\xi)
    =
    \theta_{\phi\xi} \otimes c\xi
    =
    \eta_{\psi_0\phi} \otimes c\xi
    =
    b\phi \otimes c\xi
    =
    (b\otimes c)(\phi\otimes\xi).
  \]
  And for $\phi,\xi$ with $c\xi = 0$, we have
  \[
    \overlbc(\phi\otimes\xi) = \theta_{\phi\xi} \otimes c\xi = 0 = b\phi\otimes c\xi
    =
    (b\otimes c)(\phi\otimes\xi).
  \]
  as well. Thus $\overlbc=b\otimes c$ on all $\phi\otimes\xi$. Thus $\overlbc=b\otimes c$ as desired.
\end{proof}

\begin{lemma}\label{lemma:reg.props}
  If $F:\mathbf A\to\mathbf B$ is a reference, then $F$ is linear, bounded, continuous, weak*-continuous, unital, completely positive, a *-homomorphism, and normal.
\end{lemma}

\begin{proof}
  By definition, $F$ is weak*-continuous unital *-homomorphisms.
  A *-homomorphism is linear by definition.
  By \cite[Example 34.3\,(a)]{conway00operator} *-homomorphisms are completely positive.
  Thus $F$ is completely positive.
  Positive linear maps are bounded by \cite[Proposition 33.4]{conway00operator}.
  Thus $F$ is bounded.
  Bounded linear maps are continuous, so $F$ is continuous.
  \cite[Definition III.2.15]{takesaki} defines a normal map as a continuous weak*-continuous linear map. \cite[Definition 46.1]{conway00operator} defines a normal map differently;
  by \cite[Theorem~46.4]{conway00operator} their definition of normal is equivalent to a positive weak*-continuous linear map.
  Since we already showed that $F$ is continuous, weak*-continuous, and positive, $F$ is normal according to both definitions.
\end{proof}

\begin{lemma}\label{lemma:normal.tensor1}
  If $F:\mathbf A\to\mathbf B$ is a reference, then there exists a
  Hilbert space $\calH_C$ and a unitary
  $U:\calH_A\otimes\calH_C\to\calH_B$ such that
  $F(a)=U(a\otimes 1_{\mathbf{C}})\adj U$ for all $a\in\mathbf A$.\footnote{The converse follows from \autoref{lemma:normal} together with Axiom~\ref{ax:tensor-1} that we prove later in \autoref{sec:proofs.infdim}.}
\end{lemma}

\begin{proof}
  By \autoref{lemma:reg.props}, $F$ is a normal *-homomorphism.
  \cite[Theorem IV.5.5]{takesaki} states that a normal homomorphism $F$ between von Neumann algebras $\mathbf A$ and $\mathbf B$ (and in particular our $F$) is of the form $F(a)=U(a\otimes 1_{\mathbf{D}})_{e'}\adj U$ for some Hilbert space $\calH_D$ and some projection $e'\in \comm{\mathbf A}\otimes\mathbf D$ and some unitary $U$.
  Here $\symbolindexmark\comm{\comm{\mathbf A}}$ denotes the \emph{commutant}\index{commutant} of $\mathbf A$, i.e., the set
  of all bounded operators on $\mathcal H_A$ that commute with all
  operators in $\mathbf A$.  And $\symbolindexmark\induction{\induction{x}{e'}}$ (the
  \emph{induction map}\index{induction map}) denotes the operator $x$,
  restricted in input and output to the image $\calH_e$ of the
  projection $e'$ \cite[Definition II.3.11]{takesaki}.  (I.e.,
  $\induction{x}{e'}$ is a bounded operator on $\calH_e$. Note that
  $\calH_e$ is a Hilbert space since it is a closed subspace of
  $\calH_D$.)

  Since in our case, $\mathbf A$ is the set of \emph{all} bounded
  operators on $\calH_A$, we have $\mathbf A'=\idmult$, the set of all multiples of the identity. Thus $e'\in\idmult\otimes\mathbf
  D$. Since $x\mapsto 1\otimes x$ is an isomorphism \cite[Corollary IV.1.5]{takesaki} and thus surjective onto $\idmult\otimes\mathbf
  D$, we have that $e'\in\idmult\otimes\mathbf
  D$ implies that $e'$ is of the form $1_{\mathbf{A}}\otimes e_0'$ for some $e_0'$. Thus $\calH_e=\calH_A\otimes\calH_C$ where we define  $\calH_{C}$ as the image of $e_0$. So $(a\otimes 1_{\mathbf{D}})_{e'}$ is the restriction of $a\otimes 1_{\mathbf{D}}$ to $\calH_A\otimes\calH_C$, i.e., $a\otimes 1_{\mathbf C}$. Thus 
  $F(a)=U(a\otimes 1_{\mathbf{C}})\adj U$ as desired.
\end{proof}

\begin{lemma}\label{lemma:reg.isometric}
  If $F$ is a quantum reference, then $\norm{F(a)}=\norm a$.
\end{lemma}

\begin{proof}
  By \autoref{lemma:normal.tensor1}, $F(a)=U(a\otimes 1)\adj U$ for unitary $U$.
  Thus $\norm{F(a)}=\norm{U(a\otimes1)\adj U} \starrel= \norm{a\otimes1} \starstarrel= \norm a\cdot\norm1 = \norm a$.
  Here $(*)$ follows because $U$ is unitary and $(**)$ follows by \cite[Equation IV.1\,(7)]{takesaki}.
\end{proof}

\begin{lemma}\label{lemma:reg.injective}
  If $F$ is a reference, then $F$ is injective.
\end{lemma}

\begin{proof}
  If $F(a)=0$, then $\norm a=\norm{F(a)}=0$ by \autoref{lemma:reg.isometric}.
  Thus $a=0$.
  Since $F$ is linear, this implies that $F$ is injective.

  
  
\end{proof}

\begin{lemma}\label{lemma:isoreg-decomp}
  If $F:\mathbf A\to\mathbf B$ is an iso-reference, then there exists a unitary
  $U:\calH_A\to\calH_B$ such that
  $F(a)=Ua\adj U$ for all $a\in\mathbf A$.
\end{lemma}

\begin{proof}
  By \autoref{lemma:normal.tensor1}, $F(a)=V(a\otimes 1_\mathbf{C})\adj V$ for some $\mathbf{C}$ and unitary $V$.
  Fix some arbitrary unit vector $\psi\in\calH_C$. 

  The function $x\mapsto Vx\adj V$ is bijective (with inverse $b\mapsto \adj VbV$), and $F$ is bijective, hence $a\mapsto a\otimes 1_\mathbf{C}$ is bijective as a function $\mathbf A\to\mathbf A\otimes\mathbf C$, too.
  Then exists an $a_*\in\mathbf A$ such that $a_*\otimes 1_\mathbf{C}=1_\mathbf{A}\otimes \psi\adj\psi$.
  Since the rhs is nonzero, $a_*$ is nonzero.
  Then there exists an $\alpha\in\calH_A$ such that $\adj\alpha a_*\alpha\neq 0$ \cite[Proposition II.2.15]{conway13functional}.
  Let $c := {\adj\alpha\alpha / \adj\alpha a_*\alpha}\neq0$.

  
  For all $\gamma\in\calH_C$:
  \[
    \adj\gamma 1_\mathbf{C} \gamma
    =
    \frac{\adj{(\alpha\otimes\gamma)} (a_*\otimes 1_\mathbf{C}) (\alpha\otimes\gamma)}
    {\adj\alpha a_*\alpha}
    =
    \frac{\adj{(\alpha\otimes\gamma)} (1_\mathbf{A}\otimes \psi\adj\psi) (\alpha\otimes\gamma)}
    {\adj\alpha a_*\alpha}
    =
    \frac{\adj\alpha\alpha\cdot\adj\gamma\psi\adj\psi\gamma}{\adj\alpha a_*\alpha}
    =
    \adj\gamma(c\psi\adj\psi)\gamma.
  \]
  By \cite[Proposition II.2.15]{conway13functional}, this implies $1_\mathbf{C}=c\psi\adj\psi$.
  And $c = \adj\psi(c\psi\adj\psi)\psi = \adj\psi1_\mathbf{C}\psi = 1$.
  Hence $1_\mathbf{C}=\psi\adj\psi$.

  Then $W:\calH_A\to \calH_A\otimes\calH_C$, $Wx := x\otimes\psi$ is an isometry.
  And we have
  \[
    VWa\adj W\adj V
    =
    V(a \otimes \psi\adj\psi)\adj V
    =
    V(a \otimes 1_\mathbf{C})\adj V
    =
    F(a).
  \]
  With $U:=VW$, we have $F(a)=Ua\adj U$.
  And since $F$ is surjective, so is $U$.
  And since $V$ is unitary and $W$ an isometry, $U$ is an isometry.
  A surjective isometry is unitary, so $U$ is unitary.
\end{proof}

\begin{lemma}\label{lemma:reg.tensor}
  For references $F:\mathbf A\to\mathbf C$, $G:\mathbf B\to\mathbf D$, there exists a reference $T:\mathbf A\otimes\mathbf B\to\mathbf C\tensor\mathbf D$ with $T(a\otimes b)=F(a)\otimes G(b)$ for all $a,b$.
\end{lemma}

\begin{proof}
  By \autoref{lemma:reg.props}, $F,G$ are both completely positive weak*-continuous maps.
  By \cite[Proposition IV.5.13]{takesaki}, there is a weak*-continuous completely positive linear map $T:\mathbf A\otimes\mathbf B\to\mathbf C\otimes\mathbf D$ such that $T(a\otimes b) = F(a)\otimes G(b)$ for all $a,b$.
  (Note that \cite{takesaki} uses the word $\sigma$-weak for weak*.)

  For $x=a\otimes b$, we have $T(\adj x)=\adj{T(x)}$ because $F,G$ preserve adjoints by definition of references.
  Since $T(\adj\cdot)$, $\adj{T(\cdot)}$ are antilinear, this implies $T(\adj x)=\adj{T(x)}$ for $x\in\mathbf A\atensor\mathbf B$.
  Since $\mathbf A\atensor\mathbf B$ is weak*-dense in $\mathbf A\tensor\mathbf B$ by \autoref{lemma:alg.dense}
  and $T$ and $\adj{(\cdot)}$ are weak*-continuous (\autoref{lemma:normal}), this implies $T(\adj x)=\adj{T(x)}$ on $\mathbf A\tensor\mathbf B$.

  For $x=a\otimes b$, $y=c\otimes d$, $T(xy)=T(x)T(y)$ because $F,G$ are multiplicative.
  By linearity, $T(xy)=T(x)T(y)$ for $x,y\in \mathbf A\atensor\mathbf B$.
  For fixed $x\in \mathbf A\atensor\mathbf B$, $y\mapsto T(xy)$ and $y\mapsto T(x)T(y)$ are weak*-continuous (\autoref{lemma:normal}, together with weak*-continuity of $T$), and they coincide on $\mathbf A\atensor\mathbf B$.
  Since $\mathbf A\atensor\mathbf B$ is weak*-dense in  $\mathbf A\tensor\mathbf B$, they coincide on  $\mathbf A\tensor\mathbf B$ as well.
  Thus, for fixed $y\in \mathbf A\tensor\mathbf B$, $x\mapsto T(xy)$ and $x\mapsto T(x)T(y)$ coincide on $\mathbf A\atensor\mathbf B$.
  And they are weak*-continuous by \autoref{lemma:normal}, together with weak*-continuity of $T$.
  Thus by weak*-density of $\mathbf A\atensor\mathbf B$, they coincide on 
  $ \mathbf A\tensor\mathbf B$. Hence   $T(xy)=T(x)T(y)$ for all $x,y$. Thus $T$ is multiplicative.

  We have $T(1)=T(1\otimes 1)=F(1)\otimes G(1)=1\otimes 1=1$. Thus $T$ is unital.
  
  Altogether, $T$ is a reference.
\end{proof}

\begin{lemma}\label{lemma:norm.mono}
  The norm is monotonous on the set of positive bounded operators.
  The trace norm is monotonous on the set of positive trace-class operators.
\end{lemma}

\begin{proof}
  For positive bounded $a$, we have
  \[
    \norm a \starrel= \norm{\sqrt a}^2 = \sup_\psi \norm{\sqrt a\psi}^2
    = \sup_\psi \adj\psi\adj{\sqrt a}\sqrt a\psi = \sup_\psi \adj\psi a\psi.
  \]
  Here $(*)$ uses \cite[Proposition II.2.7]{conway13functional}.
  And the suprema range over unit vectors $\psi$.
  Let $a\leq b$ be positive bounded operators. 
  Then
  \[
    \norm b
    = \sup_\psi (\psi^*a\psi + \psi^*(b-a)\psi)
    \starrel\geq \sup_\psi \psi^*a\psi
    = \norm a.
  \]
  Here $(*)$ uses that $\psi^*(b-a)\psi\geq0$ since $b-a$ is positive.
  So the norm is monotonous on the set of positive bounded operators.

  For positive trace-class $a$, $\trnorm a=\tr a$ (\lemmaref{lemma:traceclass.abs:normpos}).
  Let $a\leq b$ be positive trace-class operators. 
  Then $\trnorm b=\tr (a+(b-a)) = \tr a + \tr (b-a) = \trnorm a + \trnorm{b-a}\geq\trnorm a$.
  Here we use that $b-a$ is positive since $a\leq b$.
  So the trace norm is monotonous on the set of positive trace-class operators.
\end{proof}

\begin{lemma}\label{lemma:incr.net.lim}
  Let $a_i$ be an increasing net of positive bounded operators. Then the following are equivalent:
  \begin{compactenum}[(i)]
  \item\label{lemma:incr.net.lim:norm} ${\norm{a_i}}$ is upper-bounded.
  \item\label{lemma:incr.net.lim:sup} $a_i$ has a supremum in the set of bounded operators.
  \item\label{lemma:incr.net.lim:wot} $a_i$ converges in the weak operator topology (WOT).
  \item\label{lemma:incr.net.lim:sot} $a_i$ converges in the strong operator topology (SOT).
  \item\label{lemma:incr.net.lim:weak*} $a_i$ converges in the weak*-topology.
  \end{compactenum}
  In cases \eqref{lemma:incr.net.lim:sup}--\eqref{lemma:incr.net.lim:weak*}, the supremum and the limits coincide.
\end{lemma}    

Note that ``$a_i$ converges in the norm topology'' is not equivalent to statements \eqref{lemma:incr.net.lim:norm}--\eqref{lemma:incr.net.lim:weak*}.
For example, consider the operators $a_i := \sum_{j=1}^i \selfbutter j$ over the Hilbert space $\setC^{\setN}$. We have $a_i\to 1$ in the topologies mentioned in the theorem, but $a_i\to 1$ does not hold with respect to the norm-topology because $\norm{a_i-1}=1\not\to 0$.

\begin{proof}
  By \cite[Proposition 43.1]{conway00operator}, \eqref{lemma:incr.net.lim:norm} implies \eqref{lemma:incr.net.lim:sup}, \eqref{lemma:incr.net.lim:wot}, \eqref{lemma:incr.net.lim:sot}, \eqref{lemma:incr.net.lim:weak*}, and the supremum and the three limits coincide.
  
  The SOT and weak*-topology are both finer than the WOT. (See \cite[overview after Definition II.2.3]{takesaki}. The weak*-topology is called $\sigma$-weak there.) Hence \eqref{lemma:incr.net.lim:sot}, \eqref{lemma:incr.net.lim:weak*} each imply~\eqref{lemma:incr.net.lim:wot}.

  \medskip

  We now show that \eqref{lemma:incr.net.lim:wot} implies \eqref{lemma:incr.net.lim:sup}.
  Assume $a_i\to a$ in the WOT.
  By definition of the WOT, this means $\adj xa_ix\to \adj xax \in \setC$ for all vectors $x$.
  Furthermore, $\adj xa_ix$ is increasing since $a_i$ is increasing.
  Thus $\adj xa_ix\leq \adj xax$ for all $i$ and $x$.
  It follows that $\adj xax = \lim \adj xa_ix = \sup \adj xa_ix$.
  And $\adj xa_ix\leq \adj xax$ also implies $a_i\leq a$ \cite[Proposition 3.6]{conway00operator}.
  Fix some $b\leq a$ such that $b\geq a_i$ for all $i$.
  Then $\adj xbx\geq \adj xa_ix$ for all $i$.
  Since $\adj xax=\sup \adj xa_ix$, this implies $\adj xbx\geq \adj xax$.
  And from $b\leq a$, we get $\adj xbx\leq \adj xax$.
  Thus $\adj xbx=\adj xax$.
  Thus $b=a$ by \cite[Proposition II.2.15]{conway13functional}.
  This holds for any $b\leq a$ with $b\geq a_i$ for all $i$, hence $a=\sup a_i$.
  This shows that  \eqref{lemma:incr.net.lim:wot} implies \eqref{lemma:incr.net.lim:sup}.

  \medskip
  
  We now show that \eqref{lemma:incr.net.lim:sup} implies \eqref{lemma:incr.net.lim:norm}.
  Let $a$ be the supremum that exists by \eqref{lemma:incr.net.lim:sup}.
  Then $a\geq a_i$ for all $i$.
  Then by \autoref{lemma:norm.mono}, $\norm a\geq\norm{a_i}$ for all $i$.
  Thus $\norm{a_i}$ is bounded.

  \medskip
  
  Thus we have the following implications:
  
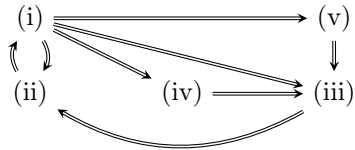
\begin{center}
    \begin{tikzpicture}[x=2cm]
      \node (norm) at (1,1) {\eqref{lemma:incr.net.lim:norm}};
      \node (sup) at (1,0) {\eqref{lemma:incr.net.lim:sup}};
      \node (wot) at (3,0) {\eqref{lemma:incr.net.lim:wot}};
      \node (sot) at (2,0) {\eqref{lemma:incr.net.lim:sot}};
      \node (weak*) at (3,1) {\eqref{lemma:incr.net.lim:weak*}};
      \tikzset{>=stealth, ->, double}
      \draw[double] (norm) to[bend left] (sup);
      \draw[double] (norm) -> (wot);
      \draw[double] (norm) -> (sot);
      \draw[double] (norm) -> (weak*);
      \draw[double] (sot) -> (wot);
      \draw[double] (weak*) -> (wot);
      \draw[double] (wot) to[bend left] (sup);
      \draw[double] (sup) to[bend left] (norm);
    \end{tikzpicture}
  \end{center}
  Hence all five statements are equivalent.
  And if any of them holds, \eqref{lemma:incr.net.lim:norm} holds and the supremum and the three limits coincide as mentioned in the beginning of the proof.
\end{proof}

\begin{lemma}\label{lemma:incr.net.lim.tr}
  Let $a_i$ be an increasing net of positive trace-class operators. Then the following are equivalent:
  \begin{compactenum}[(i)]
  \item\label{lemma:incr.net.lim.tr:norm} $\tr a_i$ (or equivalently $\trnorm {a_i}$) is upper-bounded.
  \item\label{lemma:incr.net.lim.tr:sup} $a_i$ has a supremum in the set of trace-class operators.
  \item\label{lemma:incr.net.lim.tr:wot} $a_i$ converges to a trace-class operator in the WOT.
  \item\label{lemma:incr.net.lim.tr:sot} $a_i$ converges to a trace-class operator in the SOT.
  \item\label{lemma:incr.net.lim.tr:weak*} $a_i$ converges to a trace-class operator in the weak*-topology.
  \item\label{lemma:incr.net.lim.tr:tr} $a_i$ converges to a trace-class operator w.r.t.~the trace norm.
  \item\label{lemma:incr.net.lim.tr:norm2} $a_i$ converges to a trace-class operator w.r.t.~the operator norm.
  \end{compactenum}
  In cases \eqref{lemma:incr.net.lim.tr:sup}--\eqref{lemma:incr.net.lim.tr:norm2}, the supremum and the limits coincide.
\end{lemma}

\begin{proof}
  The ``or equivalently $\trnorm{a_i}$'' part in \eqref{lemma:incr.net.lim.tr:norm} follows by \lemmaref{lemma:traceclass.abs:normpos}.
  
  Statements \eqref{lemma:incr.net.lim.tr:sup},\eqref{lemma:incr.net.lim.tr:wot},%
  \eqref{lemma:incr.net.lim.tr:sot},\eqref{lemma:incr.net.lim.tr:weak*} are equivalent by \autoref{lemma:incr.net.lim} (and the supremum/limits are the same in those four statements).
  
  Since the trace norm is monotonous (\autoref{lemma:norm.mono}), \eqref{lemma:incr.net.lim.tr:sup}
  implies \eqref{lemma:incr.net.lim.tr:norm}.

  \medskip
  
  We next show that \eqref{lemma:incr.net.lim.tr:norm} implies \eqref{lemma:incr.net.lim.tr:tr}.
  By \cite[Theorem 19.1]{conway00operator}, the space of trace-class operators is isometrically isomorphic to the dual of the space of compact operators.
  The dual of a normed space is a Banach space \cite[Proposition III.5.4]{conway13functional}, hence the dual of the space of compact operators is a Banach space.
  Thus the space of trace-class operators is isometrically isomorphic to a Banach space and thus it is a Banach space.
  Since $\tr$ is a positive linear functional \cite[Theorem 18.11\,(c)]{conway00operator}, $\tr$ is monotonous.
  By assumption $a_i$ is increasing, so $\tr a_i$ is increasing, too.
  By \eqref{lemma:incr.net.lim.tr:norm}, $\tr a_i$ is upper-bounded.
  Hence $\tr a_i$ converges and thus is a Cauchy net.
  Fix some $\varepsilon>0$.
  For sufficiently large $i,j$ we have $\abs{\tr a_j-\tr a_i}\leq\varepsilon/2$ because $\tr a_i$ is a Cauchy net.
  Then, for sufficiently large $i,j$, and with some $s\geq i,j$ ($s$ exists because $a_i$ is a net), we have
  \begin{equation*}
    \trnorm{a_i-a_j} \leq \trnorm{a_s-a_i} + \trnorm{a_s-a_j}
    \starrel=
    \abs{\tr a_s-\tr a_i} + \abs{\tr a_s-\tr a_j}
    \leq \varepsilon/2 + \varepsilon/2 = \varepsilon.
  \end{equation*}
  Here $(*)$ is by \lemmaref{lemma:traceclass.abs:normpos} and using that $a_s-a_i$, $a_s-a_j$ are positive because $a_i$ is an increasing net and $s\geq i,j$.
  Since $a_i$ is increasing, $\trnorm{a_j-a_i} = \tr( a_j-a_i) = \abs{\tr a_j-\tr a_j}\leq \varepsilon$.
  Thus $a_i$ is a Cauchy net as well (with respect to the trace norm).
  Since the trace-class operators are a Banach space (and thus complete), $a_i$ then converges to some trace-class operator.
  This shows that \eqref{lemma:incr.net.lim.tr:tr} is implied.

  \medskip

  We show that \eqref{lemma:incr.net.lim.tr:tr} implies \eqref{lemma:incr.net.lim.tr:norm2}.
  For all trace-class $t$, $\norm t\leq 2\trnorm t$.
  Hence convergence with respect to the trace norm implies convergence with respect to the norm.
  Thus \eqref{lemma:incr.net.lim.tr:tr} implies \eqref{lemma:incr.net.lim.tr:norm2}, and the limit is the same in both cases.

  We show that \eqref{lemma:incr.net.lim.tr:norm2} implies \eqref{lemma:incr.net.lim.tr:sot}.
  The norm-topology is finer than the SOT.
  (See \cite[overview after Definition II.2.3]{takesaki}. The norm-topology is called uniform there.)
  Hence convergence with respect to the norm topology implies convergence with respect to the SOT.
  Thus \eqref{lemma:incr.net.lim.tr:norm2} implies \eqref{lemma:incr.net.lim.tr:sot},
  and the limit is the same in both cases.

  \medskip

  Thus we have the following implications:
  \begin{center}
    \begin{tikzpicture}[x=1.5cm]
      \node (norm) at (-1,0) {\eqref{lemma:incr.net.lim.tr:norm}};
      \node (sup) at (0,0) {\eqref{lemma:incr.net.lim.tr:sup}};
      \node (wot) at (1,1) {\eqref{lemma:incr.net.lim.tr:wot}};
      \node (sot) at (0,1) {\eqref{lemma:incr.net.lim.tr:sot}};
      \node (weak*) at (1,0) {\eqref{lemma:incr.net.lim.tr:weak*}};
      \node (tr) at (-2,.5) {\eqref{lemma:incr.net.lim.tr:tr}};
      \node (norm2) at (-1,1) {\eqref{lemma:incr.net.lim.tr:norm2}};
      \tikzset{>=stealth, ->, double}
      \draw[double,<->] (sup) -- (sot);
      \draw[double,<->] (sot) -- (wot);
      \draw[double,<->] (wot) -- (weak*);
      \draw[double,<->] (weak*) -- (sup);
      \draw[double] (sup) -- (norm);
      \draw[double] (norm) -- (tr);
      \draw[double] (tr) -- (norm2);
      \draw[double] (norm2) -- (sot);
    \end{tikzpicture}
  \end{center}
  Hence all statements are equivalent.

  And as mentioned in the beginning of the proof,
  the supremum/limit in \eqref{lemma:incr.net.lim.tr:sup}, \eqref{lemma:incr.net.lim.tr:wot},
  \eqref{lemma:incr.net.lim.tr:sot}, \eqref{lemma:incr.net.lim.tr:weak*} is the same.
  And we showed that the limits in  \eqref{lemma:incr.net.lim.tr:tr} and \eqref{lemma:incr.net.lim.tr:norm2} are the same.
  And that the limits in \eqref{lemma:incr.net.lim.tr:norm2} and \eqref{lemma:incr.net.lim.tr:sot} are the same.
  Thus the supremum and the limits are the same in cases \eqref{lemma:incr.net.lim.tr:sup}--\eqref{lemma:incr.net.lim.tr:norm2}.
\end{proof}

\begin{lemma}\label{lemma:reference.sum}
  For a quantum reference $F$ and positive $a_i$, $F\pb\paren{\sum_i a_i}=\sum_i F(a_i)$ where the left sum converges iff the right sum converges.
  (Convergence of sums is with respect to the strong operator topology, SOT.)
\end{lemma}

\begin{proof}
  Let $s_M := \sum_{i\in M}a_i$ for finite sets $M$.
  Since $a_i$ are positive, $s_M$ are positive.
  Furthermore, $s_M$ is an increasing net (with respect to the set-inclusion ordering of finite sets $M$).
  $F$ is positive by \autoref{lemma:reg.props}.
  Thus $F(s_M)$ are positive.
  Since $F$ is positive, it is monotonous.
  Thus $F(s_M)$ is an increasing net.

  By \autoref{lemma:reg.isometric}, $\norm{s_M}=\norm{F(s_M)}$.
  Thus $\norm{s_M}$ is bounded iff $\norm{F(s_M)}$ is bounded.
  By \autoref{lemma:incr.net.lim}, this means that $s_M$ converges in the SOT iff $F(s_M)$ converges in the SOT.

  If $s_M,F(s_M)$ converge in the SOT, then they converge in the weak*-topology to the same limits by \autoref{lemma:incr.net.lim}.
  Since $F$ is weak*-continuous, $F(\lim_M s_M)=\lim_M F(s_M)$ in the weak*-topology.
  By \autoref{lemma:incr.net.lim}, this also implies $F(\lim_M s_M)=\lim_M F(s_M)$ in the SOT.

  Then we have
  \begin{equation*}
    F\pB\paren{\sum_i a_i}
    \starrel=
    F\pb\paren{\lim_M s_M}
    =
    \lim_M F(s_M)
    =
    \lim_M \sum_{i\in M} F(a_i)
    \starrel=
    \sum_i F(a_i).
  \end{equation*}
  Here each equality means that the limits/sums on the lhs converge iff they ones on the rhs converge. All limits/sums in the SOT.
  And $(*)$ is the definition of infinite sums \cite[Definition 4.11]{conway13functional}.
\end{proof}

\begin{lemma}\label{lemma:traceclass.tensor}
  If $a,b$ are trace-class operators on $\calH_A,\calH_B$, respectively, then $a\otimes b$ is trace-class and $\tr a\otimes b=\tr a\cdot\tr b$
  and $\trnorm{a\otimes b}=\trnorm a\cdot\trnorm b$. (Here $\trnorm\cdot$ is the trace norm.)
\end{lemma}

\begin{proof}
  By definition \cite[Definition 18.3]{conway00operator}, $a$ is trace-class iff $\sum_{e\in E} \adj e\abs ae$ converges for some orthonormal basis $E$. And  $\sum_{f\in F} \adj f\abs bf$ converges for some orthonormal basis $F$. With $G:=\{e\otimes f:e\in E, f\in F\}$, we have that
  \begin{multline*}
    \sum_{g\in G}\adj g\abs{a\otimes b} g
    =
    \sum_{e\in E, f\in F} (\adj e\otimes\adj f)\abs{a\otimes b}(e \otimes f)
    \ \ \txtrel{Lem.\ref{lemma:tensor.abs}}=
    \sum_{e\in E, f\in F} (\adj e\otimes\adj f)(\abs{a}\otimes\abs{b})(e \otimes f)
    \\
    = \sum_{e\in E, f\in F} \adj e\abs a e \cdot \adj f\abs bf
    = \sum_{e\in E} \adj e\abs a e \cdot \sum_{f\in F} \cdot \adj f\abs bf
  \end{multline*}
  converges. Thus $a\otimes b$ is trace-class.

  Since $\sum_{e\in E} \adj e\abs a e = \trnorm a$ by definition of the trace norm \cite[after Definition 18.3]{conway00operator}, and analogously for $\trnorm b$ and $\trnorm{a\otimes b}$, this equality also implies $\trnorm{a\otimes b}=\trnorm{a}\cdot\trnorm{b}$.

  By definition of the trace \cite[Definition 18.10]{conway00operator},
  $\tr a=\sum_{e\in E}\adj eae$ and
  $\tr b=\sum_{f\in F}\adj fbf$
  and
  $\tr (a\otimes b)=\sum_{g\in G}\adj g(a\otimes b)g$.
  (The definition of the trace is independent of the chosen basis $E,F,G$.
  And the sums converge.)
  
  Thus 
  \begin{multline*}
    \tr(a\otimes b) =
    \sum_{g\in G}\adj g(a\otimes b) g
    =
    \sum_{e\in E, f\in F} (\adj e\otimes\adj f)(a\otimes{b})(e \otimes f)
    \\
    = \sum_{e\in E, f\in F} \adj ea e \cdot \adj f bf
    \starrel= \sum_{e\in E} \adj ea e \cdot \sum_{f\in F} \adj f bf
    = \tr a \cdot \tr b.
  \end{multline*}
  For $(*)$ note that the sums are over the complex numbers.
\end{proof}

\begin{lemma}\label{lemma:tensor.generate.tracecl}
  If $S_1,S_2$ generate $\tracecl(\calH_1),\tracecl(\calH_2)$, respectively (with respect to the trace norm), then $\{\rho\otimes\sigma:\rho\in S_1,\sigma\in S_2\}$ generates $\tracecl(\calH_1\otimes\calH_2)$.
\end{lemma}

\begin{proof}
  In slight overloading of notation, for sets of trace-class operators $X,Y$, we write $X\otimes Y$ for $\{x\otimes y: x\in X, y\in Y\}$. Furthermore, let $\overline X$ denote the topological closure of $X$ with respect to the trace norm.
  With that notation, the conclusion of the lemma becomes $T(\calH_1\otimes\calH_2) \subseteq \overline{\Span(S_1\otimes S_2)}$.

  Consider first a rank-1 operator $t$ on $\calH_1\otimes\calH_2$.
  Then $t$ can be written as $t=\psi\adj\phi$ for $\psi,\phi\in\calH_1\otimes\calH_2$  \cite[Exercise II.4.8]{conway13functional}.
  We have that $\psi=\sum_{i=1}^\infty \psi_i\otimes\psi_i'$ for $\psi_i\in\calH_1,\psi_i\in\calH_2$.
  And analogously $\phi=\sum_{j=1}^\infty \phi_j\otimes\phi_j'$ for $\phi_i\in\calH_1,\phi_j\in\calH_2$.
  (Convergence of the sums is with respect to the usual topology of the Hilbert space.)
  We have
  \begin{align*}
    &\pB\trnorm{\sum\nolimits_{\substack{i\leq n\\j\leq m}}\psi_{i}\adj{\phi_j} \otimes \psi_i'\adj{{\phi_j'}} - \psi\adj\phi} \\
    &=
    \pB\trnorm{
      \pb\paren{\sum\nolimits_{i\leq n}\psi_i \otimes \psi_i' - \psi}\adj{\pb\paren{\sum\nolimits_{j\leq m}\phi_j \otimes \phi_j' - \phi}}
      + \psi\adj{\paren{\sum\nolimits_{j\leq m}\phi_j \otimes \phi_j' - \phi}}
      + \pb\paren{\sum\nolimits_{i\leq n}\psi_i \otimes \psi_i' - \psi}\adj\phi
      } \\
    &\leq
    \pB\trnorm{
      \pb\paren{\sum\nolimits_{i\leq n}\psi_i \otimes \psi_i' - \psi}\adj{\pb\paren{\sum\nolimits_{j\leq m}\phi_j \otimes \phi_j' - \phi}}}
      + \pB\trnorm{
      \psi\adj{\paren{\sum\nolimits_{j\leq m}\phi_j \otimes \phi_j' - \phi}}}
      + \pB\trnorm{
      \pb\paren{\sum\nolimits_{i\leq n}\psi_i \otimes \psi_i' - \psi}\adj\phi
      } \\
    &\leq
      \underbrace{\pB\norm{{\sum\nolimits_{i\leq n}\psi_i \otimes \psi_i' - \psi}}}_{{}\xrightarrow{n\to\infty} 0}
      \cdot \underbrace{\pB\norm{{{\sum\nolimits_{j\leq m}\phi_j \otimes \phi_j' - \phi}}}}_{{}\xrightarrow{m\to\infty} 0}
      + \pb\norm{\psi}
      \cdot\underbrace{\pB\norm{{{\sum\nolimits_{j\leq m}\phi_j \otimes \phi_j' - \phi}}}}_{{}\xrightarrow{m\to\infty} 0}
      + \underbrace{\pB\norm{{\sum\nolimits_{i\leq n}\psi_i \otimes \psi_i' - \psi}}}_{{}\xrightarrow{n\to\infty} 0}
      \cdot\norm{\phi}
    \\
    &\xrightarrow{n,m\to\infty}0.
  \end{align*}
  Then $\psi\adj\phi=\sum_{ij=1}^\infty\psi_{i}\adj{\phi_j} \otimes \psi_i'\adj{{\phi_j'}}$ (convergence with respect to the trace norm). 
  Since $\psi_i\adj{\phi_i} \otimes \psi_i'\adj{{\phi_i'}} \in \tracecl(\calH_1)\otimes\tracecl(\calH_2)$ by definition, this implies that $\psi\adj\phi \in \overline{\Span \pb\paren{\tracecl(\calH_1)\otimes\tracecl(\calH_2)}}$.
  Thus all rank-1 operators on $\calH_1\otimes\calH_2$ are in $ \overline{\Span \pb\paren{\tracecl(\calH_1)\otimes\tracecl(\calH_2)}}$.
  Since finite rank operators are linear combinations of rank-1 operators, they are in $\overline{\Span \pb\paren{\tracecl(\calH_1)\otimes\tracecl(\calH_2)}}$ as well.
  And since finite rank operators are dense in the space ${\tracecl(\calH_1\otimes\calH_2)}$ of all trace-class operators \cite[Theorem 18.11\,(d)]{conway00operator}, we have
  \begin{equation}
    \label{eq:THH.spanTT}
    {\tracecl(\calH_1\otimes\calH_2)} \subseteq
    \overline{\Span \pb\paren{\tracecl(\calH_1)\otimes\tracecl(\calH_2)}}.
  \end{equation}

  Furthermore
  \begin{align*}
    \overline{\Span \pb\paren{\tracecl(\calH_1)\otimes\tracecl(\calH_2)}}
    &\starrel=
    \overline{\Span \pb\paren{\overline{\Span S_1}\otimes \overline{\Span S_2}}}
    \starstarrel\subseteq
    \overline{\Span \overline{\pb\paren{\Span S_1\otimes \overline{\Span S_2}}}} \\
    &\starstarrel\subseteq
    \overline{\Span \overline{ \overline{\pb\paren{\Span S_1\otimes \Span S_2}}}}
    =
    \overline{\Span {{\pb\paren{\Span S_1\otimes \Span S_2}}}}
    =
    \overline{\Span {{\pb\paren{ S_1\otimes S_2}}}}.
  \end{align*}
  Here $(*)$ follows because by assumption $S_i$ generates $\tracecl(\calH_i)$, i.e., $\tracecl(\calH_i)=\overline{\Span S_i}$.
  And $(**)$ follows because $\overline X\otimes Y\subseteq \overline{X\otimes Y}$ (and $X\otimes \overline Y\subseteq \overline{X\otimes Y}$) which in turns follows from the trace norm continuity of $x\mapsto x\otimes y$ which in turn follows because $x\mapsto x\otimes y$ is bounded by \autoref{lemma:traceclass.tensor}.
  
  With \eqref{eq:THH.spanTT} this implies $T(\calH_1\otimes\calH_2) \subseteq \overline{\Span(S_1\otimes S_2)}$ and thus proves the lemma.
\end{proof}

\begin{lemma}\label{lemma:UU.qchannel}
  If $U:\calH_A\to\calH_B$ is an isometry, $\calE:\tracecl(\calH_A)\to\tracecl(\calH_B)$, $\rho\mapsto U\rho\adj U$ is a quantum channel (i.e., a completely positive trace-preserving linear map).
\end{lemma}

\begin{proof}
  $\calE$ is obviously linear.
  Since $U$ is an isometry, $\calE(\rho\sigma)=U\rho\sigma\adj U=U\rho\adj UU\sigma\adj U=\calE(\rho)\calE(\sigma)$, i.e., $\calE$ is multiplicative.
  Furthermore $\calE(\adj\rho)=U\adj\rho \adj U = \adj{\pb\paren{U\rho\adj U}} = \adj{\calE(\rho)}$, i.e., $\calE$ preserves adjoints.
  That is, $\calE$ is a *-homomorphism.
  By \cite[Example 34.3\,(a)]{conway00operator} this implies that $\calE$ is completely positive.

  Furthermore,
  \begin{equation*}
    \tr\calE(\rho) = \tr U\rho\adj U \starrel= \tr \rho \adj UU \starstarrel= \tr\rho I = \tr\rho.
  \end{equation*}
  Here $(*)$ follows by the circularity of the trace \cite[Theorem 18.11\,(e)]{conway00operator} (\lemmaref{lemma:circ.trace:circ}).
  And $(**)$ follows since $U$ is an isometry.
  Thus $\calE$ is trace-preserving.
\end{proof}

\begin{lemma}[Schrödinger-Heisenberg dual]\lemmalabel{lemma:SH.adjoint}
  Consider the following relation between functions $f:T(\calH)\to T(\calK)$ and  $g:B(\calK)\to B(\calH)$:
  \[
    \tr f(t)a = \tr g(a)t
    \qquad\text{for all $t\in\tracecl(\calH)$, $a\in\bounded(\calK)$}
    \qquad\qquad (*)
  \]
  \begin{compactenum}[(i)]
  \item\itlabel{item:SH.adjoint.unique}
    The relation $(*)$ is biunique.
    I.e., for every $f$ there is at most one $g$, and for every $g$ there is at most one $f$ satisfying~$(*)$.
  \item\itlabel{lemma:SH.adjoint.corr} Bounded linear maps $f:T(\calH)\to T(\calK)$ stand in
    1-1 correspondence to
    weak*-continuous bounded linear maps $g:B(\calK)\to B(\calH)$ via the relation $(*)$.
    (I.e., for every such $f$ there exists exactly one such $g$ and vice versa.)
  \item\itlabel{item:SH.adjoint.pos} 
    And $f$ is positive iff $g$ is positive.
    (Here and in the following statements, assume $f,g$ stand in relation~$(*)$.)
  \item \itlabel{item:SH.adjoint.cp} 
    And $f$ is completely positive iff $g$ is completely positive.
  \item \itlabel{item:SH.adjoint.trp} 
    And $f$ is trace-preserving iff $g$ is unital.
  \item \itlabel{item:SH.adjoint.trred}
    And $f$ is trace-reducing ($\tr f(t)\leq \tr t$ for positive $t$) iff $g$ is subunital.
  \end{compactenum}
\end{lemma}

\begin{proof}\anonymous{}{\!\!\footnote{\notanonymous Thanks to Robert Furber for helpful pointers on how to prove this \cite{furber21shcomments}.}}
  We first show \eqref{item:SH.adjoint.unique}.

  An operator $b\in B(\calH)$ is uniquely determined by the function $\phi,\psi\mapsto \phi^*b\psi$. We have
  $ \phi^*b\psi =  \tr \phi^*b\psi = \tr b\psi \phi^*$ where the last equality is the circularity of the trace (\lemmaref{lemma:circ.trace:circ}).
  Since $\psi\phi^*$ is trace-class, this means that $b$ is uniquely determined by $\tr bt$ for all $t\in T(\calH)$.
  Thus for fixed $f$, $g(a)$ is uniquely determined for all $a$ by $(*)$.
  Similarly, for fixed $g$, $f(t)$ is uniquely determined for all $t$ by $(*)$.

  This shows \eqref{item:SH.adjoint.unique}.
  
  \bigskip

  We show \eqref{lemma:SH.adjoint.corr}.

  We need some observations first:
  
  We have a duality $\langle B(\calH), T(\calH)\rangle$
  with canonical bilinear form $(a,t)\mapsto \tr at$ with $a\in\bounded(\calH)$, $t\in\tracecl(\calH)$.
  By definition of a duality \cite[beginning of section IV.1]{schaefer71topological}, this means that $(a,t)\mapsto \tr at$ is bilinear and satisfies $(\forall a. \tr at=0)\Longrightarrow t=0$ and
  $(\forall t. \tr at=0)\Longrightarrow a=0$.
  Bilinearity follows from the linearity of the trace, and the other two properties were shown in
  the proof of \eqref{item:SH.adjoint.unique}.
  Analogously we have a duality $\langle B(\calK), T(\calK)\rangle$.

  For a duality $\langle X,Y\rangle$, let $\sigma(X,Y)$ denote the coarsest topology on $X$ such that $x\mapsto \langle x,y\rangle$ is continuous for all $y\in Y$ \cite[before IV.1.2]{schaefer71topological}.
  Here $\langle x,y\rangle$ is the canonical bilinear form of $\langle X,Y\rangle$. 

  By definition of the weak*-topology \cite[beginning of Section 20]{conway00operator}, $\sigma(B(\calK),T(\calK))$ and $\sigma(B(\calH),T(\calH))$ are the weak*-topology on $B(\calK)$, $B(\calH)$, respectively.
  
  Now we proceed to the actual proof of \eqref{lemma:SH.adjoint.corr}.
  
  First, assume that $f:\tracecl(\calH)\to \tracecl(\calK)$ is a bounded linear map. We want to show that there exists a weak*-continuous bounded linear map $g:\bounded(\calK)\to\bounded(\calH)$ that stands in relation $(*)$ with~$f$.
  (Uniqueness of $g$ already follows from \eqref{item:SH.adjoint.unique}.)

  By \cite[discussion before IV.2.4]{schaefer71topological}, any continuous (i.e., bounded) $f:T(\calH)\to T(\calK)$ is continuous for $\sigma(\tracecl(\calH),\tracecl(\calH)')$ and $\sigma(\tracecl(\calK),\tracecl(\calK)')$ where $X'$ denotes the dual of $X$. (I.e., the space of continuous linear functionals on $X$.)
  Furthermore, $\bounded(\calH)$ is isometrically isomorphic to the dual of $\tracecl(\calH)$ if we identify $b\in\bounded(\calH)$ with $t\in\tracecl\mapsto \tr bt$ \cite[Theorem 19.2]{conway00operator}.
  Thus $\sigma(\tracecl(\calH),\tracecl(\calH)')$ is the same as $\sigma(\tracecl(\calH),\bounded(\calH))$.
  Analogously, $\sigma(\tracecl(\calK),\tracecl(\calK)')$ is the same as $\sigma(\tracecl(\calK),\bounded(\calK))$.
  Thus $f$ is continuous for $\sigma(T(\calH),B(\calH))$ and $\sigma(T(\calK),B(\calK))$.
  
  By \cite[IV.2.1]{schaefer71topological}, this implies that there exists a  $g:\bounded(\calK)\to\bounded(\calH)$ (the \emph{adjoint}) that is continuous for $\sigma(B(\calK),T(\calK))$ and $\sigma(B(\calH),T(\calH))$.
  This $g$ is furthermore linear and satisfies $\tr af(t)=\tr g(a)t$ for all $a,t$ because it is the restriction of the \emph{algebraic adjoint} to a smaller domain, and the algebraic adjoint has those properties \cite[beginning of section IV.2]{schaefer71topological}.
  So $g$ is weak*-continuous and linear.

  We next show that $g$ is bounded.
  Fix $a\in B(\calK)$. Since $f$ is bounded, there exists a $c$ such that
  $\trnorm{f(t)}\leq c\trnorm{t}$ for all $t\in T(\calH)$.
  We have
  \begin{align*}
    \norm{g(a)}
    & = \sup_{\psi,\phi} \pb\abs{\adj\phi g(a)\psi}
      \starrel = \sup_{\psi,\phi} \pb\abs{\tr g(a)\psi\adj\phi}
      = \sup_{\psi,\phi} \pb\abs{\tr af(\psi\adj\phi)}
    \\
    &      \starstarrel\leq
      \sup_{\psi,\phi} \norm a \cdot \pb\trnorm{f(\psi\adj\phi)}
      \leq \sup_{\psi,\phi} \norm a \cdot c \underbrace{\trnorm {\psi\adj\phi}}_{=1}
      = c \norm a.
  \end{align*}
  Here the suprema range over unit vectors $\psi,\phi\in \calH$.
  And $(*)$ uses the circularity of the trace, \lemmaref{lemma:circ.trace:circ}.
  And $(**)$ follows by \cite[Theorem 18.11\,(e)]{conway00operator}.
  Thus $g$ is bounded.

  Hence there exists a weak*-continuous bounded linear map $g$ that stands in relation $(*)$ with $f$.

  Next, assume that $g:\bounded(\calK)\to\bounded(\calH)$ is a weak*-continuous bounded linear map. We want to show that there exists a bounded linear map $f:T(\calH)\to T(\calK)$ that stands in relation $(*)$ with $g$.
  Since $g$ is weak*-continuous, $g$ is continuous for $\sigma(B(\calK),T(\calK))$ and $\sigma(B(\calH),T(\calH))$.
  By \cite[IV.2.1]{schaefer71topological}, this implies that there exists an $f$ that is continuous for $\sigma(T(\calH),B(\calH))$ and $\sigma(T(\calK),B(\calK))$ satisfying $\tr af(t)=\tr g(a)t$ for all $a,t$.
  This $f$ is furthermore linear and satisfies $\tr af(t)=\tr g(t)$ for all $a,t$ because it is the restriction of the \emph{algebraic adjoint} to a smaller domain, and the algebraic adjoint has those properties \cite[beginning of section IV.2]{schaefer71topological}.
  
  We show that $f$ is
  bounded. Fix $t\in T(\calH)$. By the polar decomposition \cite[Theorem~3.9]{conway00operator},
  there is a partial isometry $u\in B(\calK)$ such that $f(t)=u\abs{f(t)}$.
  Furthermore, $\adj uu$ is the projection on the closure of the range of $\abs{f(t)}$.
  (Shown in the proof of \cite[Theorem~3.9]{conway00operator}, noting that $U=W=u$ and $P=\abs A=\abs{f(t)}$ there.)
  Thus $\adj uf(t)=\adj uu\abs{f(t)}=\abs{f(t)}$.
  Then
  \begin{align*}
    \trnorm{f(t)} &=
    \pb\trnorm{u\abs{f(t)}}
    \starrel\leq
    \norm u\cdot \pb\trnorm{\abs{f(t)}}
    \starstarrel=
    \norm u\cdot \pb\abs{\tr{\abs{f(t)}}}
    =
    \norm u\cdot \pb\abs{\tr{\adj uf(t)}}
    \\&
    =
    \norm u\cdot \pb\abs{\tr{g(\adj u)t}}
    \tristarrel\leq
    \norm u\cdot \pb\norm{g(\adj u)} \cdot \trnorm t
    \leq
    \norm u\cdot \norm g \norm{\adj u} \cdot \trnorm t
    \fourstarrel\leq
    \norm g\cdot \trnorm t.
  \end{align*}
  Here $(*)$ follows by \cite[Theorem 18.11\,(g)]{conway00operator}.
  And $(**)$ by \lemmaref{lemma:traceclass.abs:normpos}.
  And $(*{*}*)$ by \cite[Theorem 18.11\,(e)]{conway00operator}.
  And the operator norm $\norm g$ exists because $g$ is bounded by assumption.
  And $(*{*}{*}*)$ uses that $\norm{u},\norm{\adj u}\leq 1$ since $u$ is a partial isometry.
  Thus $f$ is bounded.

  Hence there exists a bounded linear map $f$ that stands in relation $(*)$ with $g$.

  This shows \eqref{lemma:SH.adjoint.corr}

  \bigskip

  We now show \eqref{item:SH.adjoint.pos}.

  We need the auxiliary fact that $\tr at\geq0$ for positive $a\in B(\calH)$ and positive $t\in T(\calH)$.
  Since $a$ is positive, by definition of positivity, $a=\adj{b}b$ for some $b\in B(\calH)$ \cite[Theorem~3.4\,(c)]{conway00operator}. 
  Then $\tr at=\tr bt\adj b\geq 0$ where the $=$ follows from the circularity of the trace (\lemmaref{lemma:circ.trace:circ}),
  and $\geq$ follows since $bt\adj b$ is positive and thus has positive trace \cite[Theorem~18.11\,(c)]{conway00operator}.
  
    Assume that $f$ is positive. We show that $g$ is positive. Fix a positive $a\in B(\calK)$. 
  For every $\psi\in \calH$, $\psi\adj\psi$ is positive and trace-class. Thus $f(\psi\adj\psi)$ is positive.
  Then $\tr af(\psi\adj\psi)\geq0$. Then $\adj\psi g(a)\psi = \tr g(a) \psi\adj\psi
  = \tr a f(\psi\adj\psi) \geq 0$.
  Here the first equality uses the circularity of the trace (\lemmaref{lemma:circ.trace:circ}) and $\geq$ uses the auxiliary fact.
  Since $\adj\psi g(a)\psi\geq0$ for all $\psi$, $g(a)$ is positive \cite[Proposition 3.6]{conway00operator}.
  Since this holds for every positive $a$, $g$ is positive.

  Assume that $g$ is positive. We show that $f$ is positive. Fix a positive $t\in T(\calH)$.
  For every $\psi\in \calH$, $\psi\adj\psi$ is positive and bounded. Thus $g(\psi\adj\psi)$ is positive.
  Then $\tr g(\psi\adj\psi)t\geq0$. Then $\adj\psi f(t)\psi = \tr\psi\adj\psi f(t)
  = \tr g(\psi\adj\psi) t \geq 0$.
  Here the first equality uses the circularity of the trace (\lemmaref{lemma:circ.trace:circ}).
  Since $\adj\psi f(t)\psi\geq0$ for all $\psi$, $f(t)$ is positive. Since this holds for every positive $t$, $f$ is positive.

  Thus $f$ is positive iff $g$ is.
  This shows \eqref{item:SH.adjoint.pos}.
  
  \bigskip

  We show \eqref{item:SH.adjoint.cp}. Fix $f,g$ standing in relation $(*)$.

  Let $M_n(S)$ denote the set of $n\times n$ matrices with elements in
  $S$ where $S=T(\calH),T(\calK),B(\calH),B(\calK)$, together with the usual matrix
  addition and multiplication.
  For a linear map $f:S\to S'$, let $f^{(n)}:M_n(S)\to M_n(S')$
  be the component-wise application of $f$.

  By definition \cite[Definition~34.2]{conway00operator}, $f$ is completely
  positive iff $f^{(n)}$ is positive for all $n\geq 1$.

  Note that $M_n(B(\calH))$ are the bounded operators on $\calH^{(n)}:=\calH\oplus\dots\oplus \calH$, and analogously for $M_n(B(\calK))$.
  And $M_n(T(\calH)), M_n(T(\calK))$ are the trace-class operators on $\calH^{(n)},\calK^{(n)}$.

  Since $f,g$ stand in correspondence $(*)$, $f^{(n)},g^{(n)}$ stand in correspondence $(*)$ as well.
  (Note that $\tr t=\sum_i \tr t_{ii}$ for $t\in M_n(T(\calH))$.)
  Then $f$ is completely positive iff $f^{(n)}$ is positive for all $n\geq1$
  iff $g^{(n)}$ is positive for all $n\geq 1$ (by \eqref{item:SH.adjoint.pos} applied to $f^{(n)},g^{(n)}$) iff $g$ is completely positive.

  This shows \eqref{item:SH.adjoint.cp}.

  \bigskip

  We show \eqref{item:SH.adjoint.trp}. $f$ is trace-preserving iff $\tr f(t)=\tr t$ for all $t$
  iff $\tr f(t)1=\tr t$ for all $t$ iff $\tr t g(1)=\tr t$ for all $t$ (by relation $(*)$) iff $g(1)=1$ (i.e., $g$ is unital). The $\Rightarrow$ direction in the last ``iff'' follows since
  $\tr ta$  determines $a$ (see the proof of \eqref{item:SH.adjoint.unique}).
  This shows \eqref{item:SH.adjoint.trp}.

  \bigskip

  We show \eqref{item:SH.adjoint.trred}.
  
  Assume that $f$ is trace-reducing.
  Then for all $\psi\in\calH$,
  \begin{equation*}
    \adj\psi(1-g(1))\psi
    = \tr\adj\psi\psi - \tr g(1) \psi\adj\psi
    = \tr\adj\psi\psi - \tr 1 f(\psi\adj\psi)
    \starrel\geq 0.
  \end{equation*}
  Here $(*)$ holds because $f$ is trace-reducing and $\psi\adj\psi$ positive.
  Since this holds for all $\psi$, we have that ${1-g(1)}$ is positive \cite[Proposition 3.6]{conway00operator}.
  Thus $g(1)\leq 1$. Hence $g$ is subunital.

  Now assume that $g$ is subunital.
  Then for positive $\rho\in\tracecl(\calH)$,
  $1-g(1)$ is positive and thus $\sqrt\rho\pb\paren{1-g(1)}\rho$ is positive. Then
  \begin{equation*}
    \tr\rho - \tr f(\rho) = \tr\rho - \tr f(\rho) 1 =
    \tr\rho - \tr \rho g(1)
    \starrel=\tr\sqrt\rho\sqrt\rho - \tr \sqrt\rho g(1) \sqrt\rho
    =\tr \sqrt\rho\pb\paren{1-g(1)}\sqrt\rho \geq 0.
  \end{equation*}
  Here $(*)$ uses the circularity of the trace (\lemmaref{lemma:circ.trace:circ}).
  Hence $\tr f(\rho)\leq\tr\rho$.
  So $f$ is trace-reducing.

  This shows \eqref{item:SH.adjoint.trred}.
\end{proof}

\begin{lemma}\label{lemma:channel.bounded}
  Any quantum subchannel $\calE$ is a bounded linear map (with respect to the trace norm).
  In particular, continuous with respect to the trace norm.  
\end{lemma}

\begin{proof}
  Fix a Hermitian trace-class $t$. By \cite[Proposition 3.2]{conway00operator}, $t=t_+-t_-$ where $t_+$ and $t_-$ are the ``positive and negative part'' of $t$, and $t_+,t_-$ are positive.
  Furthermore, $\abs t=t_++t_-$ by  \cite[Exercise 3.4]{conway00operator}.
  Since $\abs t$ is trace-class (\lemmaref{lemma:traceclass.abs:tc}), and $t_+,t_i \leq\abs t$, we have that $t_+,t_-$ are trace-class by \lemmaref{lemma:traceclass.abs:between}.
  Then
  \begin{equation*}
    \pb\trnorm{\calE(t)}
    = \pb\trnorm{\calE(t_+)-\calE(t_-)}
    \leq \pb\trnorm{\calE(t_+)} + \pb\trnorm{\calE(t_-)}
    \starrel= \tr{\calE(t_+)} + \tr{\calE(t_-)}
    \starstarrel\leq \tr t_+ + \tr t_-
    = \tr t.
  \end{equation*}
  Here $(*)$ follows with \lemmaref{lemma:traceclass.abs:normpos} since $\calE$ is positive and thus $\calE(t_+),\calE(t_-)$ are positive.
  And $(**)$ follows since $\calE$ is trace-reducing.
  Thus for Hermitian $t$, $\trnorm{\calE(t)}\leq\trnorm t$.

  Fix an arbitrary (not necessarily Hermitian) trace-class $t$. Then $t=t_h + it_a$ where $t_h := (t + \adj t)/2$ and $t_a := (t - \adj t)/2i$.
  Both $t_h,t_a$ are easily verified to be Hermitian.
  Thus
  \begin{align*}
    \pb\trnorm{\calE(t)}
    &= \pb\trnorm{\calE(t_h)+i\calE(t_a)}
    \leq \pb\trnorm{\calE(t_h)} + \pb\trnorm{\calE(t_a)}
    \starrel\leq \pb\trnorm{t_h} + \pb\trnorm{t_a} \\
    &= \pb\trnorm{t/2+\adj t/2} + \pb\trnorm{t/2i - \adj t/2i} 
    \leq \tfrac12\pb\trnorm t + \tfrac12\pb\trnorm {\adj t} + \tfrac12\pb\trnorm t + \tfrac12\pb\trnorm {\adj t}
    \starstarrel= 2\trnorm t.
  \end{align*}
  Here $(*)$ uses that we showed  $\trnorm{\calE(t)}\leq\trnorm t$ for Hermitian $t$ above.
  And $(**)$ is by \cite[Theorem~18.11\,(f)]{conway00operator}.

  Thus $\trnorm{\calE(t)}\leq\trnorm t$ for all trace-class $t$, i.e., $\calE$ is bounded with respect to the trace-class.
  A bounded linear map is continuous.
\end{proof}

\begin{lemma}\label{lemma:channel.equal}
  If $g,h:\tracecl(\calH_1\otimes\calH_2)\to\tracecl(\calK)$ are quantum subchannels with $f(\rho_1\otimes\rho_2)=g(\rho_1\otimes\rho_2)$ for all trace-class $\rho_1,\rho_2$, then $f=g$.
  
  If $g,h:\tracecl\pb\paren{(\calH_1\otimes\calH_2)\otimes\calH_3}\to\tracecl(\calK)$ are quantum subchannels with $f\pb\paren{(\rho_1\otimes\rho_2)\otimes\rho_3}=g\pb\paren{(\rho_1\otimes\rho_2)\otimes\rho_3}$ for all trace-class $\rho_1,\rho_2,\rho_3$, then $f=g$.
  
  If $g,h:\tracecl\pb\paren{\calH_1\otimes(\calH_2\otimes\calH_3)}\to\tracecl(\calK)$ are quantum subchannels with $f\pb\paren{\rho_1\otimes(\rho_2\otimes\rho_3)}=g\pb\paren{\rho_1\otimes(\rho_2\otimes\rho_3)}$ for all trace-class $\rho_1,\rho_2,\rho_3$, then $f=g$.
\end{lemma}

\begin{proof}
  By \autoref{lemma:tensor.generate.tracecl}, the $\rho_1\otimes\rho_2$ (the $(\rho_1\otimes\rho_2)\otimes\rho_3$, the $\rho_1\otimes(\rho_2\otimes\rho_3)$) generate the whole space $\tracecl(\calH_1\otimes\calH_2)$ ($\tracecl\pb\paren{(\calH_1\otimes\calH_2)\otimes\calH_3}$, $\tracecl\pb\paren{\calH_1\otimes(\calH_2\otimes\calH_3)}$).
  Furthermore, $g,h$ are continuous linear maps (w.r.t.~the trace norm) by \autoref{lemma:channel.bounded}.
  Thus $g=h$.
\end{proof}

\begin{lemma}\lemmalabel{lemma:channel.tensor}
  Fix completely positive $f_i,g_i:\tracecl(\calH_i)\to \tracecl(\calK_i)$ for $i=1,2$.
  \begin{compactenum}[(i)]
  \item \itlabel{lemma:channel.tensor:cp}
    There exists a
    completely positive $f_1\otimes f_2:\tracecl(\calH_1\otimes \calH_2)\to \tracecl(\calK_1\otimes \calK_2)$ such that
    $(f_1\otimes f_2)(\rho_1\otimes\rho_2) = f_1(\rho_1)\otimes f_2(\rho_2)$.
  \item \itlabel{lemma:channel.tensor:unique}
    If $f_1,f_2$ are trace-reducing, then $f_1\otimes f_2$ as defined in \eqref{lemma:channel.tensor:cp} is unique.
  \item \itlabel{lemma:channel.tensor:distrib}
    If $f_1,f_2,g_1,g_2$ are trace-reducing, then $(g_1\otimes g_2)\circ(f_1\otimes f_2) = (g_1\circ f_1) \otimes (g_2\circ f_2)$.
  \item \itlabel{lemma:channel.tensor:tpres}
    If $f_1,f_2$ are trace-preserving, then $f_1\otimes f_2$ is.
    (I.e., the tensor product of quantum channels exists and is a quantum channel.)
  \item \itlabel{lemma:channel.tensor:tred}
    If $f_1,f_2$ are  trace-reducing (i.e., for positive $\rho$, $\tr f_i(\rho)\leq\tr\rho$), then $f_1\otimes f_2$ is.
    (I.e., the tensor product of quantum subchannels exists and is a quantum subchannel.)
  \item \itlabel{lemma:channel.tensor:id}
    $\id_{\tracecl(\calH_1)}\otimes\id_{\tracecl(\calH_2)}=\id_{\tracecl(\calH_1\otimes\calH_2)}$.
  \end{compactenum}
\end{lemma}

\begin{proof}\anonymous{}{\!\!\footnote{\notanonymous Thanks to Robert Furber for helpful pointers on how to prove this \cite{furber21tensorcomment}.}}
  We first show \eqref{lemma:channel.tensor:cp}.
  
  By \lemmaref{lemma:SH.adjoint.corr}, there are weak*-continuous completely positive $g_1,g_2:\bounded(\calK_i)\to \bounded(\calH_i)$ such that $\tr af_i(t)=\tr g_i(a)t$ for all $a\in \bounded(\calK_i)$, $t\in \tracecl(\calH_i)$.
  By \cite[Proposition~IV.5.13]{takesaki}, there is a weak*-continuous completely positive $\theta:\bounded(\calK_1)\otimes \bounded(\calK_2)\to \bounded(\calH_1)\otimes \bounded(\calH_2)$
  such that $\theta(a_1\otimes a_2) = g_1(a_1)\otimes g_2(a_2)$ for all $a_1,a_2$.
  Note that $\bounded(\calK_1)\otimes \bounded(\calK_2) = \bounded(\calK_1\otimes\calK_2)$ and
  $\bounded(\calH_1)\otimes \bounded(\calH_2) = \bounded(\calH_1\otimes\calH_2)$ \cite[(10) after Proposition IV.1.6]{takesaki}.
  (This is an equality of spaces, not just an isomorphism.)
  Positive maps are bounded \cite[Proposition~33.4]{conway00operator}.
  Thus $\theta$ is weak*-continuous bounded completely positive.
  Hence by \lemmaref{lemma:SH.adjoint.corr}, there exists a completely positive $\theta':\tracecl(\calH_1\otimes\calH_2)\to \tracecl(\calK_1\otimes\calK_2)$ such that $\tr\theta'(t)a=\tr t\theta(a)$ for all $t,a$.
  
  Fix $\rho_1\in \tracecl(\calH_1)$, $\rho_2\in \tracecl(\calH_2)$.
  $\psi_1,\phi_1\in \calK_1$, $\psi_2,\phi_2\in\calK_2$.
  Define $a_1:=\psi_1\adj{\phi_1}\in\bounded(\calK_1)$ and $a_2:=\psi_2\adj{\phi_2}\in\bounded(\calK_1)$.
  Then
  \begin{align*}
    \hskip1cm&\hskip-1cm \adj{(\phi_1\otimes\phi_2)}\theta'(\rho_1\otimes \rho_2)(\psi_1\otimes\psi_2)
    \starrel=
    \tr \theta'(\rho_1\otimes \rho_2)(\psi_1\otimes\psi_2)\adj{(\phi_1\otimes\phi_2)}
    \\
    & =
    \tr \theta'(\rho_1\otimes \rho_2) (a_1\otimes a_2)
    =
      \tr (\rho_1\otimes \rho_2) \theta(a_1\otimes a_2)
      =
      \tr \rho_1g_1(a_1) \otimes \rho_2 g_2(a_2)
    \\ &
    \starstarrel= 
      \tr \rho_1g_1(a_1) \cdot \tr \rho_2g_2(a_2) 
     = 
      \tr f_1(\rho_1)a_1 \cdot \tr f_2(\rho_2)a_2
    \\&
    \starstarrel= 
    \tr f_1(\rho_1)a_1 \otimes f_2(\rho_2)a_2
    =
    \tr \pb\paren{f_1(\rho_1) \otimes f_2(\rho_2)}(a_1 \otimes a_2)
    \\&
    \starrel=
     \adj{(\phi_1\otimes\phi_2)} (f_1(\rho_1) \otimes f_2(\rho_2))(\psi_1\otimes\psi_2).
  \end{align*}
  Here $(*)$ uses the circularity of the trace and $(**)$ uses \autoref{lemma:traceclass.tensor}.

  Since the vectors $\phi_1\otimes \phi_2$ span $\calK_1\otimes\calK_2$, and so do the vectors $\psi_1\otimes\psi_1$,
  we have that
  $\forall \phi_1\phi_2\psi_1\psi_2.\ \adj{(\phi_1\otimes\phi_2)}u(\psi_1\otimes\psi_2)=
  \adj{(\phi_1\otimes\phi_2)}u'(\psi_1\otimes\psi_2)$
  implies $u=u'$. Thus $\theta'(\rho_1\otimes \rho_2) = f_1(\rho_1)\otimes f_2(\rho_2)$.

  Thus with $(f_1\otimes f_2) := \theta'$, this shows
  \eqref{lemma:channel.tensor:cp}.

  \bigskip
  
  Next we show \eqref{lemma:channel.tensor:distrib}.
  Since $g_1\otimes g_2$ and $f_1\otimes f_2$ are completely positive, so is $(g_1\otimes g_2)\circ(f_1\otimes f_2)$.
  For trace-class $\rho_1,\rho_2$, we have $(g_1\otimes g_2)\circ(f_1\otimes f_2)(\rho_1\otimes\rho_2) = g_1\circ f_1(\rho_1) \otimes g_2\circ f_2(\rho_2)$.
  By \eqref{lemma:channel.tensor:unique}, ${(g_1\circ f_1) \otimes (g_2\circ f_2)}$ is the only completely positive map with that property.
  Thus  $ (g_1\circ f_1) \otimes (g_2\circ f_2) =  (g_1\circ f_1) \otimes (g_2\circ f_2)$.
  This shows \eqref{lemma:channel.tensor:distrib}.

  \bigskip

  Next we show \eqref{lemma:channel.tensor:tpres}.
  Assume that $f_1,f_2$ are trace-preserving. By \lemmaref{item:SH.adjoint.trp}, $g_1,g_2$ as defined above in the proof of \eqref{lemma:channel.tensor:cp} are unital.
  Then $\theta$ as defined above is unital: $\theta(1) = \theta(1\otimes 1)
  = g_1(1)\otimes g_2(1)= 1\otimes 1=1$.
  Thus by \lemmaref{item:SH.adjoint.trp}, $f_1\otimes f_2=\theta'$ is trace-preserving.
  This shows \eqref{lemma:channel.tensor:tpres}.

  \bigskip
  
  Next we show \eqref{lemma:channel.tensor:tred}.
  Assume that $f_1,f_2$ are trace-reducing.
  Then $g_1,g_2$ as defined above in the proof of \eqref{lemma:channel.tensor:cp} are subunital by \lemmaref{item:SH.adjoint.trred}.
  %
  Then $\theta$ as defined above is subunital:
  \begin{multline*}
    1-\theta(1) = 1-\theta(1\otimes 1) 
    = 1\otimes 1 - g_1(1)\otimes g_2(1) \\
    = \pB\paren{(1-g_1(1))\otimes(1-g_1(1))} + \pB\paren{(1-g_1(1))\otimes g_2(1)} + \pB\paren{g_1(1) \otimes(1-g_1(1))}
    \starrel\geq 0
  \end{multline*}
  Here $(*)$ follows by \autoref{lemma:tensor.abs} and the fact that $g_1,g_2$ are subunital and positive, so $1-g_1(1)$, $1-g_2(1)$, $g_1(1)$, $g_2(1)$ are all positive.
  Then $(f_1\otimes f_2)=\theta'$ is trace-reducing by \lemmaref{item:SH.adjoint.trred}.
  %
  This shows \eqref{lemma:channel.tensor:tred}.

  \medskip

  We show \eqref{lemma:channel.tensor:id}.
  By \eqref{lemma:channel.tensor:unique} and~\ref{lemma:channel.tensor:tpres}, $\id_{\tracecl(\calH_1)}\otimes\id_{\tracecl(\calH_2)}$ is the unique completely positive trace-preserving map that maps $\rho\otimes\sigma$ to $\rho\otimes\sigma$.
  Since $\id_{\tracecl(\calH_1\otimes\calH_2)}$ is completely positive trace-preserving and maps $\rho\otimes\sigma$ to $\rho\otimes\sigma$, we have $\id_{\tracecl(\calH_1)}\otimes\id_{\tracecl(\calH_2)}=\id_{\tracecl(\calH_1\otimes\calH_2)}$.
\end{proof}

\begin{lemma}\label{lemma:reg.conv.iff}
  Assume that $F:\mathbf A\to\mathbf B$ is a quantum reference, and $x_i$ is a net of bounded operators, and $x$ is a bounded operator.
  Then $x_i\to x$ iff $F(x_i)\to F(x)$ (both limits are in the weak*-topology).
\end{lemma}

\begin{proof}
  The direction ``if $x_i\to x$ then $F(x_i)\to F(x)$'' follows since $F$ is weak*-continuous.

  To show the direction ``if  $F(x_i)\to F(x)$ then $x_i\to x$'', we need to show $x_i\to x$ in the weak*-topology.
  By \autoref{lemma:normal.tensor1}, $F(a)=U(a\otimes 1_\mathbf{C})\adj U$ for some $\mathbf C$ and some unitary $U$.
  Let $\delta\in\tracecl(\calH_A\otimes\calH_C)$ be an arbitrary operator of trace 1 (e.g., $\delta:=\psi\adj\psi$ with $\norm\psi=1$).
  For any $t$, let $u_t:=U(t\otimes\delta)\adj U$.
  Then
  \begin{align*}
    \tr tx &= \tr t x \cdot \tr\delta 1_\mathbf{C}
    \starrel= \tr tx \otimes \delta1_\mathbf{C}
    =
             \tr (t\otimes\delta)(x\otimes1_\mathbf{C})
    \\&
    = \tr \adj U u_t U  (x\otimes1_\mathbf{C})
    \starstarrel= \tr  u_t U  (x\otimes1_\mathbf{C}) \adj U
    = \tr  u_t F(x).
  \end{align*}
  Here $(*)$ is by \autoref{lemma:traceclass.tensor}.
  And $(**)$ is by the circularity of the trace (\lemmaref{lemma:circ.trace:circ}).
  
  Analogously, $\tr tx_i=\tr u_tF(x_i)$. And $u_t$ is trace-class \cite[Theorem 18.11\,(a)]{conway00operator}.
  Since $F(x_i)\to F(x)$ by assumption, for every $t$, $\tr u_tF(x_i)\to \tr u_tF(x)$ by definition of the weak*-topology \cite[beginning of Section 20]{conway00operator}.
  Thus $\tr tx_i \to \tr tx$.
  Since this holds for all $t$, $x_i\to x$.
\end{proof}

\begin{lemma}\label{lemma:optics.uni.equiv}
  Let $U:\calH_A\otimes\calH_C\to\calH_B$, $V:\calH_B\to\calH_A\otimes\calH_C$,
  $\Hat U:\calH_A\otimes\calH_{\Hat C}\to\calH_B$, $\Hat V:\calH_B\to\calH_A\otimes\calH_{\Hat C}$ be unitaries.
  Assume that the dimension of $\calH_C$ is at least as large as that of $\calH_{\Hat C}$.
  (Otherwise apply the lemma the other way around.)
  
  Let $F(a):=U(a\otimes 1)V$ and $\Hat F(a):=\Hat U(a\otimes 1)\Hat V$. Assume $F=\Hat F$.
  Then there is a unitary $M$ such that $V=(1\otimes M)\Hat V$ and $\Hat U=U(1\otimes M)$.
  (That is, $(V,U)\cong(\Hat V,\Hat U)$ for the relation $\cong$ from \autopageref{page:def:opticequiv}.)
\end{lemma}

\begin{proof}
  Since $F=\Hat F$, we have for all bounded operators $a\in\mathbf A$ that $U(a\otimes 1)V=\Hat U(a\otimes 1)\Hat V$.
  With $a=1$, $UV=\Hat U\Hat V$.
  Since $\Hat V,U$ are unitaries, this implies $V\adj{\Hat V}=\adj U\Hat U =: W$. Then $W : \calH_A\otimes\calH_{\Hat C}\to \calH_A\otimes\calH_C$ is unitary as well.
  Let $Z$ be an arbitrary isometry $\calH_{\Hat C}\to\calH_C$.
  (It exists because the dimension of $\calH_C$ is at least as large as that of $\calH_{\Hat C}$.)
  So
  \begin{multline}
    W(1\otimes\adj Z)(a\otimes 1)
    =
    W(a\otimes 1)(1\otimes\adj Z)
    =
    \adj U\Hat U(a\otimes 1)\Hat V\adj{\Hat V}(1\otimes\adj Z)
    \\=
    \adj UU(a\otimes 1)V\adj{\Hat V}(1\otimes\adj Z)
    =
    (a\otimes 1)W(1\otimes\adj Z).
  \end{multline}
  So $W(1\otimes \adj Z)$ commutes with all $a\otimes 1$.
  By  \cite[Corollary IV.1.5]{takesaki}, $a\mapsto a\otimes 1$ is an isomorphism and thus surjective onto $\mathbf A\otimes\idmult$.
  (Here $\idmult$ is the set of the multiples of the identity.)
  Thus $W(1\otimes \adj Z)$ commutes with all of $\mathbf A\otimes\idmult$, i.e., $W(1\otimes\adj Z)\in \comm{(\mathbf A\otimes\idmult)}$, the commutant of  $\mathbf A\otimes\idmult$.
  By \cite[Proposition IV.1.6\,(i)]{takesaki}, $\comm{(\mathbf A\otimes\idmult)}=\idmult\otimes\mathbf C$.
  So $W(1\otimes\adj Z)\in\idmult\otimes\mathbf C$.
  By  \cite[Corollary IV.1.5]{takesaki}, $c\mapsto 1\otimes c$ is an isomorphism and thus surjective onto $\idmult\otimes\mathbf C$.
  Thus all elements of $\idmult\otimes\mathbf C$ are of the form $1\otimes c$.
  Hence we can find some $N\in\mathbf C$ such that $W(1\otimes \adj Z)=1\otimes N$.
  With $M:=NZ$, we have $W=W(1\otimes\adj Z)(1\otimes Z) = (1\otimes N)(1\otimes Z)=1\otimes M$.

  We have $1\otimes \adj MM=\adj WW=1=1\otimes 1$.
  Since $c\mapsto 1\otimes c$ is an isomorphism and thus injective, this implies $\adj MM=1$.
  Analogously $M\adj M=1$.
  So $M$ is unitary.

  Finally, we have
  \begin{equation}
    V=V\adj{\Hat V}\Hat V = W\Hat V = (1\otimes M)\Hat V
    \qquad\text{and}\qquad
    \Hat U= U\adj U\Hat U = UW = U(1\otimes M)
  \end{equation}
  as desired.
\end{proof}


}



\fullonly{
  \newcommand\INCLUDEONCEdfaityerlkhohiutdfutktrl{}

\section{Proofs of axioms for infinite-dimensional quantum references}
\label{sec:proofs.infdim}

We show that the axioms from
\autoref{fig:axioms} hold for the quantum reference category $\Lquantum$ as defined in \autoref{sec:infinite}. 

\paragraph{Axiom~\ref{ax:monoids}.} Immediate from the definition of the objects in $\Lquantum$ (i.e., the spaces of bounded operators).

\paragraph{Axiom~\ref{ax:preregs}.} Immediate from the definition of pre-references (i.e., weak*-continuous bounded linear operators).

\paragraph{Axiom~\ref{ax:cdot-a}.} Immediate by \autoref{lemma:normal} and the definition of pre-references.

\paragraph{Axioms~\ref{ax:tensor}.}
The fact that a bounded operator $a\otimes b$ exists is shown in \cite[discussion preceding Definition~IV.1.3]{takesaki}.

\paragraph{Axiom~\ref{ax:tensorext}.}
Fix pre-references $F,G$ such that $F(a\otimes b)=G(a\otimes b)$ for all $a\in\mathbf A$, $b\in\mathbf B$.
The algebraic tensor product $\mathbf A\atensor\mathbf B$ (see \autoref{lemma:alg.dense}) is spanned by the operators $a\otimes b$ for $a\in\mathbf A$, $b\in\mathbf B$. 
Since $F,G$ are linear, this implies that $F=G$ on $\mathbf A\atensor\mathbf B$. By \autoref{lemma:alg.dense},
$\mathbf A\atensor\mathbf B$ is weak*-dense in $\mathbf A\tensor\mathbf B$.
Since $F,G$ are weak*-continuous, and $F=G$ on a weak*-dense subset, we have $F=G$.

\paragraph{Axiom~\ref{ax:tensor.mult}.} Shown in \cite[(5) after Definition IV.1.2]{takesaki}.

\paragraph{Axiom~\ref{ax:reg.prereg}.}
References are weak*-continuous bounded linear maps by \autoref{lemma:reg.props}, and thus pre-references.

\paragraph{Axiom~\ref{ax:reg.morphisms}.}
Immediate from the definition of references (weak*-continuous unital *-homomorphisms).

\paragraph{Axiom~\ref{ax:reg.monhom}.}
A direct consequence of the fact that that references are unital and *-homomorphisms.

\paragraph{Axiom~\ref{ax:tensor-1}.}
Let $F(x):=x\otimes 1_{\mathbf B}$ for $x\in\mathbf A$.
By \cite[Corollary IV.1.5]{takesaki}, $F$ is an isomorphism
from $\mathbf A$ onto the von Neumann algebra $\mathbf A\otimes\idmult$ where $\symbolindexmark\idmult\idmult$ is the set of all multiples of the identity. $F$ also preserves adjoints and is multiplicative, so $F$ is a $*$-isomorphism.
By \cite[Proposition 46.6]{conway00operator}, any $*$-isomorphism is normal, i.e., weak*-continuous \cite[Theorem 46.4\,(c)]{conway00operator}. Thus $F$ is weak*-continuous as a map
$\mathbf A\to\mathbf A\otimes \idmult$, and thus also as a map $\mathbf A\to\mathbf A\otimes \mathbf B \supseteq \mathbf A\otimes \idmult$. $F(1)=1\otimes 1=1$, hence $F$ is also unital. Thus $F$ is a reference.

That $x\mapsto 1\otimes x$ is a reference is shown analogously.

\paragraph{Axiom~\ref{ax:pairs}.}
Fix references $F:\mathbf A\to\mathbf C$, $G:\mathbf B\to\mathbf C$ with commuting ranges.
By \autoref{lemma:normal.tensor1}, $F(a)=U(a\otimes 1_{\mathbf D})\adj U$ and $G(b)=V(b\otimes 1_{\mathbf E})\adj V$ for some Hilbert spaces $\calH_D,\calH_E$ and unitaries $U:\calH_A\otimes\calH_D\to\calH_C$, $V:\calH_B\otimes\calH_E\to\calH_C$.
By \cite[Corollary IV.1.5]{takesaki}, the range of $a\mapsto a\otimes 1_\mathbf{D}$ is $\mathbf A\otimes\idmult_{\mathbf{D}}$, and the range of $b\mapsto b\otimes\idmult_{\mathbf{E}}$ is $\mathbf B\otimes\idmult_\mathbf{E}$.
($\idmult_{\mathbf D},\idmult_{\mathbf E}$ is the set of multiples of the identity in $\mathbf D,\mathbf E$.)
Thus the ranges of $F,G$ are $F(\mathbf A)=U(\mathbf A\otimes \idmult_{\mathbf D})\adj U$ and $G(\mathbf B)=V(\mathbf B\otimes\idmult_{\mathbf E})\adj V$.
$F(\mathbf A)$, $G(\mathbf B)$ commute by assumption.
Hence  $U(\mathbf A\otimes \idmult_{\mathbf D})\adj U$ and $G(\mathbf B)$ commute.
Thus $\mathbf A\otimes \idmult_{\mathbf D}$ and $\adj UG(\mathbf B)U$ commute.
Thus $\adj UG(\mathbf B)U \subseteq \comm{\paren{\mathbf A\otimes \idmult}}$ where $\symbolindexmarkhighlight{\comm X}$ denotes the commutant of~$X$.
By \cite[Theorem IV.5.9]{takesaki}, $\comm{\paren{\mathbf A\otimes \idmult_{\mathbf D}}} = \comm{\mathbf A}\otimes\comm{(\idmult_{\mathbf D})} = \idmult_{\mathbf A}\otimes\mathbf D$.
So the range of $\Tilde G(b):=\adj UG(b)U$ is $\Tilde G(\mathbf B)=\adj UG(\mathbf B)U \subseteq \idmult_{\mathbf A}\otimes\mathbf D$, and $\Tilde G$ is a weak*-continuous unital *-homomorphism (because $G$ and $x\mapsto \adj UxU$ are, using \autoref{lemma:normal}).
Then let $\Hat G:\mathbf B\to\idmult_{\mathbf A} \otimes\mathbf D$ be result of restricting the codomain $\Tilde G$ to $\idmult_{\mathbf A} \otimes\mathbf D$.
Then $\Hat G$ is still a weak*-continuous unital *-homomorphism.
Let $\iota(d):=1_{\mathbf A}\otimes d$ for $d\in\mathbf D$.
By \cite[Corollary IV.1.5]{takesaki}, $\iota:\mathbf D\to\idmult_{\mathbf A}\otimes\mathbf D$ is an isomorphism of von Neumann algebras.
It also preserves adjoints, so it is a *-isomorphisms.
Thus $\iota^{-1}:\idmult_{\mathbf A}\otimes\mathbf D\to\mathbf D$ is a $*$-isomorphism, too. And thus by \cite[Proposition 46.6]{conway00operator}, $\iota^{-1}$ is normal and thus a weak*-continuous unital *-homomorphism.
Then $G_1 : \mathbf B\to\mathbf D := \iota^{-1} \circ \Hat G$ is a weak*-continuous unital *-homomorphism, hence a reference.
Since $\id_\mathbf A, G_1$ are references, by \autoref{lemma:reg.tensor}, there is a reference $T:\mathbf A\otimes\mathbf B\to\mathbf A\otimes\mathbf D$ such that $T(a\tensor b)=a\otimes G_1(b)$ for all $a\in\mathbf A, b\in\mathbf B$.
Let $\spair FG(x) := UT(x)\adj U$. Then $\spair FG$ is a reference because it is the composition of the references $T$ and $x\mapsto Ux\adj U$.

We have
\[
  \spair FG(a\otimes 1_\mathbf{B}) = U\pb\paren{a\otimes G_1(1_\mathbf{B})}\adj U
  \starrel= U(a\otimes 1_\mathbf{D})\adj U = F(a)
\]
where $(*)$ uses that $G_1$ is a reference and thus unital.
Moreover
\[
  \spair FG(1_\mathbf{A}\otimes b) =
  U\pb\paren{1_\mathbf{A}\otimes G_1(b)}\adj U
  =
  U\pb\paren{\iota(G_1(b))}\adj U
  =
  U\pb\paren{\Hat G(b)}\adj U
  =
  U\pb\paren{\adj UG(b)U}\adj U
  =
  G(b).
\]
Since $\spair FG$ is a reference,
\[
  \spair FG(a\otimes b)
  = \spair FG\pb\paren{(a\otimes 1)\cdot (1\otimes b)}
  = \spair FG(a\otimes 1)\cdot \spair FG(1\otimes b)
  = F(a)G(b).
\]
Thus there exists a reference $\spair FG$ with the property required by Axiom~\ref{ax:pairs}.


    
    
  


}

\fullonly{
  \newcommand\INCLUDEONCEhdsygiuthfvxjmekjf{}

\section{Proofs for complements in the quantum setting}
\label{app:proofs:complements}

The lemmas below assume the quantum reference category (finite- or infinite-dimensional).
reference/pre-reference means quantum reference/pre-reference.

\begin{lemma}\label{lemma:complement.range}
  If $F:\mathbf A\to\mathbf C$ and $G:\mathbf B\to\mathbf C$ are complements,
  then $G(\mathbf B)$ is the commutant of $F(\mathbf A)$.
  (The set of all operators that commute with all of $F(\mathbf A)$.)
\end{lemma}

\begin{proof}
  Throughout this proof, let $\symbolindexmark\comm{\comm X}$ denote the commutant of $X$.
  For any iso-reference $H:\mathbf D\to\mathbf E$,
  \begin{align}
    x\in \comm{H(\mathbf D)}
    &\iff
    \forall e\in H(\mathbf D).\ xe=ex
    \starrel\iff
    \forall e\in H(\mathbf D).\ H^{-1}(x)H^{-1}(e)=H^{-1}(e)H^{-1}(x)
      \notag\\&
    \iff
    \forall d\in \mathbf D.\ H^{-1}(x)d=dH^{-1}(x)
    \iff
      H^{-1}(x) \in \mathbf D'
    \iff
    x \in H (\mathbf D')
    \label{eq:commutant.exchange}
  \end{align}
  Here we repeatedly use that $H$ is bijective as an iso-reference, and in $(*)$ we also use the multiplicativity of $H^{-1}$.

  Let $\symbolindexmark\idmult{\idmult_\mathbf{B}}$ denote the multiples of the identity, i.e., $\{\alpha1_\mathbf{B}:\alpha\in\setC\}$.
  The map $\pi:x\mapsto x\otimes 1$ is an isomorphism of von Neumann algebras from $\mathbf{A}$ to $\mathbf{A}\otimes \idmult$ \cite[Corollary IV.1.5]{takesaki}.
  Thus $ \mathbf{A}\otimes \idmult_\mathbf{B} = \pi(\mathbf A) = \{a\otimes 1_\mathbf{B}:a\in \mathbf A\}$.
  Analogously, $  \idmult_\mathbf{A}\otimes\mathbf{B} = \{ 1_\mathbf{A}\otimes b:b\in \mathbf B\}$.

  Then
  \begin{align*}
    F(\mathbf B)'
    &\starrel = \pB\paren{\spair FG\pb\paren{\braces{a\otimes 1_\mathbf{B}: a\in\mathbf A}}}'
    \eqrefrel{eq:commutant.exchange}=
    \spair FG\pb\paren{\braces{a\otimes 1_\mathbf{B}: a\in\mathbf A}'}
    =
    \spair FG\pb\paren{\paren{\mathbf A\otimes\idmult_\mathbf{B}}'}
      \notag\\&
    \starstarrel=
    \spair FG{\paren{\idmult_\mathbf{A}\otimes \idmult_\mathbf B'}}
    =
    \spair FG{\paren{\idmult_\mathbf{A}\otimes \mathbf B}}
    \eqrefrel{eq:commutant.exchange}=
    \spair FG\pb\paren{\braces{1_\mathbf{A}\otimes b:b\in \mathbf B}}
    \starrel=
    G(\mathbf B).
  \end{align*}
  Here $(*)$ follows from Axioms~\ref{ax:pairs} and~\ref{ax:reg.monhom}.
  And $(**)$ is by \cite[Proposition IV.1.6\,(ii)]{takesaki} (with the left/right side of the tensor product exchanged).
\end{proof}

\begin{lemma}\label{lemma:same.range.equiv}
  If $F:\mathbf A\to\mathbf C$ and $G:\mathbf B\to\mathbf C$ are references and $F(\mathbf A)=G(\mathbf B)$, then $F,G$ are equivalent.
\end{lemma}

\begin{proof}
  Since $F,G$ are quantum references, they are injective (\autoref{lemma:reg.injective}).
  Let $I := F^{-1}\circ G$.
  This is well-defined because $F(\mathbf A)=G(\mathbf B)$, so $F^{-1}$ is defined on the range of $G$.
  It is elementary to check that $I$ is linear, multiplicative, unital and preserves adjoints since $F$ and $G$ are and do.
  We show that $I$ is weak*-continuous:
  Fix a net $x_i$ of bounded operators and a bounded operator $x$ such that $x_i\to x$ (w.r.t.~the weak*-topology).
  Then $G(x_i)\to G(x)$ by \autoref{lemma:reg.conv.iff}.
  So $F(F^{-1}(G(x_i))) = G(x_i) \to G(x) = F(F^{-1}(G(x)))$.
  Since $F$ is a reference, by \autoref{lemma:reg.conv.iff}, $F^{-1}(G(x_i)) \to F^{-1}(G(x))$.
  Thus $I(x_i)\to I(x)$.
  Since this holds for any net $x_i$, $I$ is weak*-continuous.
  Thus $I$ is a reference.

  Analogously, we have that $J:=G^{-1}\circ F$ is a reference. Then $I\circ J=F^{-1}\circ G\circ G^{-1}\circ F=\id$ and analogously $J\circ I=\id$. Thus $I$ is an iso-reference.

  Finally, $F\circ I = F \circ F^{-1}\circ G = G$. Thus $F$ and $G$ are equivalent.
\end{proof}

Note that only the lemmas above are used in the proof of \autoref{theo:compl} in \autoref{sec:qregs.complements}.
Thus in the proofs of the following lemmas we can already assume that \autoref{theo:compl} holds, i.e., that $\Lquantum$ has complements.

\begin{lemma}\label{lemma:complement.range.converse}
  The converse of \autoref{lemma:complement.range} holds.
  That is, for references $F:\mathbf A\to\mathbf C$ and $G:\mathbf B\to\mathbf C$, if $G(\mathbf B)$ is the commutant of $F(\mathbf A)$, then $F$ and $G$ are complements.
\end{lemma}

\begin{proof}
  By Laws~\ref{law:compl.is.compl}, \ref{law:compl.sym}, $\compl F, F$ are complements.
  Let $\mathbf D$ denote the domain of $\compl F$.
  Then by \autoref{lemma:complement.range}, $\compl F(\mathbf D)$ is the commutant of $F(\mathbf A)$.
  Since $G(\mathbf B)$ is also the commutant of $F(\mathbf A)$ by assumption, $\compl F(\mathbf D)=G(\mathbf B)$.
  By \autoref{lemma:same.range.equiv}, this implies that $\compl F$ and $G$ are equivalent.
  Since $F,\compl F$ are complements, with Law~\ref{law:compl.equiv} this implies that $F,G$ are complements.
\end{proof}

\begin{lemma}\label{lemma:unit.iff}
  $U:\mathbf A\to\mathbf B$ is a unit reference iff $U(\mathbf A)=\idmult_{\mathbf B}$ (the set of scalar multiples of $1_\mathbf{B}$).
\end{lemma}

\begin{proof}
  By definition, $U:\mathbf A\to\mathbf B$ is a unit reference iff $U$ and $\id:\mathbf B\to\mathbf B$ are complements.
  By Lemmas~\ref{lemma:complement.range} and~\ref{lemma:complement.range.converse}, this holds iff $U(\mathbf A)$ is the commutant of $\id(\mathbf B)=\mathbf B$.
  Since the commutant of $\mathbf B$ is $\idmult_\mathbf{B}$, this holds iff $U(\mathbf A)=\idmult_\mathbf{B}$.
\end{proof}

}


\fullonly{
  \section{Proofs for lifting of quantum objects}

\subsection{Elementary objects, \autoref{lemma:elementary}}
\label{app:proof:lemma:elementary}

\begin{proof}
  We show \eqref{lemma:preserve.isometry}: If $A$ is an isometry, then $\adj AA=1$ by definition. Since $F$ is a reference, this implies that $\adj{F(A)}F(A)=F(\adj AA)=F(1)=1$. Thus $F(A)$ is an isometry.

  We show \eqref{lemma:preserve.unitary}: If $A$ is unitary, then $\adj AA=1$ and $A\adj A=1$ by definition. Analogous to \eqref{lemma:preserve.isometry}, this implies that $\adj{F(A)}F(A)=1$ and $F(A)\adj{F(A)}=1$. Thus $F(A)$ is unitary.

  We show \eqref{lemma:preserve.projector}: If $A$ is a projector, then $\adj A=A$ and $A^2=A$.
  Thus $\adj{F(A)}=F(\adj A)=F(A)$ and $F(A)^2=F(A^2)=F(A)$. Thus $F(A)$ is a projector.

  We show \eqref{lemma:preserve.norm}: By \autoref{lemma:reg.isometric}.

  We show \eqref{lemma:preserve.positive}: If $A$ is positive, there exists a $B$ with $A=\adj BB$. Since $F$ is a register, $F(A)=F(\adj BB)=F(\adj B)F(B)=\adj{F(B)}F(B)\geq 0$. So $F(A)$ is positive.
\end{proof}

\subsection{Subspaces, \autoref{lemma:lift.sub}}
\label{app:proof:lemma:lift.sub}

\begin{proof}
  We show \eqref{lemma:lift.sub:subspace}.
  By \cite[comment after Definition II.3.1]{conway13functional}, the projector $P_S$ onto $S$ exists.
  By \lemmaref{lemma:preserve.projector}, $F(P_S)$ is a projector.
  By \cite[Proposition II.3.2\,(b)]{conway13functional} the image of a project is a closed subspace.
  Thus $\regsubspace FS$, the image of $F(P_S)$, is a closed subspace.

  This shows \eqref{lemma:lift.sub:subspace}.

  \medskip

  We show \eqref{lemma:lift.sub:ortho}.
  For any closed subspace $X$, if $P_X$ denotes the projector onto $X$, then $P_{\ortho X}=1-P_X$.
  (Proof: $1-P_X$ is a projector. Thus $\im(1-P_X) \starrel= \ker P_X \starstarrel= \ortho{\paren{\im P_X}} = \ortho X$
  where $(*)$ follows by \cite[Proposition II.3.2\,(b)]{conway13functional} and $(**)$ by \cite[Definition II.3.1]{conway13functional}.
  Thus $1-P_X$ is the unique projector $P_{\ortho X}$ onto $\ortho X$.)
  Thus $F(\ortho S)=\im F(P_{\ortho S}) = \im F(1-P_S) = \im 1-F(P_S) = \ortho{(\im F(P_S))} = \ortho{F(P_S)}$.

  This shows \eqref{lemma:lift.sub:ortho}.

  \medskip

  We show \eqref{lemma:lift.sub:zero}.
  The projector $P_{\braces0}$ onto $\{0\}$ is $0$. Thus $F\pb\paren{\braces0}=\im F(P_{\braces0}) = \im F(0)=\im 0=\braces0$.
  
  This shows \eqref{lemma:lift.sub:zero}.

  \medskip

  We show \eqref{lemma:lift.sub:all}.
  The projector $P_{\calH_A}$ onto $\calH_A$ is $1$. Thus $F(\calH_A)=\im F(P_{\calH_A}) = \im F(1)=\im 1=\calH_B$.

  This shows \eqref{lemma:lift.sub:all}.

  \medskip

  We show \eqref{lemma:lift.sub:mono}.

  We have $F(S)\subseteq F(T)$ iff $\im F(P_S)\subseteq\im F(P_T)$ iff $F(P_S)F(P_T)=F(P_S)$ \cite[Exercise~II.3.6]{conway13functional} iff $F(P_SP_T)=F(P_S)$ iff $P_SP_T=P_S$ (\autoref{lemma:reg.injective}) iff $S\subseteq T$ \cite[Exercise~II.3.6]{conway13functional}.

  This shows \eqref{lemma:lift.sub:mono}.

  \medskip

  We show \eqref{lemma:lift.sub:inter}.
  Since $S\cap T$ is the largest closed subspace contained in both $S,T$, \eqref{lemma:lift.sub:inter} can be equivalently stated as: $F$ preserves greatest lower bounds (in the lattice of subspaces).
  And this follows by \eqref{lemma:lift.sub:mono}.

  This shows \eqref{lemma:lift.sub:inter}.

  \medskip

  We show \eqref{lemma:lift.sub:plus}.
  The fact that $S+T$ is the smallest closed subspace containing $S,T$ follows because it is a closed subspace (since $+$ by our definition includes taking the closure of a subspace), and it is the smallest closed subspace because it is the closure of the smallest subspace containing $S,T$.

  Then $F(S+T)=F(S)+F(T)$ can be equivalently stated as: $F$ preserves least upper bounds (in the lattice of subspaces).
  And this follows by \eqref{lemma:lift.sub:mono}.
  
  This shows \eqref{lemma:lift.sub:plus}.

  \medskip
  
  We show \eqref{lemma:lift.sub:apply-op}.
  By \autoref{lemma:normal.tensor1}, there exists a Hilbert space $\calH_C$ and a unitary $U:\calH_A\otimes\calH_C\to\calH_B$ such that $F(a)=U(a\otimes 1)\adj U$ for all $a$.

  Let $P_S$ denote the projector onto $S$, and $P_{AS}$ the projector onto $AS$. (Exists by \autoref{lemma:ex.proj}.)

  When writing $XYS$ for operators $X,Y$ and a subspace $A$, we mean $X(YS)$, i.e., the closure of the image of $X$ of the closure of the image of $Y$ of $S$
  (and not the closure of the image of $XY$ of $S$).
  
  Then
  \begin{equation*}
    F(A)F(S) \starrel= F(A)F(P_S)\calH_B
    = U(A\otimes 1)\adj UU(P_S\otimes 1)\adj U\calH_B
    \starstarrel= U(A\otimes 1)(P_S\otimes 1)(\calH_A\otimes\calH_C)
  \end{equation*}
  Here $(*)$ is by definition of $F(S)$ (application of a reference to a closed subspace).
  And $(**)$ since $U$ is unitary, and thus $\adj UU=1$ and $\adj U\calH_B=\calH_A\otimes\calH_C$.
  And further
  \begin{equation*}
    \dots \starrel= U(AP_S\calH_A \otimes \calH_C)
    = U(AS \otimes \calH_C)
    = U(P_{AS}\calH_A \otimes \calH_C)
    \starrel= U(P_{AS}\otimes 1)(\calH_A \otimes \calH_C)
  \end{equation*}
  Here $(*)$ follows from \autoref{lemma:tensor.subspace.proj}. ($S\otimes T$ for subspaces $S,T$ is defined as the closed span of $\{s\otimes t:s\in S,t\in T\}$).
  And further
  \begin{equation*}
    \dots
    \starrel= U(P_{AS}\otimes 1)\adj U\calH_B
    = F(P_{AS})\calH_B
    \starstarrel= F(AS).
  \end{equation*}
  Here $(*)$ since $U$ is unitary, and thus $\adj U\calH_B=\calH_A\otimes\calH_C$.
  And $(**)$ is by definition of $F(AS)$ (application of a reference to a closed subspace).

  Altogether $F(A)F(S)=F(AS)$, as claimed.
  This shows \eqref{lemma:lift.sub:apply-op}.
  
  \medskip
  
  We show \eqref{lemma:lift.sub:inter.disjoint}.
  Note that for commuting projectors, $\im PQ\subseteq \im P$, $\im PQ=\im QP\subseteq \im Q$.
  And for $h\in\im P\cap\im Q$, $PQh=Ph=h$, hence $h\in\im PQ$.
  It follows that
  \begin{equation}\label{eq:intersect.comm.proj}
    \im P\cap\im Q=\im PQ\qquad\text{for commuting projectors }P,Q
  \end{equation}
  
  We have
  \begin{equation*}
    F(S)\cap G(V)
    =
    \im F(P_S)\cap\im G(P_V)
    \starrel=
    \im F(P_S)G(P_V)
    \starstarrel=
    F(P_S) \im G(P_V)
    =
    F(P_S)G(V).
  \end{equation*}
  Here $(*)$ follows from the \eqref{eq:intersect.comm.proj} and the fact that $F(P_S),G(P_V)$ commute because $F,G$ are disjoint.
  And $(**)$ is simply a special case of the fact that $\im f\circ g=f(\im g)$ for arbitrary functions $f,g$.
\end{proof}

\subsection{Mixed states, \autoref{lemma:mixed.states}}
\label{app:proof:lemma:mixed.states}

\begin{proof}
  Throughout this proof, we abbreviate $\rho^* := \rho_1\otimes\dots\otimes\rho_n$ and $F^* := \pairs{F_1}{F_n}$. Note that $F^*$ is an iso-reference since $F_1,\dots,F_n$ is a partition by assumption.
  And ${F_1(\rho_1)\mtensor\dots\mtensor F_n(\rho_n)} = F^*(\rho^*)$ by definition.
  Since $F^*$ is an iso-reference, by \autoref{lemma:isoreg-decomp}, $F^*(\rho^*)=U\rho^*\adj U$ for some unitary $U$ that we fix throughout the proof as well.

  \medskip
  
  We show \eqref{lemma:mixed.states:trace}.
  By \autoref{lemma:traceclass.tensor}, $\rho^*$ is a trace-class operator and $\tr\rho^* = \tr\rho_1\cdots\tr\rho_n$.
  By \cite[Theorem 18.11\,(a)]{conway00operator}, this implies that $F^*(\rho^*)=U\rho^*\adj U$ is trace-class, and $\tr F^*(\rho^*)=\tr U\rho^*\adj U\starrel=\tr\rho^*=\tr\rho_1\cdots\tr\rho_n$.
  (Here $(*)$ is by \lemmaref{lemma:circ.trace:ubu}.)
  This shows \eqref{lemma:mixed.states:trace}.

  \medskip

  We show \eqref{lemma:mixed.states:density}.
  By \autoref{lemma:tensor.abs}, $\rho^*$ is positive. Then $F^*(\rho^*)=U\rho^*\adj U$ is positive.
  By  \eqref{lemma:mixed.states:trace}, $F^*(\rho^*)$ is trace-class and $\tr F^*(\rho^*)=\tr\rho_1\cdots\tr\rho_n$.
  Thus $\tr F^*(\rho^*)=1$ (or $\leq 1$) if $\tr\rho_i=1$ (or $\leq1$) for $i=1,\dots,n$.
  By definition of (sub)density operators, this shows \eqref{lemma:mixed.states:density}.
  
  \medskip

  We show \eqref{lemma:mixed.states:permute}.
  It is sufficient to consider the case $\pi$ is a transposition of neighboring values since all permutations are generated from these.
  That is, $\pi(k)=k+1$ and $\pi(k+1)=k$ for some fixed $k\in\braces{1,\dots,n-1}$ and $\pi=\id$ everywhere else.
  Let $F^\pi := \pairs{F_{\pi(1)}}{F_{\pi(n)}}$ and $\rho^\pi := \rho_{\pi(1)}\otimes\dots\otimes\rho_{\pi(n)}$
  That is, we need to show $F^*(\rho^*)=F^\pi(\rho^\pi)$.
  By Laws~\ref{law:pair.sigma}--\ref{law:pair.alpha'}, \ref{law:pair.tensor}, $F^\pi = \chain{F^*}{\Pi}$ for $\Pi := \id \rtensor \dots \rtensor \id \rtensor (\alpha \circ (\sigma\rtensor\id) \circ \alpha')$.
  (Recall that $\otimes$ is right-associative.)
  And by Laws~\ref{law:tensor.ab}, \ref{law:sigma.ab}, \ref{law:alpha.abc}, \ref{law:alpha'.abc}, $\rho^\pi = \Pi'(\rho^*)$ for  $\Pi' := \id \rtensor \dots \rtensor \id \rtensor (\alpha \circ (\sigma\rtensor\id) \circ \alpha')$.
  (We use different variable names for $\Pi$ and $\Pi'$ even though their definitions seem identical because they have different types.)
  Furthermore, since $\sigma$ is self-inverse, and $\alpha,\alpha'$ are inverses, $\Pi\circ\Pi'=\id$.
  Thus $F^\pi(\rho^\pi) = F^*(\Pi\circ\Pi'(\rho^*)) = F^*(\rho^*)$.
  This shows \eqref{lemma:mixed.states:permute}.

  \medskip

  We show \eqref{lemma:mixed.states:apply.UV}.
  We have
  \begin{align*}
    \hskip1cm & \hskip-1cm
       F_1(U\rho_1 V)\mtensor F_2(\rho_2)\mtensor\dots\mtensor F_n(\rho_n) \\
    &= F^*((U\rho_1V)\otimes\rho_2\otimes\dots\otimes\rho_n) \\
    &= F^*\pb\paren{\paren{U\otimes\id\otimes\dots\otimes\id}\rho^*\paren{V\otimes\id\otimes\dots\otimes\id}} \\
    &= F^*\paren{U\otimes\id\otimes\dots\otimes\id}F^*(\rho^*)F^*\paren{V\otimes\id\otimes\dots\otimes\id} \\
    &= F_1\paren{U}F^*(\rho^*)F_1\paren{V}.
  \end{align*}
  This shows \eqref{lemma:mixed.states:apply.UV}.

  \medskip
  
  We show \eqref{lemma:mixed.states:rank1}.

  Assume all $\rho_i$ have rank 1.
  An operator has rank 1 iff it is of the form $\psi\adj\eta$ for nonzero vectors $\psi,\eta$ \cite[Exercise II.4.8]{conway13functional}.
  Then $\rho_i=\psi_i\adj{\eta_i}$ for nonzero $\psi_i,\eta_i$
  and then $\rho^*=\psi\adj\eta$ where $\psi:=\psi_1\otimes\dots\otimes\psi_n$ and $\eta := \eta_1\otimes\dots\otimes\eta_n$.
  Note that $\psi,\eta$ are non-zero.
  Then $F^*(\rho^*) = U\rho^*\adj U = (U\psi)\adj{(U\eta)}$.
  Since $U$ is unitary, $U\psi,U\eta$ are non-zero.
  Thus $F^*(\rho^*)$ is rank 1.

  Now assume that at least one $\rho_i$ does not have rank 1, say $\rho_1$.
  And operator has rank 1 iff its image is one-dimensional \cite[Exercise II.4.8]{conway13functional}.
  If any of the $\rho_i$ is $0$, then $F^*(\rho^*)=0$ and thus does not have rank 1.
  Since $\rho_1$ does not have rank 1 and is non-zero, there exist $\eta_1,\bar\eta_1$ such that $\rho_1\eta_1,\rho_1\bar\eta_1$ are not collinear.
  And since all $\rho_i$ are non-zero, there exist $\eta_i$ such that $\rho_i\eta_i$ is non-zero.
  Let $\eta := \eta_1\otimes\dots\otimes\eta_n$ and $\bar\eta := \bar\eta_1\otimes\eta_2\otimes\dots\otimes\eta_n$.
  Then $\rho^*\eta,\rho^*\bar\eta$ are not collinear.
  We have $F^*(\rho^*)U\eta = U\rho^*\adj UU\eta = U\rho^*\eta$ and $F^*(\rho^*)U\bar\eta = U\rho^*\bar\eta$.
  Thus $U\rho^*\eta, U\rho^*\bar\eta$ are in the image of $F^*(\rho^*)$ and not collinear.
  Thus $F^*(\rho^*)$ does not have rank 1.

  This shows \eqref{lemma:mixed.states:rank1}.

  \medskip

  We show \eqref{lemma:mixed.states:bounded}.
  We only show bounded linearity in the first argument.
  The general case follows with~\eqref{lemma:mixed.states:permute}.
  Since the tensor-product of operators is bilinear, $\Phi$ is linear in each argument.
  Furthermore
  \begin{multline}
    \pb\trnorm{F^*(\rho_1\otimes\dots\otimes\rho_n)}
    = 
    \pb\trnorm{U(\rho_1\otimes\dots\otimes\rho_n)\adj U} \\
    \starrel= \trnorm {\rho_1\otimes\dots\otimes\rho_n} 
    \starstarrel= \trnorm {\rho_1} \cdots \trnorm{\rho_n}
    = \trnorm{\rho_1} \cdot \prod_{i=2}^n \trnorm{\rho_n}.
  \end{multline}
  Here $(*)$ is by \lemmaref{lemma:circ.trace:ubu}, and $(**)$ by \autoref{lemma:traceclass.tensor}.
  Thus $\Phi$ is bounded in its first argument (with respect to the trace-norm).
  This shows \eqref{lemma:mixed.states:bounded}.
  
  \medskip

  We show \eqref{lemma:mixed.states:pair}.
  Note that $\pair{F_1}{F_2}$, $F_3$, \dots, $F_n$ is a partition.
  \begin{align*}
    \hskip1cm & \hskip-1cm
      \pair{F_1}{F_2}(\rho_1\otimes\rho_2)\mtensor F_3(\rho_3)\mtensor\dots\mtensor F_n(\rho_n) \\
    &\eqrefrel{lemma:mixed.states:permute}=
    F_3(\rho_3)\mtensor\dots\mtensor F_n(\rho_n)\mtensor \pair{F_1}{F_2}(\rho_1\otimes\rho_2) \\
    &\starrel=
    F_3\dots F_n{F_1}{F_2}
    \pb\paren{\rho_3\otimes\dots\rho_n\otimes\rho_1\otimes\rho_2} \\
    &\starrel=
    F_3(\rho_1)\mtensor\dots\mtensor F_n(\rho_n)\mtensor F_1(\rho_1)\mtensor F_2(\rho_2) \\
    &\eqrefrel{lemma:mixed.states:permute}=
    F_1(\rho_1)\mtensor F_2(\rho_2)\mtensor\dots\mtensor F_n(\rho_n).
  \end{align*}
  Both occurrences of $(*)$ are by definition of $\mtensor$ (recall that $\otimes$ and $\pair\cdot\cdot$ are right-associative.
  This shows \eqref{lemma:mixed.states:pair}.
  
  \medskip

  We show \eqref{lemma:mixed.states:separating}.
  By \autoref{lemma:tensor.generate.tracecl}, $S:=\{\rho_1\otimes\dots\otimes\rho_n:\rho_i\in S_i\}$ generates $\tracecl(\calH)$ with $\calH:=\calH_{A_1}\otimes\dots\otimes\calH_{A_n}$.
  Let $\Span S$ denote the linear span of $S$. That $S$ generates $T(\calH)$ means that $\Span S$ is dense in $T(\calH)$.
  For $\rho\in T(\calH)$,
  \begin{align*}
    \pb\trnorm{F^*(\rho)}
    = 
    \pb\trnorm{U\rho\adj U}
    \starrel\leq \norm U\cdot \trnorm \rho \cdot \norm{\adj U}
    = \trnorm\rho.
  \end{align*}
  Here $(*)$ is by \lemmaref{lemma:circ.trace:ubu}.
  So $F^*$ is bounded linear (w.r.t.~the trace-norm). Hence $F^*$ is continuous.
  Hence $F^*(\Span S)$ is dense in $F^*(\tracecl(\calH))$.
  Since $F^*$ is an iso-reference it is surjective.
  Thus $F^*(\tracecl(\calH))=\tracecl(\calH_B)$.
  Since $F^*$ is linear, $F^*(\Span S)=\Span{F^*(S)}$.
  Thus $F^*(S)$ generates $T(\calH_B)$.
  This shows the first sentence of \eqref{lemma:mixed.states:separating}.
  
  The ``in particular'' part:
  Since $\calE,\calF$ are linear and equal on $F^*(S)$, they are equal on $\Span{F^*(S)}$.
  Since $\calE,\calF$ are quantum subchannels, they are continuous with respect to the trace-norm (\autoref{lemma:channel.bounded}).
  So since $\Span{F^*(S)}$ is dense in $T(\calH_B)$, we have that $\calE,\calF$ are equal on their whole domain $T(\calH_B)$.
  This shows \eqref{lemma:mixed.states:separating}.

  \medskip

  We show \eqref{lemma:mixed.states:nested}.
  Let $G^* := \pairs{G_1}{G_m}$. Then
  \begin{align*}
    \hskip1cm & \hskip-1cm
      F_1\pB\paren{G_1(\sigma_1)\mtensor\dots\mtensor G_m(\sigma_m)}\mtensor F_2(\rho_2)\mtensor\dots\mtensor F_n(\rho_n) \\
    &\eqrefrel{lemma:mixed.states:permute}=
      F_2(\rho_2)\mtensor\dots\mtensor F_n(\rho_n)
      \mtensor F_1\pB\paren{G_1(\sigma_1)\mtensor\dots\mtensor G_m(\sigma_m)} \\
    &\starrel=
      {F_2}\dots{F_n}{F_1}
      \pB\paren{\rho_2\otimes\dots\otimes\rho_n \otimes G^*\pb\paren{\sigma_1\otimes\dots\otimes \sigma_m} }
    \\
    &=
      {F_2}\dots{F_n}{F_1}
      \circ
      \pB\paren{\id\rtensor\dots\rtensor\id\rtensor G^*}
      \pB\paren{ \rho_2\otimes\dots\otimes\rho_n\otimes {\sigma_1\otimes\dots\otimes \sigma_m}}\\
    &\starstarrel=
      \pairFdFF{F_2}{F_n}{\chain{F_1}{G^*}}
      \pB\paren{ \rho_2\otimes\dots\otimes\rho_n\otimes {\sigma_1\otimes\dots\otimes \sigma_m}}\\
    &\tristarrel=
      \pairFdFFdF{F_2}{F_n}{\chain{F_1}{G_1}}{\chain{F_1}{G_m}}
      \pB\paren{ \rho_2\otimes\dots\otimes\rho_n\otimes {\sigma_1\otimes\dots\otimes \sigma_m}}\\
    &\starrel=
      F_2(\rho_2)\mtensor\dots\mtensor F_n(\rho_n)
      \mtensor \chain{F_1}{G_1}(\sigma_1) \mtensor\dots\mtensor
      \chain{F_1}{G_m}(\sigma_m)
    \\
    &\eqrefrel{lemma:mixed.states:permute}=
    \chain{F_1}{G_1}(\sigma_1)
    \mtensor\dots\mtensor
    \chain{F_1}{G_m}(\sigma_m)
    \mtensor F_2(\rho_2)\mtensor\dots\mtensor F_n(\rho_n).
  \end{align*}
  Here $(*)$ is by definition of $\mtensor$.
  And $(**)$ is by Law~\ref{law:pair.tensor}.
  And $(*{*}*)$ uses $\chain{F_1}{G^*}={\pairs{\chain{F_1}{G_1}}{\chain{F_1}{G_m}}}$ which follows from Law~\ref{law:pair.chain}.
  Recall in this calculation that $\otimes$, $\pair\cdot\cdot$ are right-associative.
  
  This shows \eqref{lemma:mixed.states:nested}.
\end{proof}

\subsection{Pure states, \autoref{lemma:lift.pure}}
\label{app:proof:lemma:lift.pure}

\begin{proof}
  Throughout this proof, we abbreviate $\psi^* := \psi_1\otimes\dots\otimes\psi_n$ and $\eta^*:=\pureeta{\mathbf A_1}\otimes\dots\otimes\pureeta{\mathbf A_n}$ and $F^* := \pairs{F_1}{F_n}$ and $b:={F_1(\pureeta{\mathbf{A}_1}\adj{\pureeta{\mathbf{A}_1}})\mtensor\dots\mtensor F_n(\pureeta{\mathbf{A}_n}\adj{\pureeta{\mathbf{A}_n}})}=F^*(\eta^*\adj{{\eta^*}})$.
  Note that $\eta^*=\pureeta{\mathbf{A_1}\otimes\dots\otimes\mathbf{A_n}}$ by definition of $\pureeta{\dots}$.
  Note that $F^*$ is an iso-reference since $F_1,\dots,F_n$ is a partition by assumption.
  And ${F_1(\rho_1)\ptensor\dots\ptensor F_n(\rho_n)} = F^*(\psi^*\adj{{\eta^*}})\purexi b$ by definition.
  Since $F^*$ is an iso-reference, by \autoref{lemma:isoreg-decomp}, $F^*(\rho^*)=U\rho^*\adj U$ for some unitary $U$ that we fix throughout the proof as well.

  \medskip
  
  We show \eqref{lemma:lift.pure:norm}.
  \begin{multline*}
    \pb\norm {F_1(\rho_1)\ptensor\dots\ptensor F_n(\rho_n)}^2 
    =
      \pb\norm {F^*(\psi^*\adj{{\eta^*}})\purexi b}^2
    =
      \adj{{\purexi b}}
      \adj{F^*(\psi^*\adj{{\eta^*}})}
      F^*(\psi^*\adj{{\eta^*}})
      \purexi b \\
    \starrel=
      \adj{{\purexi b}}
      F^*(\eta^*\adj{{\psi^*}}\psi^*\adj{{\eta^*}})
      \purexi b
    =
      \norm{\psi^*}^2  \cdot
      \adj{{\purexi b}}
      F^*(\eta^*\adj{{\eta^*}})
      \purexi b
    \starstarrel=
      \norm{\psi^*}^2  \cdot
      \norm{\purexi b}^2 
    \tristarrel=
      \norm{\psi^*}^2 \\
    =
      \pb\paren{\norm{\psi_1}\cdots\norm{\psi_n}}^2
  \end{multline*}
  Here $(*)$ follows from the fact that references are multiplicative and preserve adjoints.
  To see $(**)$, note that $\eta^*$ is a unit vector, hence $\eta^*\adj{{\eta^*}}$ is a projector, hence $b=F^*(\eta^*\adj{{\eta^*}})$ is a non-zero projector (\lemmaref{lemma:preserve.projector}, \autoref{lemma:reg.injective}).
  And $\purexi b$ is in the image of that projector by definition.
  And $(*{*}*)$ holds because $\purexi b$ is a unit vector by definition.
  This shows the first part of \eqref{lemma:lift.pure:norm}.
  The ``in particular'' part follows because pure states are states of norm $1$.

  \medskip
    
  We show \eqref{lemma:lift.pure:permute}.
  \begin{align*}
    \hskip1em&\hskip-1em
      F_{\pi(1)}(\psi_{\pi(1)})\ptensor\dots\ptensor F_{\pi(n)}(\psi_{\pi(n)}) \\
    &\starrel=
    \pB\paren{F_{\pi(1)}(\psi_{\pi(1)}\pureeta{\mathbf A_{\pi(1)}})\mtensor\dots\mtensor F_{\pi(n)}(\psi_{\pi(n)}\pureeta{\mathbf A_{\pi(n)}})}
    \purexi{{F_{\pi(1)}(\pureeta{\mathbf{A}_{\pi(1)}}\adj{\pureeta{\mathbf A_{\pi(1)}}})\mtensor\dots\mtensor F_{\pi(n)}(\pureeta{\mathbf{A}_{\pi(n)}}\adj{\pureeta{\mathbf A_{\pi(n)}}})}} \\
    &\starstarrel=
    \pB\paren{F_{1}(\psi_{1}\pureeta{\mathbf A_{1}})\mtensor\dots\mtensor F_{n}(\psi_{n}\pureeta{\mathbf A_{n}})}
    \purexi{{F_1(\pureeta{\mathbf{A}_1}\adj{\pureeta{\mathbf A_1}})\mtensor\dots\mtensor F_n(\pureeta{\mathbf{A}_n}\adj{\pureeta{\mathbf A_n}})}} \\
    &\starrel=
    F_1(\psi_1)\ptensor F_2(\psi_2)\ptensor\dots\ptensor F_n(\psi_n).
  \end{align*}
  Here $(*)$ is by definition of $\ptensor$ and $(**)$ by two applications of \lemmaref{lemma:mixed.states:permute}.
  This shows \eqref{lemma:lift.pure:permute}.

  \medskip
  
  We show \eqref{lemma:lift.pure:bounded}. Linearity is by construction, and boundedness follows from \eqref{lemma:lift.pure:norm}.

  \medskip

  We show \eqref{lemma:lift.pure:pair}.
  \begin{align*}
    \hskip1cm & \hskip-1cm
      \pair{F_1}{F_2}(\psi_1\otimes\psi_2)\ptensor F_3(\psi_3)\ptensor\dots\ptensor F_n(\psi_n) \\
    &\starrel=
      \pB\paren{ \pair{F_1}{F_2}\pb\paren{\paren{\psi_1\otimes\psi_2}\adj{\pureeta{\mathbf A_1\otimes\mathbf A_2}}}
      \mtensor F_3(\psi_3\adj{\pureeta{\mathbf A_3}})
      \mtensor \dots
      \mtensor F_n(\psi_n\adj{\pureeta{\mathbf A_n}})      } \cdot {}\\
    &\qquad\qquad\qquad
      \purexi{{ \pair{F_1}{F_2}\paren{\pureeta{\mathbf A_1\otimes\mathbf A_2}\adj{\pureeta{\mathbf A_1\otimes\mathbf A_2}}}
      \mtensor F_3(\pureeta{\mathbf A_3}\adj{\pureeta{\mathbf A_3}})
      \mtensor \dots
      \mtensor F_n(\pureeta{\mathbf A_n}\adj{\pureeta{\mathbf A_n}})      }}
      \\
    &\starstarrel=
      \pB\paren{ \pair{F_1}{F_2}\pb\paren{\psi_1\adj{\pureeta{\mathbf A_1}}\otimes \psi_2\adj{\pureeta{\mathbf A_2}}}
      \mtensor F_3(\psi_n\adj{\pureeta{\mathbf A_3}})
      \mtensor \dots
      \mtensor F_n(\psi_n\adj{\pureeta{\mathbf A_n}})      } \cdot{}
    \\
    &\qquad\qquad\qquad
      \purexi{{ \pair{F_1}{F_2}\paren{\pureeta{\mathbf A_1}\adj{\pureeta{\mathbf A_1}}\otimes\pureeta{\mathbf A_2}\adj{\pureeta{\mathbf A_2}}}
      \mtensor F_3(\pureeta{\mathbf A_3}\adj{\pureeta{\mathbf A_3}})
      \mtensor \dots
      \mtensor F_n(\pureeta{\mathbf A_n}\adj{\pureeta{\mathbf A_n}})      }}
      \\
    &\tristarrel=
      \pB\paren{ {F_1}\pb\paren{\psi_1\adj{\pureeta{\mathbf A_1}}}
      \mtensor {F_2}\pb\paren{\psi_2\adj{\pureeta{\mathbf A_2}}}
      \mtensor F_3(\psi_n\adj{\pureeta{\mathbf A_3}})
      \mtensor \dots
      \mtensor F_n(\psi_n\adj{\pureeta{\mathbf A_n}})      } \cdot{}
      \\
    &\qquad\qquad\qquad
      \purexi{{
      {F_1}\paren{\pureeta{\mathbf A_1}\adj{\pureeta{\mathbf A_1}}}
      \mtensor {F_2}\paren{\pureeta{\mathbf A_2}\adj{\pureeta{\mathbf A_2}}}
      \mtensor F_3(\pureeta{\mathbf A_3}\adj{\pureeta{\mathbf A_3}})
      \mtensor \dots
      \mtensor F_n(\pureeta{\mathbf A_n}\adj{\pureeta{\mathbf A_n}})      }}
      \\
    &\starrel=
    F_1(\psi_1)\ptensor F_2(\psi_2)\ptensor\dots\ptensor F_n(\psi_n).
  \end{align*}
  Here $(*)$ is by definition of $\ptensor$.
  And $(**)$ follows since $\pureeta{\mathbf A_1\otimes\mathbf A_2} = \pureeta{\mathbf A_1}\otimes\pureeta{\mathbf A_2}$ by definition.
  And $(*{*}*)$ is by \lemmaref{lemma:mixed.states:pair}.
  This shows \eqref{lemma:lift.pure:pair}.

  \medskip
  
  We show \eqref{lemma:lift.pure:separating}.
  Let $S^* := \pb\braces{\psi_1\otimes\dots\otimes\psi_n: \forall i.\ \psi_i\in S_i}$.
  By assumption, $S_i$ generates $\calH_{A_i}$, i.e., $\overline{\Span S_i}=\calH_{A_i}$ where $\overline{{}\cdot{}}$ denotes the topological closure.
  The algebraic tensor product of dense subspaces (such as $\Span S_i$) of some Hilbert spaces is dense in the tensor product of those Hilbert spaces \cite[Exercise II.6.9]{prugovecki81hilbert}.
  Thus $S^*$ generates $\calH^* := \calH_{A_1}\otimes\dots\otimes\calH_{A_n}$, i.e., $\overline{\Span S^*}=\calH^*$.
  Let $S_F :=  \pb\braces{  F_1(\psi_1)\ptensor F_2(\psi_2)\ptensor\dots\ptensor F_n(\psi_n) : \psi_i\in S_i }$.
  We want to show that $S_F$ generates $\calH_B$.
  We have
  \begin{align*}
    S_F
    &= \pb\braces{ F^*\pb\paren{\psi_1\adj{\pureeta{\mathbf{A}_1}}\otimes\dots\otimes\psi_n\adj{\pureeta{\mathbf{A}_n}}} \purexi b : \psi_i\in S_i } 
    \starrel= \pb\braces{ F^*\pb\paren{\psi\adj{{\eta^*}}} \purexi b : \psi\in S^* } \\
    &= \pb\braces{ U\psi\adj{{\eta^*}}\adj U \purexi b : \psi\in S^* }
      \starstarrel= \pb\braces{ \adj{{\eta^*}}\adj U \purexi b \cdot U\psi : \psi\in S^* }
  \end{align*}
  Here $(*)$ follows by definition of $S^*$ and the fact that $\pureeta{\mathbf{A}_1\otimes\dots\otimes\mathbf A_n} = \pureeta{\mathbf A_1}\otimes\dots\otimes\pureeta{\mathbf A_n}$.
  And for $(**)$ notice that $\adj{{\eta^*}}\adj U \purexi b$ is a scalar.
  By definition, $\purexi b$ is in the image of $b^*=F(\eta^*\adj{{\eta^*}})=U\eta^*\adj{{\eta^*}}\adj U$.
  In particular, it is not orthogonal to $\im U\eta^*$.
  Thus  $\adj{{\eta^*}}\adj U \purexi b\neq 0$.
  Hence
  \begin{multline*}
    \overline{\Span S_F}
    = \overline{\Span \pb\braces{ \adj{{\eta^*}}\adj U \purexi b \cdot U\psi : \psi\in S^* }} 
    \starrel= \overline{\Span \pb\braces{ U\psi : \psi\in S^* }} = \overline{\Span US^*} \\
    \starstarrel= U\,\overline{\Span S^*} = U\calH^* \starstarrel= \calH_B.
  \end{multline*}
  Here $(*)$ uses that $\adj{{\eta^*}}\adj U \purexi b$ is a non-zero scalar factor and thus does not affect the span.
  And $(**)$ uses that $U:\calH^*\to\calH_B$ is unitary.
  Hence $S_F$ generates $\calH_B$.
  This shows the first sentence of \eqref{lemma:lift.pure:separating}.
  The ``in particular'' part follows because $U,V$ are bounded operators and thus linear and continuous.

  \medskip

  We show \eqref{lemma:lift.pure:regulars}.
  By definition, a reference $F:\mathbf A\to\mathbf B$ is $\eta$-regular if $\pair F{\compl F}(\pureetabutt{\mathbf A}\otimes c)=\pureetabutt{\mathbf B}$ for some~$c$.
  But since all complements of $F$ are equivalent, it is sufficient to show that there is \emph{some} complement $F'$ of $F$ and some $c$ such that $\pair F{F'}(\pureetabutt{\mathbf A}\otimes c)=\pureetabutt{\mathbf B}$.
  With that in mind, the fact that $\Fst:\mathbf A\to\mathbf A\otimes\mathbf B$, $\Snd:\mathbf B\to\mathbf A\otimes\mathbf B$, $\swap:\mathbf A\otimes\mathbf B\to\mathbf B\otimes\mathbf A$, $\assoc:(\mathbf A\otimes\mathbf B)\otimes\mathbf C \to \mathbf
  A\otimes(\mathbf B\otimes\mathbf C)$, $\assoc'$ are $\eta$-regular is shown by the following calculations (always keeping in mind that $\pureeta{\mathbf A\otimes\mathbf B}=\pureeta{\mathbf A}\otimes\pureeta{\mathbf B}$):
  \begin{gather*}
    \pair\Fst\Snd(\pureetabutt{\mathbf A}\otimes\pureetabutt{\mathbf B})
    =
    (\pureetabutt{\mathbf A} \otimes 1)(1\otimes \pureetabutt{\mathbf B})
    =
    \pureetabutt{\mathbf A\otimes\mathbf B}
    \\[3pt]
    \pair\Snd\Fst(\pureetabutt{\mathbf B}\otimes\pureetabutt{\mathbf A})
    =
    (1\otimes \pureetabutt{\mathbf B})(\pureetabutt{\mathbf A} \otimes 1)
    =
    \pureetabutt{\mathbf A\otimes\mathbf B}
    \\[3pt]
    \pair\sigma{\compl\swap}(\pureetabutt{\mathbf A\otimes\mathbf B} \otimes 1)
    =
    \pair\sigma{\compl\swap}\pb\paren{(\pureetabutt{\mathbf A}\otimes\pureetabutt{\mathbf B})\otimes 1}
    =
    \pureetabutt{\mathbf B}\otimes\pureetabutt{\mathbf A}
    =
    \pureetabutt{\mathbf B\otimes\mathbf A}
    \\[3pt]
    \begin{aligned}
      \pair\assoc{\compl\assoc}(\pureetabutt{(\mathbf A\otimes\mathbf B)\otimes\mathbf C} \otimes 1)
      &=
      \pair\assoc{\compl\assoc}\pb\paren{((\pureetabutt{\mathbf A}\otimes\pureetabutt{\mathbf B})\otimes\pureetabutt{\mathbf C})\otimes 1}
      \\[-2pt]&=
      \pureetabutt{\mathbf A}\otimes(\pureetabutt{\mathbf B}\otimes\pureetabutt{\mathbf C})
      =
      \pureetabutt{\mathbf A\otimes(\mathbf B\otimes\mathbf C)}
    \end{aligned}
    \\[3pt]
    \alpha'\text{ analogous}
  \end{gather*}
  This shows that $\Fst,\Snd,\swap,\assoc,\assocp$ are $\eta$-regular.

  We next show that $\spair FG$ is $\eta$-regular for $\eta$-regular disjoint $F:\mathbf A\to\mathbf C$, $G:\mathbf B\to\mathbf C$.
  By Law~\ref{law:complement.pair}, $\pair G{\compl{\spair FG}}$ is a complement of $F$.
  Then since $F$ is $\eta$-regular, there exists a $c$ such that $\pb\pair{F}{\pair G{\compl{\spair FG}}}(\pureetabutt{\mathbf A}\otimes c) = \pureetabutt{\mathbf C}$.
  Similarly, there exists a $d$ such that $\pb\pair{G}{\pair F{\compl{\spair FG}}}(\pureetabutt{\mathbf B}\otimes d) = \pureetabutt{\mathbf C}$.
  Note that for all $a,b,z$,
  $\pb\pair{G}{\pair F{\compl{\spair FG}}} (b\otimes (a\otimes z)) = \pb\pair{F}{\pair G{\compl{\spair FG}}} \paren{a\otimes(b\otimes z)}$
  and thus, since $\pb\pair{F}{\pair G{\compl{\spair FG}}}$ is an iso-reference,
  $\pb\pair{F}{\pair G{\compl{\spair FG}}}^{-1} \circ \pb\pair{G}{\pair F{\compl{\spair FG}}} (b\otimes (a\otimes z)) = a\otimes(b\otimes z)$.
  Thus
  \[
    \underbrace{\swap \circ \pb\pair{F}{\pair G{\compl{\spair FG}}}^{-1} \circ \pb\pair{G}{\pair F{\compl{\spair FG}}} \circ (1\otimes \swap)}_{{}=:H} (b\otimes (z\otimes a)) = (b\otimes z)\otimes a.
  \]
  We also have
  $\assocp (b\otimes (z\otimes a)) = (b\otimes z)\otimes a$.
  By Law~\ref{law:tensor3}, this implies that $H = \assocp$.
  Thus
  \begin{multline*}
    c\otimes\pureetabutt{\mathbf A}
    =
    \sigma(\pureetabutt{\mathbf A}\otimes c)
    =
    \sigma\circ\pb\pair{F}{\pair G{\compl{\spair FG}}}^{-1}(\pureetabutt{\mathbf C}) \\
    =
    \sigma \circ \pb\pair{F}{\pair G{\compl{\spair FG}}}^{-1} \circ \pb\pair{G}{\pair F{\compl{\spair FG}}}
    (\pureetabutt{\mathbf B}\otimes d)
    = H     (\pureetabutt{\mathbf B}\otimes \sigma(d))
    =  \assocp   (\pureetabutt{\mathbf B}\otimes \sigma(d)).
  \end{multline*}
  By \autoref{lemma:tensor.overlap}, this implies that $c=\pureetabutt{\mathbf B}\otimes c'$ for some $c'$.
  (Using that $\assocp(x)=\adj{U_\alpha}xU_\alpha$ for $U_\alpha$ as in \autoref{lemma:tensor.overlap}.)
  Then we have
  \begin{multline*}
    \pb\pair{\spair FG}{\compl{\spair FG}} (\pureetabutt{\mathbf A\otimes\mathbf B}\otimes c')
    =
    \pb\pair{\spair FG}{\compl{\spair FG}} \pb\paren{\paren{\pureetabutt{\mathbf A}\otimes\pureetabutt{\mathbf B}}\otimes c'} \\
    =
    \pb\pair F{\pair G{\compl{\spair FG}}} \pb\paren{\pureetabutt{\mathbf A}\otimes(\pureetabutt{\mathbf B}\otimes c')}
    =
    \pb\pair F{\pair G{\compl{\spair FG}}} (\pureetabutt{\mathbf A}\otimes c)
    =
    \pureetabutt{\mathbf C}.
  \end{multline*}
  This shows that $\spair FG$ is $\eta$-regular.

  We next show that $\chain FG$ is $\eta$-regular for $\eta$-regular $F:\mathbf B\to\mathbf C$, $G:\mathbf A\to\mathbf B$.
  Since $G$ is $\eta$-regular, there exists some $c$ such that $\pair{G}{\compl G}(\pureetabutt{\mathbf A}\otimes c)=\pureetabutt{\mathbf B}$.
  Since $F$ is $\eta$-regular, there exists some $d$ such that $(F,\compl F)\paren{\pureetabutt{\mathbf B}\otimes d}=\pureetabutt{\mathbf C}$.
  Thus
  \begin{equation*}
    \pb\pair{\chain FG}{\pair{\compl F}{\chain F{\compl G}}} \pb\paren{\pureetabutt{\mathbf A} \otimes (d\otimes c)}
    =
    \pb\pair{F}{\compl F} \pb\paren{\pair{G}{\compl G}(\pureetabutt{\mathbf A} \otimes c) \otimes d}
    =
    \pb\pair{F}{\compl F} \pb\paren{\pureetabutt{\mathbf B} \otimes d}
    =
    \pureetabutt{\mathbf C}.
  \end{equation*}
  Since ${\pair{\compl F}{\chain F{\compl G}}}$ is a complement of $\chain FG$ by Law~\ref{law:complement.chain}, this means that $F.G$ is $\eta$-regular.

  We next show that $F\rtensor G$ is $\eta$-regular for $\eta$-regular $F:\mathbf A\to\mathbf C$, $G:\mathbf B\to\mathbf D$.
  By Law~\ref{law:complement.tensor}, $\compl F\rtensor\compl G$ and $F\rtensor G$ are complements.
  Thus by Law~\ref{law:compl.equiv}, $\compl F\rtensor\compl G$ and $\compl{\paren{F\rtensor G}}$ are equivalent.
  Thus by definition of equivalence, there exists an iso-reference $I$ such that $\compl F\rtensor\compl G = \compl{\paren{F\rtensor G}} \circ I$.
  Since $F,G$ are $\eta$-regular, there exist $c_F,c_G$ such that $\pair F{\compl F}(\pureetabutt{\mathbf A}\otimes c_F)=\pureetabutt{\mathbf C}$ and $\pair G{\compl G}(\pureetabutt{\mathbf B}\otimes c_G)=\pureetabutt{\mathbf D}$.
  Let $c:=I(c_F\otimes c_G)$.
  Then
  \begin{align*}
    &\pB\pair{F\rtensor G}{\compl{\paren{F\rtensor G}}} \pb\paren{\pureetabutt{\mathbf A\otimes\mathbf B} \otimes c}
    =
      (F\rtensor G)(\pureetabutt{\mathbf A\otimes\mathbf B}) \cdot {\compl{\paren{F\rtensor G}}}(c)
    \\&
    =
    \pb\paren{F\rtensor G}\pb\paren{\pureetabutt{\mathbf A}\otimes\pureetabutt{\mathbf B}} \cdot {\compl{\paren{F\rtensor G}}}\circ I (c_F\otimes c_G)
    \\&
    =
    \pb\paren{F\rtensor G}\pb\paren{\pureetabutt{\mathbf A}\otimes\pureetabutt{\mathbf B}} \cdot \pb\paren{\compl F\rtensor\compl G} (c_F\otimes c_G)
    =
    \pB\paren{ F(\pureetabutt{\mathbf A}) \otimes G(\pureetabutt{\mathbf B}) }
    \cdot
    \pB\paren{\compl F(c_F)\rtensor\compl G(c_G)}
    \\&
    =
    { F(\pureetabutt{\mathbf A})\compl F(c_F) \otimes G(\pureetabutt{\mathbf B})\compl G(c_G) }
    =
    \pair F{\compl F}\paren{\pureetabutt{\mathbf A} \otimes c_F}
    \otimes
    \pair G{\compl G}\paren{\pureetabutt{\mathbf B} \otimes c_G}
    \\&
    =
    \pureetabutt{\mathbf C} \otimes \pureetabutt{\mathbf D}
    =
    \pureetabutt{\mathbf C\otimes\mathbf D}.
  \end{align*}
  Thus $F\rtensor G$ is $\eta$-regular.
  
  This shows \eqref{lemma:lift.pure:regulars}.
  
  \medskip
  
  We show \eqref{lemma:lift.pure:nested}.
  Let $G^* := \pairs{G_1}{G_m}$ and  $\xi_G^*:=\purexi{{G_1(\pureeta{\mathbf{C}_1}\adj{\pureeta{\mathbf{C}_1}})\mtensor\dots\mtensor G_m(\pureeta{\mathbf{C}_m}\adj{\pureeta{\mathbf{C}_m}})}}$
  and $\mathbf C^* := \mathbf C_1\otimes\dots\otimes \mathbf C_m$ and $\mathbf A^* := \mathbf A_1\otimes\dots\otimes \mathbf A_n$ 
  and
  \[
    \xi_{FG}^*:=\purexi{{\chain{F_1}{G_1}(\pureeta{\mathbf{C}_1}\adj{\pureeta{\mathbf{C}_1}})\mtensor\dots\mtensor \chain{F_1}{G_m}(\pureeta{\mathbf{C}_m}\adj{\pureeta{\mathbf{C}_m}})
      \mtensor F_2(\pureeta{\mathbf{A}_2}\adj{\pureeta{\mathbf{A}_2}})
      \mtensor \dots       \mtensor F_n(\pureeta{\mathbf{A}_n}\adj{\pureeta{\mathbf{A}_n}})
    }}.
  \]
  Since all $G_i$ are $\eta$-regular (and disjoint), by \eqref{lemma:lift.pure:regulars}, $G^*:\mathbf C^*\to\mathbf A_1$ is $\eta$-regular.
  Thus there is a $c$ such that $\pair{G^*}{\compl{{G^*}}}\paren{\pureeta{\mathbf C^*}\adj{\pureeta{\mathbf C^*}}\otimes c} = \pureeta{\mathbf A_1}\adj{\pureeta{\mathbf A_1}}$.
  Since $G^*$ is an iso-reference, the range of $G^*$ is $\mathbf A_1$.
  Since $\compl{{G^*}}$ is disjoint from $G^*$ (by definition of complements), the ranges of $\compl{{G^*}}$ and $G^*$ commute.
  Thus $\compl{{G^*}}(c)$ commutes with all of $\mathbf A_1$, hence  $\compl{{G^*}}(c)$ is a multiple of the identity, say $\lambda 1$.
  Thus
  \begin{equation*}
    \pureeta{\mathbf A_1}\adj{\pureeta{\mathbf A_1}} =
    \pair{G^*}{\compl{{G^*}}}\paren{\pureeta{\mathbf C^*}\adj{\pureeta{\mathbf C^*}}\otimes c} =
    G^*\paren{\pureeta{\mathbf C^*}\adj{\pureeta{\mathbf C^*}}}\cdot \compl{{G^*}}(c) =
    \lambda G^*\paren{\pureeta{\mathbf C^*}\adj{\pureeta{\mathbf C^*}}}.
  \end{equation*}
  Since $\pureeta{\mathbf A_1}$ and $\pureeta{\mathbf C^*}$ are unit vectors by definition, we have $\norm{\pureeta{\mathbf A_1}\adj{\pureeta{\mathbf A_1}}}=1$
  and $\norm{G^*\paren{\pureeta{\mathbf C^*}\adj{\pureeta{\mathbf C^*}}}} \starrel= \norm{\pureeta{\mathbf C^*}\adj{\pureeta{\mathbf C^*}}} = 1$
  (with $(*)$ by \lemmaref{lemma:preserve.norm}), we have $\abs\lambda=1$.
  Since $\pureetabutt{\mathbf A_1}$ and $G^*(\pureetabutt{\mathbf C^*})$ are positive (using the positivity of $G^*$, \autoref{lemma:reg.props}), $\pureetabutt{\mathbf A_1}=G^*(\pureetabutt{\mathbf C^*})$ implies that $\lambda\geq 0$.
  Thus $\lambda=1$.
  Thus $G^*\paren{\pureeta{\mathbf C^*}\adj{\pureeta{\mathbf C^*}}} =     \pureeta{\mathbf A_1}\adj{\pureeta{\mathbf A_1}}$.
  Then $\xi^*_G = \purexi{\im   \pureeta{\mathbf A_1}\adj{\pureeta{\mathbf A_1}}} = \pureeta{\mathbf A_1}$ where the last equality is by definition of $\pureeta{\dots}$.

  By definition, $\purexi b$ is in the image of $b=F^*(\eta^*\adj{{\eta^*}})=F^*(\pureetabutt{\mathbf A^*})$.
  Since $\pureeta{\mathbf A^*}$ and $\pureeta{\mathbf A_1}$ are unit vectors, $\pureetabutt{\mathbf A^*},\pureetabutt{\mathbf A_1}$ are projectors.
  Then by \lemmaref{lemma:preserve.projector}, $F^*(\pureetabutt{\mathbf A^*}), F_1(\pureetabutt{\mathbf A_1})$ are projectors.
  Thus
  \begin{multline*}
    F^*(\pureetabutt{\mathbf A^*}) \cdot F_1(\pureetabutt{\mathbf A_1})
    \starrel=
    F^*(\pureetabutt{\mathbf A^*}) \cdot F^*(\pureetabutt{\mathbf A_1}\otimes 1\otimes\dots\otimes 1) \\
    =
    F^*(\pureetabutt{\mathbf A^*} \cdot (\pureetabutt{\mathbf A_1}\otimes 1\otimes\dots\otimes 1))
    \starstarrel=
    F^*(\pureetabutt{\mathbf A^*}).
  \end{multline*}
  Here $(*)$ is by definition of $F^*$ and \autoref{def:pair}.
  And $(**)$ uses $\pureeta{\mathbf A^*}=\pureeta{\mathbf A_1}\otimes\dots\otimes\pureeta{\mathbf A_n}$.
  
  Thus the image of the projector $F^*(\pureetabutt{\mathbf A^*})$ is a subset of the image of the projector $F_1(\pureetabutt{\mathbf A_1})$ \cite[Exercise II.3.6]{conway13functional}.
  Hence $\purexi b$ is in the image of $F_1(\pureetabutt{\mathbf A_1})$. So  $F_1(\pureetabutt{\mathbf A_1})\purexi b=\purexi b$.

  Furthermore, we have
  \begin{align}
    \hskip1cm&\hskip-1cm
     \chain{F_1}{G_1}(\pureeta{\mathbf{C}_1}\adj{\pureeta{\mathbf{C}_1}})\mtensor\dots\mtensor \chain{F_1}{G_m}(\pureeta{\mathbf{C}_m}\adj{\pureeta{\mathbf{C}_m}})
      \mtensor F_2(\pureeta{\mathbf{A}_2}\adj{\pureeta{\mathbf{A}_2}})
      \mtensor \dots       \mtensor F_n(\pureeta{\mathbf{A}_n}\adj{\pureeta{\mathbf{A}_n}}) \notag\\
    &\starrel=
      F_1\pB\paren{{G_1}(\pureeta{\mathbf{C}_1}\adj{\pureeta{\mathbf{C}_1}})\mtensor\dots\mtensor {G_m}(\pureeta{\mathbf{C}_m}\adj{\pureeta{\mathbf{C}_m}})}
      \mtensor F_2(\pureeta{\mathbf{A}_2}\adj{\pureeta{\mathbf{A}_2}})
      \mtensor \dots       \mtensor F_n(\pureeta{\mathbf{A}_n}\adj{\pureeta{\mathbf{A}_n}})  \notag\\
    &\starstarrel=
      F_1\pB\paren{{G^*}(\pureetabutt{\mathbf{C}^*})}
      \mtensor F_2(\pureeta{\mathbf{A}_2}\adj{\pureeta{\mathbf{A}_2}})
      \mtensor \dots       \mtensor F_n(\pureeta{\mathbf{A}_n}\adj{\pureeta{\mathbf{A}_n}})  \notag\\
    &\tristarrel=
      F_1(\pureeta{\mathbf{A}_1}\adj{\pureeta{\mathbf{A}_1}})\mtensor\dots\mtensor F_n(\pureeta{\mathbf{A}_n}\adj{\pureeta{\mathbf{A}_n}}).
      \label{eq:xiFG.xi.prep}
  \end{align}
  Here $(*)$ follows by \lemmaref{lemma:mixed.states:nested}.
  And $(**)$ follows by definition of $G^*$ and $\mathbf C^*$ and the fact that $\pureeta{\mathbf C^*}=\pureeta{\mathbf C_1\otimes\dots\otimes\mathbf C_m}=\pureeta{\mathbf C_1}\otimes\dots\otimes\pureeta{\mathbf C_m}$.
  And $(*{*}*)$ follows because we established that ${G^*}(\pureetabutt{\mathbf{C}^*}) = \pureetabutt{\mathbf A_1}$.
  By definition of $\xi_{FG}^*$ and $\purexi b$, \eqref{eq:xiFG.xi.prep} implies $\xi_{FG}^*=\purexi b$.

  Combining what we have so far, we get:
  \begin{equation}
    F_1(\xi_G^*\adj{\pureeta{\mathbf A_1}})\purexi b =
    F_1(\pureetabutt{\mathbf A_1})\purexi b =
    \purexi b =
    \xi_{FG}^*.
    \label{eq:xi.xiFG}
  \end{equation}
  Then
  \begin{align*}
    \hskip1em & \hskip-1em
      F_1\pB\paren{G_1(\gamma_1)\ptensor\dots\ptensor G_m(\gamma_m)}\ptensor F_2(\psi_2)\ptensor\dots\ptensor F_n(\psi_n) \\
    &\starrel=\pB\paren{ F_1\pB\paren{\pb\paren{G_1(\gamma_1\adj{\pureeta{\mathbf C_1}})\mtensor\dots\mtensor G_m(\gamma_m\adj{\pureeta{\mathbf C_n}})}\xi_G^*\adj{\pureeta{\mathbf A_1}}}\mtensor F_2(\psi_2\adj{\pureeta{\mathbf A_2}})\mtensor\dots\mtensor F_n(\psi_n\adj{\pureeta{\mathbf A_n}}) }
      \cdot \purexi b
    \\
    &\starstarrel=\pB\paren{ F_1\pB\paren{\pb\paren{G_1(\gamma_1\adj{\pureeta{\mathbf C_1}})\mtensor\dots\mtensor G_m(\gamma_m\adj{\pureeta{\mathbf C_n}})}}\mtensor F_2(\psi_2\adj{\pureeta{\mathbf A_2}})\mtensor\dots\mtensor F_n(\psi_n\adj{\pureeta{\mathbf A_n}}) }
      \cdot F_1(\xi_G^*\adj{\pureeta{\mathbf A_1}})\purexi b
    \\
    &\tristarrel=\pB\paren{ \chain{F_1}{G_1}(\gamma_1\adj{\pureeta{\mathbf C_1}})\mtensor \dots \mtensor \chain{F_1}{G_m}(\gamma_m\adj{\pureeta{\mathbf C_m}}) \mtensor F_2(\psi_2\adj{\pureeta{\mathbf A_2}})\mtensor\dots\mtensor F_n(\psi_n\adj{\pureeta{\mathbf A_n}}) }
      \cdot F_1(\xi_G^*\adj{\pureeta{\mathbf A_1}})\purexi b
    \\
    &\eqrefrel{eq:xi.xiFG}=\pB\paren{ \chain{F_1}{G_1}(\gamma_1\adj{\pureeta{\mathbf C_1}})\mtensor \dots \mtensor \chain{F_1}{G_m}(\gamma_m\adj{\pureeta{\mathbf C_m}}) \mtensor F_2(\psi_2\adj{\pureeta{\mathbf A_2}})\mtensor\dots\mtensor F_n(\psi_n\adj{\pureeta{\mathbf A_n}}) }
      \cdot \xi_{FG}^* \\
    &\starrel= \chain{F_1}{G_1}(\gamma_1) \ptensor\dots\ptensor \chain{F_1}{G_m}(\gamma_m) \ptensor F_2(\psi_2)\ptensor\dots\ptensor F_n(\psi_n).
  \end{align*}
  Here $(*)$ is by definition of $\ptensor$.
  Here $(**)$ uses \lemmaref{lemma:mixed.states:apply.UV} with $U:=1$ and $V:=\xi_G^*\adj{\pureeta{\mathbf A_1}}$.
  And $(*{*}*)$ uses \lemmaref{lemma:mixed.states:nested}.
  This shows \eqref{lemma:lift.pure:nested}.

  \medskip

  We show \eqref{lemma:lift.pure:U}:
  \begin{multline*}
    F_1(U\psi_1)\ptensor\dots\ptensor F_n(\psi_n)
    \starrel=
      \pB\paren{F_1(U\psi_1\pureeta{\mathbf A_1})\mtensor\dots\mtensor F_n(\psi_n\pureeta{\mathbf A_n})}\purexi b \\
    \starstarrel=
      F_1(U) \pB\paren{F_1(\psi_1\pureeta{\mathbf A_1})\mtensor\dots\mtensor F_n(\psi_n\pureeta{\mathbf A_n})}\purexi b 
    \starrel= F_1(U) \pb\paren{F_1(\psi_1)\ptensor\dots\ptensor F_n(\psi_n)}.
  \end{multline*}
  Here $(*)$ is by definition of $\ptensor$.
  And $(**)$ by \lemmaref{lemma:mixed.states:apply.UV}.
  This shows \eqref{lemma:lift.pure:U}.

  \medskip
  
  We show \eqref{lemma:lift.pure:butterfly}.
  By definition $\purexi b$ is in the image of $b = F^*(\eta^*\adj{{\eta^*}})$.
  And since $\purexi b$ is a unit vector, $\purexi b\adj{{\purexi b}}$ is a rank-1 projector.
  Since $\eta^*$ is a unit vector, $\eta^*\adj{{\eta^*}}$ is a rank-1 projector and so is $F^*(\eta^*\adj{{\eta^*}})$ by \lemmaref{lemma:preserve.projector} and \lemmaref{lemma:mixed.states:rank1} (the latter in the degenerate case where $n=1$ and $F_1=F^*$).
  The fact that $\purexi b$ is in the range of $F^*(\eta^*\adj{{\eta^*}})$ implies that  $\purexi b\adj{{\purexi b}}$ and $F^*(\eta^*\adj{{\eta^*}})$ project onto the same space, hence $\purexi b\adj{{\purexi b}} = F^*(\eta^*\adj{{\eta^*}})$.
  Thus
  \begin{align*}
    \hskip1cm & \hskip-1cm
     \mathord\lozenge\pB\paren{F_1(\psi_1)\ptensor\dots\ptensor F_n(\psi_n)} \\
    &\starrel=
      \mathord\lozenge\pB\paren{\pB\paren{F_1(\psi_1\adj{\pureeta{\mathbf A_1}})\mtensor\dots\mtensor F_1(\psi_n\adj{\pureeta{\mathbf A_n}})}\purexi b} \\
    &\starstarrel=
      \mathord\lozenge\pB\paren{{F^*(\psi^*\adj{{\eta^*}})}\purexi b} \\
    &=
      {F^*(\psi^*\adj{{\eta^*}})}\purexi b
      \adj{{\purexi b}}\adj{F^*(\psi^*\adj{{\eta^*}})}\\
    &=
      {F^*(\psi^*\adj{{\eta^*}})}F^*(\eta^*\adj{{\eta^*}})\adj{F^*(\psi^*\adj{{\eta^*}})} \\
    &=
      F^*(\psi^*\adj{{\eta^*}}\eta^*\adj{{\eta^*}}\eta^*\adj{{\psi^*}}) \\
    &\tristarrel=
      F^*(\psi^*\adj{{\psi^*}}) \\
    &\starstarrel= F_1(\mathord\lozenge(\psi_1))\mtensor\dots\mtensor F_n(\mathord\lozenge(\psi_n)).
  \end{align*}
  Here $(*)$ is by definition of $\ptensor$.
  And $(**)$ is by definition of $\mtensor$.
  And $(*{*}*)$ follows because $\purexi b$ is a unit vector.
  This shows \eqref{lemma:lift.pure:butterfly}.

  \medskip

  We show \eqref{lemma:lift.pure:subspace}.
  Let $P_S$ denote the projector onto $S$.
  We have
  \begin{align*}
    \hskip1cm & \hskip-1cm
     F_1(\psi_1)\ptensor\dots\ptensor F_n(\psi_n) \in \regsubspace{F_1}S
    \\&
    \starrel\iff
      F_1(\psi_1)\ptensor\dots\ptensor F_n(\psi_n) \in \im F_1(P_S)
    \\&
    \iff
    F_1(P_S) \pB\paren{ F_1(\psi_1)\ptensor\dots\ptensor F_n(\psi_n) } = F_1(\psi_1)\ptensor\dots\ptensor F_n(\psi_n)
    \\&
    \eqrefrel{lemma:lift.pure:U}\iff
    F_1(P_S\psi_1)\ptensor F_2(\psi_2)\ptensor\dots\ptensor F_n(\psi_n)
    = F_1(\psi_1)\ptensor F_2(\psi_2)\ptensor\dots\ptensor F_n(\psi_n)
    \\&
    \starstarrel\iff
    P_S\psi_1 = \psi_1
    \iff
    \psi_1\in S.
  \end{align*}
  Here $(*)$ is by definition of $\regsubspace{F_1}S$.
  And $(**)$ follows since  $F_1(\cdot)\ptensor F_2(\psi_2)\ptensor\dots\ptensor F_n(\psi_n)$ is linear by \eqref{lemma:lift.pure:bounded} and thus injective by \eqref{lemma:lift.pure:norm} (using that $\psi_2,\dots,\psi_n\neq0$).
  
  This shows \eqref{lemma:lift.pure:subspace}.
\end{proof}

\subsection{Quantum channels, Lemmas~\ref{lemma:lift.channel} and~\ref{lemma:lift.kraus}}
\label{app:proof:qchannels}

\paragraph{Quantum channels, \autoref{lemma:lift.channel}}

\begin{proof}
  We first show \eqref{lemma:isoreg:channel}.
  By \autoref{lemma:isoreg-decomp}, $F(\rho)=U\rho \adj U$ for unitary $U$ and all $\rho$.
  By \autoref{lemma:UU.qchannel}, this implies that $F|_{\tracecl(\calH_A)}$ is a quantum channel for any iso-reference $F$.
  Since it has inverse $\rho\mapsto \adj U\rho U$, $F|_{\tracecl(\calH_A)}$ is bijective.
  Since $F,G$ and thus also $F\rtensor G$ are iso-references, we have that $F|_{\tracecl(\calH_A)}, G|_{\tracecl(\calH_B)}, (F\rtensor G)|_{\tracecl(\calH_A\otimes\calH_B)}$ are quantum channels.
  Thus $F|_{\tracecl(\calH_A)}\otimes G|_{\tracecl(\calH_B)}$ is also a quantum channel (\lemmaref{lemma:channel.tensor:tpres}).
  For all trace-class $\rho,\sigma$,
  \begin{multline*}
    (F\rtensor G)|_{\tracecl(\calH_A\otimes\calH_B)}(\rho\otimes\sigma)
    \starrel =
    (F\rtensor G)(\rho\otimes\sigma)
    = F(\rho)\otimes G(\sigma)
    \\= F|_{\tracecl(\calH_A)}(\rho)\otimes G|_{\tracecl(\calH_B)}(\sigma)
    \starstarrel= (F|_{\tracecl(\calH_A)}\otimes G|_{\tracecl(\calH_B)})(\rho\otimes\sigma)
  \end{multline*}
  Here $(*)$ uses that $\rho\otimes\sigma$ is trace-class by \autoref{lemma:traceclass.tensor}.
  And $(**)$ uses \lemmaref{lemma:channel.tensor:cp}.
  Since they coincide on all $\rho\otimes\sigma$, we have $(F\rtensor G)|_{\tracecl(\calH_A\otimes\calH_B)} = F|_{\tracecl(\calH_A)}\otimes G|_{\tracecl(\calH_B)}$ by \autoref{lemma:channel.equal}.
  This shows \eqref{lemma:isoreg:channel}.

  \medskip

  We show \eqref{lemma:lift.channel:channel}.
  $\pair F{\compl F}$ is an iso-reference by definition of $\compl F$.
  Thus $\pair F{\compl F}^{-1}$ is also an iso-reference.
  Then by \eqref{lemma:isoreg:channel}, $\pair F{\compl F}$ and $\pair F{\compl F}^{-1}$ are quantum channels.
  Thus if $\calE$ is a quantum (sub)channel, then the composition $\regchannel F\calE=
  \pair F{\compl F} \circ (\calE\otimes\id) \circ \pair F{\compl F}^{-1}$ is as well.
  This shows \eqref{lemma:lift.channel:channel}.

  \medskip
  
  We show \eqref{lemma:lift.channel:choice.compl}.
  If $F,G$ are complements, then $G$ and $\compl F$ are equivalent (Law~\ref{law:compl.equiv}).
  That means there is an iso-reference $I$ such that $G=\compl F\circ I$.
  Thus
  \begin{align*}
    \spair FG \circ (\calE\otimes\id) \circ \spair FG^{-1}
    &\starrel=
      \pair F{\compl F} \circ (\id\rtensor I) \circ (\calE\otimes\id) \circ (\id\rtensor I^{-1}) \circ \pair F{\compl F}^{-1}
    \\
    &\eqrefrel{lemma:isoreg:channel}=
      \pair F{\compl F} \circ (\id\otimes I) \circ (\calE\otimes\id) \circ (\id\otimes I^{-1}) \circ \pair F{\compl F}^{-1}
      \\
    &\starstarrel=
      \pair F{\compl F} \circ \pb\paren{\calE\otimes(I\circ I^{-1})} \circ \pair F{\compl F}^{-1}
      \\
    &=
      \pair F{\compl F} \circ \paren{\calE\otimes \id} \circ \pair F{\compl F}^{-1}
    = \regchannel F\calE.
  \end{align*}
  Here $(*)$ uses Law~\ref{law:pair.tensor} and $(**)$ uses \lemmaref{lemma:channel.tensor:distrib}.
  This shows \eqref{lemma:lift.channel:choice.compl}.

  \medskip

  We show \eqref{lemma:lift.channel:compose}.
  \begin{multline*}
    \regchannel F\calE\circ \regchannel F\calF
    \starrel=
    \pair F{\compl F}\circ(\calE\otimes\id)\circ \pair F{\compl F}^{-1}
    \circ
    \pair F{\compl F}\circ(\calF\otimes\id)\circ \pair F{\compl F}^{-1}
    \\
    \starstarrel=
    \pair F{\compl F}\circ(\calE\otimes\id)\circ(\calF\otimes\id)\circ \pair F{\compl F}^{-1}
    \tristarrel=
    \pair F{\compl F}\circ\pb\paren{(\calE\circ\calF)\otimes\id}\circ \pair F{\compl F}^{-1}
    \starrel= \regchannel F{\calE\circ\calF}.
  \end{multline*}
  Here $(*)$ is by definition of $\regchannel F{\calE},\regchannel F\calF$.
  And $(**)$ uses that $\pair F{\compl F}$ is an iso-reference and thus bijective.
  And $(*{*}*)$ uses \lemmaref{lemma:channel.tensor:distrib}.
  This shows \eqref{lemma:lift.channel:compose}.

    \medskip
    
  We show \eqref{lemma:lift.channel:pair}.
  By Law~\ref{law:complement.pair}, we have that $F$ and $\pair{G}{\compl{\spair FG}}$ are complements.
  And that $G$ and $\pair F{\compl{\spair FG}}$ are complements.
  \begin{align*}
    &F(\calE)\circ G(\calF)\\
    &\eqrefrel{lemma:lift.channel:choice.compl}=
      \pair{F}{\pair{G}{\compl{\spair FG}}} \circ (\calE \otimes \id) \circ
      \pair{F}{\pair{G}{\compl{\spair FG}}}^{-1} \circ
      \pair G{\pair F{\compl{\spair FG}}} \circ (\calF \otimes \id) \circ
      \pair G{\pair F{\compl{\spair FG}}}^{-1} \\
    &\starrel=
      \pair{\spair FG}{\compl{\spair FG}} \circ \alpha' \circ (\calE \otimes \id) \circ
      {\assoc} \circ \pair{\spair FG}{\compl{\spair FG}}^{-1}
      \\ & \qquad\qquad \circ
           \pair{\spair FG}{\compl{\spair FG}} \circ (\swap\rtensor\id) \circ \assocp \circ (\calF \otimes \id)
           \circ {\assoc} \circ  (\swap\rtensor\id) \circ \pair{\spair FG}{\compl{\spair FG}}^{-1} \\
    &=
      \pb\pair{\spair FG}{\compl{\spair FG}} \circ { \alpha' \circ (\calE \otimes \id) \circ \assoc \circ
       (\swap\rtensor\id)\circ \assocp\circ (\calF \otimes \id) \circ \assoc\circ   (\swap\rtensor\id)}
      {}\circ \pb\pair{\spair FG}{\compl{\spair FG}}^{-1} \\
    &\eqrefrel{lemma:isoreg:channel}=
      \pb\pair{\spair FG}{\compl{\spair FG}} \circ \underbrace{ \alpha' \circ (\calE \otimes \id) \circ \assoc \circ
       (\swap\otimes\id)\circ \assocp\circ (\calF \otimes \id) \circ \assoc\circ   (\swap\otimes\id)}_{\starstarrel= \ (\calE\otimes\calF)\otimes\id}
      {}\circ \pb\pair{\spair FG}{\compl{\spair FG}}^{-1} \\
    &=
      \spair FG(\calE\otimes\calF).
  \end{align*}
  Here $(*)$ uses Laws~\ref{law:pair.alpha'}, \ref{law:pair.sigma}, \ref{law:pair.tensor} and ${\assocp}^{-1}=\assoc$.
  And $(**)$ is verified by evaluating the left and right hand side of the equality on all $(\rho\otimes\tau)\otimes \pi$ for trace-class $\rho,\tau,\pi$.
  Both sides evaluate to $(\calE(\rho)\otimes\calF(\tau))\otimes\pi$, i.e., they are equal on all  $(\rho\otimes\tau)\otimes \pi$.
  Furthermore, both sides are quantum channels by \eqref{lemma:isoreg:channel}.
  Thus by \autoref{lemma:channel.equal}, both sides are equal.
  So $(**)$ holds.
  This shows \eqref{lemma:lift.channel:pair}.

  \medskip

  We show \eqref{lemma:lift.channel:id}. We have
  \begin{equation*}
    \regchannel F\id = \pair F{\compl F} \circ (\id\otimes\id) \circ \pair F{\compl F}^{-1}
    \starrel= \pair F{\compl F} \circ \pair F{\compl F}^{-1}
    \starstarrel = \id.
  \end{equation*}
  Here $(*)$ follows because $\id\otimes\id=\id$ (\lemmaref{lemma:channel.tensor:id}).
  And $(**)$ since $\pair F{\compl F}$ is an iso-reference.
  This shows \eqref{lemma:lift.channel:id}.

  We show \eqref{lemma:lift.channel:apply}.
  We abbreviate $F^* := \pairs{F_2}{F_n}$ and $\rho^* := \rho_2\otimes\dots\otimes\rho_n$. (If $n=1$, we let $F^*$ be the unit reference from \autoref{lemma:quantum.unit} and $\rho^*:=1\in\setC$.)
  Note that $F_1$ and $F^*$ are complements because $F_1,\dots,F_n$ is a partition.
  We have
  \begin{align*}
    \hskip2cm & \hskip-2cm \regchannel{F_1}\calE\pb\paren{F_1(\rho_1)\mtensor\dots\mtensor F_n(\rho_n)}
      \starrel=
      \regchannel{F_1}\calE \circ \pair{F_1}{F^*} (\rho_1\otimes\rho^*)
    \\&
    \eqrefrel{lemma:lift.channel:choice.compl}=
      \pair{F_1}{F^*} \circ (\calE\otimes\id) \circ \pair{F_1}{F^*}^{-1} \circ \pair{F_1}{F^*} (\rho_1\otimes\rho^*)
    =
    \pair{F_1}{F^*} \circ (\calE\otimes\id) \circ (\rho_1\otimes\rho^*)
    \\&
    =
    \pair{F_1}{F^*} \pb\paren{\calE(\rho_1)\otimes\rho^*} 
    \starrel=
    F_1\pb\paren{\calE(\rho_1)}\mtensor F_2(\rho_2)\mtensor\dots\mtensor F_n(\rho_n)
  \end{align*}
  Here $(*)$ is by definition of $\mtensor$.
  This shows \eqref{lemma:lift.channel:apply}.
  
  \medskip
  
  We show \eqref{lemma:lift.channel:chain}.
  \begin{align*}
    F(G(\calE))
    &=
      \pair{F}{\compl F}
      \circ
      \pB\paren{\pB\paren{\pair{G}{\compl G} \circ (\calE\otimes\id) \circ \pair{G}{\compl G}^{-1}}
      \otimes \id}
      \circ
      \pair{F}{\compl F}^{-1}
      \\
    &\starrel=
      \pair{F}{\compl F}
      \circ
      \pb\paren{\pair{G}{\compl G} \otimes \id}
      \circ
      \pb\paren{(\calE\otimes\id)\otimes\id}
      \circ
      \pb\paren{ \pair{G}{\compl G}^{-1} \otimes \id}
      \circ
      \pair{F}{\compl F}^{-1}
    \\
    &\eqrefrel{lemma:isoreg:channel}=
      \pair{F}{\compl F}
      \circ
      \pb\paren{\pair{G}{\compl G} \rtensor \id}
      \circ
      \pb\paren{(\calE\otimes\id)\otimes\id}
      \circ
      \pb\paren{ \pair{G}{\compl G} \rtensor \id}^{-1}
      \circ
      \pair{F}{\compl F}^{-1}
    \\
    &\starstarrel=
      \pair{\pair{\chain FG}{\chain F{\compl G}}}{\compl F}
      \circ
      \pb\paren{(\calE\otimes\id)\otimes\id}
      \circ
      \pair{\pair{\chain FG}{\chain F{\compl G}}}{\compl F} ^{-1}
    \\
    &\tristarrel=
      \pair{\chain FG}{\pair{\chain F{\compl G}}{\compl F}} \circ \underbrace{\assoc
      \circ
      \pb\paren{(\calE\otimes\id)\otimes\id}
      \circ \assocp}_{\fourstarrel= \ \calE\otimes\id}{} \circ
      \pair{\chain FG}{\pair{\chain F{\compl G}}{\compl F}}^{-1}
    \\[-15pt]
    &\txtrel{\ensuremath{\bigl(\overset{**}{*{*}*}\bigr)}}= (\chain FG)(\calE).
  \end{align*}
  Here $(*)$ uses \lemmaref{lemma:channel.tensor:distrib},\,\eqref{lemma:channel.tensor:id}.
  And $(**)$ uses Laws~\ref{law:pair.tensor} and~\ref{law:pair.chain}.
  And $(*{*}*)$ uses Law~\ref{law:pair.alpha}.
  And $(*{*}{*}*)$ is verified by evaluating the left and right hand side of the equality on all $\rho\otimes(\tau\otimes \pi)$ for trace-class $\rho,\tau,\pi$.
  Both sides evaluate to $\calE(\rho)\otimes(\tau\otimes\pi)$.
  Thus both sides are equal as quantum channels by \autoref{lemma:channel.equal}.
  So $(*{*}{*}*)$ holds.
  And $(\overset{\scriptstyle **}{\scriptstyle *{*}*})$ follows with \eqref{lemma:lift.channel:choice.compl} using that $\chain FG$ and $\pair{\chain{F}{\compl G}}{\compl F}$ are complements by Law~\eqref{law:complement.chain}.
  This shows~\eqref{lemma:lift.channel:chain}.

  \medskip
  
  We show \eqref{lemma:lift.channel:swap}.
  Let $u$ be the unit reference with domain $\setC$ from \autoref{lemma:quantum.unit}.
  Hence by Law~\ref{law:complement.iso.unit}, $\sigma,u$ are complements.
  We have $\pair\sigma u(x\otimes 1)=\sigma(x)$.
  Since $\pair\sigma u$ and $\sigma$ are both iso-references and thus invertible, this implies $\pair\sigma u^{-1}(y)=\sigma^{-1}(y)\otimes 1_{\setC}$.
  Furthermore, $\swap(\calE\otimes\calF)=\pair\swap u\circ\pb\paren{(\calE\otimes\calF)\otimes\id}\circ\pair\swap u^{-1}$ by \eqref{lemma:lift.channel:choice.compl}.
  Hence for all trace-class $\rho$, $\tau$, we have
  \begin{equation*}
    \rho\otimes\tau
    \xmapsto{\pair\swap u^{-1}}
    (\tau\otimes\rho)\otimes 1_{\setC}
    \xmapsto{{(\calE\otimes\calF)\otimes\id}}
    \pb\paren{\calE(\tau)\otimes\calF(\rho)}\otimes 1_{\setC}
    \xmapsto{\pair\swap u}
    \calF(\rho)\otimes\calE(\tau)
    =
    (\calF\otimes\calE)(\rho\otimes\tau).
  \end{equation*}
  (Note that in this calculation, we used that $1_{\setC}$ is trace-class.
  We used this when evaluating ${(\calE\otimes\calF)\otimes\id}$ which, as a tensor product of quantum channels, is only defined on trace-class operators.)
  Thus $\sigma(\calE\otimes\calF)(\rho\otimes\tau) = (\calF\otimes\calE)(\rho\otimes\tau)$ for all $\rho,\tau$.
  By \autoref{lemma:channel.equal}, this implies $\sigma(\calE\otimes\calF)=\calF\otimes\calE$.
  This shows~\eqref{lemma:lift.channel:swap}.

  \medskip
  
  We show \eqref{lemma:lift.channel:assoc}.
  Let $u$ be the unit reference with domain $\setC$ from \autoref{lemma:quantum.unit}.
  Hence by Law~\ref{law:complement.iso.unit}, $\assoc,u$ are complements.
  We have $\pair\assoc u(x\otimes 1)=\assoc(x)$.
  Since $\pair\assoc u$ and $\assoc$ are both iso-references and thus invertible, this implies $\pair\assoc u^{-1}(y)=\assoc^{-1}(y)\otimes 1_{\setC}$.
  Furthermore, $\assoc(\calE\otimes(\calF\otimes\calG))=\pair\assoc u\circ\pb\paren{(\calE\otimes(\calF\otimes\calG))\otimes\id}\circ\pair\assoc u^{-1}$ by \eqref{lemma:lift.channel:choice.compl}.
  Hence for all trace-class $\rho$, $\tau$, $\gamma$ we have
  \begin{multline*}
    (\rho\otimes\tau)\otimes\gamma
    \xmapsto{\pair\assoc u^{-1}}
    \pb\paren{\rho\otimes(\tau\otimes\gamma)}\otimes 1_{\setC}
    \xmapsto{{\pb\paren{\calE\otimes(\calF\otimes\calG)}\otimes\id}}
    \pb\paren{\calE(\rho)\otimes(\calF(\tau)\otimes\calG(\gamma))}\otimes 1_{\setC} \\
    \xmapsto{\pair\assoc u}
    {(\calE(\rho)\otimes\calF(\tau))\otimes\calG(\gamma)}
    =
    \pb\paren{(\calE\otimes\calF)\otimes\calG}
    \pb\paren{(\rho\otimes\tau)\otimes\gamma}
  \end{multline*}
  (Like in the proof of \eqref{lemma:lift.channel:swap}, we used that $1_{\setC}$ is trace-class.)
  Thus $\assoc(\calE\otimes(\calF\otimes\calG))\pb\paren{(\rho\otimes\tau)\otimes\gamma}=\pb\paren{(\calE\otimes\calF)\otimes\calG}\pb\paren{(\rho\otimes\tau)\otimes\gamma}$ for all $\rho,\tau,\gamma$.
  By \autoref{lemma:channel.equal}, this implies $\assoc\pb\paren{\calE\otimes(\calF\otimes\calG)} = {(\calE\otimes\calF)\otimes\calG}$.
  The $\assocp$ case of \eqref{lemma:lift.channel:assoc} is shown analogously.
  This shows \eqref{lemma:lift.channel:assoc}.

  \medskip

  We show \eqref{lemma:lift.channel:sem}.
  We first show
  \begin{equation}
    \label{eq:lift.channel:sem:subgoal}
    \Sem\pb\paren{F(U)}\circ\pair F{\compl F}
    = F\pb\paren{\Sem(U)}\circ\pair F{\compl F}.
  \end{equation}
  By \autoref{lemma:channel.equal}, it is sufficient to show this equality on all $\tau\otimes\sigma$.
  \begin{align*}
    \hskip1cm&\hskip-1cm\Sem\pb\paren{F(U)}\circ\pair F{\compl F}(\tau\otimes\sigma)
    \starrel=
    \Sem\pb\paren{F(U)}(F(\tau) \compl F(\sigma))
    =
    F(U)F(\tau) \compl F(\sigma)\adj{F(U)}
    \\&=
    F(U)F(\tau) \compl F(\sigma)F(\adj U)
    \starstarrel=
    F(U)F(\tau)F(\adj U) \compl F(\sigma)
    =
    F(U\tau\adj U) \compl F(\sigma)
    \starrel=
    \pair F{\compl F}  \pb\paren{U\tau\adj U\otimes\sigma}    
    \\&=
    \pair F{\compl F} \circ \pb\paren{\Sem(U)\otimes\id}(\tau\otimes\sigma)
    =
    \pair F{\compl F} \circ \pb\paren{\Sem(U)\otimes\id}\circ \pair F{\compl F}^{-1} \circ\pair F{\compl F}(\tau\otimes\sigma)
    \\&=
    F(\Sem(U))\circ\pair F{\compl F}(\tau\otimes\sigma).
  \end{align*}
  Here $(*)$ is by the definition of pairs of references (\autoref{def:pair}).
  And $(**)$ follows because $F,\compl F$ are disjoint.

  This this holds for all $\tau,\sigma$, we conclude \eqref{eq:lift.channel:sem:subgoal}.
  Since $\pair F{\compl F}$ is an iso-register and thus surjective,
  $\Sem(F(U)) = F(\Sem(U))$ and thus \eqref{lemma:lift.channel:sem} follows.
\end{proof}

\paragraph{Kraus-channels, \autoref{lemma:lift.kraus}}
\pagelabel{page:proof:lemma:lift.kraus}

\begin{proof}
  We show \eqref{lemma:lift.kraus:channel}.
  If $\calE$ is a Kraus-(sub)channel, then $\calE(\rho)$ is of the form $\calE(\rho)=\sum_i M_i\rho\adj{M_i}$ with $\sum_i \adj{M_i}M_i \leqq 1$.
  (Where $\leqq$ means $=$ or $\leq$ depending on whether we are talking about Kraus-channels or -subchannels.)
  Then $\regchannel F\calE(\rho) = \sum_i F(M_i)\rho \adj{F(M_i)}$.
  And
  \[
    \sum_i \adj{F(M_i)}F(M_i)
    =
    \sum_i F(M_i\adj{M_i})
    \starrel=
    F\pB\paren{\sum_i M_i\adj{M_i}}
    \starstarrel\leqq
    F(1)
    =
    1
  \]
  Here $(*)$ follows by \autoref{lemma:reference.sum}.
  And $(**)$ follows because $\sum_i M_i\adj{M_i}\leqq1$ and $F$ is positive (\autoref{lemma:reg.props}) and thus monotonous.
  This shows that $\regchannel F\calE$ is a Kraus-(sub)channel.

  \medskip

  We show \eqref{lemma:lift.kraus:samedef}.
  Let $\calE(\rho)=\sum_i M_i\rho\adj{M_i}$ be a Kraus-(sub)channel.
  We write $\regchannel F\calE_1$ for $F(\calE)$ as defined for quantum (sub)channels,
  and $\regchannel F\calE_2$ for $F(\calE)$ as defined for Kraus-(sub)channels.
  Fix trace-class $\sigma,\tau$.
  \begin{align*}
    &\regchannel F\calE_1\circ \pair F{\compl F}(\sigma\otimes\tau)
    \starrel=
      \pair F{\compl F} \circ (\calE\otimes\id) (\sigma\otimes\tau)
      =
      \pair F{\compl F}(\calE(\sigma) \otimes \tau)
      =
      F(\calE(\sigma)) \cdot \compl F(\tau)
    \\&
    =
    F\pB\paren{\sum\nolimits_i M_i\sigma \adj{M_i}} \cdot \compl F(\tau)
    \starstarrel=
    \sum\nolimits_i F\paren{M_i\sigma \adj{M_i}} \cdot \compl F(\tau)
    =
    \sum\nolimits_i F(M_1) F(\sigma) \adj{F(M_i)} \cdot \compl F(\tau)
    \\&
    \tristarrel=
    \sum\nolimits_i F(M_1) F(\sigma) \adj{F(M_i)}\compl F(\tau)
    \fourstarrel=
    \sum\nolimits_i F(M_1) F(\sigma)\compl F(\tau) \adj{F(M_i)}
    =
    \sum\nolimits_i F(M_1) \pair F{\compl F}(\sigma\otimes\tau) \adj{F(M_i)}
    \\&
    =\regchannel F\calE_2\circ \pair F{\compl F}(\sigma\otimes\tau).
  \end{align*}
  Here $(*)$ is by definition of $\regchannel F\calE_1$.
  And $(**)$ follows by \autoref{lemma:reference.sum}.
  And $(*{*}*)$ follows because multiplication from the right is easily seen to be SOT-continuous (and the sum converges in the SOT), so we can pull $\compl F(\tau)$ into the sum.
  And $(*{*}{*}*)$ follows since $F$, $\compl F$ are disjoint and thus $F(\adj{M_i})$ and $\compl F(\tau)$ commute.
  
  Thus the quantum subchannels $\regchannel F\calE_1\circ \pair F{\compl F}$
  and $\regchannel F\calE_2\circ \pair F{\compl F}$ are equal on all $\sigma\otimes\tau$.
  By \autoref{lemma:channel.equal}, this implies $\regchannel F\calE_1\circ \pair F{\compl F} =
  \regchannel F\calE_2\circ \pair F{\compl F}$.
  Since $\pair F{\compl F}$ is an iso-reference and thus invertible, this implies $\regchannel F\calE_1 = \regchannel F\calE_2$.
  This shows \eqref{lemma:lift.kraus:samedef}.
  
  \medskip

  Finally, \eqref{lemma:lift.kraus:indep} follows from \eqref{lemma:lift.kraus:samedef} since the definition of $\regchannel F\calE$ for quantum subchannels does not depend on the choice of~$M_i$.
\end{proof}

\subsection{Partial trace, \autoref{lemma:partial.tr}}
\label{app:proof:lemma:partial.tr}

\begin{proof}
  We show \eqref{lemma:partial.tr:welldef}.
  By \cite[Theorem 18.11\,(c)]{conway00operator}, $\tr:T(\calH_{A'})\to\setC$ is a positive linear functional.
  By \cite[Lemma 34.5]{conway00operator}, positive linear functionals are completely positive.
  Thus $\tr$ is completely positive.
  Note that $\bounded(\setC)=\tracecl(\setC)$ and we can identify both with $\setC$ via the trivial isomorphism $c\in\setC\mapsto c\cdot 1$.
  That is, $\tr:\tracecl(\calH_{A'})\to\tracecl(\setC)$.
  Then $\tr$ is trace-preserving.
  Hence $\tr$ is a quantum channel.
  Then $\id\otimes\tr:\tracecl(\calH_A\otimes\calH_{A'})\to\tracecl(\calH_A\otimes\setC)$ exists and is a quantum channel by \lemmaref{lemma:channel.tensor:tpres}.
  Let $u:\bounded(\setC)\to\mathbf B$, $u(c):=c\cdot 1_\mathbf{B}$.
  Then $u$ is a unit reference by \autoref{lemma:quantum.unit}.
  Thus $\pair \id u$ is an iso-reference and thus a quantum channel (\lemmaref{lemma:isoreg:channel}).
  Hence $\trinFst := \pair \id u\circ(\id\otimes\tr)$ is a quantum channel.
  Furthermore
  \[
    \trinFst(\rho\otimes\sigma)
    =
    \pair \id u\circ(\id\otimes\tr)(\rho\otimes\sigma)
    =
    \pair\id u(\rho\otimes\tr\sigma)
    = \rho \cdot (\tr\sigma \cdot 1_\mathbf{B})
    =
    \rho \cdot \tr\sigma.
  \]
  Hence $\trinFst$ defined this way satisfies the definition of $\trinFst$ given in \autoref{sec:partial.trace}.
  So $\trinFst$ exists (and is a quantum channel).
  Uniqueness of $\trinFst$ follows by \autoref{lemma:channel.equal}.
  This shows that $\trinFst$ is well-defined and thus also $\trin F$.
  This shows \eqref{lemma:partial.tr:welldef}.
  
  \medskip

  We show \eqref{lemma:partial.tr:channel}.
  In the proof of \eqref{lemma:partial.tr:welldef}, we showed that $\trinFst$ is a quantum channel.
  $\pair F{\compl F}^{-1}$ is an iso-reference and thus a quantum channel (\lemmaref{lemma:isoreg:channel}).
  Thus $\trin F = {\trinFst} \circ \pair F{\compl F}^{-1}$ is a quantum channel.
  This shows \eqref{lemma:partial.tr:channel}.
  

  \medskip

  We show \eqref{lemma:partial.tr:choice.compl}.
  Since $F,G$ are complements, by Law~\ref{law:compl.equiv}, $\compl F$ and $G$ are equivalent.
  Thus $G=\compl F\circ I$ for some iso-reference $I$.
  By \lemmaref{lemma:isoreg:channel}, $I$ is a quantum channel.
  Then $ \trinFst\circ(\id\otimes I)$ is a quantum channel as well and for all trace-class $\rho,\sigma$,
  \[
    {\trinFst}\circ(\id\otimes I)(\rho\otimes\sigma)
    =
    \trinFst(\rho\otimes I(\sigma))
    =
    \rho\cdot\tr I(\sigma)
    \starrel=
    \rho\cdot\tr\sigma
    =
    \trinFst(\rho\otimes\sigma).
  \]
  Here $(*)$ uses that $I$ is a quantum channel and thus trace-preserving by definition.
  By \autoref{lemma:channel.equal}, this implies ${\trinFst}\circ(\id\otimes I)=\trinFst$ as quantum channels.
  Then
  \begin{align*}
    \trin F
    &=
    {\trinFst} \circ \pair F{\compl F}^{-1}
    =
    {\trinFst} \circ \pair F{G\circ I^{-1}}^{-1}
    =
      {\trinFst} \circ \pb\paren{\pair F{G} \circ (\id\rtensor I^{-1})}^{-1}
      \\&
    =
    {\trinFst} \circ (\id\rtensor I)\circ {\pair F{G}}^{-1}
    \starrel=
    {\trinFst} \circ (\id\otimes I)\circ {\pair F{G}}^{-1}
    \starstarrel=
    {\trinFst} \circ {\pair F{G}}^{-1}.
  \end{align*}
  (As an equality of functions on the space of trace-class operators.)
  Here $(*)$ uses \lemmaref{lemma:isoreg:channel}.
  And $(**)$ uses that ${\trinFst}\circ(\id\otimes I)=\trinFst$ as quantum channels (shown above) and the fact that ${\pair F{G}}^{-1}$ is a quantum channel (\lemmaref{lemma:isoreg:channel}).
  This shows \eqref{lemma:partial.tr:choice.compl}.
  
  \medskip

  We show \eqref{lemma:partial.tr:iso}.
  Let $u:\bounded(\setC)\to\mathbf B$, $u(c):=c\cdot 1_\mathbf{B}$.
  Then $u$ is a unit reference (\autoref{lemma:quantum.unit}).
  Since $F$ is an iso-reference, $F,u$ are complements (Law~\ref{law:complement.iso.unit}).
  For trace-class $\rho$,
  \[
    {\trin F}\circ F|_{\tracecl(\calH_A)}(\rho)
    = {\trin F}\circ\pair Fu(\rho\otimes 1_{\setC})
    \eqrefrel{lemma:partial.tr:choice.compl}=
    {\trinFst}\circ\pair Fu^{-1}\circ\pair Fu(\rho\otimes 1_{\setC})
    =
    {\trinFst}(\rho\otimes 1_{\setC})
    =\rho.
  \]
  (Note that $1_{\setC}$ is trace-class, and iso-references are quantum channels by \lemmaref{lemma:isoreg:channel}, so we do not apply $\trin F,\trinFst$ to anything outside their domain.)
  Thus ${\trin F}\circ F|_{\tracecl(\calH_A)}=\id$.
  Since $F$ is an iso-reference, $F|_{\tracecl(\calH_A)}$ is a bijection $\tracecl(\calH_A)\to\tracecl(\calH_A)$ by \lemmaref{lemma:isoreg:channel}.
  Thus ${\trin F}\circ F|_{\tracecl(\calH_A)}=\id$ implies $\trin F=F^{-1}|_{\tracecl(\calH_A)}$.
  This shows \eqref{lemma:partial.tr:iso}.

  \medskip
    
  We show \eqref{lemma:partial.tr:pair}.
  We have
  \begin{equation*}
    \trin F \pair{F}{G}(\rho\otimes\sigma)
    \eqrefrel{lemma:partial.tr:choice.compl}=
    {\trinFst} \circ \pair{F}G^{-1} \circ \pair{F}G(\rho\otimes\sigma)
    \starrel=
    \trinFst(\rho\otimes\sigma)
    \starstarrel=
    \rho\cdot\tr\sigma.
  \end{equation*}
  Here $(*)$ uses that $\pair{F}G$ is an iso-reference and thus bijective.
  And $(**)$ is by definition of $\trinFst$.
  This shows \eqref{lemma:partial.tr:pair}.

  \medskip

  We show \eqref{lemma:partial.tr:partition}.
  By \lemmaref{lemma:mixed.states:permute}, the order of the ${F_1(\rho_1),\dots, F_n(\rho_n)}$ does not matter.
  Thus without loss of generality, we can assume $i=1$.
  Let $F^*:=\pairs{F_2}{F_n}$ and $\rho^*:=\rho_2\otimes\dots\otimes\rho_n$.
  (If $n=1$, we let $F^*$ be the unit reference from \autoref{lemma:quantum.unit} and $\rho^*:=1_{\setC}$.)
  Since $F_1,\dots,F_n$ is a partition, $F_1,F^*$ are complements.
  Then
  \begin{multline*}
    \trin {F_1} \pb\paren{F_1(\rho_1)\mtensor\dots\mtensor F_n(\rho_n)}
    \starrel=
    \trin {F_1} \pb\paren{\pairs{F_1}{F_n}(\rho_1\otimes\dots\otimes\rho_n)}
    \\
    =
    \trin {F_1}\pb\paren{ \pair{F_1}{F^*}(\rho_1\otimes\rho^*)}
    \eqrefrel{lemma:partial.tr:pair}=
    \rho_1 \otimes \tr\rho^*
    \starstarrel=
    \rho_i \cdot \prod\nolimits_{j=2}^n\tr\rho_i.
  \end{multline*}
  Here $(*)$ is by definition of $\mtensor$.
  And $(**)$ is by \autoref{lemma:traceclass.tensor}.
  (Recall that $\otimes$ and $\pair\cdot\cdot$ are right-associative.)

  This shows \eqref{lemma:partial.tr:partition}.

  \medskip
  
  We show \eqref{lemma:partial.tr:chain}.
  For any trace-class $\rho,\sigma,\tau$, we have
  \begin{multline}
    \trin{\chain FG}\pb\pair{\chain FG}{\pair{\compl F}{\chain F{\compl G}}} \pb\paren{\rho\otimes(\tau\otimes\sigma)}
    \starrel=
    {\trinFst}\circ \pb\pair{\chain FG}{\pair{\compl F}{\chain F{\compl G}}}^{-1} \circ \pb\pair{\chain FG}{\pair{\compl F}{\chain F{\compl G}}} \pb\paren{\rho\otimes(\tau\otimes\sigma)}
    \\=
    \trinFst \pb\paren{\rho\otimes(\tau\otimes\sigma)}
    = \rho \cdot \tr(\tau\otimes\sigma)
    \starstarrel= \rho \cdot \tr\sigma \cdot \tr\tau.
    \label{eq:FG.tr.tr}
  \end{multline}
  Here $(*)$ is by \eqref{lemma:partial.tr:choice.compl} using the fact that ${\pair{\compl F}{\chain F{\compl G}}}$ is a complement of $\chain FG$ by Law~\ref{law:complement.chain}.
  And $(**)$ is by \autoref{lemma:traceclass.tensor}.
  And also
  \begin{align*}
    \hskip1cm & \hskip-1cm
      \pb\pair{\chain FG}{\pair{\compl F}{\chain F{\compl G}}} \pb\paren{\rho\otimes(\tau\otimes\sigma)}
    \starrel=
    \pair{F}{\compl F} \circ
    \pb\paren{\pair G{\compl G} \rtensor \id}
    \pb\paren{(\rho\otimes\sigma)\otimes\tau}
    \\&
    \xmapsto{    \pair{F}{\compl F}^{-1}  }
    \pb\paren{\pair G{\compl G} \rtensor \id}
    \pb\paren{(\rho\otimes\sigma)\otimes\tau}
    = \pair G{\compl G} (\rho\otimes\sigma)   \otimes \tau
    \xmapsto{  \trinFst  }
    \pair G{\compl G} (\rho\otimes\sigma)   \cdot \tr \tau
    \\
    &\xmapsto{  \pair G{\compl G}^{-1}  }
    \rho\otimes\sigma   \cdot \tr \tau
    \xmapsto{  \trinFst  }
    \rho \cdot \tr \sigma   \cdot \tr \tau.
  \end{align*}
  Here $(*)$ is easily seen by evaluating both sides.
  Since $\trin F = {\trinFst} \circ \pair{F}{\compl F}^{-1} $ and $\trin G = {\trinFst} \circ \pair G{\compl G}^{-1} $, the $\xmapsto{\dots}$ correspond to applying $\trin G\circ\trin F$.
  Thus, with \eqref{eq:FG.tr.tr},
  \[
    \trin G\circ \trin F \circ \pb\pair{\chain FG}{\pair{\compl F}{\chain F{\compl G}}} \pb\paren{\rho\otimes(\tau\otimes\sigma)}
    = 
    \trin{\chain FG}\circ \pb\pair{\chain FG}{\pair{\compl F}{\chain F{\compl G}}} \pb\paren{\rho\otimes(\tau\otimes\sigma)}.
  \]
  Since $\trin F,\trin G,\trin {\chain FG}$ are quantum channels by \eqref{lemma:partial.tr:channel},
  and $\pb\pair{\chain FG}{\pair{\compl F}{\chain F{\compl G}}}$ is an iso-reference and thus a quantum channel by \lemmaref{lemma:isoreg:channel},
  the function compositions on the left and right hand side are quantum channels that are equal on all $\rho\otimes(\tau\otimes\sigma)$.
  Thus by \autoref{lemma:channel.equal}, $\trin G\circ \trin F \circ \pb\pair{\chain FG}{\pair{\compl F}{\chain F{\compl G}}} = \trin{\chain FG}\circ \pb\pair{\chain FG}{\pair{\compl F}{\chain F{\compl G}}}$.
  Since $\pb\pair{\chain FG}{\pair{\compl F}{\chain F{\compl G}}}$ is an iso-reference, this implies $\trin G\circ \trin F = \trin{\chain FG}$.
  This shows \eqref{lemma:partial.tr:chain}.

  \medskip
  
  We show \eqref{lemma:partial.tr:channel.only}.
  For trace-class $\rho,\sigma$, we have
  \begin{multline*}
    {\trin F} \circ \regchannel F\calE \circ \pair F{\compl F}(\rho\otimes\sigma)
    \starrel=
    {\trinFst} \circ \pair F{\compl F}^{-1} \circ  \pair F{\compl F}
    \circ (\calE \otimes \id)
    \circ  \pair F{\compl F}^{-1}
    \circ \pair F{\compl F}
    (\rho\otimes\sigma)
    \\
    =
    {\trinFst}
    \circ (\calE \otimes \id)
    (\rho\otimes\sigma)
    =
    \trinFst (\calE(\rho)\otimes \sigma)
    =
    \calE(\rho) \cdot \tr\sigma
    = \calE(\rho \cdot \tr\sigma)
    \eqrefrel{lemma:partial.tr:pair}=
    \calE \circ {\trin F} \circ \pair F{\compl F} (\rho\otimes\sigma).
  \end{multline*}
  Here $(*)$ is by definition of $\trin F$ and of $\regchannel F\calE$.
  Since $\trin F$, $\calE$, $\regchannel F\calE$, and $\pair F{\compl F}$ are quantum subchannels, this implies that ${\trin F} \circ \regchannel F\calE \circ \pair F{\compl F} = \calE \circ \trin F \circ \pair F{\compl F}$ by \autoref{lemma:channel.equal}.
  Since $\pair F{\compl F}$ is an iso-reference, this implies ${\trin F} \circ \regchannel F\calE = \calE \circ \trin F$.
  This shows \eqref{lemma:partial.tr:channel.only}.

  \medskip
  
  We show \eqref{lemma:partial.tr:channel.other}.
  By Law~\ref{law:complement.pair}, $F$ and $\pair{G}{\compl{\spair FG}}$ are complements.
  Let $Z := \pair F{\pair{G}{\compl{\spair FG}}}$.
  $Z$ is an iso-reference because $F$ and $\pair{G}{\compl{\spair FG}}$ are complements.
  By \lemmaref{lemma:lift.channel:pair}, \eqref{lemma:lift.channel:id}, $\regchannel Z{\id\otimes(\calE\otimes\id)} = \regchannel G\calE$.
  Thus for all trace-class $\rho,\sigma$, we have
  \begin{multline*}
    {\trin F} \circ \regchannel G\calE \circ Z (\rho\otimes\sigma)
    =
    {\trin F} \circ \regchannel Z{\id\otimes(\calE\otimes\id)} \circ Z (\rho\otimes\sigma)
    \starrel= 
    {\trin F} \circ Z\pb\paren{\pb\paren{\id\otimes(\calE\otimes\id)}\pb\paren{\rho\otimes\sigma}}
    \\
    =
    {\trin F} \circ Z\pb\paren{\rho\otimes (\calE\otimes\id)(\sigma)}
    \eqrefrel{lemma:partial.tr:pair}=
    \rho\cdot \tr (\calE\otimes\id)(\sigma)
    \starstarrel=
    \rho\cdot \tr \sigma
    \eqrefrel{lemma:partial.tr:pair}=
    {\trin F} \circ Z (\rho\otimes\sigma).
  \end{multline*}
  Here $(*)$ is by \lemmaref{lemma:lift.channel:apply} (in the special case where $Z$ is a one-element partition).
  And $(**)$ follows since $\calE\otimes\id$ is a quantum channel by \lemmaref{lemma:channel.tensor:tpres} and thus trace-preserving.
  Since $Z$ is an iso-reference, it is a quantum channel (\lemmaref{lemma:isoreg:channel}).
  Thus the function compositions on the left and right hand side are quantum channels.
  And equal on all $\rho\otimes\sigma$.
  By \autoref{lemma:channel.equal}, this implies ${\trin F} \circ \regchannel G\calE \circ Z = {\trin F} \circ Z$ as quantum channels.
  Since $Z$ is and iso-reference and thus bijective, this implies ${\trin F} \circ \regchannel G\calE = \trin F$.
  This shows \eqref{lemma:partial.tr:channel.other}.
\end{proof}

\subsection{Measurements, Lemmas~\ref{lemma:proj.meas}, \ref{lemma:povm}, and \ref{lemma:general.meas}} 
\label{app:proof:measurements}

\paragraph{Projective measurements, \autoref{lemma:proj.meas}}

\begin{proof}
  We show \eqref{lemma:proj.meas:is.meas}.
  Since $M=\braces{P_i}_i$ is a projective measurement, $P_i$ are projectors, $P_iP_j=0$ ($i\neq j$), and $\sum_i P_i=1$.
  By definition, $\regmeas FM=\braces{F(P_i)}_i$.
  By \lemmaref{lemma:preserve.projector}, $F(P_i)$ is a projector.
  For $i\neq j$, we have $F(P_i)F(P_j)=F(P_iP_j)=F(0)=0$.
  And
  \begin{equation*}
    \sum_i F(P_i) \starrel= F\pB\paren{\sum_i P_i} = F(1) = 1.
  \end{equation*}
  Here $(*)$ is by \autoref{lemma:reference.sum}.
  This shows \eqref{lemma:proj.meas:is.meas}.

  \medskip

  We show \eqref{lemma:proj.meas:uni}.
  We have
  \begin{align*}
    F(U)\,\regmeas FM
    &=
      F(U)\pb\braces{F(P_i)}_i
    =
      \pb\braces{F(U)F(P_i)\adj{F(U)}}_i
    =
      \pb\braces{F(UP_i\adj U)}_i
    =
      \pb\regmeas F{\braces{UP_i\adj U}_i}
    =
      \regmeas F{UM}
  \end{align*}
  and
  \begin{equation*}
    \regmeas F{UM}
    =
    \pb\braces{F(UP_i\adj U)}_i
    =
    \pb\braces{F(\sandwich U(P_i))}_i
    =
    \pb\braces{\chain F{\sandwich U}(P_i)}_i
    =
    \regmeas{\paren{\chain F{\sandwich U}}}M.
  \end{equation*}
  This shows \eqref{lemma:proj.meas:uni}.

  \medskip

  We show \eqref{lemma:proj.meas:pure}.
  By definition $\regmeas FM=\{F(P_i)\}_i$.
  By definition of non-normalized post-measurement states, we have $\psi_1'=P_j\psi_1$ and $\psi'=F(P_1)\psi$.
  Then
  \begin{equation*}
    \psi'
    =
    F(P_1)\psi
    \starrel=
    F(P_1\psi_1)\ptensor F(\psi_2)\ptensor\dots\ptensor F(\psi_n)
    =
    F(\psi_1')\ptensor F(\psi_2)\ptensor\dots\ptensor F(\psi_n).
  \end{equation*}
  Here $(*)$ is by \lemmaref{lemma:lift.pure:U} and the definition of $\psi$.
  This shows \eqref{lemma:proj.meas:pure}.

  \medskip

  We show \eqref{lemma:proj.meas:mixed}.
  By definition $\regmeas FM=\{F(P_i)\}_i$.
  By definition of non-normalized post-measurement states, we have $\rho_1'=P_j\rho_1\adj{P_j}$ and $\rho'=F(P_1)\rho\adj{F(P_j)}$.
  Then
  \begin{align*}
    \rho'
    &=
    F(P_1)\rho\adj{F(P_1)}
    =
    F(P_1)\rho F(\adj{P_1})
    \starrel=
      F(P_1\rho_1\adj{P_1})\mtensor F(\rho_2)\mtensor\dots\mtensor F(\rho_n)
      \\&
    =
    F(\rho_1')\mtensor F(\rho_2)\mtensor\dots\mtensor F(\rho_n).
  \end{align*}
  Here $(*)$ is by \lemmaref{lemma:mixed.states:apply.UV} and the definition of $\rho$.
  This shows \eqref{lemma:proj.meas:mixed}.
\end{proof}

\paragraph{POVMs, \autoref{lemma:povm}}
\pagelabel{page:proof:povm}

\begin{proof}
  We show \eqref{lemma:povm:is.povm}.
  Since $M$ is a POVM, $M_i$ are positive operators, and $\sum_i M_i=1$.
  By definition, $\regmeas FM=\braces{F(M_i)}_i$.
  Since $F$ is positive (\autoref{lemma:reg.props}), $F(P_i)$ is a positive operator.
  And
  \begin{equation*}
    \sum_i F(M_i) \starrel= F\pB\paren{\sum_i M_i} = F(1) = 1.
  \end{equation*}
  Here $(*)$ is by \autoref{lemma:reference.sum}.
  This shows \eqref{lemma:povm:is.povm}.

  \medskip

  We show \eqref{lemma:povm:uni}.
  We have
  \begin{align*}
    F(U)\,\regmeas FM
    &=
      F(U)\pb\braces{F(M_i)}_i
    =
      \pb\braces{F(U)F(P_i)\adj{F(U)}}_i
    =
      \pb\braces{F(UM_i\adj U)}_i
    =
      \pb\regmeas F{\braces{UM_i\adj U}_i}
    =
      \regmeas F{UM}
  \end{align*}
  and
  \begin{equation*}
    \regmeas F{UM}
    =
    \pb\braces{F(UM_i\adj U)}_i
    =
    \pb\braces{F(\sandwich U(M_i))}_i
    =
    \pb\braces{\chain F{\sandwich U}(M_i)}_i
    =
    \regmeas{\paren{\chain F{\sandwich U}}}M.
  \end{equation*}
  This shows \eqref{lemma:povm:uni}.

  \medskip

  We show \eqref{lemma:povm:pure}.
  By \lemmaref{lemma:lift.pure:permute}, without loss of generality, we can assume $k=1$.
  By definition, the probability of outcome $j$ using $M$ given $\psi_1$ is $\norm{\sqrt{M_j}\psi_1}^2$.
  And the probability of outcome $j$ using $\regmeas FM$ given $\psi$ is $\norm{\sqrt{F(M_j)}\psi}^2$.
  So we have to show $\pb\norm{\sqrt{M_j}\psi_1} = \pb\norm{\sqrt{F(M_j)}\psi}$.
  Note that $F(\sqrt{M_j})$ is positive since $F$ is positive (\autoref{lemma:reg.props}), and $\adj{F\pb\paren{\sqrt{M_j}}}F\pb\paren{\sqrt{M_j}} = F\pb\paren{\adj{\sqrt{M_j}}\sqrt{M_j}} = F(M_j)$.
  Thus $F(\sqrt{M_j})=\sqrt{F(M_j)}$.
  We have
  \begin{multline*}
    \pb\norm{\textstyle\sqrt{F(M_j)}\psi}
    =
    \pb\norm{{F(\sqrt{M_j})}\psi}
    \starrel=
    \pb\norm{F(\sqrt{M_j}\psi_1)\ptensor F(\psi_2)\ptensor\dots\ptensor F(\psi_n)}
    \\\starstarrel=
    \norm{\sqrt{M_j}\psi_1}
    \cdot \underbrace{\norm{\psi_2}\cdots\norm{\psi_n}}_{{}=1}
    =
    \norm{\sqrt{M_j}\psi_1}.
  \end{multline*}
  Here $(*)$ is by \lemmaref{lemma:lift.pure:U} and the definition of $\psi$.
  And $(**)$ is by \lemmaref{lemma:lift.pure:norm}.
  This shows \eqref{lemma:povm:pure}.
  
  \medskip
  
  We show \eqref{lemma:povm:mixed}.
  By \lemmaref{lemma:mixed.states:permute}, without loss of generality, we can assume $k=1$.
  By definition, the probability of outcome $j$ using $M$ given $\rho_1$ is $\tr{M_j\rho_1}$.
  And the probability of outcome $j$ using $\regmeas FM$ given $\rho$ is $\tr F(M_j)\rho$.
  So we have to show $\tr{M_j\rho_1} = \tr F(M_j)\rho$.
  We have
  \begin{equation*}
    \tr F(M_j)\rho
    \starrel=
    \tr\pB\paren{F(M_j\rho_1)\mtensor F(\rho_2)\mtensor\dots\mtensor F(\rho_n)}
    \starstarrel=
    \tr{M_j\rho_1}
    \cdot \underbrace{\tr\rho_2\cdots\tr\rho_n}_{{}=1}
    =
    \tr{M_j\rho_1}.
  \end{equation*}
  Here $(*)$ is by \lemmaref{lemma:mixed.states:apply.UV} and the definition of $\rho$.
  And $(**)$ is by \lemmaref{lemma:mixed.states:trace}.
  This shows \eqref{lemma:povm:mixed}.

  \medskip

  We show \eqref{lemma:povm:projmeas}.
  By definition, the outcome probability in the POVM case is $\norm{\sqrt{F(M_j)}\psi}^2$ or $\tr F(M_j)\rho$, respectively.
  And the probability in the projective measurement case is $\norm{F(M_j)\psi}^2$ or $\tr F(M_j)\rho\adj{F(M_j)}$, respectively.
  Since $M_j$ is a projector, so is $F(M_j)$ by \lemmaref{lemma:preserve.projector}.
  Thus $F(M_j)=\adj{F(M_j)}=\sqrt{F(M_j)} = \adj{F(M_j)}F(M_j)$.
  Thus $\norm{\sqrt{F(M_j)}\psi}^2 = \norm{{F(M_j)}\psi}^2$ (i.e., the outcome probabilities for pure states are equal).
  And $\tr F(M_j)\rho = \tr \adj{F(M_j)}F(M_j)\rho \starrel= \tr F(M_j)\rho\adj{F(M_j)}$ (i.e., the outcome probabilities for mixed states are equal).
  Here $(*)$ uses the circularity of the trace (\lemmaref{lemma:circ.trace:circ}).
  %
\end{proof}

\paragraph{General measurements, \autoref{lemma:general.meas}}
\pagelabel{page:proof:general.meas}

\begin{proof}
  We show \eqref{lemma:general.meas:is.meas}.
  Since $M$ is a general measurement, $\sum_i \adj{M_i}M_i=1$.
  By definition, $\regmeas FM=\braces{F(M_i)}_i$.
  We have
  \begin{equation*}
    \sum_i \adj{F(M_i)}F(M_i)
    =
    \sum_i F(\adj{M_i}M_i)
    \starrel=
    F\pB\paren{\sum_i \adj{M_i}M_i} = F(1) = 1.
  \end{equation*}
  Here $(*)$ is by \autoref{lemma:reference.sum}.
  This shows \eqref{lemma:general.meas:is.meas}.

  \medskip

  We show \eqref{lemma:general.meas:pure}.
  By definition $\regmeas FM=\{F(M_i)\}_i$.
  By definition of non-normalized post-measurement states, we have $\psi_1'=M_j\psi_1$ and $\psi'=F(M_1)\psi$.
  Then
  \begin{equation*}
    \psi'
    =
    F(M_1)\psi
    \starrel=
    F(M_1\psi_1)\ptensor F(\psi_2)\ptensor\dots\ptensor F(\psi_n)
    =
    F(\psi_1')\ptensor F(\psi_2)\ptensor\dots\ptensor F(\psi_n).
  \end{equation*}
  Here $(*)$ is by \lemmaref{lemma:lift.pure:U} and the definition of $\psi$.
  This shows \eqref{lemma:general.meas:pure}.

  \medskip

  We show \eqref{lemma:general.meas:mixed}.
  By definition $\regmeas FM=\{F(M_i)\}_i$.
  By definition of non-normalized post-measurement states, we have $\rho_1'=M_j\rho_1\adj{M_j}$ and $\rho'=F(M_1)\rho\adj{F(M_j)}$.
  Then
  \begin{align*}
    \rho'
    &=
    F(M_1)\rho\adj{F(M_1)}
    =
    F(M_1)\rho F(\adj{M_1})
    \starrel=
      F(M_1\rho_1\adj{M_1})\mtensor F(\rho_2)\mtensor\dots\mtensor F(\rho_n)
      \\&
    =
    F(\rho_1')\mtensor F(\rho_2)\mtensor\dots\mtensor F(\rho_n).
  \end{align*}
  Here $(*)$ is by \lemmaref{lemma:mixed.states:apply.UV} and the definition of $\rho$.
  This shows \eqref{lemma:general.meas:mixed}.

  \medskip

  We show \eqref{lemma:general.meas:povm}.
  By definition of the ``corresponding POVM'', $M'=\braces{\adj{M_i}M_i}_i$.
  And $F(M)=\braces{F(M_i)}_i$ and $F(M')=\braces{F(\adj{M_i}M_i)}$.
  Thus $F(M')=\braces{F(\adj{M_i}M_i)}=\braces{\adj{F(M_i)}F(M_i)}$ which is the POVM corresponding to $F(M)$.
  This shows \eqref{lemma:general.meas:povm}.
\end{proof}


}

  
  


\fullonly{
\renewcommand\symbolindexentry[4]{
  \noindent\hbox{\hbox to 1.3in{$#2$\hfill}\parbox[t]{4.2in}{\raggedright#3}\hbox to 1cm{\hfill #4}}\\[2pt]}

  {\sectionToToc
   \printsymbolindex}

  \renewenvironment{theindex}{%
    \section*{Keyword Index}%
    \parindent\z@
    \parskip\z@ \@plus .3\p@\relax
    \let\item\@idxitem
    \begin{multicols}{2}}{\end{multicols}}
  {\sectionToToc
    \printindex}

}




\fullonly{
  {\sectionToToc  \printbibliography}
}

\shortonly{
  \delaytextwarnonly
}

\end{document}